\documentclass{article}
\usepackage{amsmath}
\usepackage[dvipsnames]{xcolor}
\usepackage{graphicx}
\usepackage{enumerate}
\usepackage{natbib}
\usepackage{amsthm}
\newtheoremstyle{mainspacedplain}
    {\dimexpr\topsep+2pt\relax}  
    {\topsep}                    
    {\itshape}                   
    {}                           
    {\bfseries}                  
    {.}                          
    {.5em}                       
    {}
\theoremstyle{mainspacedplain}
\newtheorem{assumption}{Assumption}
\newtheorem{theorem}{Theorem}
\usepackage[left=1in,top=1in,right=1in,bottom=1.in,head=.2in]{geometry}
\definecolor{undergreen}{HTML}{7BC96F}
\definecolor{overgreen}{HTML}{1B7837}
\usepackage[utf8]{inputenc}
\usepackage{algorithm}
\usepackage[noend]{algpseudocode}
\usepackage[bookmarksopen=true,bookmarks=true,pdfencoding=auto,psdextra,colorlinks]{hyperref}
\hypersetup{ 
    colorlinks,
    linkcolor={blue!80!black},
    citecolor={green!40!black},
    urlcolor={red!50!black}
}
\usepackage{subcaption}
\usepackage{bookmark}

\usepackage{mystyle}

\graphicspath{{../}{./}}
\algrenewcommand\algorithmicrequire{\textbf{Input:}}
\algrenewcommand\algorithmicensure{\textbf{Output:}}
\Crefname{algorithm}{Algorithm}{Algorithms}
\crefname{algorithm}{Algorithm}{Algorithms}
\makeatletter
\def\theHALG@line{\thealgorithm.\arabic{ALG@line}}
\makeatother
\usepackage{bibunits}
\usepackage{booktabs}
\usepackage{threeparttable}
\allowdisplaybreaks
\usepackage{xspace}

\usepackage[svgnames,dvipsnames,table]{xcolor}

\usepackage[font=small,labelfont=bf]{caption}
\usepackage{tabularx}
\usepackage{multirow}
\usepackage{array}
\newcolumntype{Y}{>{\centering\arraybackslash}X}
\newcolumntype{P}{>{\raggedleft\arraybackslash}X}
\newcolumntype{C}{ >{\centering\arraybackslash} m{4cm} }
\newcolumntype{D}{>{\raggedright\arraybackslash}X}

\usepackage[table]{xcolor}

\definecolor{coverblue}{RGB}{221,235,247}
\definecolor{missred}{RGB}{252,225,225}
\definecolor{Gray}{RGB}{227,227,227}
\usepackage{parskip}

\makeatletter
\newcounter{affiliation}
\def\aff@texts{}
\newcommand{\affiliation}[2]{%
    \stepcounter{affiliation}%
    \expandafter\xdef\csname aff@num@#1\endcsname{\the\value{affiliation}}%
    \begingroup\edef\x{\endgroup\noexpand\g@addto@macro\noexpand\aff@texts{\noexpand\footnotetext[\the\value{affiliation}]{\unexpanded{#2}}}}\x
}
\newcommand{\affnum}[1]{\csname aff@num@#1\endcsname}
\newcommand{\titlefootnotes}{\aff@texts}
\newcommand{\notetext}[2]{{\renewcommand{\thefootnote}{#1}\footnotetext{#2}}}
\makeatother

\affiliation{cufe}{School of Statistics and Mathematics, Central University of Finance and Economics, Beijing, China.}
\affiliation{hkustat}{Department of Statistics and Actuarial Science, The University of Hong Kong, Pokfulam, Hong Kong SAR, China.}
\affiliation{ids}{Institute of Data Science, The University of Hong Kong, Pokfulam, Hong Kong SAR, China.}

\title{Risk-Calibrated Balancing for \\High-Dimensional Causal Extrapolation}
\author{%
    Fenglin Yang\textsuperscript{\affnum{cufe},\textdagger}
    \and Haoran Lei\textsuperscript{\affnum{hkustat},\textdagger}
    \and Yan Chen\textsuperscript{\affnum{hkustat}}
    \and Jin-Hong Du\textsuperscript{\affnum{hkustat},\affnum{ids},\textasteriskcentered}%
}
\date{\today}

\begin{document}

\maketitle
\titlefootnotes
\notetext{\textdagger}{Equal contribution.}
\notetext{\textasteriskcentered}{Corresponding author: \href{mailto:jinhongd@hku.hk}{jinhongd@hku.hk}.}

\begin{abstract}
In observational causal inference, covariate balancing is widely used to reduce source-target covariate shift, but under weak overlap in high dimensions, stronger balance can induce concentrated weights and increase variance.
Balance measures how well the target covariate distribution is represented, but does not by itself determine how reliably the counterfactual mean can be estimated.
We develop risk-calibrated balancing for the average treatment effect on the treated, which applies ridge augmentation to any normalised base weights and selects its penalty using conditional prediction risk of the counterfactual mean.
Under a random-effects predictive model, we derive an exact finite-sample decomposition of this risk into residual covariate imbalance and weight-induced variance.
For design-independent base weights under proportional asymptotics, we characterise how limiting risk depends on source and target covariance geometry, population mean shift, and weight concentration.
For covariate-adaptive base weights, we develop a uniformly consistent target-aware risk estimator whose minimiser attains vanishing scaled oracle excess risk.
Simulations show that the high-dimensional risk predictions remain informative for adaptive balancing and that target-aware tuning generally reduces excess target risk.
Empirical analyses of job-training and single-cell perturbation data show that risk-calibrated balancing improves on most base estimators, with larger gains under weaker overlap.
\end{abstract}


\section{Introduction}\label{sec:intro}
In observational studies, data are often collected from one population while the scientific question concerns another.
Clinical-trial evidence may need to be generalised to a broader patient population \citep{cole2010generalizing,stuart2011use,dahabreh2019extend,dahabreh2021studydesigns}; nonrepresentative samples are calibrated to population targets in survey sampling \citep{deville1992calibration}; and high-throughput genomic studies seek to compare molecular responses across heterogeneous cellular populations \citep{Du2025Genomic,replogle2022perturbseq}.
Although these applications differ substantively, they share a source-to-target structure, where outcomes observed under one covariate distribution must inform a target population that may follow a different covariate distribution.
When the target is poorly represented by the source, identification alone does not ensure accurate estimation: statistical accuracy also depends on the extent and cost of the required extrapolation.

This source-to-target problem arises naturally in observational causal inference.
For the average treatment effect on the treated (ATT), the treated population is the target and the controls provide the source outcomes.
The counterfactual mean \(\mu_0=\EE\{y(0)\mid a=1\}\) requires transporting the control response surface from the control covariate distribution to the treated covariate distribution.
When the usual causal identification assumptions hold, the ATT can remain identified under weak overlap, while estimating \(\mu_0\) can still be unstable because the source-to-target density ratio may be highly variable, leading to irregular estimation and extreme weights \citep{crump2009overlap,khan2010irregular}.
In high dimensions, simultaneous overlap across many covariates can deteriorate rapidly, making accurate estimation still more difficult \citep{damour2021overlap}.
Covariate balancing mitigates this difficulty by reweighting the source toward the target distribution \citep{hainmueller2012entropy,zubizarreta2015stable,wang2020minimal,li2018balancing}, but stronger balance can require highly concentrated weights, while more stable weights leave greater residual imbalance and reliance on outcome-model extrapolation.

When the covariate dimension is comparable to the sample size, this tradeoff raises two statistical challenges.
\emph{First}, balance measures residual covariate imbalance, while counterfactual prediction error also depends on the noise induced by concentrated weights.
Thus stronger balance need not yield lower conditional prediction risk for the treated.
Moreover, the empirical target mean is itself noisy, and this sampling variation can remain non-negligible when the covariate dimension is comparable to the target sample size.
\emph{Second}, once target prediction risk is taken as the tuning criterion, it must itself be estimated and minimised from data.
Existing augmented and minimax procedures regularise through covariate balance, weight dispersion, or worst-case approximation-variance criteria \citep{athey2018arb,hirshberg2021aml,arbour2025regularizing}, rather than directly tuning target prediction risk.
Under covariate shift, a penalty that performs well for source prediction need not minimise out-of-distribution risk \citep{dobriban2018ridge,patil2024optimal}.
Moreover, when balancing weights are estimated using target covariates, reusing those observations for risk evaluation induces dependence that must be accounted for.

To address these challenges, we propose risk-calibrated balancing (RCB), which applies ridge augmentation to any normalised base weights and selects the penalty using conditional prediction risk of the counterfactual mean for the treated.

\subsection{Contribution}

This paper makes four main contributions.

\begin{enumerate}[(a)]

    \item \textbf{General ridge augmentation and conditional prediction risk.}
    RCB can be applied to any normalised base weights, allowing the initial weighting rule to be chosen separately from risk-based regularisation.
    For base weights constructed without the control outcomes, \Cref{prop:conditional-risk} gives an exact finite-sample decomposition of the conditional random-effects prediction risk into residual covariate imbalance and weight-induced noise on a common prediction-risk scale.

    \item \textbf{High-dimensional risk asymptotics and extrapolation regimes.}
    For design-independent base weights under proportional asymptotics, \Cref{thm:risk-asym} gives uniform deterministic equivalents for the conditional prediction risk, separating the contributions of source geometry, source-target mean shift, weight concentration, and target sampling.
    Here, \(\eta\in[0,1]\) indexes the coupled scaling of source-target mean-shift energy and weight concentration, with smaller \(\eta\) corresponding to more severe extrapolation.
    The resulting regimes are distinct: at \(\eta=0\), risk can remain nonvanishing, whereas for \(\eta>0\) it decays at the \(n_0^{-\eta}\) scale when the limiting risk is positive (\Cref{cor:fixed-risk-limit}).
    In the boundary case \(\eta=1\), source and target empirical-mean effects and affine normalisation enter the leading risk; for \(\eta<1\), these terms are lower order and the leading risk is invariant to the target aspect ratio and target covariance spectrum (\Cref{cor:target-sampling-invariant}).
    \Cref{prop:overlap-weight-concentration} further connects the weight-concentration scaling to density-ratio overlap through the source-to-target \(\chi^2\) divergence.

    \item \textbf{Feasible target-aware risk estimation and tuning.}
    For data-adaptive base weights, we construct a feasible criterion for the same conditional prediction risk using an independent target evaluation sample, without estimating the full source-to-target density ratio.
    A spectral quasi-likelihood estimator consistently estimates the required signal and residual variance components (\Cref{prop:spectral-quasi-likelihood}).
    Under evaluation-sample separation and base-weight rate conditions, the estimated risk is uniformly consistent along the regularisation path and attains vanishing scaled oracle excess risk; penalty consistency additionally follows when the scaled risk path has a unique deterministic minimiser (\Cref{thm:risk-estimation}).
    Consistent risk estimation further yields model-based prediction intervals for the counterfactual mean under the Gaussian predictive law (\Cref{thm:predictive-calibration}).

    \item \textbf{Empirical implications and applications.}
    Simulations show that the high-dimensional predictions remain informative for data-adaptive balancing procedures and that target-aware tuning generally tracks the oracle penalty for the conditional prediction risk more closely than compared methods.
    In the LaLonde analysis \citep{lalonde1986evaluating}, risk-calibrated augmentation is stable across different base weights; under larger source-target discrepancy, it reduces the RMSE of \(\ell_2\) balancing weights and improves on the best double-ridge estimator.
    In the K562 Perturb-seq analysis \citep{replogle2022perturbseq}, where 48 control profiles represent targets in 500 covariate dimensions, standard balancing weights concentrate on few controls while barely reducing population imbalance; RCB reduces the prediction error of every base in most perturbation-outcome pairs under weak overlap and adjusts only for perturbations with a detectable covariate shift.
\end{enumerate}

\subsection{Related work}

Our work is rooted in the literature on weighting, covariate balance, and calibration for causal inference.
Propensity-score methods use treatment assignment probabilities to construct comparable treatment groups \citep{rosenbaum1983central}, whereas a broad class of more recent methods targets covariate balance directly, including entropy balancing \citep{hainmueller2012entropy}, covariate balancing propensity scores \citep{imai2014cbps}, stable balancing weights \citep{zubizarreta2015stable}, minimal approximately balancing weights \citep{wang2020minimal}, kernel balancing \citep{hazlett2020kernelbalancing}, and general balancing weights for user-specified target populations \citep{li2018balancing}. 
Calibration-based results further connect moment balance to efficiency and robustness \citep{chan2016calibration,zhao2019tailored}, including regularised calibration for high-dimensional propensity-score models \citep{tan2020regularized}.
Related source-to-target weighting problems arise in generalising or transporting causal effects across populations \citep{dahabreh2019extend,dahabreh2021studydesigns,chen2023entropygeneralization}.
Balancing criteria are meaningful relative to the outcome structure they are intended to protect against \citep{kallus2020gom}; our analysis makes this structure explicit through a predictive model and evaluates residual imbalance together with weight-induced noise on the scale of target-mean prediction risk.

The resulting problem is particularly consequential under limited overlap.
Large inverse-probability weights can make an identified causal parameter irregular or difficult to estimate at conventional rates \citep{khan2010irregular}, motivating trimming or restriction to regions of sufficient overlap \citep{crump2009overlap}, diagnostics for practical positivity violations \citep{petersen2012positivity}, and alternative target populations such as those induced by overlap weights \citep{li2018overlapweights}. 
More broadly, extrapolation has long been recognised as a source of model dependence in counterfactual inference \citep{king2006extreme}, while methods such as synthetic control and augmented synthetic control constrain or regularise extrapolation from the observed donor support \citep{benmichael2021augsynth}. 
Related high-dimensional work studies how interpolation and ridge regularisation affect causal risk under confounding \citep{vankadara2022causal}. 
Most directly, \citet{arbour2025regularizing} regularise extrapolation through a worst-case error bound that trades feature imbalance and model-misspecification bias against variance.
Our criterion instead targets the conditional prediction risk of the counterfactual target mean.

Our estimator is also related to augmented, debiased, and minimax procedures that combine weighting with outcome regression. 
Augmented inverse-probability weighting and doubly robust estimators use outcome models to correct weighted estimators \citep{robins1994aipw,bang2005dr}, while orthogonal-score methods extend this principle to flexible nuisance estimation \citep{chernozhukov2018dml,chernozhukov2022riesz}. 
Approximate residual balancing \citep{athey2018arb} and augmented minimax linear estimation \citep{hirshberg2021aml} construct weights or augmentation terms by controlling worst-case bias over function classes; related finite-sample minimax theory explicitly accounts for bias when overlap is weak \citep{armstrong2021finite}. 
Linear regression itself can be represented through weighting, making balance and extrapolation properties of regression adjustment transparent \citep{chattopadhyay2023implied}. 
Most directly, \citet{brunssmith2025augmented} show broad algebraic equivalences between augmented balancing estimators and linear regression, and relate double-ridge augmentation to an undersmoothed ridge fit.
We analyse the ridge-augmentation estimator, but with a different focus on the conditional prediction risk for the treated and penalty selection.

Finally, our analysis connects causal weighting to covariate-shift prediction and high-dimensional ridge regression. 
Importance weighting and kernel mean matching treat distribution shift by reweighting the training distribution toward the target distribution \citep{gretton2008kmm,sugiyama2012covariateshift}, while \citet{patil2024optimal} show that optimal ridge regularisation for out-of-distribution prediction depends on the relationship between the training and test distributions.
\citet{du2023subsample} establish high-dimensional equivalences between subsampling and ridge regularisation and develop generalised cross-validation for ridge ensembles, while \citet{du2024implicit} characterise implicit regularisation paths induced by observation weighting.

\section{Causal Setup and Regularised Balancing Estimation}\label{sec:setup}

\subsection{Causal target and source-to-target covariate shift}

We use the potential-outcomes framework for the superpopulation vector
\((x,a,y(0),y(1))\), where \(x\in\RR^p\), \(a\in\{0,1\}\), and
\(y(0),y(1)\) are the control and treated potential outcomes. We observe i.i.d.\ data \(o_i=(x_i,a_i,y_i)\), \(i=1,\ldots,n\), and let \(n_1=\sum_{i=1}^n a_i\) and \(n_0=n-n_1\).
To ensure treatment effects are well-defined, we require regularity conditions on integrability
\(
\EE\left\{|y(0)|+|y(1)|\right\}<\infty
\) and nondegenerate treatment prevalence
\(
0<\PP(a=1)<1
\).
Furthermore, we require the following causal assumption throughout.
\begin{assumption}[Causal identification]
    \label{asm:causal-identification}
    The following conditions hold.
    \begin{enumerate}[(a)]
    \item\label{asm:causal-identification-i} Consistency:
    \(
    y = a y(1) + (1-a)y(0)
    \).
    \item\label{asm:causal-identification-ii} Unconfoundedness:
    \(
    y(0) \indep a \mid x
    \).
    \item\label{asm:causal-identification-iii} Positivity on the treated support:
    \(
    \PP(a=0 \mid x) > 0
    \) for \(\PP(\,\cdot \mid a=1)\)-almost every \(x\).
    \end{enumerate}
\end{assumption}

Let \( \pi_t:=\PP(a=t) \) and \(\PP_t:=\PP(x\in\cdot\mid a=t) \) for \(t\in\{0,1\}\) denote the treatment prevalences and treatment-specific covariate distributions. 
Define the treatment-specific response surfaces, the propensity score, and the propensity odds by 
\[ m_t(x):=\EE\{y(t)\mid x\}, \quad t\in\{0,1\}, \qquad e(x):=\PP(a=1\mid x), \qquad \omega(x):=\frac{e(x)}{1-e(x)}. \] 
The average treatment effect on the treated (ATT) is defined as
\[ \tau_{\mathrm{ATT}} :=\EE\{y(1)-y(0)\mid a=1\} =\mu_1-\mu_0, \] 
where \( \mu_1:=\EE\{y(1)\mid a=1\}\) and \(\mu_0:=\EE\{y(0)\mid a=1\}\).
We have the following identification result.
\begin{proposition}[Identification of the ATT] \label{prop:att-identification} 
    Under \Cref{asm:causal-identification}, the ATT is nonparametrically identified from the observed-data distribution. In particular, 
    \[
        \tau_{\mathrm{ATT}}
        =\EE(y\mid a=1)-\int \EE(y\mid x,a=0)\,\rd\PP_1(x)
        =\int\!\left\{\EE(y\mid x,a=1)-\EE(y\mid x,a=0)\right\}\,\rd\PP_1(x).
    \]
    Moreover, the Radon--Nikodym derivative \( r(x):=\frac{\rd\PP_1}{\rd\PP_0}(x) \) exists and satisfies \[ r(x) = \frac{\pi_0}{\pi_1}\omega(x) = \frac{\pi_0}{\pi_1} \frac{e(x)}{1-e(x)} \qquad \text{for \(\PP_0\)-almost every \(x\)}. \] 
    Consequently, the counterfactual mean admits the equivalent identification formulas 
    \[ \begin{aligned} \mu_0 &= \int m_0(x)\,\rd\PP_1(x)  = \int r(x)m_0(x)\,\rd\PP_0(x) = \EE\{r(x)y\mid a=0\} = \frac{1}{\pi_1} \EE\left[ (1-a)\frac{e(x)}{1-e(x)}y \right]. \end{aligned} \] 
    Equivalently, \[ \tau_{\mathrm{ATT}} = \frac{1}{\pi_1} \EE\left[ ay-(1-a)\frac{e(x)}{1-e(x)}y \right]. \] 
\end{proposition}

Proposition~\ref{prop:att-identification} shows that only the counterfactual component \(\mu_0\) requires transport across covariate distributions.
The treated mean \(\mu_1\) is directly estimable from the treated outcomes, with the natural sample analog \(\bar y_1:=n_1^{-1}\sum_{i:a_i=1}y_i\).
In contrast, estimating \(\mu_0=\int m_0(x)\,\rd\PP_1(x)\) requires using control outcomes observed under the source distribution \(\PP_0\) to recover the mean of the control response surface under the target distribution \(\PP_1\).

\Cref{asm:causal-identification}~\eqref{asm:causal-identification-iii} implies the absolute-continuity condition \(\PP_1\ll\PP_0\), but does not require the density ratio \(r=\rd\PP_1/\rd\PP_0\) to be bounded or well concentrated.
Thus, under \Cref{asm:causal-identification}, the ATT may remain identified even when the source-to-target shift is statistically severe and estimation of \(\mu_0\) is unstable.
Under weak overlap, especially in high dimensions, estimating \(\mu_0\) is therefore an out-of-distribution prediction problem: the control response surface is learned under \(\PP_0\) but evaluated under the potentially different target distribution \(\PP_1\).

This source-to-target perspective is not specific to the ATT.
Other causal estimands correspond to different target measures and covariate shifts and may require stronger identification conditions: the average treatment effect on the controls (ATC) transports the treated response surface from \(\PP_{x\mid a=1}\) to \(\PP_{x\mid a=0}\), the average treatment effect (ATE) transports both response surfaces to the marginal distribution \(\PP_x\), and conditional effects at a point \(x\) require both surfaces at \(\delta_x\).
We focus on the ATT because it isolates a single source-to-target extrapolation problem, from controls to treated units, while retaining the central difficulty created by weak overlap.
The remainder of the paper develops covariate balancing and ridge augmentation as complementary tools for controlling the resulting conditional prediction risk for the treated.

\subsection{Base balancing weights and residual imbalance}\label{subsec:base}

Collect the control covariates and outcomes in \(X_0\in\mathbb R^{n_0\times p}\) and \(y_0\in\mathbb R^{n_0}\), respectively, with rows \(x_i^\top\) and entries \(y_i\) for units satisfying \(a_i=0\).
Define the control and treated covariate means and the control outcome mean by
\[
    \bar x_0
    =
    \frac{1}{n_0}X_0^\top 1_{n_0},
    \qquad
    \bar x_1
    =
    \frac{1}{n_1}\sum_{i:a_i=1}x_i,
    \qquad
    \bar y_0
    =
    \frac{1}{n_0}1_{n_0}^\top y_0.
\]
Define the centred control covariates and outcomes by
\begin{align}
    X_0^c=C_0X_0,
    \qquad
    y_0^c=C_0y_0,
    \qquad
    \text{where } C_0
    =
    I_{n_0}
    -
    \frac{1}{n_0}\one_{n_0}\one_{n_0}^\top.\label{eq:C}
\end{align}

Covariate balancing methods construct a normalised base-weight vector \(\gamma\in\RR^{n_0}\) so that the weighted control covariate mean approximates the treated covariate mean.
The corresponding weighted control outcome mean \(\gamma^\top y_0\) provides a base estimator of the counterfactual mean \(\mu_0\).
We restrict attention to affinely normalised weights satisfying \(\one_{n_0}^\top\gamma=1\).
Under this normalisation, the residual covariate imbalance is
\begin{align}
    \Delta(\gamma)
    &=
    \bar x_1-X_0^\top\gamma
    =
    (\bar x_1-\bar x_0)-(X_0^c)^\top\gamma,
    \label{eq:Delta}
\end{align}
which expresses the remaining discrepancy as the original treated-control mean shift minus the centred shift induced by reweighting.
Additional convex restrictions, such as nonnegativity or upper bounds on individual weights, may be imposed when required by a particular balancing method.

A broad class of regularised balancing estimators trades residual imbalance against weight dispersion,
\(
    \mathcal D(\gamma)
    :=
    \sum_{i=1}^{n_0}\mathcal C_i(\gamma_i),
\) where each \(\mathcal C_i\) is convex.
For a balance tolerance \(\delta\geq0\), define
\[
    \widehat\gamma_\delta
    \in
    \argmin_{\substack{\gamma\in\RR^{n_0}\\
                       \one_{n_0}^\top\gamma=1}}
    \mathcal D(\gamma)
    \quad
    \text{subject to}
    \quad
    \|\Delta(\gamma)\|_2^2\leq\delta.
\]
A corresponding penalised formulation is
\[
    \widehat\gamma_\zeta
    \in
    \argmin_{\substack{\gamma\in\RR^{n_0}\\
                       \one_{n_0}^\top\gamma=1}}
    \left\{
        \mathcal D(\gamma)
        +
        \zeta\|\Delta(\gamma)\|_2^2
    \right\},
    \qquad
    \zeta\geq0,
\]
where smaller \(\delta\) or larger \(\zeta\) places greater emphasis on balance.
Familiar examples include:
\begin{enumerate}[(a)]
    \item the quadratic dispersion
    \(        \mathcal D(\gamma)
            =
            \frac12\|\gamma\|_2^2
    \),
    which penalises weight concentration and underlies stable balancing weights
    \citep{zubizarreta2015stable};

    \item the relative-entropy dispersion
    \(
        \mathcal D(\gamma)
        =
        \sum_{i=1}^{n_0}
        \gamma_i\log(\gamma_i/q_i)
    \),
    with nonnegative weights \(q_i>0\) and \(0\log0:=0\), as in entropy balancing
    \citep{hainmueller2012entropy}; and

    \item general convex dispersion criteria under bounded covariate imbalance
    \citep{wang2020minimal}.
\end{enumerate}
Under feasibility and a suitable constraint qualification, a constrained solution with an active balance constraint also solves the penalised problem for an associated Karush--Kuhn--Tucker multiplier \(\zeta\), while each \(\zeta\) determines the balance tolerance attained by its penalised solution.
The two formulations therefore generate the same set of solutions as their tuning parameters vary under the usual regularity conditions, although a prespecified \(\delta\) need not correspond to a prespecified \(\zeta\).

Another class of estimators uses the propensity score \(e(x)\) to reweight the control sample.
By \Cref{prop:att-identification}, the propensity odds are proportional to the source-to-target density ratio:
\[
    \frac{\rd\PP_1}{\rd\PP_0}(x)
    =
    \frac{\PP(a=0)}{\PP(a=1)}\omega(x),    
    \qquad
    \mu_0
    =
    \EE\left[
        \frac{1-a}{\PP(a=1)}
        \omega(x) y
    \right],
    \qquad
    \omega(x):=\frac{e(x)}{1-e(x)}.
\]
This identity motivates two common inverse-probability-weighted estimators:
\begin{enumerate}[(a)]
    \item The Horvitz--Thompson estimator: \(\widehat\mu_{0}^{\textsc{ht}}
        =
        \frac{1}{n_1}
        \sum_{i:a_i=0}\widehat\omega_i y_i,\) where \(\widehat\omega_i
        =
        {\widehat e(x_i)}/({1-\widehat e(x_i)})\).
    
    \item The self-normalised H\'ajek estimator:
    \(
        \widehat\mu_{0}^{\textsc{h}}
        =
        \sum_{i:a_i=0}\widehat\gamma_i^{\ipw}y_i\) where
    \(
        \widehat\gamma_i^{\ipw}
        =
        {\widehat\omega_i}/
            ({\sum_{j:a_j=0}\widehat\omega_j}).
    \)
\end{enumerate}
The H\'ajek weights satisfy the affine normalisation exactly and therefore fit the normalised base-weight framework above, whereas the Horvitz--Thompson weights need not sum to one in a realised sample \citep{rosenbaum1983central,lunceford2004ps}.
Under weak overlap, the propensity odds can be highly variable, causing the normalised weights to concentrate on a small number of controls.
Propensity-based weighting, like direct covariate balancing, can therefore improve source-to-target representation at the cost of increased weight concentration.

\subsection{Ridge-augmented balancing}\label{subsec:ridge-augmented}

The weighting procedures in \Cref{subsec:base} need not eliminate residual covariate imbalance. 
Outcome-model augmentation adjusts the base weighted estimator for the remaining imbalance.
This follows the general principle of combining weighting with outcome regression \citep{robins1994aipw,bang2005dr} and is closely related to approximate residual balancing and augmented minimax linear estimation \citep{athey2018arb,hirshberg2021aml}.

Let \(\widehat\gamma\) be a normalised base-weight vector from \Cref{subsec:base}, and denote its residual imbalance by \(\Delta=\Delta(\widehat\gamma)\) as in \eqref{eq:Delta}.
The construction below applies to any such normalised base weights and does not require a particular balancing criterion.
For augmentation, consider the affine working class \(m_{0,\beta}(x):=\alpha_0+x^\top\beta\).
We fit ridge regression with an unpenalised intercept:
\[
    \{\widehat\alpha_0(\lambda),\widehat\beta(\lambda)\}
    \in
    \argmin_{\alpha\in\RR,\,b\in\RR^p}
    \left\{
        \frac{1}{2n_0}
        \|y_0-\alpha\one_{n_0}-X_0b\|_2^2
        +
        \frac{\lambda}{2}\|b\|_2^2
    \right\},
    \qquad \lambda>0.
\]
The covariate vector \(x\) need not consist only of raw baseline measurements.
If \(w\) denotes underlying baseline covariates and \(x=\phi(w)\in\RR^p\) is a chosen finite-dimensional feature representation, including polynomial or spline expansions, random-feature maps, and finite-dimensional kernel approximations, then \(m_{0,\beta}\{\phi(w)\}=\alpha_0+\phi(w)^\top\beta\) can be nonlinear in \(w\).
When the same representation is used throughout \Cref{sec:setup}, residual balance and ridge augmentation are defined in the same feature space.
For the estimator to retain a causal interpretation, the chosen features must preserve the identifying information required by \Cref{asm:causal-identification}; in particular, the corresponding unconfoundedness condition must hold conditional on those features.

Define the sample covariance matrix and resolvent by \( S_0^c
    =
    n_0^{-1}(X_0^c)^\top X_0^c\) and \( M_\lambda^c
    =
    (S_0^c+\lambda I_p)^{-1}\).
The explicit ridge solution is
\begin{align}
    \widehat\beta(\lambda)
    &=
    \frac{1}{n_0}M_\lambda^c(X_0^c)^\top y_0^c
    =
    \frac{1}{n_0}M_\lambda^c(X_0^c)^\top y_0,
    \qquad
    \widehat\alpha_0(\lambda)
    =
    \bar y_0-\bar x_0^\top\widehat\beta(\lambda).
    \label{eq:ridge}
\end{align}
The ridge penalty stabilises extrapolation by shrinking fitted slopes more strongly in directions that are weakly supported by the control design.
Following the augmented-balancing construction \citep{brunssmith2025augmented,hirshberg2021aml,athey2018arb} with ridge outcome adjustment \citep{benmichael2021augsynth}, define
\begin{align*}
    \widehat\mu_{0,\lambda}
    &=
    \widehat\alpha_0(\lambda)
    +
    \bar x_1^\top\widehat\beta(\lambda)
    +
    \widehat\gamma^\top
    {y_0-\widehat\alpha_0(\lambda)\one_{n_0}
    -X_0\widehat\beta(\lambda)} 
    =
    \widehat\gamma^\top y_0
    +
    \Delta^\top\widehat\beta(\lambda).
\end{align*}
Thus the base weights determine the source representation of the target, while ridge augmentation corrects the remaining discrepancy along \(\Delta\).

Substituting \eqref{eq:ridge} gives the equivalent linear-weight representation
\begin{align}
    \widehat\mu_{0,\lambda}
    &=
    \widehat\gamma_\lambda^\top y_0,
    \qquad
    \widehat\gamma_\lambda
    :=
    \widehat\gamma
    +
    \frac{1}{n_0}X_0^cM_\lambda^c\Delta.
    \label{eq:gamma-aug}
\end{align}
Ridge augmentation therefore adjusts the base weights according to both the residual covariate imbalance and the spectral geometry of the control design.
Since \(\one_{n_0}^\top X_0^c=0\), the augmented weights remain normalised, \(\one_{n_0}^\top\widehat\gamma_\lambda=1\).
Equivalently, they solve the weight-space problem
\begin{equation}
    \widehat\gamma_\lambda
    =
    \argmin_{\substack{\widetilde\gamma\in\RR^{n_0}\\
                        \one_{n_0}^\top\widetilde\gamma=1}}
    \left\{
        \frac{1}{2n_0}
        \|\bar x_1-X_0^\top\widetilde\gamma\|_2^2
        +
        \frac{\lambda}{2}
        \|\widetilde\gamma-\widehat\gamma\|_2^2
    \right\}.
    \label{eq:gamma-aug-ridge}
\end{equation}
The first term reduces residual imbalance in the chosen feature representation, whereas the second keeps the augmented weights close to the base weights.
The penalty \(\lambda\) therefore controls the strength of augmentation: smaller \(\lambda\) places greater emphasis on reducing residual imbalance, while larger \(\lambda\) retains more of the base weights.

The resulting ATT estimator is \(\widehat\tau_\lambda=\bar y_1-\widehat\mu_{0,\lambda}\).
For any constant \(c\), replacing \(y_0\) by \(y_0+c\one_{n_0}\) leaves \(\widehat\beta(\lambda)\) and the weights unchanged and shifts \(\widehat\mu_{0,\lambda}\) by \(c\).
Consequently, shifting all observed and potential outcomes by the same constant shifts both \(\bar y_1\) and \(\widehat\mu_{0,\lambda}\) by \(c\), leaving \(\widehat\tau_\lambda\) unchanged.

\section{Conditional Prediction Risk}\label{sec:risk}

Given the causal identification established in the previous section, we study the prediction error of the ridge-augmented estimator under a model-based law over affine control response surfaces in the chosen feature representation.
Our analysis contains two complementary parts.
First, \Cref{sec:risk:exact-risk} gives an exact finite-sample decomposition of the conditional prediction risk for normalised base weights constructed without the control outcomes.
Second, \Cref{sec:risk:asymptotic-risk} introduces a design-independent high-dimensional reference setting, from which we obtain explicit risk characterisations and limiting extrapolation regimes (\Cref{sec:risk:implication}).

\subsection{Exact conditional risk decomposition}
\label{sec:risk:exact-risk}

Recall the treated-sample covariate mean \(\bar x_1=n_1^{-1}\sum_{i:a_i=1}x_i\), and write \(\nu_1=\EE[x\mid a=1]\) and \(\epsilon_1=\bar x_1-\nu_1\).
Let \(\mathcal F_{\mathrm{aux}}\) denote the auxiliary information used to construct the weights; it is the trivial sigma-field when no such information is used.
When honest target splitting is used, \(\mathcal F_{\mathrm{aux}}\) also contains the split assignment \(\mathcal S_n:=\sigma(\mathcal I_{\mathrm P},\mathcal I_{\mathrm E})\).
Define the conditioning field \(\mathcal G_n:=\sigma(X_0,X_1,\mathcal F_{\mathrm{aux}})\).
Throughout the conditional theory, we suppress hats after conditioning on \(\mathcal G_n\): \(\widehat\gamma\) and \(\widehat\gamma_\lambda\) denote the computed weights, whereas \(\gamma\) and \(\gamma_\lambda\) denote their conditional theoretical counterparts.
For any realised sample sizes \((n_0,n_1)\), we work under the following predictive model for the control outcomes.

\begin{assumption}[Affine random-effects predictive model]
    \label{asm:predictive-model}
    The control response model is
    \(
        y_i(0)=\alpha_0+x_i^\top\beta+\epsilon_{0i}
    \) for \(i=1,\ldots,n_0\), where \(\alpha_0\) is a deterministic intercept, \(\beta\) is a coefficient vector and \(\epsilon_0\) is a random noise vector, satisfying the following conditions.
    \begin{enumerate}[(a)]
        \item (Coefficient) Independent of \(\mathcal G_n\), the coefficient admits a prior distribution \(\beta\sim\Pi\) such that \(\EE_\Pi(\beta)=0\) and \(\EE_\Pi(\beta\beta^\top)=\frac{r^2}{p}I_p\) for constant \(r^2\in(0,\infty)\).
        \item (Noise) Conditional on \(\mathcal G_n\) and \(\beta\),
        \(\EE(\epsilon_0\mid\mathcal G_n,\beta)=0\) and \( \Var(\epsilon_0\mid\mathcal G_n,\beta)=\sigma_0^2 I_{n_0}\) for constant \(\sigma_0^2\in(0,\infty)\).
    \end{enumerate}
\end{assumption}

The predictive law \(\Pi\) defines an integrated performance criterion by averaging prediction error over the response surfaces.
For a realised \(\beta\), define the model-based predictive target \(\mu_{0,\beta}:=\int m_{0,\beta}(x)\,\mathrm{d}\PP_1(x)=\alpha_0+\nu_1^\top\beta\).
Under \Cref{asm:predictive-model}, each realisation of \(\beta\) indexes a possible control response surface and hence a corresponding counterfactual mean.
Across realisations of \(\beta\), this target may vary, while for a realised response surface it is fixed.
When the realised affine surface coincides with the identified control response surface, \(\mu_{0,\beta}\) is the causal counterfactual mean \(\mu_0\) for that data-generating process.

For \(\widehat\mu_{0,\lambda}=\gamma_\lambda^\top y_0\), define the extrapolation direction \(d_n(\lambda):=X_0^\top\gamma_\lambda-\nu_1\).
The prediction error satisfies \(\widehat\mu_{0,\lambda}-\mu_{0,\beta}=d_n(\lambda)^\top\beta+\gamma_\lambda^\top\epsilon_0\).
The intercept cancels because \(\one_{n_0}^\top\gamma_\lambda=1\).
For a realised \(\beta\), the conditional mean-squared prediction error is
\begin{align}
    R(\beta;\lambda)
    &:={\EE}_{\epsilon_0}\left[
    \{\widehat\mu_{0,\lambda}-\mu_{0,\beta}\}^2
    \mid\mathcal G_n,\beta\right] \notag\\
    &=\{\beta^\top d_n(\lambda)\}^2
      +\sigma_0^2\|\gamma_\lambda\|_2^2
      =: B(\beta;\lambda)+ V(\lambda).
      \label{eq:R-BV-decom}
\end{align}
Here \(B(\beta;\lambda)\) is the fixed-surface squared conditional bias, and \(V(\lambda)\) is the conditional noise variance.
\Cref{prop:conditional-risk} characterises the conditional prediction risk averaged over \(\Pi\).

\begin{proposition}[Conditional prediction risk]\label{prop:conditional-risk}
    Suppose \(\gamma\) is \(\mathcal G_n\)-measurable and satisfies
    \(\one_{n_0}^\top\gamma=1\). Then \(\gamma_\lambda\) is also
    \(\mathcal G_n\)-measurable and satisfies
    \(\one_{n_0}^\top\gamma_\lambda=1\). Under
    \Cref{asm:predictive-model}, define
    \begin{align*}
        B_n(\lambda)
        &:=\EE_{\beta\sim\Pi}\{B(\beta;\lambda)\mid\mathcal G_n\}
        =\frac{r^2}{p}\|d_n(\lambda)\|_2^2,\qquad
        V_n(\lambda)
        :=V(\lambda)
        =\sigma_0^2\|\gamma_\lambda\|_2^2.
    \end{align*}
    Thus \(B_n\) is the integrated squared-bias component,
    \(V_n\) is the noise variance component, and the conditional prediction risk is
    \begin{align}\label{eq:conditional-risk}
        R_n(\lambda)
        &:=\EE_{\beta\sim\Pi,\epsilon_0}\left[
        \{\widehat\mu_{0,\lambda}-\mu_{0,\beta}\}^2\mid\mathcal G_n\right]
        =B_n(\lambda)+V_n(\lambda)\notag\\
        &=\frac{r^2}{p}\left\{\|\epsilon_1\|_2^2
        -2\lambda\epsilon_1^\top M_\lambda^c\Delta
        +\lambda^2\Delta^\top (M_\lambda^c)^2\Delta\right\} +\sigma_0^2\left\{\|\gamma\|_2^2
        +\frac{2}{n_0}\gamma^\top X_0^c M_\lambda^c\Delta
        +\frac1{n_0}\Delta^\top M_\lambda^c
        \{I_p-\lambda M_\lambda^c\}\Delta\right\}.
    \end{align}
\end{proposition}

The finite-sample identity in \Cref{prop:conditional-risk} applies to any normalised \(\mathcal G_n\)-measurable base weight vector, including weights adapted to the realised covariates but constructed without the control outcomes.
The decomposition expresses residual extrapolation \(B_n\) and weight-induced noise \(V_n\) in the same risk units, making explicit the tradeoff controlled by ridge augmentation.

\subsection{High-dimensional extrapolation risk}\label{sec:risk:asymptotic-risk}

\Cref{prop:conditional-risk} gives the exact conditional prediction risk, but the realised quantities in \eqref{eq:conditional-risk} do not by themselves reveal its systematic behaviour when \(p\) is comparable to the sample sizes.
They jointly depend on the source spectrum, the orientation and concentration of the base weights, the source-target mean shift, and empirical fluctuations of the source and target covariate means.
To separate these effects, we study a proportional random-design sequence with design-independent base weights, yielding explicit deterministic equivalents for the high-dimensional risk.
For an ATT interpretation along this sequence, the causal identification conditions continue to hold at each \(n\), including absolute continuity of the target covariate law with respect to the source covariate law.
The next two assumptions specify the high-dimensional covariate design and the scaling of the design-independent base weights.
For \(t=0,1\), define \(\phi_{t,n}:=p/n_t\).

\begin{assumption}[High-dimensional covariate design]
    \label{asm:rmt}
    The following conditions hold.
    \begin{enumerate}[(a)]
        \item (Proportional asymptotic regime)        
        Along a sequence indexed by total sample size $n$ with \(p,n_0,n_1\to\infty\), there exist constants \(0<c_\phi<C_\phi<\infty\) such that \(c_\phi\leq\phi_{t,n}\leq C_\phi\) for \(t=0,1\).

        \item (Source-target design)
        There exist deterministic vectors \(\nu_0,\nu_1\in\RR^p\) and deterministic symmetric matrices
        \(\Sigma_0,\Sigma_1\in\RR^{p\times p}\) such that
        \[
            X_0
            =
            \one_{n_0}\nu_0^\top+Z_0\Sigma_0^{1/2},
            \qquad
            x_{1i}
            =
            \nu_1+\Sigma_1^{1/2}z_{1i},
            \qquad
            i=1,\ldots,n_1,
        \]
        where the entries of \(Z_0\in\RR^{n_0\times p}\) are i.i.d., the entries of
        \(z_{11},\ldots,z_{1n_1}\in\RR^p\) are mutually independent, and the two arrays are independent.
        There exists \(\delta>0\) such that
        \[
            \EE Z_{0,ij}
            =
            \EE z_{1i,j}
            =
            0,
            \qquad
            \EE Z_{0,ij}^2
            =
            \EE z_{1i,j}^2
            =
            1,
            \qquad
            \sup_{n,i,j}\EE|Z_{0,ij}|^{8+\delta}<\infty,
            \qquad
            \sup_{n,i,j}\EE|z_{1i,j}|^{4+\delta}<\infty.
        \]

        \item (Covariance spectra)
        There exist constants \(0<c_\Sigma<C_\Sigma<\infty\) such that
        \(c_\Sigma I_p\preceq\Sigma_t\preceq C_\Sigma I_p\) for \(t=0,1\).
    \end{enumerate}
\end{assumption}

\begin{assumption}[Design-independent balancing]
    \label{asm:design-independent-regime}
    The following conditions hold.
    \begin{enumerate}[(a)]
        \item\label{asm:design-independent-regime-i} (Design independence)
        The sigma-field \(\mathcal F_{\mathrm{aux}}\) is independent of
        \((X_0,X_1)\), and the base weight vector
        \(\gamma\in\RR^{n_0}\) is \(\mathcal F_{\mathrm{aux}}\)-measurable
        with \(\sum_i\gamma_i=1\).

        \item\label{asm:design-independent-regime-ii} (Coupled extrapolation scaling)
        For some \(\eta\in[0,1]\), define
        \(\nu_\Delta:=\nu_1-\nu_0\), \(\rho_{\eta,n}^2:=(n_0^\eta/p)\|\nu_\Delta\|_2^2\), and \(\varrho_{\eta,n}^2:=n_0^\eta\|C_0\gamma\|_2^2\).
        Assume \(\rho_{\eta,n}^2=\cO(1)\) and, for \(\PP\)-almost every
        realisation of \(\mathcal F_{\mathrm{aux}}\),
        \(
            \sup_n\varrho_{\eta,n}^2<\infty.
        \)
    \end{enumerate}
\end{assumption}

\Cref{asm:rmt} specifies proportional asymptotic regimes, bounded source and target covariance spectra, and the moment conditions.
\Cref{asm:design-independent-regime}~\eqref{asm:design-independent-regime-i} requires the base weights to be independent of the realised covariates conditional on \(\mathcal F_{\mathrm{aux}}\).
Under \Cref{asm:design-independent-regime}~\eqref{asm:design-independent-regime-ii}, \(\rho_{\eta,n}^2\) measures source-target mean-shift energy and \(\varrho_{\eta,n}^2\) measures weight concentration, with \(\eta\) indexing the common \(n_0^{-\eta}\) scale.
\Cref{prop:overlap-weight-concentration} connects the scale of the concentration constant \(\varrho_{\eta,n}^2\) to density-ratio overlap.

The deterministic-equivalent analysis also requires concentration of the source random design; we impose a standard sub-Gaussian sufficient condition below, while the weaker resolvent condition used in the proof is stated in \Cref{asm:technical-source-resolvent}.

\begin{assumption}[Sub-Gaussian source design]
    \label{asm:global-source-resolvent}
    The source coordinates are uniformly sub-Gaussian: for some constant \(K<\infty\),
    \(
        \sup_{n,i,j}\|Z_{0,ij}\|_{\psi_2}\leq K.
    \)
\end{assumption}

We next introduce the spectral and resolvent notation used in the high-dimensional risk characterisation.
For the remainder of this subsection, fix \(\Lambda=[\lambda_-,\lambda_+]\subset(0,\infty)\), write \(\otr(A):=p^{-1}\tr(A)\) for matrices in \(\RR^{p\times p}\), and let \((s_j,w_j)_{j=1}^p\) be an orthonormal eigensystem of \(\Sigma_0\).
Define the spectral measures
\begin{align*}
    H_{0,p}
    &=
    \frac1p\sum_{j=1}^p\delta_{s_j},
    \qquad
    H_{1\mid0,p}
    =
    \frac1p\sum_{j=1}^p
    (w_j^\top\Sigma_1w_j)\delta_{s_j},
    \qquad
    G_{\nu,p}
    =
    \begin{cases}
        \displaystyle
        \frac1{\|\nu_\Delta\|_2^2}
        \sum_{j=1}^p(\nu_\Delta^\top w_j)^2\delta_{s_j},
        & \nu_\Delta\neq0,\\
        H_{0,p},
        & \nu_\Delta=0.
    \end{cases}
\end{align*}
Here \(H_{0,p}\) records the source covariance spectrum, \(H_{1\mid0,p}\) the target covariance geometry along the source eigenspaces, and \(G_{\nu,p}\) the orientation of the source-target mean shift relative to those eigenspaces.
For \(\lambda\in\Lambda\), let \(v_n(\lambda)>0\) be the unique solution to
\[
    \frac1{v_n(\lambda)}
    =
    \lambda
    +
    \phi_{0,n}
    \int
    \frac{s}{1+v_n(\lambda)s}
    \,\rd H_{0,p}(s),
\]
and define
\vspace{-2mm}
\begin{align*}
    \widetilde v_{v,n}(\lambda)
    &:=
    \left[
        v_n(\lambda)^{-2}
        -
        \phi_{0,n}
        \int
        \frac{s^2}{\{1+v_n(\lambda)s\}^2}
        \,\rd H_{0,p}(s)
    \right]^{-1},\\
    A_n(\lambda)
    &:=
    \{v_n(\lambda)\Sigma_0+I_p\}^{-1}.
\end{align*}
For a deterministic matrix sequence \(A\in\RR^{p\times p}\) satisfying \(A=A^\top\succeq0\) and \(\sup_n\|A\|_{\oper}<\infty\), define
\begin{align*}
    \widetilde v_{b,n}(\lambda;A)
    &:=
    \frac{
        \phi_{0,n}
        \otr\{A\Sigma_0(v_n(\lambda)\Sigma_0+I_p)^{-2}\}
    }{
        \widetilde v_{v,n}(\lambda)^{-1}
    },\\
    K_{1,n}(\lambda;A)
    &:=
    \otr\{AA_n(\lambda)\},\\
    K_{2,n}(\lambda;A)
    &:=
    \otr\left[
        A_n(\lambda)
        \left\{\widetilde v_{b,n}(\lambda;A)\Sigma_0+A\right\}
        A_n(\lambda)
    \right].
\end{align*}
Write \(\widetilde v_{b,n}(\lambda)\) when \(A=I_p\).
The quantities \(K_{1,n}(\lambda;A)\) and \(K_{2,n}(\lambda;A)\) are the deterministic equivalents of \(\lambda\otr\{AM_\lambda^c\}\) and \(\lambda^2\otr\{A(M_\lambda^c)^2\}\), respectively.
With this notation, we can characterise the high-dimensional conditional prediction risk uniformly over the regularisation path.

\begin{theorem}[Uniform risk deterministic equivalents under design-independent base weights]
    \label{thm:risk-asym}
    Let \(\Lambda=[\lambda_-,\lambda_+]\subset(0,\infty)\).
    Under \Cref{asm:predictive-model,asm:rmt,asm:design-independent-regime,asm:global-source-resolvent}, define
    \begingroup\small
    \begin{align*}
        \sB_{\eta,n}(\lambda)
        :=&
        \frac{\lambda^2r^2\varrho_{\eta,n}^2}{\phi_{0,n}}
        \{v_n(\lambda)-\lambda\widetilde v_{v,n}(\lambda)\}
        +r^2\rho_{\eta,n}^2
        \int\frac{\widetilde v_{b,n}(\lambda)s+1}
        {\{v_n(\lambda)s+1\}^2}\,\rd G_{\nu,p}(s)\\
        &+
        \ind{(\eta=1)}\frac{r^2\phi_{1,n}}{\phi_{0,n}}
        \left[
            \otr(\Sigma_1)
            -2\otr\{\Sigma_1A_n(\lambda)\}
            +\otr\left\{
                A_n(\lambda)
                [\widetilde v_{b,n}(\lambda;\Sigma_1)\Sigma_0+\Sigma_1]
                A_n(\lambda)
            \right\}
        \right]\\
        &\quad+
        \ind{(\eta=1)}r^2K_{2,n}(\lambda;\Sigma_0),\\[1mm]
        \sV_{\eta,n}(\lambda)
        :=&
        \sigma_0^2\varrho_{\eta,n}^2\lambda^2
        \widetilde v_{v,n}(\lambda)
        +\frac{\sigma_0^2\phi_{0,n}\rho_{\eta,n}^2}{\lambda}
        \int\frac{\{v_n(\lambda)-\widetilde v_{b,n}(\lambda)\}s}
        {\{v_n(\lambda)s+1\}^2}\,\rd G_{\nu,p}(s)\\
        &+
        \ind{(\eta=1)}\frac{\sigma_0^2\phi_{1,n}}{\lambda}
        \left[
            \otr\{\Sigma_1A_n(\lambda)\}
            -\otr\left\{
                A_n(\lambda)
                [\widetilde v_{b,n}(\lambda;\Sigma_1)\Sigma_0+\Sigma_1]
                A_n(\lambda)
            \right\}
        \right]\\
        &\quad+
        \ind{(\eta=1)}\sigma_0^2
        \left[
            1+\frac{\phi_{0,n}}{\lambda}
            \{K_{1,n}(\lambda;\Sigma_0)-K_{2,n}(\lambda;\Sigma_0)\}
        \right].
    \end{align*}\par
    \endgroup
    \vspace{-4mm}
    Let \(\sR_{\eta,n}(\lambda):=\sB_{\eta,n}(\lambda)+\sV_{\eta,n}(\lambda)\).
    For \(\PP\)-almost every realisation of \(\mathcal F_{\mathrm{aux}}\), under the conditional law of \((X_0,X_1)\) given \(\mathcal F_{\mathrm{aux}}\), every pair \((Q_n,Q_{\eta,n})\in\{(B_n,\sB_{\eta,n}),(V_n,\sV_{\eta,n}),(R_n,\sR_{\eta,n})\}\) satisfies
    \[
        \sup_{\lambda\in\Lambda}
        \left|n_0^\eta Q_n(\lambda)-Q_{\eta,n}(\lambda)\right|
        \pto0.
        \label{eq:finite-de-risk}
    \]
    The same convergence holds in probability under the joint law.
\end{theorem}

\Cref{thm:risk-asym} uniformly approximates the entire scaled conditional risk path by explicit quantities determined by the source spectrum, source-target mean shift, weight concentration, and empirical-mean sampling.
The deterministic equivalents retain the finite-dimensional spectral measures and aspect ratios; fixed limits are introduced in \Cref{sec:risk:implication}.

\subsection{Limiting regimes and qualitative implications}\label{sec:risk:implication}

\Cref{thm:risk-asym} gives sequence-dependent deterministic equivalents for the scaled conditional risk path; \Cref{asm:spectral-limits} converts these into fixed limiting risk functions by imposing convergence of the aspect ratios, scaling constants, and spectral measures.
The resulting limits identify which terms in \Cref{thm:risk-asym} determine the leading risk across the extrapolation regimes indexed by \(\eta\).

\begin{assumption}[Fixed spectral limits]\label{asm:spectral-limits}
    In addition to \Cref{asm:rmt,asm:design-independent-regime}, suppose that, for \(t=0,1\), \(\phi_{t,n}\to\phi_t\in(0,\infty)\), \(\rho_{\eta,n}^2\to\rho_\eta^2\), and \(\varrho_{\eta,n}^2\asto\varrho_\eta^2\) for deterministic constants \(\rho_\eta^2,\varrho_\eta^2\in[0,\infty)\).
    Suppose further that, weakly as finite measures, \(H_{0,p}\dto H_0\), \(G_{\nu,p}\dto G_\nu\), and \(H_{1\mid0,p}\dto H_{1\mid0}\).
\end{assumption}

Let \(\sB_\eta,\sV_\eta\colon\Lambda\to[0,\infty)\) denote the continuous fixed-limit integrated squared-bias and noise variance functions induced by \Cref{asm:spectral-limits}, and set \(\sR_\eta=\sB_\eta+\sV_\eta\).
Their explicit formulas are given in \eqref{eq:fixed-limit-functions} of Appendix~\ref{app:proof-fixed-limit-corollaries}.
The following corollary gives the fixed-limit form of \Cref{thm:risk-asym} and the resulting risk scale.

\begin{corollary}[Fixed risk limits and extrapolation regimes]
    \label{cor:fixed-risk-limit}
    Under the assumptions of \Cref{thm:risk-asym} and \Cref{asm:spectral-limits}, for \(\PP\)-almost every admissible realisation of \(\mathcal F_{\mathrm{aux}}\), for every pair \((Q_n,Q_\eta)\in\{(B_n,\sB_\eta),(V_n,\sV_\eta),(R_n,\sR_\eta)\}\), the following convergence holds in probability under the conditional law of \((X_0,X_1)\) given \(\mathcal F_{\mathrm{aux}}\):
    \(
        \sup_{\lambda\in\Lambda}
        \left|n_0^\eta Q_n(\lambda)-Q_\eta(\lambda)\right|
        \pto0.
    \)
    The same convergence holds in probability under the joint law.
    For each fixed \(\lambda\), each such pair satisfies \(Q_n(\lambda)=n_0^{-\eta}\{Q_\eta(\lambda)+\op(1)\}\).
    Consequently, 
    \begin{enumerate}[(a)]
        \item when \(\eta>0\), each component vanishes; if \(Q_\eta(\lambda)>0\), then \(n_0^\eta Q_n(\lambda)\pto Q_\eta(\lambda)\), while if \(Q_\eta(\lambda)=0\), then \(Q_n(\lambda)=\op(n_0^{-\eta})\);
        \item when \(\eta=0\), \(R_n(\lambda)\pto\sR_0(\lambda)\), so the conditional prediction risk vanishes in probability if and only if \(\sR_0(\lambda)=0\).
    \end{enumerate}
\end{corollary}

Under \Cref{asm:design-independent-regime}, the source-target mean-shift energy \(p^{-1}\|\nu_\Delta\|_2^2\) and weight concentration \(\|C_0\gamma\|_2^2\) are both of order \(n_0^{-\eta}\), which \Cref{thm:risk-asym} carries into the prediction risk.
For oracle density-ratio weights, \Cref{prop:overlap-weight-concentration} further shows that the effective sample size grows at rate \(n_0^\eta\) for \(0<\eta\leq1\), so smaller positive \(\eta\) corresponds to slower growth of the effective sample size.

At \(\eta=0\), the source-target mean-shift and weight-concentration contributions in \Cref{thm:risk-asym} can remain first order, so target prediction risk can remain nonvanishing as the nominal sample sizes diverge.
Thus, identification can coexist with persistent prediction risk under severe extrapolation.
For \(0<\eta<1\), the risk decays at rate \(n_0^{-\eta}\) when its limiting risk is positive, while the \(n_0^{-1}\)-scale contributions from the source and target empirical means and affine normalisation are lower order.

In the boundary case \(\eta=1\), the affine-normalisation contribution is of order \(n_0^{-1}\), matching the source-target mean-shift and weight-concentration contributions, and the source and target empirical-mean fluctuations also enter the leading risk at the same order.
Consequently, the target aspect ratio \(\phi_1\) and target covariance geometry \(H_{1\mid0}\) enter the leading risk at \(\eta=1\), while for \(\eta<1\) their contributions are lower order and the source covariance spectrum and aspect ratio continue to shape the leading risk.

\begin{corollary}[Lower-order target-sampling effects under design-independent base weights]
    \label{cor:target-sampling-invariant}
    Under the assumptions of \Cref{cor:fixed-risk-limit}, if \(\eta\in[0,1)\), then the fixed limits \(\sB_\eta(\lambda)\), \(\sV_\eta(\lambda)\), and \(\sR_\eta(\lambda)\) do not depend on \(H_{1\mid0}\) or \(\phi_1\).
\end{corollary}

For \(\eta<1\), the target enters the leading risk only through the source-target mean shift \(\nu_\Delta=\nu_1-\nu_0\), its magnitude \(\rho_\eta^2\), and its orientation relative to the source eigenspaces through \(G_\nu\), so the target-specific part of the leading risk reflects extrapolation of the systematic source-target discrepancy through the source design.
\Cref{prop:bias-monotonicity} also shows that \(\sB_\eta(\lambda)\) is nondecreasing in \(\lambda\), consistent with \Cref{subsec:ridge-augmented}: increasing \(\lambda\) weakens augmentation and leaves more of this discrepancy unresolved.
Because \(\sV_\eta(\lambda)\) also varies with \(\lambda\), the total-risk minimising penalty balances extrapolation against noise variance.
\Cref{subsec:num-adaptive} shows that these predictions remain informative for data-adaptive base weights.

\section{Risk Estimation, Tuning, and Prediction Intervals}\label{sec:risk-estimation}
\subsection{Estimating Variance Components}\label{sec:variance-components}

Evaluating the conditional prediction risk requires the random-effects signal variance \(r^2\) and residual variance \(\sigma_0^2\) in \Cref{asm:predictive-model}.
These quantities may be specified as predictive-model parameters or estimated from the control sample; here we focus on estimation from the controls.

Under \Cref{asm:predictive-model}, the centred control outcomes satisfy
\[
    \EE_{\beta\sim\Pi,\epsilon_0}
    \left[
        y_0^c(y_0^c)^\top
        \mid \mathcal G_n
    \right]
    =
    \frac{r^2}{p}X_0^c(X_0^c)^\top
    +\sigma_0^2 C_0.
\]
Let \(m=n_0-1\), and choose \(Q_0\in\RR^{n_0\times m}\) such that \(Q_0^\top Q_0=I_m\) and \(Q_0Q_0^\top=C_0\).
Writing \(\frac1p Q_0^\top X_0^c(X_0^c)^\top Q_0=U_m\diag(d_1,\ldots,d_m)U_m^\top\) and defining \(\tilde y_0=U_m^\top Q_0^\top y_0^c\), we obtain \(\EE_{\beta\sim\Pi,\epsilon_0}(\tilde y_{0i}^2\mid\mathcal G_n)=r^2d_i+\sigma_0^2\) for \(i=1,\ldots,m\).
The signal contribution therefore varies with the source eigenvalue \(d_i\), while the residual contribution \(\sigma_0^2\) is common across spectral coordinates, motivating the use of variation in the source spectrum to separate the two variance components.

For \(\vartheta=(a,b)^\top\) in a fixed compact rectangle \(\Theta\subset\RR_{++}^2\), define the spectral criterion and a global minimiser
\begin{equation}
    \label{eq:spectral-quasi-likelihood}
    L_m(\vartheta)
    =
    \frac1m\sum_{i=1}^m
    \left\{
        \log(ad_i+b)
        +
        \frac{\tilde y_{0i}^2}{ad_i+b}
    \right\},
    \qquad
    \widehat\vartheta
    =
    (\widehat r^2,\widehat\sigma_0^2)^\top
    \in
    \argmin_{\vartheta\in\Theta}L_m(\vartheta).
\end{equation}
We use \(L_m\) as a spectral quasi-likelihood for estimating \((r^2,\sigma_0^2)^\top\); when \(\beta\) and \(\epsilon_0\) are Gaussian, it agrees, up to constants, with the conditional Gaussian negative log-likelihood.

To establish consistency of these estimates, we impose the following source-side conditions.

\begin{assumption}[Source model for variance estimation]
    \label{asm:variance-component-model}
    The conditions in \Cref{asm:rmt} involving the control design \(X_0\) hold, and \(p/n_0\to\phi_0\in(0,\infty)\).
    In addition to \Cref{asm:predictive-model}, for some \(\delta>0\), the following conditions hold:
    \begin{enumerate}[(a)]
        \item (Coefficient)
        The random-effects coefficient has the representation \(\beta=(r/\sqrt p)\,\xi_n\), with \(\EE\xi_{n,j}=0\), \(\EE\xi_{n,j}^2=1\), and \(\sup_{n,j}\EE|\xi_{n,j}|^{4+\delta}<\infty\), where the entries of \(\xi_n\) are independent and \(\xi_n\) is independent of \((\mathcal G_n,\epsilon_0)\).

        \item (Noise)
        Conditional on \(\mathcal G_n\), the entries of \(\epsilon_0\) are independent and \(\sup_{i\leq n_0}\EE(|\epsilon_{0i}|^{4+\delta}\mid\mathcal G_n)=\Op(1)\).
    \end{enumerate}
\end{assumption}

\begin{proposition}[Consistency of spectral quasi-likelihood]
    \label{prop:spectral-quasi-likelihood}
    Suppose \Cref{asm:variance-component-model} holds, and let
    \(\vartheta_0=(r^2,\sigma_0^2)^\top\in\Theta^\circ\).
    Then every global minimiser \(\widehat\vartheta\) satisfies \(\widehat\vartheta\pto\vartheta_0\) and \(\PP(\widehat\vartheta\in\Theta^\circ)\to1\).
\end{proposition}

For optional initialisation and diagnostics, Appendix~\ref{app:proof-variance-components} gives a consistent closed-form moment estimator, based on two aggregate spectral moments; see \Cref{prop:two-moment-initializer}.

\subsection{Target-aware risk estimation and tuning}
\label{sec:ood-risk-tuning}

With the variance components estimated, we next consider feasible estimation and tuning of the conditional prediction risk for the treated along the regularisation path.
When the base weights are constructed using target covariates, reusing the same target observations for risk evaluation creates dependence between weight construction and risk evaluation.
We therefore separate the target observations used for risk evaluation from the target information used to construct the base weights.

\begin{assumption}[Evaluation-sample separation]
    \label{asm:evaluation-sample}
    Let \(\mathcal H_n\subseteq\mathcal G_n\) contain \(X_0\) and all auxiliary or pilot target information used to construct the base weights, and let \(\mathcal I_{\mathrm E}\) and \(n_{1,\mathrm E}=|\mathcal I_{\mathrm E}|\) be \(\mathcal H_n\)-measurable.
    The normalised base weights \(\gamma\) are \(\mathcal H_n\)-measurable and satisfy \(\one_{n_0}^\top\gamma=1\).
    Conditional on \(\mathcal H_n\), the evaluation observations satisfy \((x_i)_{i\in \mathcal I_{\mathrm E}}\mid\mathcal H_n \sim P_{1,n}^{\otimes n_{1,\mathrm E}}\) almost surely.
    The target law \(P_{1,n}\) satisfies the target-side design and moment conditions in \Cref{asm:rmt}.
\end{assumption}

Under \Cref{asm:evaluation-sample}, let \(\bar x_{1,\mathrm E}=n_{1,\mathrm E}^{-1}\sum_{i\in\mathcal I_{\mathrm E}}x_i\) and \(\epsilon_{1,\mathrm E}=\bar x_{1,\mathrm E}-\nu_1\).
Then \Cref{asm:evaluation-sample} implies
\[
    \EE\left(\epsilon_{1,\mathrm E}\mid\mathcal H_n\right)=0,
    \qquad
    \EE\left(
        \epsilon_{1,\mathrm E}\epsilon_{1,\mathrm E}^{\top}
        \mid\mathcal H_n
    \right)
    =
    \frac{\Sigma_1}{n_{1,\mathrm E}}.
\]
The pilot/evaluation split in Appendix~\ref{app:honest-target-split} is one sufficient construction; if the base weights do not use target covariates, the full target sample may serve as the evaluation sample.
Throughout this subsection, write \(\bar x_1=\bar x_{1,\mathrm E}\) and \(n_1=n_{1,\mathrm E}\), and define \(\Delta=\bar x_1-X_0^\top\gamma\) and \(\gamma_\lambda=\gamma+n_0^{-1}X_0^cM_\lambda^c\Delta\).
Conditional on \(\mathcal H_n\), \(X_0\) and \(\gamma\) are fixed, while \(\bar x_1\), \(\Delta\), and \(\gamma_\lambda\) are random.

For each realised evaluation sample, the tuning target remains the \(\mathcal G_n\)-conditional prediction risk defined in \Cref{sec:risk:exact-risk}:
\begin{equation}
    \label{eq:random-effects-risk}
    R_n(\lambda)
    =
    \frac{r^2}{p}
    \left\|X_0^\top\gamma_\lambda-\nu_1\right\|_2^2
    +
    \sigma_0^2\|\gamma_\lambda\|_2^2.
\end{equation}
Conditioning on \(\mathcal H_n\) isolates the randomness of the evaluation sample in estimating this risk.
Since \(X_0^\top\gamma_\lambda-\bar x_1=-\lambda M_\lambda^c\Delta\), the conditional moments above give
\[
    \EE\left[
        \frac{\lambda^2}{p}
        \Delta^\top(M_\lambda^c)^2\Delta
        +
        \frac{1}{n_1}
        \otr\!\left\{
            (I_p-2\lambda M_\lambda^c)\Sigma_1
        \right\}
        \,\middle|\,\mathcal H_n
    \right]
    =
    \EE\left[
        \frac{1}{p}
        \left\|X_0^\top\gamma_\lambda-\nu_1\right\|_2^2
        \,\middle|\,\mathcal H_n
    \right].
\]
Thus, the trace term corrects for sampling variation in \(\bar x_1\), so the corrected quadratic term targets the squared-bias component of \(R_n(\lambda)\).

Let \(\widehat\Sigma_1=(n_1-1)^{-1}\sum_{i\in\mathcal I_{\mathrm E}}(x_i-\bar x_1)(x_i-\bar x_1)^\top\) be the evaluation-sample covariance.
Using the variance-component estimates from \Cref{prop:spectral-quasi-likelihood}, define
\begin{equation}
    \label{eq:ood-risk-estimator}
    \widehat R_n(\lambda;\gamma)
    =
    \widehat r^2
    \left[
        \frac{\lambda^2}{p}
        \Delta^\top(M_\lambda^c)^2\Delta
        +
        \frac{1}{n_1}
        \otr\!\left\{
            (I_p-2\lambda M_\lambda^c)\widehat\Sigma_1
        \right\}
    \right]_+
    +
    \widehat\sigma_0^2\|\gamma_\lambda\|_2^2.
\end{equation}
The positive part prevents a negative finite-sample estimate of the squared-bias component and has asymptotically negligible effect; with the spectral quasi-likelihood plug-in, \(\widehat R_n(\lambda;\gamma)>0\) because \(\widehat\sigma_0^2>0\) and \(\|\gamma_\lambda\|_2^2\geq n_0^{-1}\).
The criterion is target-aware: its quadratic term follows the realised residual imbalance, while the trace term corrects for sampling variation in \(\bar x_1\).

Evaluation-sample separation controls dependence between weight construction and risk evaluation, but places no rate restriction on residual imbalance or weight concentration.
We therefore impose a separate rate condition on these two quantities.
Here \(\eta\in[0,1]\) indexes the order of the resulting risk path; in \Cref{sec:risk:asymptotic-risk}, the same index also parameterises the coupled source-target mean-shift and weight-concentration regime.

\begin{assumption}[Base-weight rate conditions]
    \label{asm:risk-rate-admissibility}
    For some \(\eta\in[0,1]\), \((n_0^\eta/p)\|\nu_1-X_0^\top\gamma\|_2^2=\Op(1)\) and \(n_0^\eta\|C_0\gamma\|_2^2=\Op(1)\).
\end{assumption}

Because \(\gamma\) is normalised, \(\|\gamma\|_2^2=\|C_0\gamma\|_2^2+n_0^{-1}\), so the centred and uncentred weight-concentration rates are equivalent for \(\eta\in[0,1]\).
Together with \Cref{asm:evaluation-sample}, \Cref{asm:risk-rate-admissibility} allows the base weights to depend on the realised source design and permitted target information.
As a special case, \Cref{asm:rmt,asm:design-independent-regime} imply \Cref{asm:risk-rate-admissibility}.
Together with the preceding design and evaluation-sample conditions, it yields the uniform rate control along a regularisation~path.

\begin{theorem}[Uniform risk estimation and consistency of risk-based tuning]
    \label{thm:risk-estimation}
    Suppose \Cref{asm:predictive-model,asm:rmt,asm:evaluation-sample,asm:risk-rate-admissibility} hold for some \(\eta\in[0,1]\).
    Assume \(p/n_{1,\mathrm E}\to\phi_{1,\mathrm E}\in(0,\infty)\), and let \(\Lambda\) be either a fixed finite grid or a compact interval contained in \([\lambda_-,\lambda_+]\subset(0,\infty)\).
    Suppose \(\widehat r^2\pto r^2\) and \(\widehat\sigma_0^2\pto\sigma_0^2\), and define \(\widehat\lambda=\min\argmin_{\lambda\in\Lambda}\widehat R_n(\lambda;\gamma)\).
    Then the following holds:
    \begin{enumerate}[(a)]
        \item\label{thm:risk-estimation-i}
        \(
            n_0^\eta
            \sup_{\lambda\in\Lambda}
            \left|
                \widehat R_n(\lambda;\gamma)-R_n(\lambda)
            \right|
            \pto 0.
        \)

        \item\label{thm:risk-estimation-ii}
        \(
            n_0^\eta
            \left\{
                R_n(\widehat\lambda)
                -
                \inf_{\lambda\in\Lambda}R_n(\lambda)
            \right\}
            \pto 0.
        \)

        \item\label{thm:risk-estimation-iii}
        If further for a deterministic continuous function \(R_{\eta,\star}\),
        \(
            \sup_{\lambda\in\Lambda}
            \left|
                n_0^\eta R_n(\lambda)-R_{\eta,\star}(\lambda)
            \right|
            \pto 0
        \) 
        and
        \(
            M
            =
            \argmin_{\lambda\in\Lambda}
            R_{\eta,\star}(\lambda)
        \),
        then
        \(
            \dist(\widehat\lambda,M)\pto0.
        \)
        If \(R_{\eta,\star}\) has a unique minimiser \(\lambda_{\eta,\star}^\ast\), then
        \(
            \widehat\lambda\pto\lambda_{\eta,\star}^\ast.
        \)
    \end{enumerate}
\end{theorem}

\Cref{thm:risk-estimation} shows that \(\widehat R_n(\lambda;\gamma)\) uniformly estimates the realised conditional prediction risk for the treated along the regularisation path with error \(\op(n_0^{-\eta})\).
Consequently, minimising the estimated path yields oracle excess risk \(\op(n_0^{-\eta})\).
If the scaled risk path converges uniformly to a deterministic limit, the selected penalty approaches its minimiser set; when the minimiser is unique, \(\widehat\lambda\) is consistent for that penalty.

\Cref{thm:risk-estimation} applies to base weights adapted to the realised source covariates and pilot target information when \Cref{asm:evaluation-sample,asm:risk-rate-admissibility} hold; sufficient conditions for common balancing procedures are given in Appendix~\ref{app:sample-split-rate-verification}.
\Cref{prop:ridge-profile-verification} treats same-sample target-adaptive ridge weights separately, and \Cref{cor:target-split-ridge-post-tuning} verifies the deterministic risk limit required for penalty consistency for sample-split ridge base weights.

\begin{remark}[No-augmentation endpoint]
    \label{rem:no-augmentation-endpoint}
    Set \(\Lambda^+=\Lambda\cup\{\infty\}\), with \(R_n(\infty)\) defined by \eqref{eq:random-effects-risk} evaluated at \(\gamma_\infty:=\gamma\).
    Define
    \(
        \widehat R_n(\infty;\gamma)
        :=
        \widehat r^2
        \left[
            \frac{1}{p}\|\Delta\|_2^2
            -
            \frac{1}{n_1}\otr(\widehat\Sigma_1)
        \right]_+
        +
        \widehat\sigma_0^2\|\gamma\|_2^2.
    \)
    This is the \(\lambda\to\infty\) limit of \eqref{eq:ood-risk-estimator} and leaves the base weights unchanged.
    Under the conditions of \Cref{thm:risk-estimation}, if \(\widehat\lambda=\min\argmin_{\lambda\in\Lambda^+}\widehat R_n(\lambda;\gamma)\), then \Cref{thm:risk-estimation}~\eqref{thm:risk-estimation-i}--\eqref{thm:risk-estimation-ii} remain valid with \(\Lambda\) replaced by \(\Lambda^+\).
\end{remark}

The resulting risk-calibrated balancing (RCB) procedure is summarised in \Cref{alg:ood-tuning}.
For a tuning grid, the spectral decomposition of \(S_0^c\) can be reused to compute \(M_\lambda^c\Delta\), \(\Delta^\top(M_\lambda^c)^2\Delta\), the augmented weights, and the trace correction along the full path, so a separate matrix inversion is unnecessary for each \(\lambda\).
Henceforth, \(\widehat R_n(\lambda)\) abbreviates \(\widehat R_n(\lambda;\widehat\gamma)\).

\begin{algorithm}[!t]
\caption{Risk-calibrated balancing (RCB)}
\label{alg:ood-tuning}
    \small
    \begin{algorithmic}[1]
    \Require Control data \((X_0,y_0)\); an evaluation treated covariate sample; the full treated-sample outcome mean \(\bar y_1\); normalised base weights \(\widehat\gamma\) constructed without the evaluation sample; candidate grid \(\Lambda\subset(0,\infty)\).
    \Ensure Selected penalty \(\widehat\lambda\), counterfactual estimate \(\widehat\mu_{0,\widehat\lambda}\), ATT estimate \(\widehat\tau_{\widehat\lambda}\), and augmented weights \(\widehat\gamma_{\widehat\lambda}\).
    \State Compute \(\bar x_0\), \(X_0^c\), and one spectral decomposition of \(S_0^c=(X_0^c)^\top X_0^c/n_0\).
    \State Estimate \((\widehat r^2,\widehat\sigma_0^2)\) by \eqref{eq:spectral-quasi-likelihood}, optionally initialised by the two-moment estimator in \Cref{prop:two-moment-initializer}.
    \State Compute \(\bar x_1\), \(\widehat\Sigma_1\), and \(\Delta=\bar x_1-X_0^\top\widehat\gamma\).
    \For{\(\lambda\in\Lambda\)}
        \State Compute \(u_\lambda=M_\lambda^c\Delta\) using the spectral decomposition.
        \State Set \(\widehat\gamma_\lambda=\widehat\gamma+X_0^cu_\lambda/n_0\).
        \State Evaluate \(\widehat R_n(\lambda;\widehat\gamma)\) using \eqref{eq:ood-risk-estimator}, with \(\Delta^\top(M_\lambda^c)^2\Delta=\|u_\lambda\|_2^2\).
    \EndFor
    \State Select
    \[
        \widehat\lambda
        =
        \min\argmin_{\lambda\in\Lambda}
        \widehat R_n(\lambda;\widehat\gamma).
    \]
    \State Set \(\widehat\mu_{0,\widehat\lambda}=\widehat\gamma_{\widehat\lambda}^{\top}y_0\) and \(\widehat\tau_{\widehat\lambda}=\bar y_1-\widehat\mu_{0,\widehat\lambda}\).
    \end{algorithmic}
\end{algorithm}

\subsection{Model-based prediction intervals for the counterfactual mean}
\label{subsec:predictive-calibration}

Having constructed a uniformly consistent estimator of the conditional prediction risk for the treated along the regularisation path in \Cref{thm:risk-estimation}, we now use this risk to quantify predictive uncertainty for the counterfactual mean \(\mu_{0,\beta}\) under the predictive law.
For any deterministic or \(\mathcal G_n\)-measurable \(\lambda\), \Cref{prop:conditional-risk} gives
\[
    \EE_{\beta\sim\Pi,\epsilon_0}
    \left(
        \widehat\mu_{0,\lambda}-\mu_{0,\beta}
        \mid \mathcal G_n
    \right)=0,
    \qquad
    \Var_{\beta\sim\Pi,\epsilon_0}
    \left(
        \widehat\mu_{0,\lambda}-\mu_{0,\beta}
        \mid \mathcal G_n
    \right)
    =R_n(\lambda).
\]
Thus \(R_n(\lambda)^{1/2}\) is the conditional predictive standard deviation, with \(\widehat R_n(\lambda)^{1/2}\) as its feasible estimate.
For a fixed penalty, the same conditional prediction risk used to assess regularisation therefore also quantifies predictive uncertainty in the counterfactual mean.

For deterministic or \(\mathcal G_n\)-measurable \(\lambda\), conditioning on \(\mathcal G_n\) fixes the coefficients in the prediction error \(\widehat\mu_{0,\lambda}-\mu_{0,\beta}\).
Under the default plug-in implementation, however, \(\widehat R_n(\lambda)\) and its minimiser \(\widehat\lambda\) generally depend on \(y_0\) through \(\widehat r^2\) and \(\widehat\sigma_0^2\).
Thus \(\widehat\lambda\) need not be \(\mathcal G_n\)-measurable, so the conditional normal pivot for a fixed or \(\mathcal G_n\)-measurable penalty does not extend to \(\widehat\lambda\) by direct substitution.
The result below therefore treats fixed and selected penalties separately.

\begin{assumption}[Gaussian predictive law]\label{asm:gaussian-predictive}
    Conditional on \(\mathcal G_n\), \(\beta\sim\Pi=\mathcal N(0,r^2I_p/p)\) and \(\epsilon_0\sim\mathcal N(0,\sigma_0^2I_{n_0})\), with \(\beta\) and \(\epsilon_0\) conditionally independent.
\end{assumption}

\begin{theorem}[Prediction intervals for the counterfactual mean]
    \label{thm:predictive-calibration}
    Suppose \Cref{asm:predictive-model,asm:gaussian-predictive} hold.
    \begin{enumerate}[(a)]
        \item\label{thm:predictive-calibration-i} \emph{Fixed penalty.}
        For any deterministic or \(\mathcal G_n\)-measurable \(\lambda\in\Lambda\), for every \(n\),
        \[
            \left.
            \frac{\widehat\mu_{0,\lambda}-\mu_{0,\beta}}
                 {R_n(\lambda)^{1/2}}
            \,\right|\,\mathcal G_n
            \sim\cN(0,1).
        \]
        If, in addition, \(\widehat R_n(\lambda)/R_n(\lambda)\pto1\), then \(\PP\{\widehat R_n(\lambda)>0\}\to1\) and
        \[
            \sup_{t\in\RR}
            \left|
            \PP\left(
                \frac{\widehat\mu_{0,\lambda}-\mu_{0,\beta}}
                     {\widehat R_n(\lambda)^{1/2}}
                \leq t
                \,\middle|\,
                \mathcal G_n
            \right)-\Phi(t)
            \right|\pto0.
        \]

        \item\label{thm:predictive-calibration-ii} \emph{Selected penalty.}
        Suppose, in addition, that the conditions of \Cref{thm:risk-estimation} hold for some \(\eta\in[0,1]\), with \(\Lambda=[\lambda_-,\lambda_+]\subset(0,\infty)\) compact.
        Let \(\widehat\lambda\) be the minimiser defined in \Cref{thm:risk-estimation}, which need not be \(\mathcal G_n\)-measurable.
        Suppose further that, for a deterministic continuous function \(\sR_{\eta,\star}\), \(\sup_{\lambda\in\Lambda}|n_0^\eta R_n(\lambda)-\sR_{\eta,\star}(\lambda)|\pto0\).
        If \(\sR_{\eta,\star}\) has a unique minimiser \(\lambda_{\eta,\star}^*\) satisfying \(\sR_{\eta,\star}(\lambda_{\eta,\star}^*)>0\), then \(\PP\{\widehat R_n(\widehat\lambda)>0\}\to1\) and
        \[
            \sup_{t\in\RR}
            \left|
            \PP\left(
                \frac{\widehat\mu_{0,\widehat\lambda}-\mu_{0,\beta}}
                     {\widehat R_n(\widehat\lambda)^{1/2}}
                \leq t
                \,\middle|\,
                \mathcal G_n
            \right)-\Phi(t)
            \right|\pto0.
        \]        
    \end{enumerate}
\end{theorem}

For \Cref{thm:predictive-calibration}~\eqref{thm:predictive-calibration-i}, \Cref{thm:risk-estimation}~\eqref{thm:risk-estimation-i} implies \(\widehat R_n(\lambda)/R_n(\lambda)\pto1\) whenever \(n_0^\eta R_n(\lambda)\pto c\) for some constant \(c>0\).
For \Cref{thm:predictive-calibration}~\eqref{thm:predictive-calibration-ii}, the selected-penalty prediction interval \(\left[\widehat\mu_{0,\widehat\lambda}\pm z_{1-\alpha/2}\widehat R_n(\widehat\lambda)^{1/2}\right]\) has conditional predictive coverage converging in probability to \(1-\alpha\).
The studentised statistic may be defined arbitrarily on \(\{\widehat R_n(\lambda)\leq0\}\) or \(\{\widehat R_n(\widehat\lambda)\leq0\}\), whose probability tends to zero.

The interval in \Cref{thm:predictive-calibration} is a model-based prediction interval for the counterfactual mean \(\mu_{0,\beta}\) under \(\Pi\).
It does not provide frequentist confidence coverage for the fixed causal estimand \(\mu_0\) or \(\tau_{\mathrm{ATT}}\).
Confidence inference for \(\tau_{\mathrm{ATT}}\) would require a joint sampling theory for \(\bar y_1-\mu_1\) and \(\widehat\mu_{0,\widehat\lambda}-\mu_0\), including their covariance and the effect of data-adaptive tuning.
For deterministic or \(\mathcal G_n\)-measurable penalties, \Cref{thm:predictive-normality-nongaussian} gives a non-Gaussian coefficient extension under an envelope condition while retaining Gaussian residual noise.

\section{Numerical Experiments}\label{sec:num-exp}
This section uses a common Gaussian source-target design.
The source and target covariates follow \(x_{0i}=\Sigma_0^{1/2}z_{0i}\) and \(x_{1i}=\nu+\Sigma_1^{1/2}z_{1i}\), where \(z_{0i},z_{1i}\sim\cN(0,I_p)\) are mutually independent and \(\Sigma_1=\Sigma_0\) is the AR(1) covariance \((\Sigma_0)_{jk}=0.5^{|j-k|}\).
Conditional on the covariates, the source outcomes follow \(y_{0i}=x_{0i}^{\top}\beta+\varepsilon_{0i}\), where \(\beta\sim\cN(0,r^2I_p/p)\), \(\varepsilon_0\sim\cN(0,\sigma_0^2I_{n_0})\), \(r^2=1\), and \(\sigma_0^2=0.5\).
We evaluate the conditional prediction risk \(R_n(\lambda)=B_n(\lambda)+V_n(\lambda)\) and its components as defined in \Cref{prop:conditional-risk}.

\subsection{Risk paths in the design-independent setting}\label{subsec:num-benchmark}

We use design-independent base weights satisfying the conditions of \Cref{thm:risk-asym}.
Write the source-target mean shift as \(\nu=\nu_{1,n}=(\rho_\eta^2p/n_0^\eta)^{1/2}u\), where \(\|u\|_2=1\), and set \(\gamma=n_0^{-1}\one_{n_0}+(\varrho_\eta^2/n_0^\eta)^{1/2}q\), where \(\one_{n_0}^{\top}q=0\) and \(\|q\|_2=1\).
Then \((n_0^\eta/p)\|\nu_{1,n}\|_2^2=\rho_\eta^2\), \(n_0^\eta\|C_0\gamma\|_2^2=\varrho_\eta^2\), and \(\one_{n_0}^{\top}\gamma=1\), so \(\rho_\eta^2\) and \(\varrho_\eta^2\) set the mean-shift energy and centred weight concentration.
For \Cref{fig:risk}, we set \(p=640\), \(n_0=n_1=853\), use \(\eta\in\{0,0.5,1\}\) and \((\rho_\eta^2,\varrho_\eta^2)\in\{(1,0),(0,1),(1,1)\}\), and align the mean shift with the largest-eigenvalue eigenvector of \(\Sigma_0\).
The figure compares the exact scaled risk components from \Cref{prop:conditional-risk}, their deterministic equivalents from \Cref{thm:risk-asym}, and the feasible criterion \(\widehat R_n\); the deterministic equivalents closely track the exact paths, and the feasible criterion identifies the same low-risk regions.

\begin{figure}[!t]
    \centering
    \includegraphics[width=0.9\linewidth]{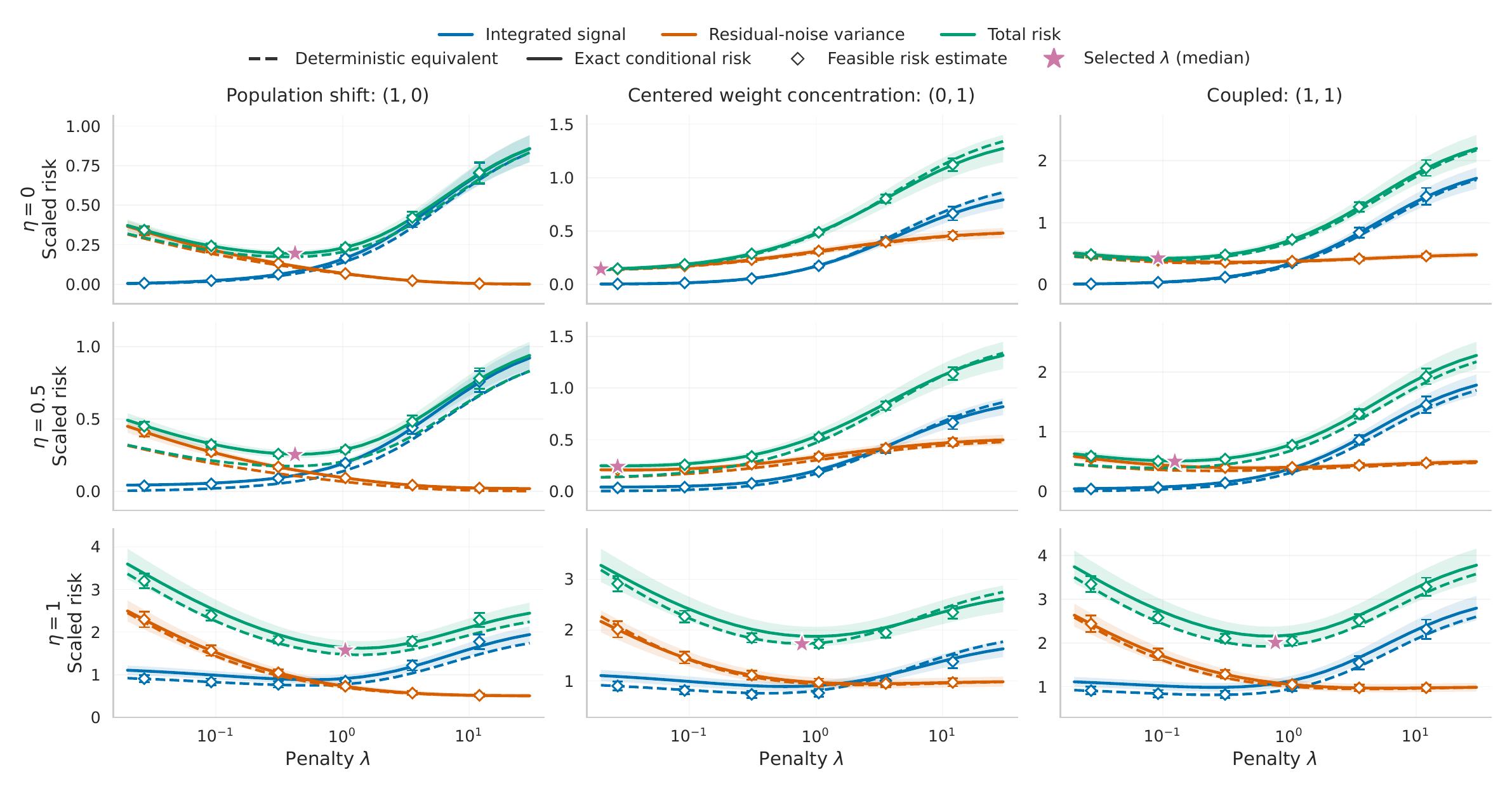}
    \caption{
AR(1) risk paths and target-aware tuning with the source-target mean shift aligned with the largest-eigenvalue eigenvector direction of the source covariance.
Rows index \(\eta\in\{0,0.5,1\}\), and columns set \((\rho_\eta^2,\varrho_\eta^2)\) to \((1,0)\), \((0,1)\), and \((1,1)\).
All risks are scaled by \(n_0^\eta\).
Dashed curves are the deterministic equivalents from \Cref{thm:risk-asym}, solid curves are the exact conditional quantities from \Cref{prop:conditional-risk}, and diamonds are feasible plug-in estimates.
Stars summarise the median selected penalty over 200 replications.
Uncertainty conventions are given in Appendix~\ref{subsec:mc-conventions}.}
    \label{fig:risk}
\end{figure}

With the mean shift spread across the eigenvector directions of \(\Sigma_0\), \Cref{fig:aspect ratio risk} shows that source and target aspect ratios affect risk differently: at \(\eta=0\), total risk is essentially invariant to \(\phi_1\), as predicted by \Cref{cor:target-sampling-invariant}, while at \(\eta=1\), it increases with \(\phi_1\) as target empirical-mean variation enters at leading order.
The source aspect ratio \(\phi_0\) affects risk in both regimes, particularly when the mean-shift component dominates; the isotropic controls in Appendix~\ref{subsec:isotropic-controls} show the same qualitative pattern.

\begin{figure}[!t]
    \centering
    \includegraphics[width=0.9\linewidth]{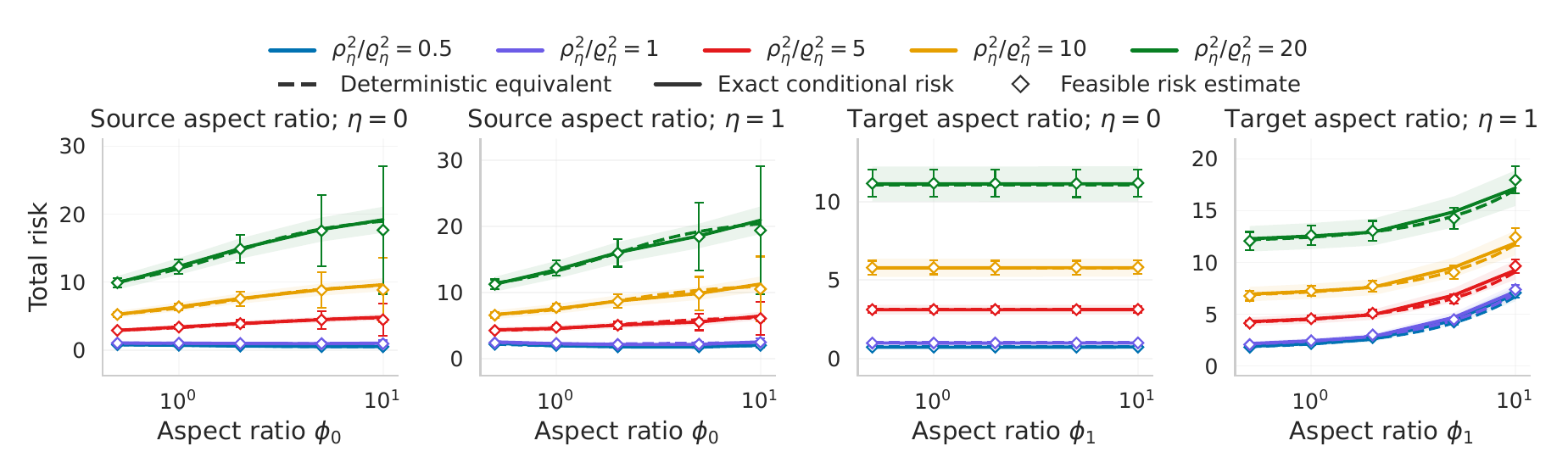}
    \caption{
Source-target aspect-ratio dependence under the AR(1) design.
We fix \(p=640\), \(\lambda=1\), and \(\varrho_\eta^2=1\), vary either \(\phi_0\) or \(\phi_1\), and hold the other at \(0.75\).
Colors index \(\rho_\eta^2/\varrho_\eta^2\).
Dashed, solid, and diamond curves denote the deterministic equivalent, exact conditional risk, and feasible risk estimate, respectively.}
    \label{fig:aspect ratio risk}
\end{figure}
At \(p=320\), the selected studentised prediction errors from \Cref{thm:predictive-calibration} are broadly close to the standard normal reference across \(\eta\in\{0,0.5,0.75,1\}\); details are reported in Appendix~\ref{subsec:exp-predictive}.

\begin{figure}[!b]
    \centering
    \includegraphics[width=1\linewidth]{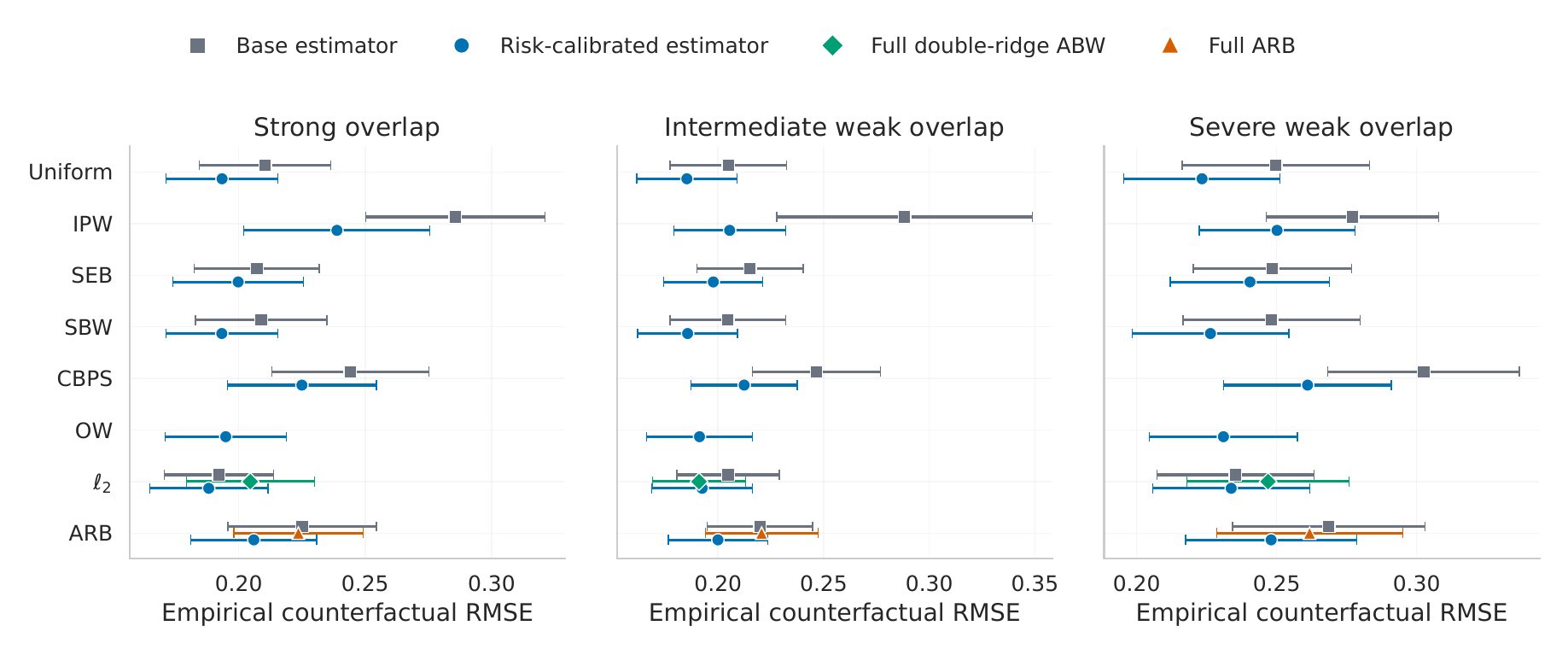}
    \caption{Counterfactual prediction RMSE for the complete estimators in the overparameterised source design \(p=50\), \(n_0=40\), and \(n_1=100\).
    Rows show the eight initial weighting rules alone and with target-aware ridge augmentation; the \(\ell_2\) and ARB rows additionally show the default full double-ridge ABW and the full ARB estimators built on those weighting components.
    Columns correspond to strong, intermediate, and severe weak overlap.
    Points are root mean squared errors over 120 Monte Carlo replications and bars are \(\pm2\) Monte Carlo standard errors.}
    \label{fig:complete-estimator-rmse}
\end{figure}

\subsection{Adaptive balancing and target-aware tuning}\label{subsec:num-adaptive}

We next use base weights estimated from the realised source and pilot-target covariates.
We compare Uniform, IPW \citep{lunceford2004ps}, stabilised entropy balancing (SEB) \citep{hainmueller2012entropy}, stable balancing weights (SBW) \citep{zubizarreta2015stable}, CBPS \citep{imai2014cbps}, overlap weights (OW) \citep{li2018overlapweights}, the \(\ell_2\) balancing component of augmented balancing weights (ABW) \citep{brunssmith2025augmented}, and approximate residual balancing (ARB) \citep{athey2018arb}, pairing each with target-aware ridge augmentation.
The complete double-ridge ABW and ARB estimators are also included; for the \(\ell_2\) base, the augmented path coincides algebraically with the double-ridge path, so their comparison concerns penalty selection.
For the \(\ell_2\) base, the balancing penalty is selected from a fixed positive grid using only the source design and pilot target information, and \Cref{cor:target-split-ridge-selected-alpha} gives uniformly consistent risk estimation and vanishing scaled oracle excess risk.
For the remaining procedures, we report finite-sample performance; simulation conventions, tuning grids, and additional diagnostics are given in Appendices~\ref{subsec:mc-conventions}, \ref{subsec:sup-num}, and~\ref{subsec:response-robustness}.

\begin{figure}[!t]
    \centering
    \begin{minipage}[t]{0.9\linewidth}
        \centering
        \includegraphics[width=\linewidth]{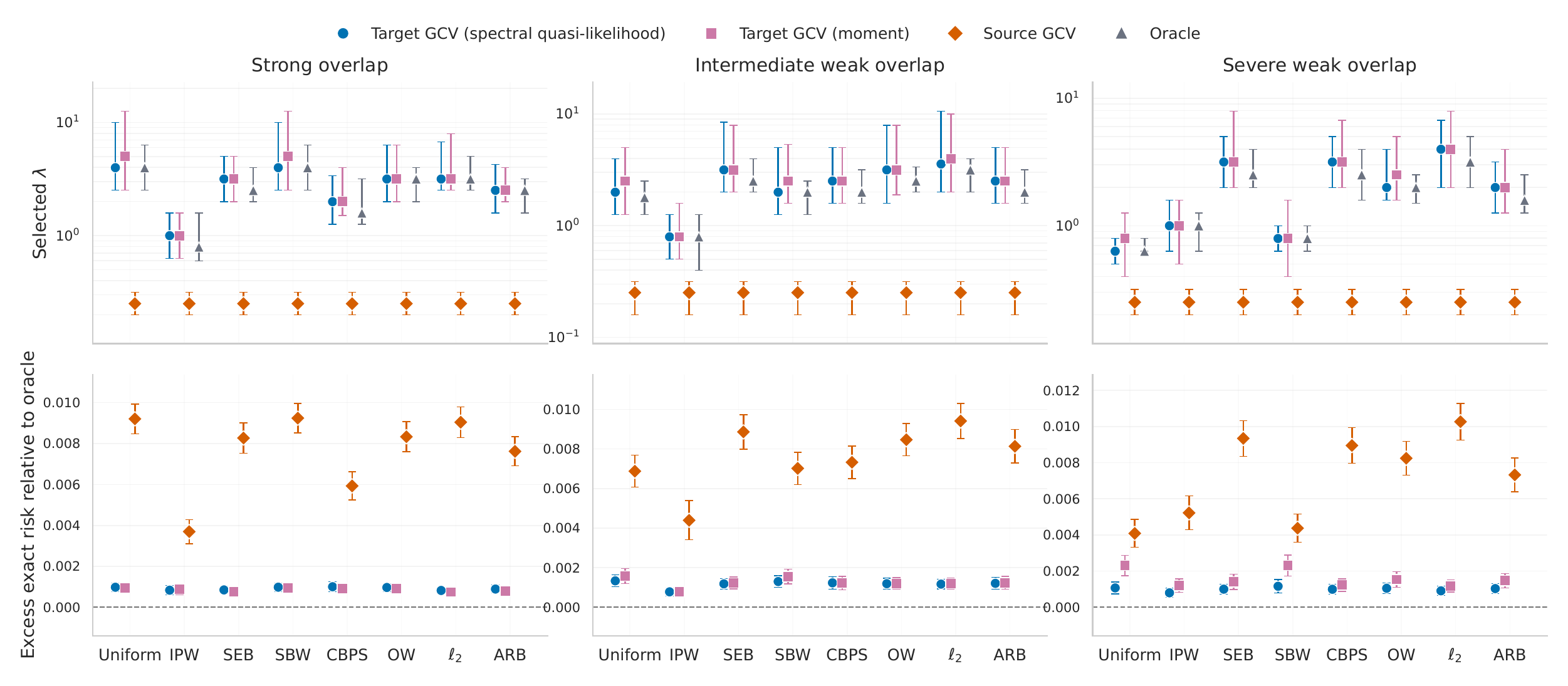}
        \par\smallskip\small\textbf{(a)} Underparameterised source design: \(p/n_0=0.5\).
    \end{minipage}
    \par\medskip
    \begin{minipage}[t]{0.9\linewidth}
        \centering
        \includegraphics[width=\linewidth]{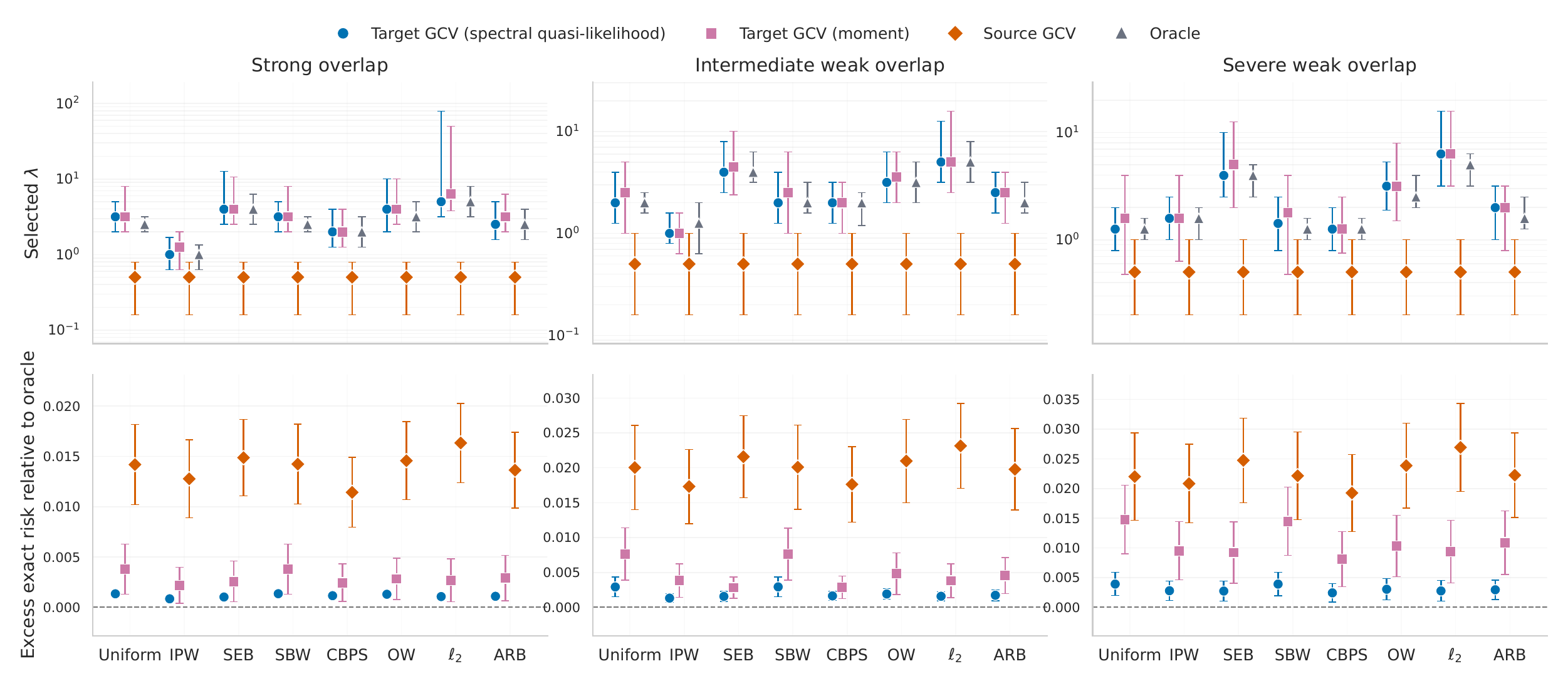}
        \par\smallskip\small\textbf{(b)} Overparameterised source design: \(p/n_0=1.25\).
    \end{minipage}
    \caption{Target-aware tuning compared with source-only GCV and the infeasible oracle.
    Panels (a) and (b) correspond to source aspect ratios \(p/n_0=0.5\) and \(p/n_0=1.25\), respectively.
    Columns correspond to strong, intermediate, and severe weak overlap.
    The upper row reports the median selected penalty with interquartile-range bars.
    The lower row reports mean excess exact risk relative to the infeasible oracle, with \(\pm2\) Monte Carlo standard-error bars, over 120 replications.}
    \label{fig:exp_tuning}
\end{figure}

\noindent\textbf{Experimental setup.}
We set \(p=50\), \(n_1=100\), and \(n_0\in\{100,40\}\), giving \(\phi_0\in\{0.5,1.25\}\), and split the target sample equally into pilot and evaluation folds.
The mean shift is spread across the source-covariance eigenvector directions and scaled so that \(\nu_\delta^\top\Sigma_0^{-1}\nu_\delta=\delta^2\), with \(\delta\in\{0.5,1,\sqrt{3}\}\) corresponding to strong, intermediate, and severe weak overlap; the treatment effect is fixed at \(\tau=1\), and each configuration uses 120 Monte Carlo replications.

\noindent\textbf{Adaptive balancing performance.}
\Cref{fig:complete-estimator-rmse} compares complete-estimator RMSE.
Across the three overlap regimes, target-aware augmentation of the \(\ell_2\) and ARB weighting components has similar RMSE to the corresponding complete ABW and ARB estimators, with no uniform ordering.
Across the displayed weighting families, augmentation reduces integrated squared bias and generally lowers total risk by trading residual extrapolation against weight concentration; residual imbalance and weight concentration provide complementary diagnostics, and both selected-weight concentration and evaluation-sample imbalance increase as overlap weakens (\Cref{fig:comparison-phi-1-25}).

\noindent\textbf{Target-aware tuning.}
For comparison, ordinary generalised cross-validation (GCV) evaluates prediction under the source covariate distribution; we denote its selected penalty by \(\widehat\lambda_{\mathrm{GCV}}\).
\Cref{fig:exp_tuning} compares the spectral target-aware rule with the two-moment variant, source GCV, and the infeasible conditional-risk oracle.
The spectral rule generally selects penalties near the oracle, while source GCV selects substantially smaller penalties and incurs larger target risk.
The spectral rule attains the smallest mean excess exact risk in 39 of the 48 displayed configurations; the remaining cases occur in the underparameterised design and favour the two-moment rule.
Across the displayed adaptive experiments, target-aware tuning generally yields lower excess target risk than source GCV.

\section{Real-Data Applications}
\label{sec:real-data}

We use two real-data applications to study target-specific extrapolation.

\subsection{LaLonde job-training study}
\label{subsec:lalonde}

We analyse the National Supported Work Demonstration (NSW)--Panel Study of Income Dynamics (PSID) LaLonde data \citep{lalonde1986evaluating,dehejia1999causal} using the 171-feature specification of \citet{farrell2015robust} and \citet{brunssmith2025augmented}, with \(n_0=727\) PSID controls, \(n_1=185\) NSW treated units, and \(p=171\); the outcome is 1978 earnings.
We apply RCB with uniform, \(\ell_2\), or regularised-entropy base weights, using the pilot/evaluation split for target-adaptive bases.
Implementation details and diagnostics are reported in Appendix~\ref{app:real-data}.

Across the three base-weight choices, the RCB estimates of ATT are positive and similar in magnitude to the randomised NSW reference estimate; numerical estimates and diagnostics are given in \Cref{tab:lalonde171-full}.
Because untreated outcomes for the NSW treated group are unobserved, we also use the 260 randomised NSW controls as a pseudo-target with held-out outcomes: RCB attains similar RMSE across base weights, far below uniform weighting and comparable to regularised-entropy weighting, although above \(\ell_2\) weighting.

As comparators, the complete double-ridge ABW estimator of \citet{brunssmith2025augmented} tunes its outcome-ridge penalty by outcome cross-validation and sets its balancing penalty equal to it (outcome CV) or selects it by held-out imbalance (imbalance CV) or Riesz-loss (Riesz CV) cross-validation.
We next consider moderate and strong controlled shifts, corresponding to approximately \(4\times\) and \(16\times\) the baseline classifier-induced source-target discrepancy while holding the target covariates and outcome surface fixed.
Under the moderate shift, the uniform- and \(\ell_2\)-base RCB estimators have similar RMSE to imbalance-CV double ridge, and under the strong shift both have lower RMSE.
Relative to its own base, RCB reduces RMSE by \(28\%\) and \(37\%\) for \(\ell_2\) weighting and by about \(80\%\) for uniform weighting under the moderate and strong shifts, so its benefit grows with the severity of extrapolation.
Regularised-entropy weights are already closely balanced but highly concentrated (maximum standardised mean difference \(0.30\) and ESS \(14.9\) of 727 controls in \Cref{tab:lalonde171-full}), and RCB selects a large penalty that leaves them nearly unchanged.

\begin{figure}[htbp]
\centering
\includegraphics[width=\textwidth]{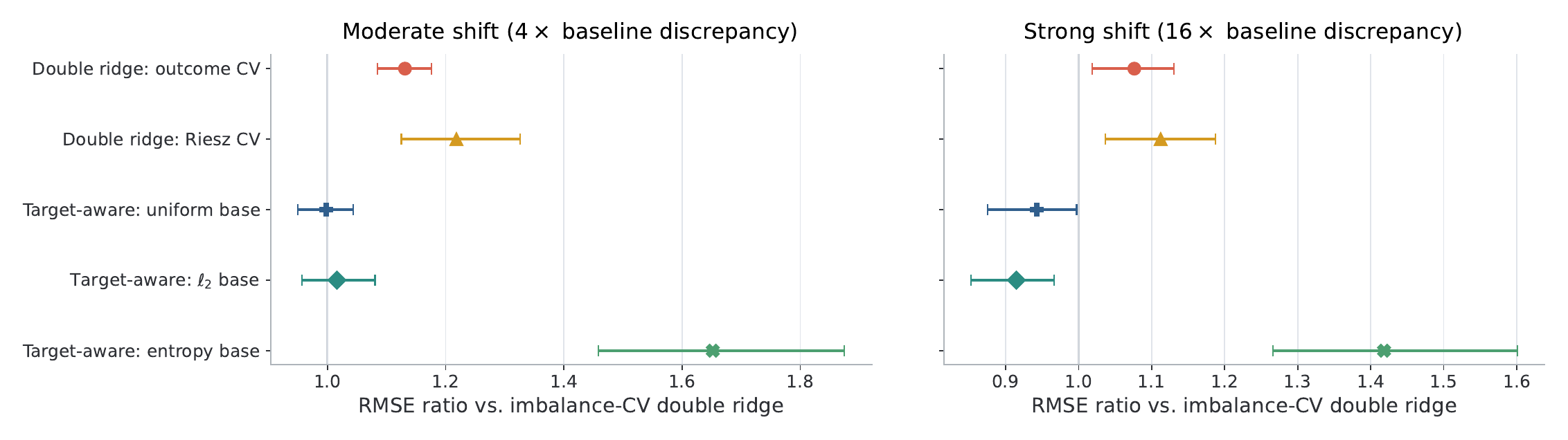}
\caption{
LaLonde controlled overlap-sensitivity analysis.
Points show paired RMSE ratios relative to imbalance-CV double ridge over 120 common replications, with 95\% Monte Carlo bootstrap intervals.
RCB with uniform or \(\ell_2\) base weights shows no detectable gain under the moderate controlled shift but lower RMSE under the strong controlled shift.}
\label{fig:r9_chi2}
\end{figure}

\subsection{K562 Perturb-seq extrapolation study}\label{subsec:crispr}

We apply RCB to the K562 Essential Perturb-seq data of \citet{replogle2022perturbseq}, which measure genome-wide transcriptional responses to genetic perturbations, to study source-to-target extrapolation in high-dimensional gene-expression space.
Each observation is a pseudobulk profile obtained by averaging \(\log(1+\text{count})\) over cells within one \((\text{perturbation},\text{batch})\) group, so batch-level variation enters both the covariates and outcomes, where covariate adjustment can partly account for it.
For each perturbation \(P\) and response gene \(G\), let \(y_G\) denote pseudobulk expression, \(x\in\mathbb R^p\) the covariates, and, for \(a\in\{0,P\}\), define \(m_{aG}(x)=\EE(y_G\mid x,a)\) and \(\tau_{PG}(x)=m_{PG}(x)-m_{0G}(x)\).
We use \(p=500\) standardised expression features screened for limited perturbation-associated differential expression and excluding perturbation and outcome genes; the screening rule and implementation details are given in Appendix~\ref{app:k562-implementation}.
Because pre-perturbed cell states are unobserved, these features are necessarily measured after perturbation, so we interpret the resulting quantities primarily as source-target predictive response contrasts.
A causal ATT interpretation would additionally require the selected features to be unaffected by perturbation and to satisfy the identifying conditions in \Cref{asm:causal-identification}.

\noindent\textbf{Overlap-sensitivity analysis.}
For each perturbation-outcome pair, we repeatedly sample 24 perturbation observations as the target and split them equally into pilot and evaluation folds, so that \(n_{1,\mathrm P}=n_{1,\mathrm E}=12\) and \(p/n_{1,\mathrm E}\approx41.7\); the observed K562 controls form the source population in the good-overlap regime.
For the weak-overlap regime, we resample the controls using a classifier-based covariate tilt constructed from the pilot target while retaining the observed outcomes, increasing the discrepancy to approximately \(16\) times its baseline level and reducing the average effective source sample size from \(48.0\) to \(10.8\).
Adaptive base weights are constructed using the pilot fold, and target-aware tuning uses the held-out fold.
Panels (a)--(b) of \Cref{fig:crispr_k562_fixed_panel} compare pair-level RMSE for each base estimator and its RCB augmentation under good and weak overlap.
Relative to its base, RCB reduces RMSE for most perturbation-outcome pairs with \(\ell_2\) weighting and the three double-ridge rules in both regimes (\(72\%\)--\(96\%\) of the 125 pairs), changes little for regularised-entropy weighting, and improves on uniform weighting only under weak overlap, where the uniform-base selection also moves away from the no-augmentation endpoint.
The balance diagnostics in \Cref{tab:k562-balance} explain this pattern: the base procedures concentrate the weights on about 4--13 effective controls yet reduce the imbalance relative to the mean of all perturbed profiles by at most about \(2\%\) under good overlap and \(4\%\) under weak overlap, whereas uniform-base RCB reduces it by about \(4\%\) and \(5\%\) while retaining 48 and 28 effective controls.
Because 48 controls span at most 47 of the 500 covariate directions, reweighting cannot remove most of the shift, and uniform-base RCB improves balance without the weight concentration that the base procedures incur.

\begin{figure}[t]
    \centering
    \includegraphics[width=\textwidth]{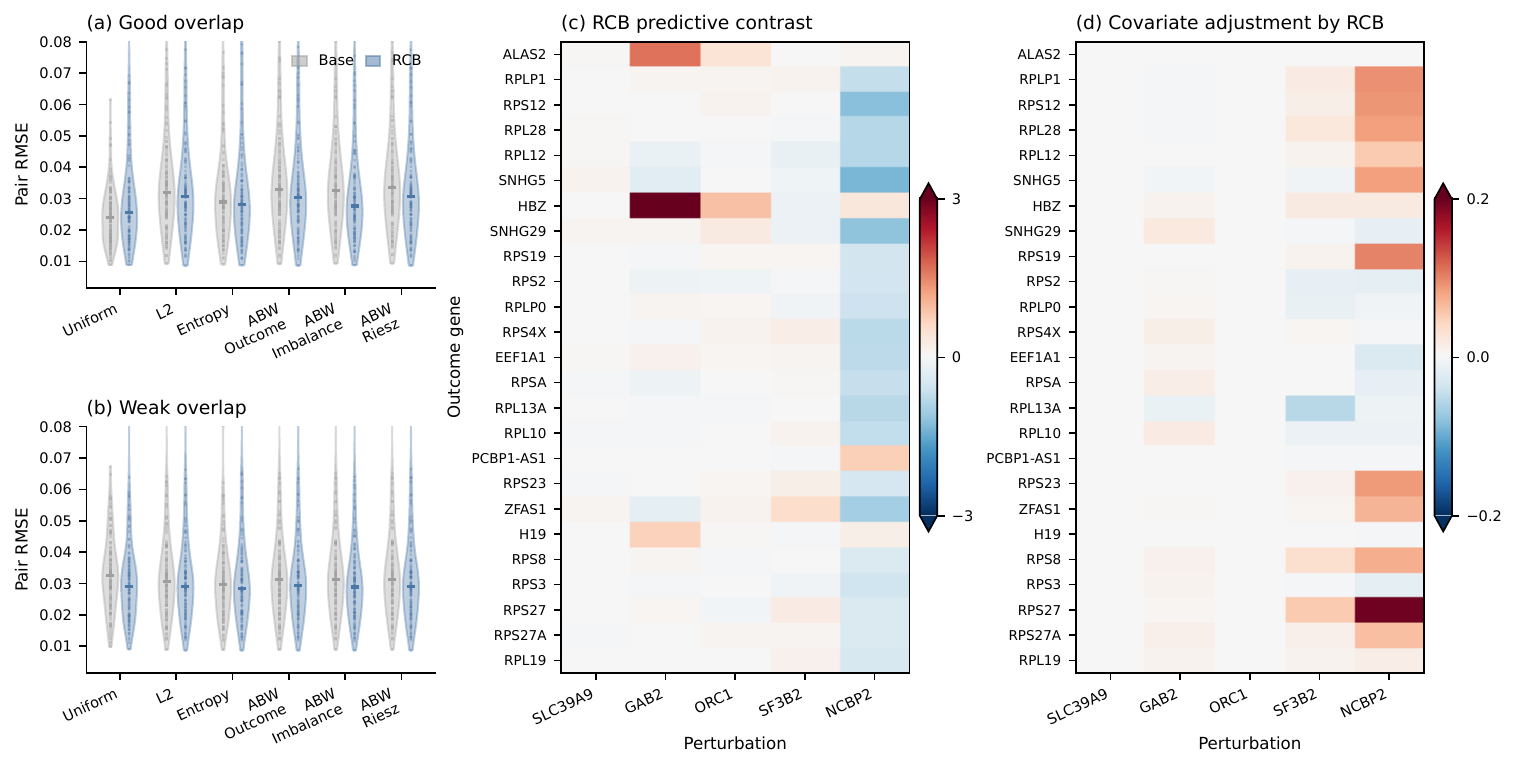}
    \caption{
    K562 Perturb-seq results.
    Panels (a)--(b) compare deviations of the estimated perturbation contrasts from empirical reference contrasts under the original and weak-overlap regimes.
    Panel (c) shows RCB predictive response contrasts, perturbed-minus-counterfactual differences in mean \(\log_2(1+\mathrm{count})\), and panel (d) shows, on a separate colour scale, the covariate adjustment \(-\Delta^\top\widehat\beta(\widehat\lambda)/\log 2\) applied by RCB.
    Columns index perturbations and rows index outcome genes.}
    \label{fig:crispr_k562_fixed_panel}
\end{figure}

\noindent\textbf{Full-data response profiles.}
Using all available K562 control and perturbation observations, we estimate predictive response contrasts on a common 25-gene display panel selected once from a marginal differential-expression screen across the five perturbations; Appendix~\ref{app:k562-implementation} gives the selection rule.
The selected penalties vary substantially across perturbations: across the 25 outcome genes, the median penalty is the no-augmentation endpoint, \(\lambda=\infty\), for \textit{SLC39A9} and \textit{ORC1}, and approximately \(7.94\), \(6.31\), and \(3.16\) for \textit{GAB2}, \textit{SF3B2}, and \textit{NCBP2}, respectively.
The estimated target risk therefore favours different degrees of augmentation across perturbations.

\Cref{fig:crispr_k562_fixed_panel}(d) shows the covariate adjustment that RCB applies on the selected outcome panel, the difference between its predictive contrast in panel (c) and the unadjusted perturbed-minus-control contrast.
The adjustment is target-specific: for \textit{SLC39A9} and \textit{ORC1}, RCB selects the no-augmentation endpoint and retains uniform weights on all 48 controls, whereas for \textit{GAB2}, \textit{SF3B2}, and \textit{NCBP2} it selects finite penalties with median effective sample sizes of about 22, 16, and 8 controls and changes the contrasts by less than \(0.2\), without reversing the sign of any contrast.
RCB thus adjusts only where the estimated shift term in its endpoint risk is positive: that term is zero for \textit{SLC39A9} and \textit{ORC1} and largest for \textit{NCBP2}, which receives the largest adjustment (\Cref{fig:k562-adjustment-diagnostics}).

\section{Discussion}

We have developed a criterion for tuning ridge-augmented balancing estimators under weak overlap based on the conditional prediction risk for the treated.
When the causal target is identified, stronger covariate balance can reduce extrapolation while increasing noise through more concentrated weights; the finite-sample risk decomposition places these two effects on a common prediction-risk scale.
The high-dimensional analysis further shows how source geometry, source-target shift, and target sampling shape this tradeoff along the regularisation path.
In the simulations, target-aware tuning generally reduces excess target risk relative to source-only tuning.
In both applications, augmentation adapts to the base weights and the severity of extrapolation: RCB attains similar accuracy across base weights under the observed LaLonde shift and reduces the RMSE of \(\ell_2\) weighting by \(37\%\) under a severe controlled covariate shift, and in the high-dimensional K562 design, where standard balancing concentrates the weights without transferring balance to the target population, RCB improves most base procedures while adjusting only for perturbations with a detectable shift.

The RCB procedure applies to any normalised base weights, allowing the initial weighting rule to be chosen separately from ridge augmentation and tuning.
Under the affine random-effects model, \Cref{prop:conditional-risk} gives an exact finite-sample risk decomposition for normalised base weights constructed without the control outcomes, including covariate-adaptive weights, while \Cref{thm:risk-asym} gives the high-dimensional characterisation under the additional design-independence condition.
Under the conditions of \Cref{thm:risk-estimation}, covariate-adaptive weights are allowed, with uniform consistency of the feasible risk estimate and vanishing scaled oracle excess risk for risk-based tuning.

Fixed finite-dimensional nonlinear feature maps are accommodated by the estimator; extending the risk theory to general random-feature or kernel models requires control of feature dependence and approximation error.
The source-target formulation can also be developed for the ATC and ATE under estimand-specific identification and design conditions; pointwise conditional effects require separate analysis.
More broadly, target prediction risk gives a direct criterion for how much source-target discrepancy to correct when stronger balance also increases noise.

\section*{Data Availability}
The LaLonde data \citep{farrell2015robust} and K562 Perturb-seq data \citep{replogle2022perturbseq} are publicly available; code and processed analysis files are at \url{https://github.com/JinHongDu-Lab/RCB}.

\section*{Acknowledgements}
The corresponding author acknowledges support from the National Natural Science Foundation of China (No. 12601487), and the Institute of Data Science (IDS) at the University of Hong Kong through the HKU100 award.
Part of the computations was performed using research computing facilities offered by Information Technology Services, the University of Hong Kong.

\clearpage
\bibliographystyle{apalike}
\bibliography{references}

\clearpage
\appendix
\counterwithin{theorem}{section}
\counterwithin{lemma}{section}
\counterwithin{proposition}{section}
\counterwithin{corollary}{section}
\counterwithin{assumption}{section}
\counterwithin{remark}{section}
\counterwithin{definition}{section}
\counterwithin{example}{section}
\counterwithin{property}{section}
\counterwithin{equation}{section}

\renewcommand{\thesection}{\Alph{section}}
\setcounter{table}{0}
\renewcommand{\thetable}{\thesection\arabic{table}}
\setcounter{figure}{0} 
\renewcommand{\thefigure}{\thesection\arabic{figure}}

\crefname{section}{Appendix}{Appendices}
\Crefname{section}{Appendix}{Appendices}
\crefname{subsection}{Appendix}{Appendices}
\Crefname{subsection}{Appendix}{Appendices}

\begin{center}
\textbf{\Large Appendices}
\end{center}

The appendices contain the proofs of the main theoretical results and the supplementary numerical and real-data analyses; their organisation and notation are summarised below.

\vspace{2mm}

\noindent\textbf{Organisation.}
\vspace{-4mm}

\begin{table}[!ht]
    \centering \small
\begin{tabularx}{0.98\textwidth}{l l D}
\toprule
\multicolumn{2}{c}{Appendix} & Content \\
\midrule
\multirow{3}{*}{\Cref{app:overlap-splitting}}
& \ref{app:overlap-calibration}
& Density-ratio overlap and oracle-weight concentration \(\varrho_{\eta,n}^2\). \\
& \ref{app:honest-target-split}
& Target sample splitting (verification of \Cref{asm:evaluation-sample}). \\
& \ref{app:sample-split-rate-verification}
& Rates of base balancing weights (verification of \Cref{asm:risk-rate-admissibility}). \\
\cmidrule(lr){1-3}

\multirow{9}{*}{\Cref{app:main-proofs}}
    & \ref{app:proof-prop-att-identification} & Proof of ATT identification (\Cref{prop:att-identification}). \\
    & \ref{app:proof-prop-risk} & Proof of the exact conditional random-effects risk decomposition (\Cref{prop:conditional-risk}). \\
    & \ref{app:proof-thm-risk-asym} & Proof of the random-matrix deterministic equivalent (\Cref{thm:risk-asym}). \\
    & \ref{app:proof-fixed-limit-corollaries} & Proofs of fixed risk limits (\Cref{cor:fixed-risk-limit}) and lower-order target-sampling effects (\Cref{cor:target-sampling-invariant}); limiting-risk regularity (\Cref{lem:risk-lipschitz}) and monotonicity of the integrated squared bias (\Cref{prop:bias-monotonicity}). \\
    & \ref{app:proof-variance-components} & Variance-component estimation: spectral quasi-likelihood consistency (\Cref{prop:spectral-quasi-likelihood}), two-moment initialisation, and computation. \\
    & \ref{app:proof-risk-tuning} & Proof of uniform target-aware risk estimation and tuning consistency (\Cref{thm:risk-estimation}). \\ 
    & \ref{app:proof-predictive-calibration} & Proof of the prediction-interval result (\Cref{thm:predictive-calibration}), tuning stability and non-Gaussian extension. \\
\cmidrule(lr){1-3}

\multirow{3}{*}{\Cref{app:rmt-concentration}}
& \ref{app:uncentered-resolvent}
& Verification of source-resolvent concentration; higher-resolvent and cross-block bounds. \\
& \ref{app:centered-resolvent}
& Centred source resolvents, design-independence trace replacement, and source-mean identities. \\
& \ref{app:treated-mean-concentration}
& Treated-mean quadratic-form and bilinear concentration. \\
\cmidrule(lr){1-3}

\multirow{3}{*}{\Cref{app:profile-transfer}}
& \ref{app:general-profile-transfer}
& Risk-path components and ridge verification tools. \\
& \ref{app:target-split-ridge-profile-verification}
& Target-split ridge base weights and consequences for penalty selection. \\
& \ref{app:ridge-profile-verification}
& Same-sample ridge base weights. \\
\cmidrule(lr){1-3}

\multirow{5}{*}{\Cref{app:sup-exp}}
& \ref{subsec:mc-conventions}
& Simulation implementation and Monte Carlo conventions. \\
& \ref{subsec:isotropic-controls}
& Isotropic risk-path and aspect-ratio controls. \\
& \ref{subsec:exp-predictive}
& Prediction-interval coverage. \\
& \ref{subsec:sup-num}
& Adaptive-balancing diagnostics and sensitivity. \\
& \ref{subsec:response-robustness}
& Response-model robustness. \\
\cmidrule(lr){1-3}

\multirow{2}{*}{\Cref{app:real-data}}
& \ref{app:lalonde-observational}
& LaLonde analyses and diagnostics. \\
& \ref{app:k562-implementation}
& K562 Perturb-seq implementation. \\
\bottomrule
\end{tabularx}
\end{table}
\addcontentsline{toc}{part}{\appendixname}

\noindent\textbf{Notation.}
Lowercase symbols denote scalars or vectors as clear from context, and uppercase symbols denote matrices.
We write \(\one_d\) for the \(d\)-vector of ones, \(I_d\) for the \(d\)-dimensional identity matrix, and \(A^\top\) for the transpose of \(A\).
The Euclidean and operator norms are denoted by \(\|\cdot\|_2\) and \(\|\cdot\|_{\oper}\), respectively.
We use \(\tr(A)\) for the trace, \(\otr(A)=p^{-1}\tr(A)\) for the normalised trace of a \(p\times p\) matrix, and \(A\succeq0\) to indicate that \(A\) is positive semidefinite.

Expectation, probability, variance, and covariance are denoted by \(\EE\), \(\PP\), \(\Var\), and \(\Cov\).
Convergence in probability and distribution are written \(\pto\) and \(\dto\), while \(\op(1)\) and \(\Op(1)\) have their usual stochastic-order meanings.
The notation \(\dto\) also denotes weak convergence of finite measures when the context is unambiguous.
For results involving design-independent weights, convergence is first interpreted conditionally on \(\mathcal F_{\mathrm{aux}}\) for \(\PP\)-almost every admissible auxiliary realisation.
The corresponding unconditional conclusions follow by averaging the conditional error probabilities.

The index \(t=0\) denotes the control or source population, and \(t=1\) denotes the treated or target population.
We use \(n_t\), \(\PP_t\), \(\nu_t\), and \(\Sigma_t\) for the corresponding sample size, covariate distribution, mean, and covariance.
The control design and outcome vector are \(X_0\) and \(y_0\), and \(\bar x_t\) denotes a treatment-specific covariate mean.
The centring matrix is \(C_0\), with \(X_0^c=C_0X_0\).
A normalised base-weight vector is denoted by \(\gamma\), its residual imbalance by \(\Delta=\bar x_1-X_0^\top\gamma\), and its ridge-augmented version by \(\gamma_\lambda\), with centred ridge resolvent \(M_\lambda^c\).
Finally, \(B_n(\lambda)\), \(V_n(\lambda)\), and \(R_n(\lambda)\) denote the integrated squared bias, noise variance, and total conditional random-effects prediction risk.

\clearpage
\section{Overlap Calibration and Target Sample Splitting}
\label{app:overlap-splitting}

\subsection[Density-Ratio Overlap and Oracle-Weight Concentration]{Density-Ratio Overlap and Oracle-Weight Concentration \(\varrho_{\eta,n}^2\)}
\label{app:overlap-calibration}

For \(t=0,1\), let
\[
    \PP_{t,n}:=\PP(x\in\cdot\mid a=t)
\]
and suppose \(\PP_{1,n}\ll\PP_{0,n}\). Define the source-to-target density
ratio and its second moment by
\[
    w_n(x)=\frac{\rd \PP_{1,n}}{\rd \PP_{0,n}}(x),
    \qquad
    \kappa_n=\EE_{0,n}\{w_n(x)^2\}
    =1+\chi^2(\PP_{1,n}\,\|\,\PP_{0,n}).
\]
Thus \(\kappa_n\) is an overlap-divergence measure.

\begin{proposition}[Overlap divergence and oracle-weight concentration]
    \label{prop:overlap-weight-concentration}
    For independent control covariates \(x_{0i}\sim \PP_{0,n}\), define
    \[
        \gamma_{i,n}^{\mathrm{or}}
        =\frac{w_n(x_{0i})}{\sum_{j=1}^{n_0}w_n(x_{0j})}.
    \]
    Fix \(\eta\in(0,1]\). 
    Suppose \(n_0^{\eta-1}\kappa_n\rightarrow\kappa\in(0,\infty)\) and \({\EE_{0,n}\{w_n(x)^4\}}/({n_0\kappa_n^2})\rightarrow0\).
    Then
    \[
        n_0^\eta\|\gamma_n^{\mathrm{or}}\|_2^2\pto\kappa,\quad n_0^{-\eta}\operatorname{ESS}(\gamma_n^{\mathrm{or}})
        \pto\kappa^{-1} ,
    \]
    where \(\operatorname{ESS}(\gamma)=\|\gamma\|_2^{-2}\) is the effective sample size of the control weights.
\end{proposition}

\begin{proof}[Proof of \Cref{prop:overlap-weight-concentration}]
    Write
    \[
        \overline w_n=\frac1{n_0}\sum_{i=1}^{n_0}w_n(x_{0i}),
        \qquad
        \overline{w_n^2}=\frac1{n_0}\sum_{i=1}^{n_0}w_n(x_{0i})^2.
    \]
    Because \(\EE_{0,n}w_n(x)=1\),
    \[
        \Var(\overline w_n)=\frac{\kappa_n-1}{n_0}=\cO(n_0^{-\eta}),
    \]
    and hence \(\overline w_n\pto1\). Moreover,
    \[
        \Var\left(\frac{\overline{w_n^2}}{\kappa_n}\right)
        \leq\frac{\EE_{0,n}\{w_n(x)^4\}}{n_0\kappa_n^2}
        \rightarrow0.
    \]
    Therefore,
    \[
        n_0^\eta\|\gamma_n^{\mathrm{or}}\|_2^2
        =n_0^{\eta-1}\frac{\overline{w_n^2}}{\overline w_n^2}
        \pto\kappa.
    \]
    The effective-sample-size conclusion follows by the continuous mapping theorem, which completes the proof.
\end{proof}

For every normalised weight vector satisfying \Cref{asm:design-independent-regime}~\eqref{asm:design-independent-regime-i}, centring gives the exact bridge
\[
    \|C_0\gamma\|_2^2=\|\gamma\|_2^2-\frac1{n_0}.
\]
Consequently, under \Cref{prop:overlap-weight-concentration},
\[
    \varrho_{\eta,n}^2 := n_0^\eta\|C_0\gamma_n^{\mathrm{or}}\|_2^2
    \pto \kappa-\ind{(\eta=1)} =: \varrho_{\eta}^2.
\]
In particular, the centred concentration limit is \(\kappa-1\) at
\(\eta=1\). Since \(\kappa_n\geq1\), convergence at this boundary implies
\(\kappa\geq1\), so the limiting centred concentration is nonnegative.

For \(\eta=1\), the density-ratio second moment remains bounded and the oracle effective sample size is proportional to \(n_0\). 
For \(0<\eta<1\),
\[
    \chi^2(\PP_{1,n}\,\|\, \PP_{0,n})\asymp n_0^{1-\eta},
\]
and the oracle effective sample size is of order \(n_0^\eta\). 
The boundary case \(\eta=0\) corresponds to an effective sample size of constant order. 
At this boundary, the normalised density-ratio denominator need not satisfy a law of large numbers under second-moment scaling alone, so the concentration condition \(\|\gamma\|_2^2=\cO(1)\) is required, as in \Cref{asm:risk-rate-admissibility}.

\Cref{prop:overlap-weight-concentration} therefore provides an overlap interpretation for the order of weight concentration used in
\Cref{asm:design-independent-regime}~\eqref{asm:design-independent-regime-ii}.
Because the oracle density-ratio weights depend on the realised control covariates, the proposition interprets the scale of \(\varrho_{\eta,n}^2\) rather than verifying \Cref{asm:design-independent-regime}.

\subsection[Target Sample Splitting (Verification of Assumption \ref{asm:evaluation-sample})]{Target Sample Splitting (Verification of \Cref{asm:evaluation-sample})}
\label{app:honest-target-split}

Partition the treated sample into a pilot subset \(\mathcal I_{\mathrm P}\)
and an evaluation subset \(\mathcal I_{\mathrm E}\), independently of the
observed covariate values. Let their sizes be \(n_{1,\mathrm P}\) and
\(n_{1,\mathrm E}\), and define
\[
    \begin{aligned}
        \bar{x}_{1,\mathrm P}
        &=\frac1{n_{1,\mathrm P}}\sum_{i\in\mathcal I_{\mathrm P}}x_i,
        &\bar{x}_{1,\mathrm E}
        &=\frac1{n_{1,\mathrm E}}\sum_{i\in\mathcal I_{\mathrm E}}x_i,\\
        \widehat\Sigma_{1,\mathrm E}
        &=\frac1{n_{1,\mathrm E}-1}
        \sum_{i\in\mathcal I_{\mathrm E}}
        (x_i-\bar x_{1,\mathrm E})(x_i-\bar x_{1,\mathrm E})^\top.&&
    \end{aligned}
\]
Construct weight through any covariate balancing method \(\mathsf{Balance}\):
\begin{equation}
    \label{eq:target-split-imbalance}
        \widehat\gamma
        =\mathsf{Balance}(X_0,\bar{x}_{1,\mathrm P}),
        \qquad
        \Delta_{\mathrm E}
        =\bar{x}_{1,\mathrm E}-X_0^\top\widehat\gamma.
\end{equation}

\begin{lemma}[Target sample splitting]
    \label{lem:target-split-separation}
    Suppose the treated observations are independent with common law
    \(\PP_{1,n}\), the source and treated samples are independent, and the
    partition \(\left(\mathcal I_{\mathrm P},\mathcal I_{\mathrm E}\right)\)
    is generated independently of the observed covariates. Let
    \(
    \mathcal H_n
    =\sigma\!\left[
        \mathcal S_n,X_0,
        \left\{x_i:i\in\mathcal I_{\mathrm P}\right\}
    \right],
    \)
    and suppose \(\mathsf{Balance}\) returns normalised weights. 
    Then \(\widehat\gamma=\mathsf{Balance}\left(X_0,\bar x_{1,\mathrm P}\right)\) is \(\mathcal H_n\)-measurable and
    \(
    \PP\left(
        \left(x_i\right)_{i\in\mathcal I_{\mathrm E}}
        \,\middle|\,
        \mathcal H_n
    \right)
    =\PP_{1,n}^{\otimes n_{1,\mathrm E}}\) almost surely.
    Further, if \(\PP_{1,n}\) satisfies the target-side design and moment conditions in \Cref{asm:rmt}, then \Cref{asm:evaluation-sample} holds.
\end{lemma}

\begin{proof}[Proof of \Cref{lem:target-split-separation}]
    Conditional on the split assignment, the pilot and evaluation subsets of the
    independent treated observations remain independent. Independence of the
    source and treated samples then implies that conditioning additionally on
    \(X_0\) does not change the joint law of the evaluation observations.
    Their conditional law is therefore
    \(\PP_{1,n}^{\otimes n_{1,\mathrm E}}\). Measurability and normalisation
    of \(\widehat\gamma\) follow from its construction.
\end{proof}

\begin{algorithm}[!t]
\caption{Ridge-augmented balancing with target sample splitting at a fixed penalty}
\label{alg:target-split-ridge-augmentation}
\small
\begin{algorithmic}[1]
\Require Control data \((X_0,y_0)\); treated covariates; the full
treated-sample outcome mean \(\bar y_1\); base-balancing map
\(\mathsf{Balance}\); penalty \(\lambda>0\); treated covariate split
\((\mathcal I_{\mathrm P},\mathcal I_{\mathrm E})\).
\Ensure Counterfactual estimate
\(\widehat\mu_{0,\mathrm E,\lambda}\), ATT estimate
\(\widehat\tau_{\mathrm E,\lambda}\), and augmented weights
\(\widehat\gamma_{\mathrm E,\lambda}\).
\State Compute \(\bar x_{1,\mathrm P}\) and \(\bar x_{1,\mathrm E}\)
from the corresponding treated subsets.
\State Compute \(\bar x_0=X_0^\top\one_{n_0}/n_0\) and
\(X_0^c=X_0-\one_{n_0}\bar x_0^\top\).
\State Construct \(\widehat\gamma\) and \(\Delta_{\mathrm E}\) using
\eqref{eq:target-split-imbalance}, and solve
\(\{(X_0^c)^\top X_0^c/n_0+\lambda I_p\}u_\lambda
=\Delta_{\mathrm E}\).
\State After fixing \(\lambda\), set
\(\widehat\gamma_{\mathrm E,\lambda}
=\widehat\gamma+X_0^cu_\lambda/n_0\),
\(\widehat\mu_{0,\mathrm E,\lambda}
=\widehat\gamma_{\mathrm E,\lambda}^{\top}y_0\), and
\(\widehat\tau_{\mathrm E,\lambda}
=\bar y_1-\widehat\mu_{0,\mathrm E,\lambda}\).\\
\Return the estimates and augmented weights.
\end{algorithmic}
\end{algorithm}

The full treated-sample outcome mean \(\bar y_1\) enters only after the penalty is fixed. 
Its sampling error is outside the counterfactual-risk analysis.
Thus, the target sample splitting construction uses the evaluation-split augmented weights
\[
    \widehat\gamma_{\mathrm E,\lambda}
    =\widehat\gamma
    +\frac1{n_0}X_0^c M_\lambda^c\Delta_{\mathrm E}.
\]
Recall that
\(
    \mathcal S_n=\sigma(\mathcal I_{\mathrm P},\mathcal I_{\mathrm E})
\)
denotes the split assignment and is contained in
\(\mathcal F_{\mathrm{aux}}\). Conditional on
\[
    \mathcal H_n
    =\sigma\left[\mathcal S_n,X_0,
    \left\{x_i:i\in\mathcal I_{\mathrm P}\right\}\right],
\]
\Cref{lem:target-split-separation} verifies
\Cref{asm:evaluation-sample} with \(n_{1,\mathrm E}\) and the
evaluation-sample covariance \(\widehat\Sigma_{1,\mathrm E}\).
For target-aware tuning, apply \Cref{alg:ood-tuning} and
\Cref{thm:risk-estimation} with the substitutions
\[
    (\bar x_1,\widehat\Sigma_1,n_1,\Delta,\gamma_\lambda)
    \longmapsto
    (\bar x_{1,\mathrm E},\widehat\Sigma_{1,\mathrm E},n_{1,\mathrm E},
    \Delta_{\mathrm E},\gamma_{\mathrm E,\lambda}).
\]
This gives the end-to-end target sample splitting procedure: construct the base weights on the pilot subset, evaluate and minimise the risk path on the evaluation subset, and return the selected augmented weights and estimates.
The oracle-excess-risk conclusion of \Cref{thm:risk-estimation} applies to such a sample split.

\subsection[Rates of Base Balancing Weights (Verification of Assumption \ref{asm:risk-rate-admissibility})]{Rates of Base Balancing Weights (Verification of \Cref{asm:risk-rate-admissibility})}
\label{app:sample-split-rate-verification}

The pilot/evaluation construction in \Cref{app:honest-target-split} verifies the evaluation-sample separation required by \Cref{asm:evaluation-sample}. 
This subsection gives sufficient
conditions for the separate base-weight rates in \Cref{asm:risk-rate-admissibility}. 
\Cref{tab:risk-rate-applicability} separates two logically distinct requirements.

\begin{table}[!t]
    \small
    \caption{Sufficient conditions for evaluation-sample separation and the
    base-weight rate conditions}
    \label{tab:risk-rate-applicability}
    \begin{tabular}{
        @{}
        p{0.29\textwidth}
        p{0.19\textwidth}
        p{0.48\textwidth}
        @{}
    }
    \toprule
    Base-weight construction
    & Evaluation separation
    & Verification of \Cref{asm:risk-rate-admissibility}
    \\
    \midrule
    Uniform weights
    & Target independent
    & \Cref{cor:uniform-risk-rate} under the source-target mean-shift rate
    \\
    Sample-split ridge weights
    & Target sample split
    & \Cref{cor:target-split-ridge-post-tuning,cor:target-split-ridge-selected-alpha}
    \\
    Quadratic-dispersion weights, SBW
    & Target sample split
    & \Cref{cor:sample-split-sbw} under the balance-tolerance and
    feasible-comparator conditions
    \\
    Entropy balancing with an explicit weight cap
    & Target sample split
    & \Cref{cor:sample-split-capped-entropy} under the balance-tolerance and
    cap rates
    \\
    Uncapped entropy balancing
    & Target sample split
    & \Cref{cor:normalized-score-risk-rate} under the pilot-imbalance and
    score-moment rates
    \\
    CBPS, IPW, or overlap weights
    & Target sample split
    & \Cref{cor:normalized-score-risk-rate} when the induced positive scores
    satisfy the same two rates
    \\
    Same-sample target-adaptive ridge weights
    & None
    & Separate verification in \Cref{prop:ridge-profile-verification}
    \\
    \bottomrule
    \end{tabular}
\end{table}

Recall that with target sample splitting, the pilot information field is defined as \(
    \mathcal H_n
    =
    \sigma\!\left[
        \mathcal S_n,X_0,
        \{x_i:i\in\mathcal I_{\mathrm P}\}
    \right],
\) as defined in \Cref{asm:evaluation-sample}.
For a normalised weight vector constructed from the pilot target sample, define its pilot imbalance by
\[
    \Delta_{\mathrm P}(\gamma)
    :=\bar x_{1,\mathrm P}-X_0^\top\gamma.
\]
The design-independent setting in \Cref{asm:rmt,asm:design-independent-regime} also satisfies \Cref{asm:risk-rate-admissibility}.
Conditional on \(\mathcal F_{\mathrm{aux}}\), affine normalisation and design independence give
\[
    \EE\left[
        \left\|X_0^\top\gamma-\nu_0\right\|_2^2
        \,\middle|\,
        \mathcal F_{\mathrm{aux}}
    \right]
    =
    \tr(\Sigma_0)\|\gamma\|_2^2
    \leq
    C_\Sigma p
    \left(
        \|C_0\gamma\|_2^2+\frac{1}{n_0}
    \right).
\]
Hence
\[
    \frac{n_0^\eta}{p}
    \left\|X_0^\top\gamma-\nu_0\right\|_2^2
    =
    \Op(1),
\]
and the source-target mean-shift and weight-concentration rates in \Cref{asm:design-independent-regime}, together with \(\eta\leq1\), yield \Cref{asm:risk-rate-admissibility}.

We first present an intermediate result for verifying rate conditions under varying cases.

\begin{proposition}[From pilot rates to base-weight rate conditions]
    \label{prop:pilot-to-population-rate}
    Suppose the pilot target observations satisfy the target-side design and
    moment conditions in \Cref{asm:rmt}, and let \(\eta\in[0,1]\) satisfy
    \[
        \frac{n_0^\eta}{n_{1,\mathrm P}}=\cO(1).
    \]
    Let \(\widehat\gamma\) be normalised, measurable with respect to \(\mathcal H_n\), and constructed without \(y_0\). 
    If
    \[
        n_0^\eta
        \frac{\|\Delta_{\mathrm P}(\widehat\gamma)\|_2^2}{p}
        =\Op(1),
        \qquad
        n_0^\eta\|C_0\widehat\gamma\|_2^2=\Op(1),
    \]
    then \(\widehat\gamma\) satisfies
    \Cref{asm:risk-rate-admissibility}.
\end{proposition}

\begin{proof}[Proof of \Cref{prop:pilot-to-population-rate}]
    By the target-side covariance bound in \Cref{asm:rmt},
    \[
        \EE\|\bar x_{1,\mathrm P}-\nu_1\|_2^2
        =\frac{\tr(\Sigma_1)}{n_{1,\mathrm P}}
        \leq \frac{C_\Sigma p}{n_{1,\mathrm P}}.
    \]
    Hence Markov's inequality gives
    \[
        n_0^\eta
        \frac{\|\bar x_{1,\mathrm P}-\nu_1\|_2^2}{p}
        =\Op(1).
    \]
    Since
    \(
        \nu_1-X_0^\top\widehat\gamma
        =(\nu_1-\bar x_{1,\mathrm P})
        +\Delta_{\mathrm P}(\widehat\gamma),
    \)
    the inequality
    \(\|u+v\|_2^2\leq2\|u\|_2^2+2\|v\|_2^2\)
    yields the first condition in
    \Cref{asm:risk-rate-admissibility}. The second condition is assumed
    directly.
\end{proof}

Based on \Cref{prop:pilot-to-population-rate}, we next verify rate conditions for balancing weights listed on \Cref{tab:risk-rate-applicability}.

\begin{corollary}[Uniform base weights]
    \label{cor:uniform-risk-rate}
    Suppose the source-side design conditions in \Cref{asm:rmt} hold and, for
    some \(\eta\in[0,1]\),
    \[
        n_0^\eta {\|\nu_1-\nu_0\|_2^2}/{p}=\cO(1).
    \]
    Then
    \[
        \gamma^{\mathrm{unif}}
        ={n_0}^{-1}\one_{n_0}
    \]
    satisfies \Cref{asm:risk-rate-admissibility}.
\end{corollary}

\begin{proof}[Proof of \Cref{cor:uniform-risk-rate}]
    The centred concentration is zero. Moreover,
    \[
        \nu_1-X_0^\top\gamma^{\mathrm{unif}}
        =(\nu_1-\nu_0)-(\bar x_0-\nu_0),
    \]
    and the source-side covariance bound gives
    \[
        \EE\|\bar x_0-\nu_0\|_2^2
        \leq \frac{C_\Sigma p}{n_0}.
    \]
    The conclusion follows because \(\eta\leq1\).
\end{proof}

\begin{corollary}[Quadratic-dispersion balancing]
    \label{cor:sample-split-sbw}
    Let
    \(
        \Gamma_n
        \subseteq
        \left\{
            \gamma\in\RR^{n_0}:
            \one_{n_0}^\top\gamma=1
        \right\}
    \)
    be closed and encode any additional outcome-free side constraints, such as
    nonnegativity or upper bounds, and define
    \(
        \mathcal Q_n(\delta_n)
        :=
        \left\{
            \gamma\in\Gamma_n:
            \|\Delta_{\mathrm P}(\gamma)\|_2\leq\delta_n
        \right\}.
    \)
    On the event
    \(\{\mathcal Q_n(\delta_n)\neq\varnothing\}\), let
    \[
        \widehat\gamma^{\mathrm{quad}}
        \in
        \argmin_{\gamma\in\mathcal Q_n(\delta_n)}
        \|\gamma\|_2^2,
    \]
    and use the uniform weights otherwise. Suppose the target-side and
    pilot-size conditions of \Cref{prop:pilot-to-population-rate} hold,
    \[
        n_0^\eta\frac{\delta_n^2}{p}=\Op(1),
    \]
    and there exists a possibly random comparator sequence
    \(\gamma_n^\circ\) such that
    \[
        \PP\left\{
            \gamma_n^\circ\in\mathcal Q_n(\delta_n)
        \right\}\to1,
        \qquad
        n_0^\eta\|C_0\gamma_n^\circ\|_2^2=\Op(1).
    \]
    Then \(\widehat\gamma^{\mathrm{quad}}\) satisfies
    \Cref{asm:risk-rate-admissibility}. In particular, this applies to the
    quadratic-dispersion formulation of stable balancing weights whenever the
    stated feasibility and comparator conditions hold.
\end{corollary}

\begin{proof}[Proof of \Cref{cor:sample-split-sbw}]
    The defining constraint gives
    \[
        \|\Delta_{\mathrm P}
            (\widehat\gamma^{\mathrm{quad}})\|_2
        \leq\delta_n.
    \]
    On the event
    \(\{\gamma_n^\circ\in\mathcal Q_n(\delta_n)\}\), optimality and affine
    normalisation give
    \begin{align*}
        \|C_0\widehat\gamma^{\mathrm{quad}}\|_2^2
        &=
        \|\widehat\gamma^{\mathrm{quad}}\|_2^2-\frac1{n_0}
        \leq
        \|\gamma_n^\circ\|_2^2-\frac1{n_0}
        =
        \|C_0\gamma_n^\circ\|_2^2.
    \end{align*}
    The conclusion follows from
    \Cref{prop:pilot-to-population-rate}.
\end{proof}

\begin{corollary}[Capped Stable Entropy Balance]
    \label{cor:sample-split-capped-entropy}
    Let \(q_i>0\), \(\sum_{i=1}^{n_0}q_i=1\), and define
    \[
        \mathcal E_n(\delta_n,L_n)
        :=
        \left\{
            \gamma\in\RR^{n_0}:
            \one_{n_0}^\top\gamma=1,\quad
            0\leq\gamma_i\leq L_n,\quad
            \|\Delta_{\mathrm P}(\gamma)\|_2\leq\delta_n
        \right\}.
    \]
    On the event
    \(\{\mathcal E_n(\delta_n,L_n)\neq\varnothing\}\), let
    \[
        \widehat\gamma^{\mathrm{ent}}
        \in
        \argmin_{\gamma\in\mathcal E_n(\delta_n,L_n)}
        \sum_{i=1}^{n_0}
        \gamma_i\log\left(\frac{\gamma_i}{q_i}\right),
    \]
    and use the uniform weights otherwise, with the convention
    \(0\log0=0\). Suppose the target-side and pilot-size conditions of
    \Cref{prop:pilot-to-population-rate} hold, the feasible set is nonempty
    with probability tending to one, and
    \[
        n_0^\eta\frac{\delta_n^2}{p}=\Op(1),
        \qquad
        n_0^\eta L_n=\cO(1).
    \]
    Then \(\widehat\gamma^{\mathrm{ent}}\) satisfies
    \Cref{asm:risk-rate-admissibility}.
\end{corollary}

\begin{proof}[Proof of \Cref{cor:sample-split-capped-entropy}]
    The balance constraint gives the required pilot-imbalance rate. Since the
    weights are nonnegative and normalised,
    \[
        \|C_0\widehat\gamma^{\mathrm{ent}}\|_2^2
        \leq
        \|\widehat\gamma^{\mathrm{ent}}\|_2^2
        \leq
        \left(
            \max_{i\leq n_0}\widehat\gamma_i^{\mathrm{ent}}
        \right)
        \sum_{i=1}^{n_0}\widehat\gamma_i^{\mathrm{ent}}
        \leq L_n.
    \]
    Apply \Cref{prop:pilot-to-population-rate}.
\end{proof}

\begin{corollary}[Normalised positive-score weights]
    \label[corollary]{cor:normalized-score-risk-rate}
    Suppose the target-side and pilot-size conditions of
    \Cref{prop:pilot-to-population-rate} hold. Let
    \(
        \widetilde w_{1,n},\ldots,\widetilde w_{n_0,n}\geq0
    \)
    be pilot-measurable scores with
    \[
        \overline w_n
        :=
        \frac1{n_0}
        \sum_{i=1}^{n_0}\widetilde w_{i,n}>0,
        \qquad
        \overline{w_n^2}
        :=
        \frac1{n_0}
        \sum_{i=1}^{n_0}\widetilde w_{i,n}^2,
    \]
    and define
    \[
        \widehat\gamma_i
        =
        \frac{\widetilde w_{i,n}}
        {\sum_{j=1}^{n_0}\widetilde w_{j,n}}.
    \]
    If
    \[
        n_0^\eta
        \frac{
            \|\Delta_{\mathrm P}(\widehat\gamma)\|_2^2
        }{p}
        =\Op(1),
        \qquad
        n_0^{\eta-1}
        \frac{\overline{w_n^2}}{\overline w_n^2}
        =\Op(1),
    \]
    then \(\widehat\gamma\) satisfies
    \Cref{asm:risk-rate-admissibility}.
\end{corollary}

\begin{proof}[Proof of \Cref{cor:normalized-score-risk-rate}]
    The normalisation gives the exact identity
    \[
        n_0^\eta\|\widehat\gamma\|_2^2
        =
        n_0^{\eta-1}
        \frac{\overline{w_n^2}}{\overline w_n^2}.
    \]
    Since
    \[
        \|C_0\widehat\gamma\|_2^2
        \leq
        \|\widehat\gamma\|_2^2,
    \]
    the centred-concentration rate follows. The pilot-imbalance rate and
    \Cref{prop:pilot-to-population-rate} complete the proof.
\end{proof}

\begin{remark}[Uncapped entropy and propensity-derived weights]
    \label{rem:uncapped-score-weights}
    For normalised positive weights,
    \[
        \operatorname{ESS}(\widehat\gamma)
        :=
        \|\widehat\gamma\|_2^{-2}
        =
        n_0
        \frac{\overline w_n^2}{\overline{w_n^2}}.
    \]
    Thus the score-moment condition in
    \Cref{cor:normalized-score-risk-rate} is equivalent to
    \[
        \frac{
            n_0^\eta
        }{
            \operatorname{ESS}(\widehat\gamma)
        }
        =\Op(1),
    \]
    and is implied by
    \[
        \frac{
            \max_{i\leq n_0}\widetilde w_{i,n}
        }{
            \overline w_n
        }
        =
        \Op(n_0^{1-\eta}).
    \]
    This route applies to uncapped SEB, CBPS-based or fitted-IPW
    weights, and overlap weights after separately verifying the displayed
    pilot-imbalance and score-moment rates. The entropy objective alone does not
    supply the required upper bound on
    \(\|\widehat\gamma\|_2^2\); target sample splitting supplies evaluation
    separation, not weight concentration.
\end{remark}

\clearpage
\section{Proofs of the Main Theoretical Results}
\label{app:main-proofs}

\subsection{Proof of Proposition~\ref{prop:att-identification}}
\label{app:proof-prop-att-identification}

\begin{proof}[Proof of \Cref{prop:att-identification}]
    By consistency and unconfoundedness for \(y(0)\),
    \[
        \EE(y\mid x,a=0)
        =
        \EE\{y(0)\mid x,a=0\}
        =
        \EE\{y(0)\mid x\}
        =
        m_0(x)
    \]
    \(\PP_0\)-almost everywhere.
    To verify absolute continuity, let \(B\) be measurable with \(\PP_0(B)=0\).
    Then
    \[
        0
        =
        \pi_0\PP_0(B)
        =
        \EE\left[\ind\{x\in B\}\{1-e(x)\}\right].
    \]
    Because \(\PP_1\ll\PP_x\) automatically and \(1-e(x)>0\) for \(\PP_1\)-almost every \(x\), positivity implies \(\PP_1(B)=0\).
    Hence \(\PP_1\ll\PP_0\), so the observed control regression is also identified \(\PP_1\)-almost everywhere.
    The law of iterated expectations therefore gives
    \[
        \begin{aligned}
            \mu_0
            &=
            \EE\!\left[
                \EE\{y(0)\mid x,a=1\}
                \,\middle|\,a=1
            \right] \\
            &=
            \EE\{m_0(x)\mid a=1\}
            =
            \int m_0(x)\,\rd\PP_1(x),
        \end{aligned}
    \]
    whereas consistency and the law of iterated expectations give
    \[
        \mu_1
        =
        \EE\{y(1)\mid a=1\}
        =
        \EE(y\mid a=1)
        =
        \int \EE(y\mid x,a=1)\,\rd\PP_1(x).
    \]
    These identities yield the two regression representations of \(\tau_{\mathrm{ATT}}\).

    Absolute continuity ensures that the density ratio \(r=\rd\PP_1/\rd\PP_0\) exists, and Bayes' rule gives
    \[
        r(x)
        =
        \frac{\pi_0}{\pi_1}
        \frac{e(x)}{1-e(x)}
        \qquad
        \PP_0\text{-almost everywhere}.
    \]
    Hence
    \[
        \begin{aligned}
            \mu_0
            &=
            \int r(x)m_0(x)\,\rd\PP_0(x) \\
            &=
            \EE\{r(x)y\mid a=0\} \\
            &=
            \frac{1}{\pi_1}
            \EE\left[
                (1-a)\frac{e(x)}{1-e(x)}y
            \right],
        \end{aligned}
    \]
    which proves the remaining claims.
\end{proof}

\subsection{Proof of Proposition~\ref{prop:conditional-risk}}
\label{app:proof-prop-risk}

\begin{proof}[Proof of \Cref{prop:conditional-risk}]  
    Since
    \begin{align}
        \frac{1}{n_0}X_0^\top X_0^c=\frac{1}{n_0}(X_0^c)^\top X_0^c
        =(M_\lambda^c)^{-1}-\lambda I_p, \label{eq:prop:conditional-risk-proof-1}
    \end{align}
    we have
    \begin{align*}
        X_0^\top \gamma_\lambda &= X_0^\top \gamma + \frac{1}{n_0}X_0^\top X_0^c M_\lambda^c \Delta \\
        &= (\bar{x}_1 - \Delta) + \left[(M_\lambda^c)^{-1} - \lambda I_p\right]M_\lambda^c\Delta \\
        &= \bar{x}_1 - \lambda M_\lambda^c\Delta.
    \end{align*}
    
    Substitution into \eqref{eq:R-BV-decom} gives
    \begin{align*}
        B(\beta;\lambda) &= [(\gamma_\lambda)^{\top}X_0\beta -{\nu_1}^\top \beta]^2\\
        &= [\left(\bar{x}_1 - \nu_1 - \lambda M_\lambda^c\Delta\right)^\top \beta ]^2.
    \end{align*}
    Hence, by \Cref{asm:predictive-model},
    \begin{align*}
        B_n(\lambda)&:=\EE_{\beta\sim\Pi}[B(\beta;\lambda)
        \mid\mathcal G_n] \\
        &=\left(\epsilon_1-\lambda M_\lambda^c\Delta\right)^\top
        \EE_\Pi(\beta\beta^\top)
        \left(\epsilon_1-\lambda M_\lambda^c\Delta\right) \\
        &=\frac{r^2}{p}\|\epsilon_1-\lambda M_\lambda^c\Delta\|_2^2\\
        &=\frac{r^2}{p}\left\{\|\epsilon_1\|_2^2
        -2\lambda\epsilon_1^\top M_\lambda^c\Delta
        +\lambda^2\Delta^\top (M_\lambda^c)^2\Delta\right\}.
    \end{align*}
    For the conditional residual variance \(V(\lambda)\) in
    \eqref{eq:R-BV-decom}, using \eqref{eq:prop:conditional-risk-proof-1},
    \begin{align*}
        V_n(\lambda) &:= V(\lambda) \\
        &= \sigma_0^2 \left[\| \gamma\|_2^2 + \frac{2}{n_0}\gamma^\top\left( X_0^c M_\lambda^c\Delta \right) + \frac{1}{n_0^2}\Delta^\top M_\lambda^c (X_0^c)^\top X_0^c M_\lambda^c \Delta \right]\\
        &= \sigma_0^2\left[\| \gamma\|_2^2 + \frac{2}{n_0}\gamma^\top\left( X_0^c M_\lambda^c\Delta \right) + \frac{1}{n_0}
        \Delta^\top\left\{M_\lambda^c - \lambda(M_\lambda^c)^2 \right\} \Delta\right]\\
        &= \sigma_0^2\left[\| \gamma\|_2^2 + \frac{2}{n_0}\gamma^\top\left( X_0^c M_\lambda^c\Delta \right) + \frac{1}{n_0}
        \Delta^\top M_\lambda^c
        \left(I_p-\lambda M_\lambda^c\right)\Delta\right] .
    \end{align*}
    Since \(R_n(\lambda)=B_n(\lambda)+V_n(\lambda)\), these two
    displays give \eqref{eq:conditional-risk}.
\end{proof}

\subsection{Proof of Theorem~\ref{thm:risk-asym}}
\label{app:proof-thm-risk-asym}

This subsection proves \Cref{thm:risk-asym}.
The first two subsubsections develop the deterministic-equivalent and quadratic-form tools, the next two analyse the pure and mixed components, and the final subsubsection assembles the proof.

Throughout this subsection, fix an admissible realisation of \(\mathcal F_{\mathrm{aux}}\).
Unless stated otherwise, all probability limits and stochastic-order statements are under the conditional law of \((X_0,X_1)\) given \(\mathcal F_{\mathrm{aux}}\).

\subsubsection{Uniform ridge deterministic equivalents}
Under the source-design \((8+\delta)\)-moment bound in
\Cref{asm:rmt}, we invoke the generalised ridge deterministic equivalent
of \citet[Lemma F.6]{du2023subsample}, specialised to the control design.

For matrix sequences \(K_n(\lambda)\) and \(D_n(\lambda)\), write
\(
    K_n(\lambda)\asympequi D_n(\lambda)
\)
uniformly on \(\Lambda\) if, for every deterministic matrix sequence \(C_n\)
satisfying \(\sup_n\|C_n\|_{\mathrm{tr}}<\infty\),
\[
    \sup_{\lambda\in\Lambda}
    \left|\tr\{C_n(K_n(\lambda)-D_n(\lambda))\}\right|\pto0.
\]
Throughout this subsection, \(\asympequi\) has this uniform trace-test meaning.

For every deterministic matrix sequence \(A=A_n\in\RR^{p\times p}\) satisfying \(A=A^\top\succeq0\) and \(\sup_n\|A\|_{\oper}<\infty\), we establish the centred deterministic equivalents
\begin{align*}
    \lambda^2M_\lambda^cAM_\lambda^c
    &\asympequi
    A_n(\lambda)
    \{\widetilde v_{b,n}(\lambda;A)\Sigma_0+A\}
    A_n(\lambda),\\
    \lambda M_\lambda^c
    &\asympequi A_n(\lambda).
\end{align*}
Here \(v_n\), \(\widetilde v_{b,n}\), and \(A_n\) are the finite-dimensional quantities defined in Section~\ref{sec:risk:asymptotic-risk} from \((\phi_{0,n},\Sigma_0,H_{0,p})\), and \(H_{0,p}\) is used throughout.

\begin{lemma}[Uniform ridge deterministic equivalents]
    \label{lem:uniform-ridge-de}
    Under \Cref{asm:rmt,asm:technical-source-resolvent}, both preceding ridge deterministic equivalents hold uniformly over every compact \(\Lambda\subset(0,\infty)\) in the trace-test sense defined above, for every deterministic matrix sequence \(A=A_n\in\RR^{p\times p}\) satisfying \(A=A^\top\succeq0\) and \(\sup_n\|A\|_{\oper}<\infty\).
\end{lemma}

\begin{proof}[Proof of \Cref{lem:uniform-ridge-de}]
    For each fixed \(\lambda>0\), \citet[Lemma F.6]{du2023subsample} gives the corresponding uncentred deterministic equivalents for \(\lambda M_0\) and \(\lambda^2M_0AM_0\) for deterministic \(A=A^\top\succeq0\) with \(\sup_n\|A\|_{\oper}<\infty\); its i.i.d.\ coordinate and \((8+\delta)\)-moment conditions are imposed in \Cref{asm:rmt}.
    We now transfer these equivalents to \(M_\lambda^c\).
    With \(e_0=B_0^\top u\), Sherman--Morrison and \Cref{lem:source-mean-resolvent} give
    \[
        D_\lambda
        :=
        M_\lambda^c-M_0
        =
        \frac{M_0e_0e_0^\top M_0}{\lambda a_n},
        \qquad
        \inf_{\lambda\in\Lambda}a_n(\lambda)>c
    \]
    with probability tending to one for some \(c>0\).
    For a deterministic rank-one trace test \(C_n=xy^\top\), \Cref{lem:global-linearized-resolvent} gives
    \[
        \sup_{\lambda\in\Lambda}
        \left|\tr(C_nD_\lambda)\right|
        =
        \Op(n_0^{-1+2\zeta})\|x\|_2\|y\|_2
        =
        \op(1)\|x\|_2\|y\|_2,
    \]
    because \(\zeta<1/4\).
    For a general deterministic \(C_n\) with \(\|C_n\|_{\mathrm{tr}}\leq K\), fix \(\tau>0\) and write \(C_n=C_n^{>\tau}+C_n^{\leq\tau}\) by its singular-value decomposition.
    Since \(\rank(C_n^{>\tau})\leq K/\tau\), the preceding uniform rank-one bound gives
    \[
        \sup_{\lambda\in\Lambda}
        \left|\tr(C_n^{>\tau}D_\lambda)\right|
        =
        \op(1)
    \]
    for fixed \(\tau\).
    Since \(\|C_n^{\leq\tau}\|_{\oper}\leq\tau\) and \(\sup_{\lambda\in\Lambda}\|D_\lambda\|_*=\Op(1)\),
    \[
        \sup_{\lambda\in\Lambda}
        \left|\tr(C_n^{\leq\tau}D_\lambda)\right|
        =
        \Op(\tau).
    \]
    Letting \(n_0\to\infty\) and then \(\tau\downarrow0\) proves the first trace-test transfer.

    For the second deterministic equivalent, use the exact decomposition
    \[
        M_\lambda^cAM_\lambda^c-M_0AM_0
        =
        D_\lambda AM_\lambda^c+M_0AD_\lambda.
    \]
    For \(C_n=xy^\top\), the first term satisfies
    \[
        \tr(C_nD_\lambda AM_\lambda^c)
        =
        \frac{(y^\top M_0e_0)(e_0^\top M_0AM_\lambda^cx)}
        {\lambda a_n},
    \]
    and the second is analogous.
    Uniformly over \(\Lambda\), \Cref{lem:global-linearized-resolvent} and the bounded operator norms of \(A\), \(M_0\), and \(M_\lambda^c\) give
    \[
        \left|\tr\{C_n(M_\lambda^cAM_\lambda^c-M_0AM_0)\}\right|
        =
        \Op(n_0^{-1/2+\zeta})\|x\|_2\|y\|_2
        =
        \op(1)\|x\|_2\|y\|_2.
    \]
    Moreover,
    \[
        \sup_{\lambda\in\Lambda}
        \|M_\lambda^cAM_\lambda^c-M_0AM_0\|_*
        =
        \Op(1).
    \]
    The same singular-value argument as above therefore gives the required trace-test convergence.

    Writing \(\lambda_-:=\inf\Lambda>0\), the resolvent identity gives, for each
    fixed \(k\),
    \[
        \sup_{\lambda\in\Lambda}\|(M_\lambda^c)^k\|_{\oper}
        \leq\lambda_-^{-k},
        \qquad
        \sup_{\lambda\in\Lambda}
        \|\partial_{\lambda}(M_\lambda^c)^k\|_{\oper}
        \leq k\lambda_-^{-(k+1)}.
    \]
    Hence, for every deterministic trace test with
    \(\sup_n\|C_n\|_{\mathrm{tr}}<\infty\), the random trace tests are uniformly
    Lipschitz. The fixed-point stability bound
    \[
        \widetilde v_{v,n}(\lambda)^{-1}
        \geq \lambda/v_n(\lambda)\geq\lambda_-^2
    \]
    and implicit differentiation of the fixed-point equation give uniform
    derivative bounds for \(v_n\), \(\widetilde v_{v,n}\),
    \(\widetilde v_{b,n}\), and \(A_n\). 
    Thus the deterministic trace tests are uniformly Lipschitz as well.
    For a deterministic \(\delta\)-net \(\mathcal N_\delta\subset\Lambda\), the supremum error is bounded by the maximum error on \(\mathcal N_\delta\) plus \(\delta\Op(1)\).
    For fixed \(\delta\), the net is finite, so the maximum is \(\op(1)\).
    Letting \(n_0\to\infty\) and then \(\delta\downarrow0\) proves uniform convergence.
\end{proof}

\subsubsection{Quadratic-form extraction and component decomposition}

\begin{lemma}[Scaled quadratic forms from matrix deterministic equivalents]
    \label{lem:quadratic-form-extraction}
    Suppose \(K_n(\lambda)\asympequi D_n(\lambda)\) uniformly over
    \(\lambda\in\Lambda\). If \(u_n\in\RR^p\) is deterministic and
    \[
        \frac{n_0^\eta}{p}\|u_n\|_2^2=\cO(1),
    \]
    then
    \[
        \sup_{\lambda\in\Lambda}\frac{n_0^\eta}{p}
        \left|u_n^\top\{K_n(\lambda)-D_n(\lambda)\}u_n\right|\pto0.
    \]
\end{lemma}

\begin{proof}[Proof of \Cref{lem:quadratic-form-extraction}]
    Take \(C_n=(n_0^\eta/p)u_nu_n^\top\). Then
    \[
        \|C_n\|_{\mathrm{tr}}
        =\frac{n_0^\eta}{p}\|u_n\|_2^2=\cO(1),
    \]
    so the conclusion follows from the definition of \(\asympequi\).
\end{proof}

For the component decomposition, fix \(\lambda\in\Lambda\) and write
\[
    e_0=\bar x_0-\nu_0,
    \qquad
    \epsilon_1=\bar x_1-\nu_1,
    \qquad
    g=C_0\gamma,
    \qquad
    G_0^c=\frac1{n_0}X_0^c(X_0^c)^\top,
    \qquad
    M=M_\lambda^c.
\]
Suppressing the dependence on \(\lambda\), define the signal components
\(d_\nu,d_1,d_e,d_g\in\RR^p\) by
\[
    \begin{aligned}
    d_\nu&=-\lambda M\nu_\Delta,
    &
    d_1&=(I_p-\lambda M)\epsilon_1,\\
    d_e&=\lambda Me_0,
    &
    d_g&=\lambda M(X_0^c)^\top g,
    \end{aligned}
\]
and the centred-weight components
\(h_g,h_\nu,h_1,h_e\in\RR^{n_0}\) by
\[
    \begin{aligned}
    h_g&=\lambda(G_0^c+\lambda I_{n_0})^{-1}g,
    &
    h_\nu&=\frac1{n_0}X_0^cM\nu_\Delta,\\
    h_1&=\frac1{n_0}X_0^cM\epsilon_1,
    &
    h_e&=-\frac1{n_0}X_0^cMe_0.
    \end{aligned}
\]
These definitions give
\begin{align*}
    \Delta
    &=\nu_\Delta+\epsilon_1-e_0-(X_0^c)^\top g,\\
    d_n(\lambda)
    &=d_\nu+d_1+d_e+d_g,\\
    \gamma_\lambda
    &=\frac1{n_0}\one_{n_0}+h_g+h_\nu+h_1+h_e.
\end{align*}
Each centred-weight component belongs to \(\one_{n_0}^{\perp}\).

\subsubsection{Pure-component deterministic equivalents}

\begin{lemma}[Pure component deterministic equivalents]
    \label{lem:pure-component-de}
    Under \Cref{asm:rmt,asm:technical-source-resolvent}, together with \Cref{asm:design-independent-regime},
    \Cref{lem:external-trace-replacement} applies. The pure quadratic terms
    associated with \(d_\nu,d_1,d_e,d_g\) and
    \(h_g,h_\nu,h_1,h_e\) admit the finite-dimensional deterministic
    equivalents displayed in \Cref{thm:risk-asym}, uniformly over \(\Lambda\).
\end{lemma}

\begin{proof}[Proof of \Cref{lem:pure-component-de}]
    Using the component decomposition above, put
    \[
        R_\lambda^c=(G_0^c+\lambda I_{n_0})^{-1}.
    \]
    We use the ridge identities
    \begin{align*}
        M(X_0^c)^\top&=(X_0^c)^\top R_\lambda^c,\\
        X_0^cM^2(X_0^c)^\top
        &=n_0G_0^c(R_\lambda^c)^2
        =n_0\left\{R_\lambda^c-\lambda(R_\lambda^c)^2\right\},\\
        \frac1{n_0}M(X_0^c)^\top X_0^cM&=M-\lambda M^2.
    \end{align*}
    All stochastic remainders below are uniform over \(\Lambda\).

    \noindent\textbf{Centred base-weight components.}
    The exact quadratic forms are
    \begin{align*}
        \frac{n_0^\eta}{p}\|d_g\|_2^2
        &=\frac{n_0^\eta}{p}\lambda^2n_0
        g^\top\left\{R_\lambda^c-\lambda(R_\lambda^c)^2\right\}g,\\
        n_0^\eta\|h_g\|_2^2
        &=n_0^\eta\lambda^2g^\top(R_\lambda^c)^2g.
    \end{align*}
    By \Cref{lem:external-trace-replacement},
    \begin{align*}
        g^\top R_\lambda^cg
        &=\|g\|_2^2\left\{v_n(\lambda)+\op(1)\right\},\\
        g^\top(R_\lambda^c)^2g
        &=\|g\|_2^2
        \left\{\widetilde v_{v,n}(\lambda)+\op(1)\right\}.
    \end{align*}
    Since \(n_0^\eta\|g\|_2^2=\varrho_{\eta,n}^2\) and
    \(n_0/p=\phi_{0,n}^{-1}\), these identities give
    \begin{align*}
        \frac{n_0^\eta}{p}\|d_g\|_2^2
        &=\frac{\lambda^2\varrho_{\eta,n}^2}{\phi_{0,n}}
        \left\{v_n(\lambda)-\lambda\widetilde v_{v,n}(\lambda)\right\}
        +\op(1),\\
        n_0^\eta\|h_g\|_2^2
        &=\lambda^2\varrho_{\eta,n}^2
        \widetilde v_{v,n}(\lambda)+\op(1).
    \end{align*}

    \noindent\textbf{Source-target mean-shift components.}
    For the signal component,
    \[
        \frac{n_0^\eta}{p}\|d_\nu\|_2^2
        =\frac{n_0^\eta}{p}
        \nu_\Delta^\top(\lambda^2M^2)\nu_\Delta.
    \]
    By \Cref{lem:uniform-ridge-de,lem:quadratic-form-extraction},
    \[
        \frac{n_0^\eta}{p}\|d_\nu\|_2^2
        =\rho_{\eta,n}^2
        \int
        \frac{\widetilde v_{b,n}(\lambda)s+1}
        {\{v_n(\lambda)s+1\}^2}\,\rd G_{\nu,p}(s)
        +\op(1).
    \]
    For the centred-weight component,
    \[
        n_0^\eta\|h_\nu\|_2^2
        =\frac{n_0^\eta}{n_0}
        \nu_\Delta^\top(M-\lambda M^2)\nu_\Delta.
    \]
    The deterministic equivalent of \(M-\lambda M^2\) is
    \[
        \frac1\lambda
        \left[
        A_n(\lambda)
        -A_n(\lambda)
        \{\widetilde v_{b,n}(\lambda)\Sigma_0+I_p\}
        A_n(\lambda)
        \right].
    \]
    The scalar identity
    \[
        \frac1{1+v_ns}
        -\frac{\widetilde v_{b,n}s+1}{(1+v_ns)^2}
        =\frac{(v_n-\widetilde v_{b,n})s}{(1+v_ns)^2},
    \]
    with the arguments \(\lambda\) suppressed, and \(p/n_0=\phi_{0,n}\)
    therefore give
    \[
        n_0^\eta\|h_\nu\|_2^2
        =\frac{\phi_{0,n}\rho_{\eta,n}^2}{\lambda}
        \int
        \frac{\{v_n(\lambda)-\widetilde v_{b,n}(\lambda)\}s}
        {\{v_n(\lambda)s+1\}^2}\,\rd G_{\nu,p}(s)
        +\op(1).
    \]

    \noindent\textbf{Treated empirical-mean components.}
    Let \(P_\lambda=I_p-\lambda M\). Then
    \[
        \|d_1\|_2^2=\epsilon_1^\top P_\lambda^2\epsilon_1,
        \qquad
        \|h_1\|_2^2
        =\frac1{n_0}\epsilon_1^\top(M-\lambda M^2)\epsilon_1.
    \]
    Using \Cref{lem:treated-mean-qf}, the expansion
    \(P_\lambda^2=I_p-2\lambda M+\lambda^2M^2\), and
    \Cref{lem:uniform-ridge-de} yields
    \begin{align*}
        \frac{n_0^\eta}{p}\|d_1\|_2^2
        &=n_0^{\eta-1}\frac{\phi_{1,n}}{\phi_{0,n}}
        \left[
        \otr(\Sigma_1)-2K_{1,n}(\lambda;\Sigma_1)
        +K_{2,n}(\lambda;\Sigma_1)
        \right]+\op(1),\\
        n_0^\eta\|h_1\|_2^2
        &=n_0^{\eta-1}\frac{\phi_{1,n}}{\lambda}
        \left[
        K_{1,n}(\lambda;\Sigma_1)-K_{2,n}(\lambda;\Sigma_1)
        \right]+\op(1).
    \end{align*}
    Thus both contributions vanish for \(\eta<1\) and give the corresponding \(\ind{(\eta=1)}\) terms in \Cref{thm:risk-asym} at \(\eta=1\).

    \noindent\textbf{Source empirical-mean and affine components.}
    By \Cref{lem:source-mean-resolvent},
    \begin{align*}
        \frac{n_0^\eta}{p}\|d_e\|_2^2
        &=n_0^{\eta-1}
        \left\{K_{2,n}(\lambda;\Sigma_0)+\op(1)\right\},\\
        n_0^\eta\|h_e\|_2^2
        &=n_0^{\eta-1}
        \left[
        \frac{\phi_{0,n}}{\lambda}
        \{K_{1,n}(\lambda;\Sigma_0)-K_{2,n}(\lambda;\Sigma_0)\}
        +\op(1)
        \right].
    \end{align*}
    These contributions likewise vanish for \(\eta<1\) and give the
    corresponding \(\ind{(\eta=1)}\) terms at \(\eta=1\).
    The affine-normalisation direction \(n_0^{-1}\one_{n_0}\) has squared norm
    \(n_0^{-1}\) and is orthogonal to \(h_g,h_\nu,h_1,h_e\), because all
    four centred components belong to \(\one_{n_0}^{\perp}\). Its scaled
    contribution is therefore \(n_0^{\eta-1}\), which yields the remaining
    \(\ind{(\eta=1)}\) term in the variance equivalent.

    Multiplying the signal contributions by \(r^2\) and the weight contributions by \(\sigma_0^2\) gives the corresponding summands in \(\sB_{\eta,n}(\lambda)\) and \(\sV_{\eta,n}(\lambda)\), uniformly over \(\Lambda\).

\end{proof}

\subsubsection{Mixed-component negligibility}

\begin{lemma}[Mixed component negligibility]
    \label{lem:mixed-component-negligibility}
    Under the conditions of \Cref{lem:pure-component-de}, let
    \(d_\nu,d_1,d_e,d_g\) and \(h_g,h_\nu,h_1,h_e\) be the components
    defined above. Uniformly on \(\Lambda\), each of the six distinct
    signal pairwise products, multiplied by
    \(n_0^\eta/p\), is \(\op(1)\), and each of the six distinct pairwise products
    among the centred-weight components, multiplied by \(n_0^\eta\), is
    \(\op(1)\).
\end{lemma}

\begin{proof}[Proof of \Cref{lem:mixed-component-negligibility}]
    Write \(n=n_0\), so \(p\asymp n\).
    The six signal products are \(d_\nu^\top d_1\), \(d_\nu^\top d_e\), \(d_\nu^\top d_g\), \(d_1^\top d_e\), \(d_1^\top d_g\), and \(d_e^\top d_g\); the six centred-weight products are \(h_g^\top h_\nu\), \(h_g^\top h_1\), \(h_g^\top h_e\), \(h_\nu^\top h_1\), \(h_\nu^\top h_e\), and \(h_1^\top h_e\).
    Conditional on \(\mathcal F_{\mathrm{aux}}\), \(\nu_\Delta\) and \(g\) satisfy deterministic bounds, while
    \[
        \begin{aligned}
        \|\nu_\Delta\|_2&=\cO(\sqrt p\,n^{-\eta/2}), &
        \|\epsilon_1\|_2&=\Op(1), &
        \|g\|_2&=\cO(n^{-\eta/2}),\\
        \|e_0\|_2&=\Op(1), &
        \|X_0^c\|_{\oper}&=\Op(\sqrt n).&&
        \end{aligned}
    \]
    Also, \(\inf_{\lambda\in\Lambda}a_n(\lambda)\) is bounded away from zero with probability tending to one.
    Fix \(\zeta\in(0,1/4)\) as in \Cref{asm:technical-source-resolvent}.
    For the treated-mean bounds, take the source-measurable directions
    \[
        v_\nu=\nu_\Delta,\qquad v_e=e_0,\qquad
        v_g^{(d)}=(X_0^c)^\top g,\qquad
        v_g^{(h)}=(X_0^c)^\top h_g.
    \]
    Their norms are \(\cO(\sqrt p\,n^{-\eta/2})\), \(\Op(1)\), and \(\Op(\sqrt n\,n^{-\eta/2})\) for each \(g\)-direction.
    The bounded \(\mathcal H_0\)-measurable resolvent factors and their \(\lambda\)-derivatives preserve these orders, so \Cref{lem:treated-mean-bilinear} applies with the corresponding \(r_n\).

    \begin{center}
    \scriptsize
    \setlength{\tabcolsep}{3pt}
    \renewcommand{\arraystretch}{1.16}
    \begin{tabular}{p{0.13\textwidth}p{0.20\textwidth}p{0.22\textwidth}p{0.36\textwidth}}
    \hline
    Product & Unscaled order & Normalised order & Controlling input \\
    \hline
    \(d_\nu^\top d_1\)
    & \(\Op(n^{-\eta/2})\)
    & \(\Op(n^{-1+\eta/2})\)
    & \Cref{lem:treated-mean-bilinear}, conditional on the source sample \\
    \(d_\nu^\top d_e\)
    & \(\Op(n^{-\eta/2+\zeta})\)
    & \(\Op(n^{-1+\eta/2+\zeta})\)
    & \Cref{lem:global-linearized-resolvent,lem:source-mean-resolvent} \\
    \(d_\nu^\top d_g\)
    & \(\Op(n^{1/2-\eta+\zeta})\)
    & \(\Op(n^{-1/2+\zeta})\)
    & \Cref{lem:global-linearized-resolvent} \\
    \(d_1^\top d_e\)
    & \(\Op(n^{-1/2})\)
    & \(\Op(n^{-3/2+\eta})\)
    & \Cref{lem:treated-mean-bilinear}, conditional on \(X_0\) \\
    \(d_1^\top d_g\)
    & \(\Op(n^{-\eta/2})\)
    & \(\Op(n^{-1+\eta/2})\)
    & \Cref{lem:treated-mean-bilinear} and the pure \(g\)-quadratic bound \\
    \(d_e^\top d_g\)
    & \(\op(n^{(1-\eta)/2})\)
    & \(\op(n^{(\eta-1)/2})\)
    & Exact reduction below and
    \Cref{lem:global-linearized-resolvent,lem:source-mean-resolvent} \\
    \hline
    \(h_g^\top h_\nu\)
    & \(\Op(n^{-1/2-\eta+\zeta})\)
    & \(\Op(n^{-1/2+\zeta})\)
    & \Cref{lem:global-linearized-resolvent} \\
    \(h_g^\top h_1\)
    & \(\Op(n^{-1-\eta/2})\)
    & \(\Op(n^{-1+\eta/2})\)
    & \Cref{lem:treated-mean-bilinear} and the pure \(g\)-quadratic bound \\
    \(h_g^\top h_e\)
    & \(\op(n^{-(1+\eta)/2})\)
    & \(\op(n^{(\eta-1)/2})\)
    & Exact reduction below and
    \Cref{lem:centered-row-compression,lem:global-linearized-resolvent} \\
    \(h_\nu^\top h_1\)
    & \(\Op(n^{-1-\eta/2})\)
    & \(\Op(n^{-1+\eta/2})\)
    & \Cref{lem:treated-mean-bilinear} \\
    \(h_\nu^\top h_e\)
    & \(\Op(n^{-1-\eta/2+\zeta})\)
    & \(\Op(n^{-1+\eta/2+\zeta})\)
    & \Cref{lem:global-linearized-resolvent,lem:source-mean-resolvent} \\
    \(h_1^\top h_e\)
    & \(\Op(n^{-3/2})\)
    & \(\Op(n^{-3/2+\eta})\)
    & \Cref{lem:treated-mean-bilinear}, conditional on \(X_0\) \\
    \hline
    \end{tabular}
    \end{center}

    For terms involving \(\epsilon_1\), condition on
    \(\mathcal H_0=\sigma(X_0,\mathcal F_{\mathrm{aux}})\) and apply
    \Cref{lem:treated-mean-bilinear} with the four directions displayed above.
    Uniformity follows by differentiating the bounded resolvent factors; the
    resulting scaled direction and matrix derivative bounds are precisely those
    in \Cref{lem:treated-mean-bilinear}. The explicit \(n^{-1}\) factors in the
    weight components give the corresponding rows. All fixed resolvent powers
    use contours contained in the same
    \(\mathcal D_\Lambda(r_\Lambda)\).

    For \(d_\nu^\top d_g\) and \(h_g^\top h_\nu\), conditional on \(\mathcal F_{\mathrm{aux}}\), the vectors \(g\) and \(\nu_\Delta\) are fixed and independent of \(Z_0\).
    Applying
    the cross-block conclusion of \Cref{lem:global-linearized-resolvent} and the exact Sherman--Morrison identity
    \[
        M_\lambda^c-M_0=\frac{M_0e_0e_0^\top M_0}{\lambda a_n}
    \]
    gives the stated orders for \(d_\nu^\top d_g\) and
    \(h_g^\top h_\nu\).

    For an \(\epsilon_1\)-term controlled through the source resolvent, use
    \[
        \mathcal A_n:=\sigma(\mathcal F_{\mathrm{aux}},X_1).
    \]
    This sigma-field is independent of \(Z_0\), while \(\epsilon_1\), \(g\), and
    \(\nu_\Delta\) are \(\mathcal A_n\)-measurable. The conditional clause of
    \Cref{asm:technical-source-resolvent} therefore applies under the
    conditional source-design law; the
    \(\Op(1)\) norm is handled on the usual increasing high-probability events.

    For the source-mean cross-form bounds, use
    \[
        M_\lambda^c e_0=\frac{B_0^\top R_0u}{\lambda a_n},
        \qquad
        h_e=\frac{C_0R_0u}{\sqrt n\,a_n},
    \]
    where the expression for \(h_e\) follows from
    \[
        h_e=-\frac1nX_0^cM_\lambda^ce_0
        \qquad\text{and}\qquad
        C_0B_0M_\lambda^ce_0=-\frac{C_0R_0u}{a_n}.
    \]
    The cross-form conclusion of \Cref{lem:source-mean-resolvent} gives
    \[
        w^\top (M_\lambda^c)^ke_0
        =\Op(n^{-1/2+\zeta}\|w\|_2),
        \qquad w\in\{\nu_\Delta,\epsilon_1\}.
    \]

    It remains to treat \(d_e^\top d_g\) and \(h_g^\top h_e\). Write
    \[
        a=a_n(\lambda)=u^\top R_0u,
        \qquad
        b=b_n(\lambda)=u^\top R_0^2u,
        \qquad
        c_j=g^\top R_0^ju,\quad j=1,2.
    \]
    Let \(M_0=(B_0^\top B_0+\lambda I_p)^{-1}\). By Sherman--Morrison and
    \(M_0B_0^\top=B_0^\top R_0\),
    \[
        M_\lambda^ce_0=\frac{B_0^\top R_0u}{\lambda a}.
    \]
    Moreover, because \(u^\top g=0\),
    \begin{align*}
        e_0^\top M_0B_0^\top g
        &=u^\top B_0M_0B_0^\top g\\*
        &=u^\top(I_n-\lambda R_0)g
        =-\lambda c_1.
    \end{align*}
    Hence
    \[
        M_\lambda^cB_0^\top g
        =B_0^\top R_0g-\frac{c_1}{a}B_0^\top R_0u.
    \]
    Using \(X_0^c=\sqrt n\,C_0B_0\) and \(C_0g=g\), it follows that
    \[
        d_e=\frac{B_0^\top R_0u}{a},
        \qquad
        d_g=\lambda\sqrt n
        \left(B_0^\top R_0g-\frac{c_1}{a}B_0^\top R_0u\right).
    \]
    Since \(R_0G_0R_0=R_0-\lambda R_0^2\), direct multiplication gives
    \begin{align*}
        d_e^\top d_g
        &=\frac{\lambda\sqrt n}{a}
        \left[
        u^\top R_0G_0R_0g
        -\frac{c_1}{a}u^\top R_0G_0R_0u
        \right]\\
        &=\frac{\lambda^2\sqrt n}{a^2}
        \left(-ac_2+bc_1\right).
    \end{align*}

    The centred row-resolvent identity in
    \Cref{lem:centered-row-compression} gives
    \[
        h_g=\lambda\left(R_0g-\frac{c_1}{a}R_0u\right),
        \qquad
        h_e=\frac{C_0R_0u}{\sqrt n\,a}.
    \]
    Because \(h_g\in\one_n^\perp\),
    \(h_g^\top C_0R_0u=h_g^\top R_0u\). Therefore,
    \begin{align*}
        h_g^\top h_e
        &=\frac{\lambda}{\sqrt n\,a}
        \left[
        g^\top R_0^2u-\frac{c_1}{a}u^\top R_0^2u
        \right]\\
        &=\frac{\lambda}{\sqrt n\,a^2}
        \left(ac_2-bc_1\right).
    \end{align*}

    Since \(g^\top u=0\), the \(k=1,2\) row-resolvent conclusions in
    \Cref{lem:global-linearized-resolvent} yield, uniformly over \(\Lambda\),
    \[
        \sup_{\lambda\in\Lambda}|c_j(\lambda)|
        =\op(\|g\|_2),
        \qquad j=1,2.
    \]
    Moreover, \(\inf_{\lambda\in\Lambda}a_n(\lambda)\) is bounded away from
    zero with probability tending to one, while
    \(\sup_{\lambda\in\Lambda}b_n(\lambda)=\Op(1)\). Since
    \(\|g\|_2=\cO(n^{-\eta/2})\), the exact formulas imply
    \begin{align*}
        \sup_{\lambda\in\Lambda}|d_e^\top d_g|
        &=\op\left(n^{(1-\eta)/2}\right),\\
        \sup_{\lambda\in\Lambda}|h_g^\top h_e|
        &=\op\left(n^{-(1+\eta)/2}\right).
    \end{align*}
    Because \(p\asymp n\), after normalisation,
    \[
        \frac{n^\eta}{p}
        \sup_{\lambda\in\Lambda}|d_e^\top d_g|
        =\op\left(n^{(\eta-1)/2}\right)=\op(1),
        \qquad
        n^\eta\sup_{\lambda\in\Lambda}|h_g^\top h_e|
        =\op\left(n^{(\eta-1)/2}\right)=\op(1).
    \]
    For \(\eta<1\), the displayed deterministic factor vanishes; at \(\eta=1\), the little-\(\op\) row-resolvent bounds give convergence.
    Together with the bounds in the table, all twelve normalised pairwise products are \(\op(1)\), completing the proof.
\end{proof}

\subsubsection[Proof of Theorem \ref{thm:risk-asym}]{Proof of \Cref{thm:risk-asym}}

\begin{proof}[Proof of \Cref{thm:risk-asym}]
    By \Cref{prop:global-source-resolvent-sufficient},
    \Cref{asm:rmt,asm:global-source-resolvent} imply
    \Cref{asm:technical-source-resolvent}.
    Recall the component decompositions \(d_n=d_\nu+d_1+d_e+d_g\) and \(\gamma_\lambda=n_0^{-1}\one_{n_0}+h_g+h_\nu+h_1+h_e\).
    By \Cref{lem:pure-component-de}, the pure signal and centred-weight components admit the stated uniform deterministic equivalents, including the source empirical-mean terms.
    The affine-normalisation direction contributes the remaining \(\ind{(\eta=1)}\sigma_0^2\) term.

    By \Cref{lem:mixed-component-negligibility}, all distinct pairwise products among the signal components and among the centred-weight components are negligible after the corresponding normalisation.
    The affine direction is orthogonal to the centred-weight components.
    Therefore,
    \[
        \sup_{\lambda\in\Lambda}
        \left|n_0^\eta B_n(\lambda)-B_{\eta,n}(\lambda)\right|
        \pto0
    \]
    and
    \[
        \sup_{\lambda\in\Lambda}
        \left|n_0^\eta V_n(\lambda)-V_{\eta,n}(\lambda)\right|
        \pto0.
    \]
    Since \(R_n=B_n+V_n\), the corresponding conclusion for \(R_n\) follows.

    These convergences hold under the conditional design law for \(\PP\)-almost every admissible realisation of \(\mathcal F_{\mathrm{aux}}\).
    Since the corresponding conditional error probabilities converge almost surely to zero and are bounded by one, dominated convergence gives the same conclusions under the joint law.
\end{proof}

\subsection{Exact risk characterisation and analytic properties}
\label{app:proof-fixed-limit-corollaries}

\subsubsection[Proof of Corollary \ref{cor:fixed-risk-limit}]{Proof of \Cref{cor:fixed-risk-limit}}

\begin{proof}[Proof of \Cref{cor:fixed-risk-limit}]
    By \Cref{asm:rmt}, the prelimit measures, and hence the limits in \Cref{asm:spectral-limits}, are supported on the common compact interval \([c_\Sigma,C_\Sigma]\).
    Under \Cref{asm:spectral-limits}, let \(v(\lambda)>0\) solve
    \[
        \frac1{v(\lambda)}
        =\lambda+\phi_0\int\frac{s}{1+v(\lambda)s}\,\rd H_0(s),
    \]
    and define
    \begin{align*}
        \widetilde v_v(\lambda)
        &:=\left[v(\lambda)^{-2}
        -\phi_0\int\frac{s^2}{\{1+v(\lambda)s\}^2}\,\rd H_0(s)\right]^{-1},\\
        \widetilde v_{b,0}(\lambda)
        &:=\frac{\phi_0\int s\{1+v(\lambda)s\}^{-2}\,\rd H_0(s)}
        {\widetilde v_v(\lambda)^{-1}},\\
        \widetilde v_{b,1}(\lambda)
        &:=\frac{\phi_0\int s\{1+v(\lambda)s\}^{-2}\,\rd H_{1\mid0}(s)}
        {\widetilde v_v(\lambda)^{-1}},\\
        \widetilde v_{b,00}(\lambda)
        &:=\frac{\phi_0\int s^2\{1+v(\lambda)s\}^{-2}\,\rd H_0(s)}
        {\widetilde v_v(\lambda)^{-1}}.
    \end{align*}
    Write
    \begin{align*}
        J_{\nu,b}(\lambda)
        &:=\int\frac{\widetilde v_{b,0}(\lambda)s+1}
        {\{1+v(\lambda)s\}^2}\,\rd G_\nu(s),\\
        J_{\nu,v}(\lambda)
        &:=\int\frac{\{v(\lambda)-\widetilde v_{b,0}(\lambda)\}s}
        {\{1+v(\lambda)s\}^2}\,\rd G_\nu(s),\\
        T_0&:=\int1\,\rd H_{1\mid0}(s),\\
        T_1(\lambda)&:=\int\frac1{1+v(\lambda)s}\,\rd H_{1\mid0}(s),\\
        T_2(\lambda)&:=\widetilde v_{b,1}(\lambda)
        \int\frac{s}{\{1+v(\lambda)s\}^2}\,\rd H_0(s)
        +\int\frac1{\{1+v(\lambda)s\}^2}\,\rd H_{1\mid0}(s),\\
        U_1(\lambda)
        &:=\int\frac{s}{1+v(\lambda)s}\,\rd H_0(s),\\
        U_2(\lambda)
        &:=[1+\widetilde v_{b,00}(\lambda)]
        \int\frac{s}{\{1+v(\lambda)s\}^2}\,\rd H_0(s).
    \end{align*}
    The fixed-limit functions are
    \begin{align}
        \sB_\eta(\lambda)
        :=&\frac{\lambda^2r^2\varrho_\eta^2}{\phi_0}
        \{v(\lambda)-\lambda\widetilde v_v(\lambda)\}
        +r^2\rho_\eta^2J_{\nu,b}(\lambda)\notag\\
        &+\ind{(\eta=1)}\frac{r^2\phi_1}{\phi_0}
        \{T_0-2T_1(\lambda)+T_2(\lambda)\}
        +\ind{(\eta=1)}r^2U_2(\lambda),\notag\\
        \sV_\eta(\lambda)
        :=&\sigma_0^2\varrho_\eta^2\lambda^2
        \widetilde v_v(\lambda)
        +\frac{\sigma_0^2\phi_0\rho_\eta^2}{\lambda}
        J_{\nu,v}(\lambda)\notag\\
        &+\ind{(\eta=1)}\frac{\sigma_0^2\phi_1}{\lambda}
        \{T_1(\lambda)-T_2(\lambda)\}\notag\\
        &+\ind{(\eta=1)}\sigma_0^2
        \left[1+\frac{\phi_0}{\lambda}
        \{U_1(\lambda)-U_2(\lambda)\}\right],\notag\\
        \sR_\eta:=&\sB_\eta+\sV_\eta.
        \label{eq:fixed-limit-functions}
    \end{align}

    Fix an admissible auxiliary realisation in the probability-one set on
    which \Cref{asm:spectral-limits} holds and for which the conditional
    conclusions of \Cref{thm:risk-asym} apply. Throughout this proof,
    probability limits are taken under the conditional law of \((X_0,X_1)\)
    given \(\mathcal F_{\mathrm{aux}}\). The fixed-point equations and
    the common spectral bounds imply that \(v_n(\lambda)\) and \(v(\lambda)\) lie in a
    common compact subset of \((0,\infty)\), uniformly over
    \(\lambda\in\Lambda\). Weak convergence of the spectral measures therefore
    yields uniform convergence of the corresponding resolvent integrals over
    \(\Lambda\) and this compact set.

    The fixed-point equations are uniformly stable. Indeed,
    \[
        \widetilde v_{v,n}(\lambda)^{-1}
        =
        \frac{\lambda}{v_n(\lambda)}
        +
        \phi_{0,n}\int
        \frac{s}
        {v_n(\lambda)\{1+v_n(\lambda)s\}^2}
        \,\rd H_{0,p}(s)
        \geq \lambda_-^2.
    \]
    A uniform mean-value argument then gives
    \[
        \sup_{\lambda\in\Lambda}
        |v_n(\lambda)-v(\lambda)|
        \to0.
    \]
    Since the denominators defining the derivative transforms are uniformly
    bounded away from zero, it follows that, uniformly on \(\Lambda\),
    \[
        \begin{aligned}
            \widetilde v_{v,n}
            &\to \widetilde v_v,
            &
            \widetilde v_{b,n}(\,\cdot\,;I_p)
            &\to \widetilde v_{b,0},\\
            \widetilde v_{b,n}(\,\cdot\,;\Sigma_1)
            &\to \widetilde v_{b,1},
            &
            \widetilde v_{b,n}(\,\cdot\,;\Sigma_0)
            &\to \widetilde v_{b,00}.
        \end{aligned}
    \]
    Here the limits involving \(\Sigma_1\) follow from
    \begin{equation}\label{eq:weighted-spectral-trace}
        \otr\{\Sigma_1f(\Sigma_0)\}
        =
        \int f(s)\,\rd H_{1\mid0,p}(s)
    \end{equation}
    for every bounded Borel function \(f\); no commutativity between
    \(\Sigma_0\) and \(\Sigma_1\) is required. The same argument, using
    \(G_{\nu,p}\dto G_\nu\), gives uniform convergence of
    \(J_{\nu,b,n}\) and \(J_{\nu,v,n}\) to
    \(J_{\nu,b}\) and \(J_{\nu,v}\). Similarly, the trace terms converge
    uniformly to \(T_0,T_1,T_2,U_1\), and \(U_2\). In particular,
    \[
        \begin{aligned}
            K_{2,n}(\lambda;\Sigma_0)
            &\to U_2(\lambda),\\
            K_{1,n}(\lambda;\Sigma_0)
            -K_{2,n}(\lambda;\Sigma_0)
            &\to U_1(\lambda)-U_2(\lambda),
        \end{aligned}
    \]
    uniformly over \(\lambda\in\Lambda\).

    Combining these limits with
    \(\phi_{t,n}\to\phi_t\),
    \(\rho_{\eta,n}^2\to\rho_\eta^2\), and
    \(\varrho_{\eta,n}^2\to\varrho_\eta^2\), and substituting into the
    deterministic equivalents in \Cref{thm:risk-asym}, yields
    \[
        \sup_{\lambda\in\Lambda}
        \max\left\{
            |\sB_{\eta,n}(\lambda)-\sB_\eta(\lambda)|,
            |\sV_{\eta,n}(\lambda)-\sV_\eta(\lambda)|,
            |\sR_{\eta,n}(\lambda)-\sR_\eta(\lambda)|
        \right\}
        \to0.
    \]
    The triangle inequality and \Cref{thm:risk-asym} then give
    \begin{align*}
        &\sup_{\lambda\in\Lambda}
        \max\left\{
            |n_0^\eta B_n(\lambda)-\sB_\eta(\lambda)|,
            |n_0^\eta V_n(\lambda)-\sV_\eta(\lambda)|,
            |n_0^\eta R_n(\lambda)-\sR_\eta(\lambda)|
        \right\}\\
        &\quad\leq
        \sup_{\lambda\in\Lambda}
        \max\left\{
            |n_0^\eta B_n(\lambda)-\sB_{\eta,n}(\lambda)|,
            |n_0^\eta V_n(\lambda)-\sV_{\eta,n}(\lambda)|,
            |n_0^\eta R_n(\lambda)-\sR_{\eta,n}(\lambda)|
        \right\}\\
        &\qquad+
        \sup_{\lambda\in\Lambda}
        \max\left\{
            |\sB_{\eta,n}(\lambda)-\sB_\eta(\lambda)|,
            |\sV_{\eta,n}(\lambda)-\sV_\eta(\lambda)|,
            |\sR_{\eta,n}(\lambda)-\sR_\eta(\lambda)|
        \right\}
        \pto0.
    \end{align*}

    For each fixed \(\lambda\in\Lambda\), the preceding convergence implies
    \[
        \begin{aligned}
            B_n(\lambda)
            &=n_0^{-\eta}\{\sB_\eta(\lambda)+\op(1)\},\\
            V_n(\lambda)
            &=n_0^{-\eta}\{\sV_\eta(\lambda)+\op(1)\},\\
            R_n(\lambda)
            &=n_0^{-\eta}\{\sR_\eta(\lambda)+\op(1)\}.
        \end{aligned}
    \]
    If \(\eta>0\), then \(R_n(\lambda)\pto0\). If, in addition,
    \(\sR_\eta(\lambda)>0\), then
    \(n_0^\eta R_n(\lambda)\pto\sR_\eta(\lambda)>0\); if
    \(\sR_\eta(\lambda)=0\), then
    \(R_n(\lambda)=\op(n_0^{-\eta})\). When \(\eta=0\),
    \(R_n(\lambda)\pto\sR_0(\lambda)\), so the conditional random-effects
    prediction risk vanishes in probability if and only if
    \(\sR_0(\lambda)=0\). The componentwise conclusions follow identically.

    The unconditional conclusion follows from the same conditional-probability
    and dominated-convergence argument used in the proof of
    \Cref{thm:risk-asym}.
\end{proof}

\subsubsection[Proof of Corollary \ref{cor:target-sampling-invariant}]{Proof of \Cref{cor:target-sampling-invariant}}
\begin{proof}[Proof of \Cref{cor:target-sampling-invariant}]
    Under \Cref{asm:spectral-limits}, when \(\eta\in[0,1)\), all terms
    involving \((\phi_1,H_{1\mid0})\) vanish from the fixed limits.
\end{proof}

\subsubsection[Regularity of the limiting risk path (Lemma \ref{lem:risk-lipschitz})]{Regularity of the limiting risk path (\Cref{lem:risk-lipschitz})}

\begin{lemma}[Regularity of the limiting risk path]
    \label{lem:risk-lipschitz}
    Under \Cref{asm:rmt,asm:design-independent-regime,asm:spectral-limits}, for every compact \(\Lambda=[\lambda_-,\lambda_+]\subset(0,\infty)\), the maps \(\sB_\eta\), \(\sV_\eta\), and \(\sR_\eta\) are Lipschitz continuous on \(\Lambda\) and real analytic on \(\RR^{++}\).
\end{lemma}

\begin{proof}[Proof of \Cref{lem:risk-lipschitz}]
    For the fixed-point equation
    \(v^{-1}=\lambda+\phi_0\int s(1+vs)^{-1}\,\rd H_0(s)\), we have
    \(v^{-1}\geq\lambda\geq\lambda_-\) and
    \(v^{-1}\leq\lambda_++\phi_0C_\Sigma\), so
    \[
        0 < \underline{v}:= \frac{1}{\lambda_+ + \phi_0 C_\Sigma}
        \leq v(\lambda) \leq \frac{1}{\lambda_-}.
    \]
    Multiplying the fixed-point equation by \(1/v\) gives
    \[
        \frac{1}{v^2} = \frac{\lambda}{v}
        + \phi_0\int \frac{s}{v(1 + vs)}\,\rd H_0(s),
    \]
    and since
    \[
        \frac{s}{v(1 + sv)} - \frac{s^2}{(1 + sv)^2} = \frac{s}{v(1 + vs)^2},
    \]
    we obtain
    \begin{equation}
        \widetilde v_v(\lambda)^{-1}
        = \frac{1}{v^2}
        - \phi_0\int \frac{s^2}{(1 + sv)^2}\,\rd H_0(s)
        = \frac{\lambda}{v}
        + \phi_0\int\frac{s}{v(1 + vs)^2}\,\rd H_0(s)
        \geq \frac{\lambda}{v} \geq \lambda_-^2.
        \label{eq:tvv-bounded}
    \end{equation}
    Thus \(\widetilde v_v\) is positive and bounded by \(\lambda_-^{-2}\).
    For \(F(v,\lambda)=v^{-1}-\lambda-\phi_0\int s(1+sv)^{-1}\,\rd H_0(s)\), \(\partial_vF=-\widetilde v_v^{-1}<0\), so the implicit function theorem gives
    \[
        \frac{\rd}{\rd\lambda}v(\lambda)
        =-\widetilde v_v(\lambda)\in[-\lambda_-^{-2},0].
    \]

    The bounded spectral supports imply that the component maps defining \(\sB_\eta\) and \(\sV_\eta\), and their first derivatives with respect to \((v,\lambda)\), are uniformly bounded for \(\lambda\in\Lambda\) and \(v\in[\underline v,\lambda_-^{-1}]\).
    Together with \eqref{eq:tvv-bounded} and the bound on \(v'\), this proves Lipschitz continuity on \(\Lambda\).
    Since \(F\) and the component maps are real analytic on their domains and \(\partial_vF\neq0\), the real-analytic implicit function theorem gives real analyticity of \(\sB_\eta\), \(\sV_\eta\), and \(\sR_\eta\) on \(\RR^{++}\).
\end{proof}

\subsubsection[Monotonicity (Proposition \ref{prop:bias-monotonicity})]{Monotonicity (\Cref{prop:bias-monotonicity})}
\begin{proposition}[Monotonicity of the integrated squared bias for \(\eta<1\)]
    \label{prop:bias-monotonicity}
    Under the assumptions of \Cref{cor:fixed-risk-limit}, if
    \(\eta\in[0,1)\), then
    \(\lambda\mapsto\sB_\eta(\lambda)\) is nondecreasing on \(\Lambda\).
\end{proposition}
\begin{proof}[Proof of \Cref{prop:bias-monotonicity}]
    Let
    \[
        S_n^c=\frac{1}{n_0}(X_0^c)^\top X_0^c,
        \qquad
        M_{n,\lambda}^c=(S_n^c+\lambda I_p)^{-1}.
    \]
    
    From \Cref{cor:fixed-risk-limit}, the leading integrated squared-bias
    component is
    \[
        \sB_{\eta}(\lambda)
        =
        \frac{\lambda^2 r^2\varrho_\eta^2}{\phi_0}
        \{v(\lambda)-\lambda\widetilde v_v(\lambda)\}
        +
        {r^2\rho_\eta^2}
        \int
        \frac{\widetilde v_{b,0}(\lambda)s+1}{\{v(\lambda)s+1\}^2}
        \,\mathrm d G_\nu(s).
    \]
    Here \(s\) denotes the spectral integration variable. The two terms above
    are deterministic equivalents of the two prelimit integrated extrapolation quadratic forms
    \[
        \mathscr B_{\gamma,n}(\lambda)
        :=
        \frac{r^2 n_0^\eta}{p}
        \lambda^2
        g^\top X_0^c (M_{n,\lambda}^c)^2 (X_0^c)^\top g
    \]
    and
    \[
        \mathscr B_{\nu,n}(\lambda)
        :=
        \frac{r^2 n_0^\eta}{p}
        \lambda^2
        \nu_\Delta^\top (M_{n,\lambda}^c)^2\nu_\Delta,
    \]
    up to the normalisations used in \Cref{thm:risk-asym}. By
    \Cref{lem:pure-component-de}, \Cref{lem:mixed-component-negligibility} and \Cref{cor:fixed-risk-limit},
    for each fixed
    \(\lambda\in[\lambda_-,\lambda_+]\), where \(0<\lambda_-<\lambda_+<\infty\),
    \[
        \mathscr B_{\gamma,n}(\lambda)
        \pto
        \frac{\lambda^2 r^2\varrho_\eta^2}{\phi_0}
        \{v(\lambda)-\lambda\widetilde v_v(\lambda)\},
    \]
    and
    \[
        \mathscr B_{\nu,n}(\lambda)
        \pto
        r^2\rho_\eta^2
        \int
        \frac{\widetilde v_{b,0}(\lambda)s+1}{\{v(\lambda)s+1\}^2}
        \,\mathrm d G_\nu(s).
    \]
    
    We first show that the prelimit quadratic forms are monotone in \(\lambda\). Since
    \[
        \frac{\mathrm d}{\mathrm d\lambda}M_{n,\lambda}^c
        =
        -(M_{n,\lambda}^c)^2,
    \]
    for any vector \(u\in\RR^p\) that does not vary with \(\lambda\),
    \[
        \frac{\mathrm d}{\mathrm d\lambda}
        \left\{
        \lambda^2 u^\top (M_{n,\lambda}^c)^2u
        \right\}
        =
        2\lambda u^\top (M_{n,\lambda}^c)^2u
        -
        2\lambda^2 u^\top (M_{n,\lambda}^c)^3u.
    \]
    Using
    \[
        I_p-\lambda M_{n,\lambda}^c
        =
        I_p-\lambda(S_n^c+\lambda I_p)^{-1}
        =
        S_n^c(S_n^c+\lambda I_p)^{-1}
        =
        S_n^c M_{n,\lambda}^c,
    \]
    we obtain
    \[
        \frac{\mathrm d}{\mathrm d\lambda}
        \left\{
        \lambda^2 u^\top (M_{n,\lambda}^c)^2u
        \right\}
        =
        2\lambda u^\top (M_{n,\lambda}^c)^2
        \{I_p-\lambda M_{n,\lambda}^c\}u
        =
        2\lambda u^\top S_n^c (M_{n,\lambda}^c)^3u
        \ge 0.
    \]
    The last inequality follows because \(S_n^c\succeq0\),
    \(M_{n,\lambda}^c\succ0\), and \(S_n^c\) commutes with
    \(M_{n,\lambda}^c\). Hence
    \[
        \lambda\mapsto
        \lambda^2u^\top (M_{n,\lambda}^c)^2u
    \]
    is nondecreasing on \([\lambda_-,\lambda_+]\).
    
    Applying this identity with \(u=\nu_\Delta\) gives
    \[
        \frac{\mathrm d}{\mathrm d\lambda}
        \mathscr B_{\nu,n}(\lambda)
        =
        \frac{2r^2 n_0^\eta\lambda}{p}
        \nu_\Delta^\top S_n^c (M_{n,\lambda}^c)^3\nu_\Delta
        \ge0.
    \]
    Similarly, applying it with \(u=(X_0^c)^\top g\) gives
    \[
        \frac{\mathrm d}{\mathrm d\lambda}
        \mathscr B_{\gamma,n}(\lambda)
        =
        \frac{2r^2 n_0^\eta\lambda}{p}
        g^\top X_0^c S_n^c (M_{n,\lambda}^c)^3
        (X_0^c)^\top g
        \ge0.
    \]
    Therefore both \(\mathscr B_{\nu,n}(\lambda)\) and
    \(\mathscr B_{\gamma,n}(\lambda)\) are nondecreasing functions of
    \(\lambda\).
    
    It remains to transfer this monotonicity to the deterministic equivalents. Define
    \[
        B_{\gamma}(\lambda)
        :=
        \frac{\lambda^2 r^2\varrho_\eta^2}{\phi_0}
        \{v(\lambda)-\lambda\widetilde v_v(\lambda)\},
    \]
    and
    \[
        B_{\nu}(\lambda)
        :=
        {r^2\rho_\eta^2}
        \int
        \frac{\widetilde v_{b,0}(\lambda)s+1}{\{v(\lambda)s+1\}^2}
        \,\mathrm d G_\nu(s).
    \]
    Fix any
    \(\lambda_1,\lambda_2\in[\lambda_-,\lambda_+]\) with
    \(\lambda_1<\lambda_2\).
    Pointwise deterministic-equivalent convergence at these two values gives
    \[
        \mathscr B_{\gamma,n}(\lambda_2)-\mathscr B_{\gamma,n}(\lambda_1)
        \pto
        B_{\gamma}(\lambda_2)-B_{\gamma}(\lambda_1),
    \]
    and
    \[
        \mathscr B_{\nu,n}(\lambda_2)-\mathscr B_{\nu,n}(\lambda_1)
        \pto
        B_{\nu}(\lambda_2)-B_{\nu}(\lambda_1).
    \]
    Each prelimit difference is nonnegative, so both deterministic limits are nonnegative. Uniform convergence is not required for this transfer.
    Thus both \(B_{\gamma}\) and \(B_{\nu}\) are nondecreasing on \([\lambda_-,\lambda_+]\). Consequently,
    \[
        \sB_{\eta}(\lambda)
        =
        B_{\gamma}(\lambda)+B_{\nu}(\lambda)
    \]
    is also nondecreasing in \(\lambda\).
    
    By \Cref{lem:risk-lipschitz}, \(\sB_\eta\) is differentiable on \((0,\infty)\), and hence
    \[
        \frac{\mathrm d}{\mathrm d\lambda}
        \sB_{\eta}(\lambda)
        \geq0,
        \qquad
        \lambda\in[\lambda_-,\lambda_+].
    \]
    Hence the leading integrated squared-bias component is
    nondecreasing on every compact subinterval of \((0,\infty)\). This
    argument does not cover the ridgeless boundary \(\lambda\downarrow0\) or
    the limit \(\lambda\to\infty\).
\end{proof}

\subsection{Variance-component estimation}
\label{app:proof-variance-components}

This subsection establishes the results for variance-component estimation and proves \Cref{prop:spectral-quasi-likelihood}.

\subsubsection[Source trace moments (Lemma \ref{lem:source-trace-moments})]{Source trace moments (\Cref{lem:source-trace-moments})}

\begin{lemma}[Source trace moments]
    \label{lem:source-trace-moments}
    Suppose the conditions on the control design \(X_0\) in \Cref{asm:rmt} hold and \(p/n_0\to\phi_0\in(0,\infty)\).
    Write
    \[
        \begin{aligned}
            m&=n_0-1, &
            \widetilde\phi_{0,n}&=\frac pm, &
            W_0&=\frac1m(X_0^c)^\top X_0^c,\\
            D_0&=\otr\{(W_0)^2\}
            -\widetilde\phi_{0,n}\otr(W_0)^2, &
            \tau_{k,p}&=\otr(\Sigma_0^k).&&
        \end{aligned}.
    \]
    With \(Q_0\) as in Section~\ref{sec:variance-components}, define \(K_m=p^{-1}Q_0^\top X_0^c(X_0^c)^\top Q_0\), and let \(d_1,\ldots,d_m\) be its eigenvalues.
    Define
    \[
        \overline d_m=\frac1m\sum_{i=1}^m d_i,
        \qquad
        q(d)=(d,1)^\top,
        \qquad
        G_m=\frac1m\sum_{i=1}^m q(d_i)q(d_i)^\top.
    \]
    Then
    \[
        \|W_0\|_{\oper}=\Op(1),
        \qquad
        \otr(W_0)=\tau_{1,p}+\op(1),
        \qquad
        \otr\{(W_0)^2\}
        =\widetilde\phi_{0,n}\tau_{1,p}^2+\tau_{2,p}+\op(1).
    \]
    Consequently,
    \[
        D_0=\tau_{2,p}+\op(1),
        \qquad
        \inf_p\tau_{2,p}\geq c_\Sigma^2>0,
        \qquad
        \max_{2\leq k\leq4}\otr\{(W_0)^k\}=\Op(1).
    \]
    Moreover,
    \[
        \max_{i\leq m}d_i=\Op(1),
        \qquad
        \lambda_{\min}(G_m)^{-1}=\Op(1).
    \]
\end{lemma}

\begin{proof}[Proof of \Cref{lem:source-trace-moments}]
    Write \(n=n_0\), let \(y_i^\top\) denote the rows of \(Y_0=Z_0\Sigma_0^{1/2}\), and define
    \[
        S_n=\frac1nY_0^\top Y_0,
        \qquad
        \overline y=\frac1n\sum_{i=1}^n y_i=\bar x_0-\nu_0,
        \qquad
        S_n^c=S_n-\overline y\,\overline y^\top.
    \]
    Then
    \[
        \frac1n(X_0^c)^\top X_0^c=S_n^c,
        \qquad
        W_0=\frac nmS_n^c.
    \]
    Since \(\|S_n\|_{\oper}\leq\|\Sigma_0\|_{\oper}\|n^{-1}Z_0^\top Z_0\|_{\oper}\), the Bai--Yin extreme-eigenvalue bound gives \(\|S_n\|_{\oper}=\Op(1)\); see \citet{bai2010spectral}.
    Also, \(\EE\|\overline y\|_2^2=\tr(\Sigma_0)/n=\cO(p/n)=\cO(1)\), so \(\|\overline y\|_2^2=\Op(1)\), \(\|S_n^c\|_{\oper}=\Op(1)\), and \(\|W_0\|_{\oper}=\Op(1)\).

    For the first trace moment,
    \[
        \EE\otr(S_n)=\tau_{1,p},
        \qquad
        \Var\{\otr(S_n)\}
        \leq
        \frac{C\tr(\Sigma_0^2)}{p^2n}
        =
        \cO\!\left(\frac1{pn}\right),
    \]
    where the variance bound follows from the independence and bounded fourth moments of the coordinates of \(Z_0\), together with the bounded spectrum of \(\Sigma_0\).
    Hence \(\otr(S_n)=\tau_{1,p}+\op(1)\).
    For the second trace moment,
    \[
        \otr(S_n^2)
        =
        \frac1{pn^2}\sum_{i,j=1}^n(y_i^\top y_j)^2.
    \]
    Independence and the bounded fourth moments give
    \[
        \EE(y_i^\top y_j)^2
        =
        p\tau_{2,p},
        \qquad i\neq j,
        \qquad
        \EE\|y_i\|_2^4
        =
        p^2\tau_{1,p}^2+\cO(p),
    \]
    and therefore
    \[
        \EE\otr(S_n^2)
        =
        \tau_{2,p}
        +
        \frac pn\tau_{1,p}^2
        +
        o(1).
    \]
    Let \(H_{ij}=(y_i^\top y_j)^2\).
    The \((8+\delta)\)-moment bound gives
    \[
        \EE H_{ij}^2=\cO(p^2),
        \qquad i\neq j,
        \qquad
        \EE H_{ii}^2=\cO(p^4).
    \]
    Covariances vanish for disjoint index sets, while Cauchy--Schwarz bounds the overlapping off-diagonal, diagonal-off-diagonal, and repeated diagonal contributions by \(\cO(p^2)\), \(\cO(p^3)\), and \(\cO(p^4)\), respectively.
    Since their numbers are \(\cO(n^3)\), \(\cO(n^2)\), and \(\cO(n)\), and \(p\asymp n\),
    \[
        \Var\{\otr(S_n^2)\}
        =
        \cO\!\left(
            \frac{n^3p^2+n^2p^3+np^4}{p^2n^4}
        \right)
        =
        \cO(n^{-1}).
    \]
    Thus \(\otr(S_n^2)=(p/n)\tau_{1,p}^2+\tau_{2,p}+\op(1)\).

    Centring satisfies
    \[
        \left|\otr(S_n^c)-\otr(S_n)\right|
        =
        \frac{\|\overline y\|_2^2}{p}
        =
        \op(1),
    \]
    and
    \[
        \left|
            \otr\{(S_n^c)^2-S_n^2\}
        \right|
        \leq
        \frac{
            2\|S_n\|_{\oper}\|\overline y\|_2^2
            +\|\overline y\|_2^4
        }{p}
        =
        \op(1).
    \]
    Since \(m=n-1\),
    \[
        \frac nm=1+\cO(n^{-1}),
        \qquad
        \left(\frac nm\right)^2\frac pn
        =
        \widetilde\phi_{0,n}+\cO(n^{-1}).
    \]
    Consequently,
    \[
        \otr(W_0)=\tau_{1,p}+\op(1),
        \qquad
        \otr(W_0^2)
        =
        \widetilde\phi_{0,n}\tau_{1,p}^2+\tau_{2,p}+\op(1),
    \]
    and hence \(D_0=\tau_{2,p}+\op(1)\).
    The covariance lower bound in \Cref{asm:rmt} gives \(\tau_{2,p}\geq c_\Sigma^2\), while
    \[
        \otr(W_0^k)
        \leq
        \|W_0\|_{\oper}^{k-2}\otr(W_0^2),
        \qquad
        k=2,3,4,
    \]
    gives the stated higher trace bounds.

    The eigenvalues of \(K_m\) are \(\widetilde\phi_{0,n}^{-1}\) times the nonzero eigenvalues of \(W_0\), with additional zeros if needed.
    Therefore,
    \[
        \max_{i\leq m}d_i
        =
        \widetilde\phi_{0,n}^{-1}\|W_0\|_{\oper}
        =
        \Op(1),
        \qquad
        \overline d_m=\otr(W_0),
    \]
    and
    \[
        \frac1m\sum_{i=1}^m(d_i-\overline d_m)^2
        =
        \frac{D_0}{\widetilde\phi_{0,n}}.
    \]
    Finally,
    \[
        \det(G_m)
        =
        \frac{D_0}{\widetilde\phi_{0,n}},
        \qquad
        \tr(G_m)
        =
        1+\frac1m\sum_{i=1}^m d_i^2
        \leq
        1+\max_{i\leq m}d_i^2.
    \]
    Since \(D_0=\tau_{2,p}+\op(1)\), \(\tau_{2,p}\geq c_\Sigma^2\), and \(\widetilde\phi_{0,n}\to\phi_0\), the determinant is bounded away from zero in probability and \(\tr(G_m)=\Op(1)\).
    Hence
    \[
        \lambda_{\min}(G_m)^{-1}
        \leq
        \frac{\tr(G_m)}{\det(G_m)}
        =
        \Op(1).
    \]
\end{proof}

\subsubsection[Proof of Proposition \ref{prop:spectral-quasi-likelihood}]{Proof of \Cref{prop:spectral-quasi-likelihood}}
\begin{proof}[Proof of \Cref{prop:spectral-quasi-likelihood}]
    Write \(\vartheta=(a,b)^\top\), \(\vartheta_0=(r^2,\sigma_0^2)^\top\), \(v_i(\vartheta)=ad_i+b\), and \(v_{i0}=v_i(\vartheta_0)\).
    Define \(\breve y_0=Q_0^\top y_0^c\) and \(V_m(\vartheta)=aK_m+bI_m\).
    Since \(K_m=U_m\diag(d_1,\ldots,d_m)U_m^\top\) and \(\tilde y_0=U_m^\top\breve y_0\),
    \[
        L_m(\vartheta)
        =
        \frac1m
        \left\{
            \log\det V_m(\vartheta)
            +
            \breve y_0^\top V_m(\vartheta)^{-1}\breve y_0
        \right\}.
    \]
    Let \(\overline L_m(\vartheta)=\EE_{\beta\sim\Pi,\epsilon_0}\{L_m(\vartheta)\mid\mathcal G_n\}\).

    \noindent\textbf{Step 1: Separation.}
    Diagonalising \(K_m\) gives
    \[
        \overline L_m(\vartheta)-\overline L_m(\vartheta_0)
        =
        \frac1m\sum_{i=1}^m
        \left\{
            \log\frac{v_i(\vartheta)}{v_{i0}}
            +
            \frac{v_{i0}}{v_i(\vartheta)}
            -1
        \right\}.
    \]
    Let \(\psi(x)=\log x+x^{-1}-1\).
    The ratio \(\psi(x)/(x-1)^2\) extends continuously at \(x=1\) with value \(1/2\) and is positive on \((0,\infty)\).
    Since \(\Theta\) is compact and bounded away from the coordinate axes and \(\max_{i\leq m}d_i=\Op(1)\) by \Cref{lem:source-trace-moments}, there exists \(a_m>0\), with \(a_m^{-1}=\Op(1)\), such that
    \[
        \psi\left\{\frac{v_i(\vartheta)}{v_{i0}}\right\}
        \geq
        a_m\{v_i(\vartheta)-v_{i0}\}^2
    \]
    uniformly over \(i\leq m\) and \(\vartheta\in\Theta\).
    For \(h=\vartheta-\vartheta_0\), \(m^{-1}\sum_{i=1}^m\{v_i(\vartheta)-v_{i0}\}^2=h^\top G_mh\).
    Set \(c_m=a_m\lambda_{\min}(G_m)\).
    By \Cref{lem:source-trace-moments}, \(c_m^{-1}=\Op(1)\), and uniformly over \(\Theta\),
    \begin{equation}
        \label{eq:spectral-quasi-likelihood-separation}
        \overline L_m(\vartheta)-\overline L_m(\vartheta_0)
        \geq
        c_m\|\vartheta-\vartheta_0\|_2^2.
    \end{equation}

    \noindent\textbf{Step 2: Uniform approximation.}
    Let \(\zeta_0=\epsilon_0/\sigma_0\), \(\upsilon_m=(\xi_n^\top,\zeta_0^\top)^\top\), and
    \[
        H_m
        =
        \left[
            \frac r{\sqrt p}Q_0^\top X_0^c,
            \ \sigma_0Q_0^\top
        \right].
    \]
    Then \(\breve y_0=H_m\upsilon_m\) and \(H_mH_m^\top=V_m(\vartheta_0)\).
    Conditional on \(\mathcal G_n\), the coordinates of \(\upsilon_m\) are independent, centred, standardised, and have uniformly bounded fourth moments in probability.

    Define \(A_m(\vartheta)=H_m^\top V_m(\vartheta)^{-1}H_m\) and \(Z_m(\vartheta)=L_m(\vartheta)-\overline L_m(\vartheta)\).
    Then
    \[
        Z_m(\vartheta)
        =
        \frac1m
        \left\{
            \upsilon_m^\top A_m(\vartheta)\upsilon_m
            -
            \tr A_m(\vartheta)
        \right\}.
    \]
    Compactness of \(\Theta\), its positive distance from the coordinate axes, and \(\max_i d_i=\Op(1)\) imply
    \[
        \sup_{\vartheta\in\Theta}\|A_m(\vartheta)\|_{\oper}
        \leq
        \sup_{\vartheta\in\Theta}\|V_m(\vartheta)^{-1}\|_{\oper}
        \|V_m(\vartheta_0)\|_{\oper}
        =
        \Op(1).
    \]
    Moreover, \(\rank\{A_m(\vartheta)\}\leq m\), so \(\tr\{A_m(\vartheta)^2\}\leq m\|A_m(\vartheta)\|_{\oper}^2=\Op(m)\).
    The conditional fourth-moment bound in \Cref{asm:variance-component-model} and the independent-coordinate quadratic-form variance bound therefore give, for each fixed \(\vartheta\),
    \[
        \Var\{Z_m(\vartheta)\mid\mathcal G_n\}
        =
        \Op(m^{-1}),
    \]
    and hence \(Z_m(\vartheta)=\op(1)\).

    Moreover, \(m^{-1}\|\breve y_0\|_2^2=\Op(1)\), because \(m^{-1}\EE(\|\breve y_0\|_2^2\mid\mathcal G_n)=m^{-1}\tr\{V_m(\vartheta_0)\}=r^2\overline d_m+\sigma_0^2=\Op(1)\), so the claim follows by the conditional Markov inequality.
    The resolvent identity and the preceding operator-norm bounds then imply
    \[
        \sup_{\substack{\vartheta,\vartheta'\in\Theta\\\vartheta\neq\vartheta'}}
        \frac{|Z_m(\vartheta)-Z_m(\vartheta')|}
        {\|\vartheta-\vartheta'\|_2}
        =
        \Op(1).
    \]
    For any fixed \(\rho>0\), a finite deterministic \(\rho\)-net \(\mathcal N_\rho\) of \(\Theta\) therefore satisfies
    \[
        \sup_{\vartheta\in\Theta}|Z_m(\vartheta)|
        \leq
        \max_{\vartheta\in\mathcal N_\rho}|Z_m(\vartheta)|
        +
        \rho\Op(1).
    \]
    The first term is \(\op(1)\) for fixed \(\rho\).
    Letting \(m\to\infty\) and then \(\rho\downarrow0\) yields
    \[
        \Delta_m
        :=
        \sup_{\vartheta\in\Theta}
        |L_m(\vartheta)-\overline L_m(\vartheta)|
        =
        \op(1).
    \]

    \noindent\textbf{Step 3: Global minimisers.}
    Since \(L_m\) is continuous on compact \(\Theta\), its set \(\widehat{\mathcal M}_m\) of global minimisers is nonempty.
    For every \(\widehat\vartheta\in\widehat{\mathcal M}_m\), the minimising property and \eqref{eq:spectral-quasi-likelihood-separation} give
    \[
        c_m\|\widehat\vartheta-\vartheta_0\|_2^2
        \leq
        \overline L_m(\widehat\vartheta)-\overline L_m(\vartheta_0)
        \leq
        \{\overline L_m(\widehat\vartheta)-L_m(\widehat\vartheta)\}
        +
        \{L_m(\vartheta_0)-\overline L_m(\vartheta_0)\}
        \leq
        2\Delta_m.
    \]
    Since \(c_m^{-1}=\Op(1)\) and \(\Delta_m=\op(1)\),
    \[
        \sup_{\widehat\vartheta\in\widehat{\mathcal M}_m}
        \|\widehat\vartheta-\vartheta_0\|_2
        \pto0.
    \]
    Thus every global minimiser is consistent.
    Since \(\vartheta_0\in\Theta^\circ\), every global minimiser lies in \(\Theta^\circ\) with probability tending to one.
\end{proof}

\subsubsection[Two-moment initialiser (Proposition \ref{prop:two-moment-initializer})]{Two-moment initialiser (\Cref{prop:two-moment-initializer})}

With \(m\), \(\widetilde\phi_{0,n}\), \(W_0\), and \(D_0\) as in \Cref{lem:source-trace-moments}, let \(a_1=\otr(W_0)\), \(T_1=m^{-1}\|y_0^c\|_2^2\), and \(T_2=m^{-2}\|(X_0^c)^\top y_0^c\|_2^2\).
For a deterministic sequence \(\kappa_n\downarrow0\), define
\[
    \begin{aligned}
        \widehat r_{\mathrm{MM}}^2
        &:={}
        \begin{cases}
            \displaystyle
            \left[
                \frac{T_2-\widetilde\phi_{0,n}a_1T_1}{D_0}
            \right]_+,
            & D_0>\kappa_n,\\[2.5ex]
            0,
            & D_0\leq\kappa_n,
        \end{cases}
        \\
        \widehat\sigma_{0,\mathrm{MM}}^2
        &:={}
        \left[T_1-a_1\widehat r_{\mathrm{MM}}^2\right]_+,
        \qquad
        \widehat\vartheta_{\mathrm{MM}}
        =
        \begin{pmatrix}
            \widehat r_{\mathrm{MM}}^2\\
            \widehat\sigma_{0,\mathrm{MM}}^2
        \end{pmatrix}.
    \end{aligned}
\]

\begin{proposition}[Consistency of the two-moment initialiser]
    \label{prop:two-moment-initializer}
    Suppose \Cref{asm:variance-component-model} holds.
    For every deterministic sequence \(\kappa_n\downarrow0\),
    \[
        \widehat\vartheta_{\mathrm{MM}}
        \pto
        \begin{pmatrix}
            r^2\\
            \sigma_0^2
        \end{pmatrix}.
    \]
\end{proposition}

\begin{proof}[Proof of \Cref{prop:two-moment-initializer}]
    Because \(y_0^c=X_0^c\beta+C_0\epsilon_0\) and \((X_0^c)^\top C_0=(X_0^c)^\top\), conditional on \(\mathcal G_n\),
    \[
        \EE(T_1\mid\mathcal G_n)
        =
        r^2a_1+\sigma_0^2,
        \qquad
        \EE(T_2\mid\mathcal G_n)
        =
        r^2\otr(W_0^2)
        +
        \sigma_0^2\widetilde\phi_{0,n}a_1.
    \]
    Writing \(\breve y_0=H_m\upsilon_m\) as in the proof of \Cref{prop:spectral-quasi-likelihood}, and using \(Q_0Q_0^\top=C_0\), gives \(T_1=m^{-1}\breve y_0^\top\breve y_0\) and \(T_2=(p/m^2)\breve y_0^\top K_m\breve y_0\).
    By \(\|K_m\|_{\oper}=\Op(1)\) and \(\|H_mH_m^\top\|_{\oper}=\|V_m(\vartheta_0)\|_{\oper}=\Op(1)\), the corresponding quadratic-form matrices in \(\upsilon_m\) have rank at most \(m\) and operator norm \(\Op(m^{-1})\).
    The conditional fourth-moment bound in \Cref{asm:variance-component-model} and the independent-coordinate quadratic-form variance bound therefore give \(\Var(T_j\mid\mathcal G_n)=\Op(m^{-1})\), and hence \(T_j-\EE(T_j\mid\mathcal G_n)=\op(1)\) for \(j=1,2\).
    Hence
    \[
        T_2-\widetilde\phi_{0,n}a_1T_1
        =
        r^2D_0+\op(1).
    \]
    By \Cref{lem:source-trace-moments}, \(D_0=\tau_{2,p}+\op(1)\) and \(\inf_p\tau_{2,p}\geq c_\Sigma^2>0\), so \(\PP(D_0>c_\Sigma^2/2)\to1\) and hence \(\PP(D_0>\kappa_n)\to1\).
    On this event, the untruncated ratio equals \(r^2+\op(1)\), so it is positive with probability tending to one and \(\widehat r_{\mathrm{MM}}^2\pto r^2\).
    Finally, \(a_1=\Op(1)\) and
    \[
        T_1-a_1\widehat r_{\mathrm{MM}}^2
        =
        \sigma_0^2
        +
        a_1\left(r^2-\widehat r_{\mathrm{MM}}^2\right)
        +
        \op(1)
        \pto
        \sigma_0^2.
    \]
    Since \(\sigma_0^2>0\), the positive part leaves the second estimator unchanged with probability tending to one, giving \(\widehat\sigma_{0,\mathrm{MM}}^2\pto\sigma_0^2\).
\end{proof}

\subsubsection{Computation of the spectral quasi-likelihood estimator}

Write \(\Theta=[a_-,a_+]\times[b_-,b_+]\subset\mathbb R_{++}^2\) and \(\ell=a/b\).
The feasible range of \(\ell\) is \(\mathcal L(\Theta)=[a_-/b_+,a_+/b_-]\), and, for each \(\ell\in\mathcal L(\Theta)\), the feasible values of \(b\) form
\[
    \mathcal B(\ell;\Theta)
    =
    \left[
        \max\left\{b_-,\frac{a_-}{\ell}\right\},
        \min\left\{b_+,\frac{a_+}{\ell}\right\}
    \right].
\]
Under \(a=\ell b\), the criterion in \eqref{eq:spectral-quasi-likelihood} becomes
\[
    L_m(\ell,b)
    =
    \log b
    +
    \frac1m\sum_{i=1}^m\log(1+\ell d_i)
    +
    \frac{\widehat\sigma_0^2(\ell)}{b},
\]
where \(\widehat\sigma_0^2(\ell)=m^{-1}\sum_{i=1}^m\tilde y_{0i}^2/(1+\ell d_i)\).
For fixed \(\ell\), the constrained minimiser over \(b\) is \(\widehat b(\ell;\Theta)=\proj_{\mathcal B(\ell;\Theta)}\{\widehat\sigma_0^2(\ell)\}\).
Hence the original optimisation reduces to
\[
    \widehat\ell
    \in
    \argmin_{\ell\in\mathcal L(\Theta)}
    \left\{
        \log\widehat b(\ell;\Theta)
        +
        \frac1m\sum_{i=1}^m\log(1+\ell d_i)
        +
        \frac{\widehat\sigma_0^2(\ell)}
             {\widehat b(\ell;\Theta)}
    \right\}.
\]
When \(\widehat\sigma_0^2(\ell)\in\mathcal B(\ell;\Theta)\), the profiled criterion, up to an additive constant, is \(L_m^{\mathrm{prof}}(\ell)=\log\widehat\sigma_0^2(\ell)+m^{-1}\sum_{i=1}^m\log(1+\ell d_i)\).
Given a global minimiser \(\widehat\ell\), set \(\widehat\sigma_0^2=\widehat b(\widehat\ell;\Theta)\) and \(\widehat r^2=\widehat\ell\,\widehat\sigma_0^2\).

The profile can be evaluated directly from a thin singular-value decomposition.
Let \(k=\rank(X_0^c)\) and write \(X_0^c=U_k\diag(s_1,\ldots,s_k)V_k^\top\).
The positive eigenvalues of \(K_m\) are \(d_j=s_j^2/p\), \(j=1,\ldots,k\), and, taking the corresponding columns of \(U_m\) to be \(Q_0^\top u_j\), their response coordinates are \(\tilde y_{0j}=u_j^\top y_0^c\), where \(u_j\) is the \(j\)th column of \(U_k\), so \(Q_0\) need not be constructed.
The remaining \(m-k\) eigenvalues are zero, with total response energy \(E_\perp=\|y_0^c\|_2^2-\sum_{j=1}^k\tilde y_{0j}^2\).
Consequently,
\[
    \widehat\sigma_0^2(\ell)
    =
    \frac1m
    \left\{
        E_\perp
        +
        \sum_{j=1}^k
        \frac{\tilde y_{0j}^2}{1+\ell d_j}
    \right\},
    \qquad
    \frac1m\sum_{i=1}^m\log(1+\ell d_i)
    =
    \frac1m\sum_{j=1}^k\log(1+\ell d_j).
\]
After one thin singular-value decomposition, each profile evaluation requires \(\cO(k)\) operations.

\subsection{Proof of Theorem~\ref{thm:risk-estimation}}
\label{app:proof-risk-tuning}

This subsection proves \Cref{thm:risk-estimation}.
We first establish the auxiliary rate and uniformity results used in the proof, then assemble the theorem.

Throughout this subsection, calculations involving the evaluation target sample are conditional on \(\mathcal H_n\).
As in \Cref{sec:ood-risk-tuning}, write \(\bar x_1=\bar x_{1,\mathrm E}\) and \(n_1=n_{1,\mathrm E}\).
After the evaluation sample is realised, \(R_n(\lambda)\) remains the \(\mathcal G_n\)-conditional predictive risk defined in Section~\ref{sec:risk:exact-risk}.

\subsubsection[Rate preservation under ridge augmentation (Lemma \ref{lem:base-to-augmented-rate})]{Rate preservation under ridge augmentation (\Cref{lem:base-to-augmented-rate})}

\begin{lemma}[Rate preservation under ridge augmentation]
    \label{lem:base-to-augmented-rate}
    Suppose \Cref{asm:rmt,asm:evaluation-sample} hold, \(p/n_1=\cO(1)\), and the nonempty set \(\Lambda\) is contained in \([\lambda_-,\lambda_+]\subset(0,\infty)\).
    If, for some
    \(\eta\in[0,1]\),
    \[
        n_0^\eta
        \frac{\|\nu_1-X_0^\top\gamma\|_2^2}{p}
        =\Op(1),
    \]
    then
    \[
        n_0^\eta\frac{\|\Delta\|_2^2}{p}=\Op(1).
    \]
    Moreover, the following conditions are equivalent:
    \[
        n_0^\eta\|\gamma\|_2^2=\Op(1),
        \qquad
        n_0^\eta
        \sup_{\lambda\in\Lambda}\|\gamma_\lambda\|_2^2
        =\Op(1).
    \]
    Under either condition,
    \[
        n_0^\eta
        \sup_{\lambda\in\Lambda}
        \frac{\|X_0^\top\gamma_\lambda-\nu_1\|_2^2}{p}
        =\Op(1).
    \]
\end{lemma}

\begin{proof}[Proof of \Cref{lem:base-to-augmented-rate}]
    Write
    \[
        g=\nu_1-X_0^\top\gamma,
        \qquad
        \epsilon_1=\bar x_1-\nu_1,
        \qquad
        \Delta=g+\epsilon_1.
    \]
    By \Cref{asm:evaluation-sample} and the target-side moment conditions,
    \(\|\epsilon_1\|_2^2=\Op(p/n_1)=\Op(1)\). Hence
    \[
        n_0^\eta\frac{\|\Delta\|_2^2}{p}
        \leq
        2n_0^\eta\frac{\|g\|_2^2}{p}
        +2n_0^\eta\frac{\|\epsilon_1\|_2^2}{p}
        =\Op(1),
    \]
    where the second term is \(\Op(n_0^{\eta-1})\).

    Next,
    \[
        \sup_{\lambda\in\Lambda}
        \|\gamma_\lambda-\gamma\|_2^2
        \leq
        \frac{\|X_0^c\|_{\oper}^2}
        {n_0^2\lambda_-^2}\|\Delta\|_2^2.
    \]
    Under \Cref{asm:rmt},
    \(\|X_0^c\|_{\oper}^2/n_0=\Op(1)\), while the preceding display and
    \(p/n_0\asymp1\) give
    \(n_0^\eta\|\Delta\|_2^2/n_0=\Op(1)\). Therefore,
    \[
        n_0^\eta\sup_{\lambda\in\Lambda}
        \|\gamma_\lambda-\gamma\|_2^2=\Op(1).
    \]
    The triangle inequality gives the forward implication. For the reverse
    implication, fix any \(\lambda_0\in\Lambda\) and write
    \(\gamma=\gamma_{\lambda_0}
    -(\gamma_{\lambda_0}-\gamma)\); the same bound applies.

    Finally, with
    \(Q_\lambda=\lambda M_\lambda^c\) and
    \(P_\lambda=I_p-\lambda M_\lambda^c\),
    \[
        X_0^\top\gamma_\lambda-\nu_1
        =-Q_\lambda g+P_\lambda\epsilon_1.
    \]
    Both matrix families have operator norm at most one. Consequently,
    \[
        \sup_{\lambda\in\Lambda}
        \frac{\|X_0^\top\gamma_\lambda-\nu_1\|_2^2}{p}
        \leq
        2\frac{\|g\|_2^2}{p}
        +2\frac{\|\epsilon_1\|_2^2}{p},
    \]
    which gives the claimed uniform scaled rate.
\end{proof}

\subsubsection[Uniform control of the risk criterion (Lemma \ref{lem:ood-risk-equicontinuity})]{Uniform control of the risk criterion (\Cref{lem:ood-risk-equicontinuity})}

\begin{lemma}[Uniform control of the risk criterion]
    \label{lem:ood-risk-equicontinuity}
    Under the compact-interval conditions of
    \Cref{thm:risk-estimation}, let \(\widetilde q_n(\lambda)\) denote
    the expression inside the positive part in
    \eqref{eq:ood-risk-estimator}, and let
    \[
        q_n(\lambda)
        =\frac1p\|X_0^\top\gamma_\lambda-\nu_1\|_2^2.
    \]
    Then
    \[
        n_0^\eta\sup_{\lambda\in\Lambda}
        \left\{|\partial_\lambda\widetilde q_n(\lambda)|
        +|\partial_\lambda q_n(\lambda)|\right\}=\Op(1),
    \]
    and
    \[
        \frac{n_0^\eta}{n_1}
        \sup_{\lambda\in\Lambda}
        \left|\otr\left[
        \{I_p-2\lambda M_\lambda^c\}
        (\widehat\Sigma_1-\Sigma_1)
        \right]\right|\pto0.
    \]
    The derivative bounds imply that
    \(n_0^\eta\{[\widetilde q_n(\lambda)]_+-q_n(\lambda)\}\)
    is stochastically equicontinuous on \(\Lambda\).
\end{lemma}

\begin{proof}[Proof of \Cref{lem:ood-risk-equicontinuity}]
    Write \(g=\nu_1-X_0^\top\gamma\) and
    \(\epsilon_1=\bar x_1-\nu_1\), so \(\Delta=g+\epsilon_1\).
    By \Cref{lem:base-to-augmented-rate},
    \[
        n_0^\eta\frac{\|\Delta\|_2^2}{p}=\Op(1).
    \]
    Because \(\partial_\lambda M_\lambda^c=-(M_\lambda^c)^2\) and the first two derivatives of \(M_\lambda^c\) are uniformly bounded on \(\Lambda\), differentiation of the quadratic term in \(\widetilde q_n\) gives the stated scaled derivative bound.
    The trace derivative is also \(\Op(1)\) after scaling, since
    \(n_0^\eta/n_1=\cO(1)\) and
    \(\otr(\widehat\Sigma_1)=\Op(1)\).

    Writing \(P_\lambda=I_p-\lambda M_\lambda^c\) and
    \(Q_\lambda=\lambda M_\lambda^c\), the vector in \(q_n\) is
    \(-Q_\lambda g+P_\lambda\epsilon_1\). Its derivative is
    \(-M_\lambda^cP_\lambda\Delta\), so the same rate proves the
    bound for \(\partial_\lambda q_n\).

    For the covariance process, set
    \(A_\lambda=I_p-2\lambda M_\lambda^c\) and
    \(T_n(\lambda)=\otr\{A_\lambda
    (\widehat\Sigma_1-\Sigma_1)\}\). Conditional on \(\mathcal H_n\), the
    centring in \(\widehat\Sigma_1\) is handled as follows. Put
    \(m=n_1\), \(y_i=x_i-\nu_1\), and
    \(\bar y=m^{-1}\sum_i y_i\). Then, exactly,
    \[
        \tr\{A_\lambda(\widehat\Sigma_1-\Sigma_1)\}
        =\frac{m}{m-1}\left[
        \frac1m\sum_{i=1}^m\{y_i^\top A_\lambda y_i
        -\tr(A_\lambda\Sigma_1)\}
        -\left\{\bar y^\top A_\lambda\bar y
        -\frac1m\tr(A_\lambda\Sigma_1)\right\}
        \right].
    \]
    By the conditional product law in
    \Cref{asm:evaluation-sample}, the evaluation observations are
    conditionally independent with common law \(\PP_{1,n}\) given
    \(\mathcal H_n\). The target-side moment conditions in \Cref{asm:rmt} and
    the \(\mathcal H_n\)-measurability of \(A_\lambda\) therefore imply that the
    first centred term has conditional variance at most \(Cp/m\), while the
    second has conditional variance at most \(Cp/m^2\). Thus the elementary inequality
    \(\Var(U+V)\leq2\Var(U)+2\Var(V)\) gives
    \[
        \Var\!\left[
        \tr\{A_\lambda(\widehat\Sigma_1-\Sigma_1)\}
        \,\middle|\,\mathcal H_n\right]
        \leq \frac{C}{n_1}
        \tr\{(\Sigma_1^{1/2}A_\lambda\Sigma_1^{1/2})^2\}
        \leq \frac{Cp}{n_1},
    \]
    and hence, because \(\otr=p^{-1}\tr\),
    \[
        \Var\{T_n(\lambda)\mid\mathcal H_n\}
        \leq\frac{C}{p^2n_1}
        \tr\{(\Sigma_1^{1/2}A_\lambda\Sigma_1^{1/2})^2\}
        \leq\frac{C}{pn_1}.
    \]
    Consequently,
    \[
        \Var\!\left\{\frac{n_0^\eta}{n_1}T_n(\lambda)
        \,\middle|\,\mathcal H_n\right\}
        \leq \frac{Cn_0^{2\eta}}{pn_1^3}=\cO(n_0^{-2})
    \]
    at the worst boundary \(\eta=1\), so the required scaled process converges pointwise to zero.

    Its increments are bounded by
    \[
        C|\lambda-\lambda'|p^{-1}
        \|\widehat\Sigma_1-\Sigma_1\|_*,
    \]
    where \(\|\cdot\|_*\) denotes the nuclear norm. Both covariance matrices
    are positive semidefinite, so
    \[
        p^{-1}\|\widehat\Sigma_1-\Sigma_1\|_*
        \leq
        \otr(\widehat\Sigma_1)+\otr(\Sigma_1)
        =\Op(1),
    \]
    where the last equality follows from the treated-sample trace bound above
    and the bounded spectrum of \(\Sigma_1\). Thus the random Lipschitz
    envelope is \(\Op(1)\). A grid of mesh
    \(n_0^{-1}\) has \(\cO(n_0)\) points; the last
    variance bound and a union bound control the grid, while the Lipschitz bound
    controls the gaps.
    This proves the uniform covariance bound.
    The derivative bounds and the \(1\)-Lipschitz property of the positive part give the stated stochastic equicontinuity.
\end{proof}

\subsubsection[Proof of Theorem \ref{thm:risk-estimation}]{Proof of \Cref{thm:risk-estimation}}

\begin{proof}[Proof of \Cref{thm:risk-estimation}]
    For \Cref{thm:risk-estimation}~\eqref{thm:risk-estimation-i}, by the moment consequences of \Cref{asm:evaluation-sample},
    \[
        \EE(\epsilon_1\mid\mathcal H_n)=0,
        \qquad
        \EE(\epsilon_1\epsilon_1^\top
        \mid\mathcal H_n)=\frac{\Sigma_1}{n_1}.
    \]
    Recall that
    \[
        X_0^\top\gamma_\lambda
        =(I_p-\lambda M_\lambda^c)\bar{x}_1
        +\lambda M_\lambda^cX_0^\top\gamma.
    \]
    With \(C_1=\Sigma_1/n_1\), the target-error correction identity is
    \begin{equation}
        \begin{split}
            \EE\left[\|X_0^\top\gamma_\lambda
            -\bar{x}_1\|_2^2
            +\tr\{(2(I_p-\lambda M_\lambda^c)-I_p)C_1\}
            \mid\mathcal H_n\right]
            &=\EE\left[\|X_0^\top\gamma_\lambda-\nu_1\|_2^2
            \mid\mathcal H_n\right].
        \end{split}
        \label{eq:target-error-correction}
    \end{equation}

    To verify \eqref{eq:target-error-correction}, write
    \(b_0=X_0^\top\gamma\), \(g=\nu_1-b_0\),
    \(Q_\lambda=\lambda M_\lambda^c\), and
    \(P_\lambda=I_p-\lambda M_\lambda^c\). Conditional on \(\mathcal H_n\),
    these quantities are fixed, while \(\epsilon_1\) is centred
    with conditional covariance \(C_1\). Therefore,
    \[
        \begin{aligned}
        \EE[\|X_0^\top\gamma_\lambda-\bar{x}_1\|_2^2
        \mid\mathcal H_n]
        &=g^\top Q_\lambda^2g+\tr(Q_\lambda^2C_1),\\
        \EE[\|X_0^\top\gamma_\lambda-\nu_1\|_2^2
        \mid\mathcal H_n]
        &=g^\top Q_\lambda^2g+\tr(P_\lambda^2C_1).
        \end{aligned}
    \]
    Because \(Q_\lambda=I_p-P_\lambda\),
    \(Q_\lambda^2-P_\lambda^2=I_p-2P_\lambda\).

    With \(\widetilde q_n\) defined in
    \Cref{lem:ood-risk-equicontinuity}, \eqref{eq:ood-risk-estimator} can be
    written as
    \[
        \widehat R_n(\lambda;\gamma)
        =\widehat r^2[\widetilde q_n(\lambda)]_+
        +\widehat\sigma_0^2\|\gamma_\lambda\|_2^2.
    \]

    The target-error correction identity \eqref{eq:target-error-correction}
    implies that, before replacing \(\Sigma_1\) by
    \(\widehat\Sigma_1\), the corrected quadratic form differs from
    the realised target error by
    \[
        2g^\top Q_\lambda\epsilon_1
        +\epsilon_1^\top(2Q_\lambda-I_p)\epsilon_1
        -\tr\{(2Q_\lambda-I_p)C_1\}.
    \]
    Because \(\|M_\lambda^c\|_{\oper}\leq\lambda_-^{-1}\), the variance
    of the linear term is bounded by
    \(C\|g\|_2^2/n_1\), while the variance of the centred quadratic
    term is bounded by
    \[
        C\frac{\tr(\Sigma_1^2)}{n_1^2}
        =\cO(p/n_1^2).
    \]
    After division by \(p^2\) and multiplication by \(n_0^{2\eta}\), both
    bounds converge to zero for every \(\eta\in[0,1]\). The covariance-process
    bound in \Cref{lem:ood-risk-equicontinuity} makes replacement of
    \(\Sigma_1\) by \(\widehat\Sigma_1\) negligible at the same scale.
    Thus the corrected signal criterion is pointwise consistent.
    Because \(q_n(\lambda)\geq0\) and \(x\mapsto[x]_+\) is \(1\)-Lipschitz, the same conclusion holds after applying the positive part:
    \[
        n_0^\eta\{[\widetilde q_n(\lambda)]_+-q_n(\lambda)\}\pto0.
    \]
    For a fixed finite grid, taking the maximum over \(\lambda\in\Lambda\) gives
    uniform convergence. For a compact interval, the scaled stochastic
    equicontinuity in \Cref{lem:ood-risk-equicontinuity} and a finite-net
    argument give the same conclusion. Moreover, by normalisation,
    \Cref{asm:risk-rate-admissibility} implies
    \(n_0^\eta\|\gamma\|_2^2=\Op(1)\). Hence
    \Cref{lem:base-to-augmented-rate} gives
    \[
        n_0^\eta\sup_{\lambda\in\Lambda}q_n(\lambda)=\Op(1),
        \qquad
        n_0^\eta\sup_{\lambda\in\Lambda}
        \|\gamma_\lambda\|_2^2=\Op(1).
    \]
    Finally,
    \[
        \begin{aligned}
        n_0^\eta\sup_{\lambda\in\Lambda}
        |\widehat R_n(\lambda;\gamma)-R_n(\lambda)|
        &\leq|\widehat r^2-r^2|n_0^\eta\sup_{\lambda\in\Lambda}q_n(\lambda)\\
        &\quad+\widehat r^2n_0^\eta\sup_{\lambda\in\Lambda}
        |[\widetilde q_n(\lambda)]_+-q_n(\lambda)|\\
        &\quad+|\widehat\sigma_0^2-\sigma_0^2|n_0^\eta
        \sup_{\lambda\in\Lambda}\|\gamma_\lambda\|_2^2
        \pto0.
        \end{aligned}
    \]

    For \Cref{thm:risk-estimation}~\eqref{thm:risk-estimation-ii}, let
    \(\Delta_n=\sup_{\lambda\in\Lambda}
    |\widehat R_n(\lambda;\gamma)-R_n(\lambda)|\). For every
    \(\lambda\in\Lambda\),
    \begin{align*}
        R_n(\widehat\lambda)
        \leq\widehat R_n(\widehat\lambda;\gamma)+\Delta_n
        \leq\widehat R_n(\lambda;\gamma)+\Delta_n
        \leq R_n(\lambda)+2\Delta_n.
    \end{align*}
    Taking the infimum over \(\lambda\) yields,
    \[
        R_n(\widehat\lambda)
        -\inf_{\lambda\in\Lambda}R_n(\lambda)\leq2\Delta_n.
    \]
    Multiplying by \(n_0^\eta\) and applying \Cref{thm:risk-estimation}~\eqref{thm:risk-estimation-i} proves \Cref{thm:risk-estimation}~\eqref{thm:risk-estimation-ii}.

    For \Cref{thm:risk-estimation}~\eqref{thm:risk-estimation-iii}, let \(f_n(\lambda)=n_0^\eta R_n(\lambda)\),
    \(\widehat f_n(\lambda)
    =n_0^\eta\widehat R_n(\lambda;\gamma)\),
    and \(f(\lambda)=R_{\eta,\star}(\lambda)\). Under the premise of \Cref{thm:risk-estimation}~\eqref{thm:risk-estimation-iii}, \(f\) is continuous on the compact set \(\Lambda\), so
    \(M=\argmin_\Lambda f\) is nonempty. Fix
    \(\varepsilon>0\) and set
    \(\Lambda_\varepsilon=\{\lambda:\operatorname{dist}(\lambda,M)
    \geq\varepsilon\}\). If this set is empty, the claim is immediate.
    Otherwise, compactness gives
    \(c_\varepsilon:=\inf_{\Lambda_\varepsilon}f-\min_\Lambda f>0\). Moreover,
    \[
        \sup_{\lambda\in\Lambda}|\widehat f_n(\lambda)-f(\lambda)|
        \leq n_0^\eta\Delta_n
        +\sup_{\lambda\in\Lambda}|f_n(\lambda)-f(\lambda)|\pto0,
    \]
    by \Cref{thm:risk-estimation}~\eqref{thm:risk-estimation-i} and the uniform-limit premise of \Cref{thm:risk-estimation}~\eqref{thm:risk-estimation-iii}, respectively.

    On the event that this supremum is below \(c_\varepsilon/3\), no
    minimiser of \(\widehat R_n(\cdot;\gamma)\) lies in
    \(\Lambda_\varepsilon\). Hence,
    \[
        \PP\{\operatorname{dist}(\widehat\lambda,M)\geq\varepsilon\}
        \to0.
    \]
    Thus \(\dist(\widehat\lambda,M)\pto0\).
    If \(M=\{\lambda_{\eta,\star}^\ast\}\), then
    \(\widehat\lambda\pto\lambda_{\eta,\star}^\ast\).
    This completes the proof.
\end{proof}

\begin{proof}[Proof of \Cref{rem:no-augmentation-endpoint}]
    Write \(g=\nu_1-X_0^\top\gamma\), so that \(\Delta=g+\epsilon_1\).
    Then
    \[
        \EE\left[
        \frac{1}{p}\|\Delta\|_2^2
        -
        \frac{1}{n_1}\otr(\widehat\Sigma_1)
        \,\middle|\,
        \mathcal H_n
        \right]
        =
        \frac{1}{p}\|g\|_2^2.
    \]
    The target-mean and covariance concentration bounds used above, \Cref{asm:risk-rate-admissibility}, and consistency of \(\widehat r^2\) and \(\widehat\sigma_0^2\) give
    \[
        n_0^\eta
        \left|
        \widehat R_n(\infty;\gamma)-R_n(\infty)
        \right|
        \pto0.
    \]
    Since \(\infty\) adds one candidate point, \Cref{thm:risk-estimation}~\eqref{thm:risk-estimation-i}--\eqref{thm:risk-estimation-ii} remain valid on \(\Lambda\cup\{\infty\}\).
\end{proof}

\subsection{Proof of Theorem~\ref{thm:predictive-calibration}}
\label{app:proof-predictive-calibration}

This subsection proves \Cref{thm:predictive-calibration}.
Supporting results precede the theorem proof, followed by a non-Gaussian coefficient extension.

\subsubsection[Stability of RCB selection (Lemma \ref{lem:ood-post-tuning-stability})]{Stability of RCB selection (\Cref{lem:ood-post-tuning-stability})}

\begin{lemma}[Stability of RCB selection]
    \label{lem:ood-post-tuning-stability}
    Suppose \Cref{asm:predictive-model} and all conditions of
    \Cref{thm:risk-estimation} hold for some \(\eta\in[0,1]\), with
    \(\Lambda=[\lambda_-,\lambda_+]\subset(0,\infty)\).
    Suppose further that a deterministic continuous function \(\sR_{\eta,\star}\) satisfies
    \[
        \sup_{\lambda\in\Lambda}
        |n_0^\eta R_n(\lambda)-\sR_{\eta,\star}(\lambda)|
        \pto0.
    \]
    Suppose \(\sR_{\eta,\star}\) has a unique minimiser \(\lambda_{\eta,\star}^*\) with \(\sR_{\eta,\star}(\lambda_{\eta,\star}^*)>0\).
    Then
    \[
        n_0^{\eta/2}\sup_{\lambda\in\Lambda}
        |\partial_\lambda\widehat\mu_{0,\lambda}|=\Op(1),
    \]
    and
    \[
        \frac{R_n(\widehat\lambda)}
        {R_n(\lambda_{\eta,\star}^*)}\pto1,
        \qquad
        \frac{\widehat R_n(\widehat\lambda)}
        {R_n(\lambda_{\eta,\star}^*)}\pto1.
    \]
\end{lemma}

\begin{proof}[Proof of \Cref{lem:ood-post-tuning-stability}]
    By \Cref{lem:base-to-augmented-rate},
    \[
        n_0^\eta\frac{\|\Delta\|_2^2}{p}=\Op(1).
    \]
    Since \(\partial_\lambda M_\lambda^c=-(M_\lambda^c)^2\),
    \[
        \partial_\lambda\widehat\mu_{0,\lambda}
        =-\frac1{n_0}\Delta^\top
        (M_\lambda^c)^2(X_0^c)^\top y_0.
    \]
    Substituting \(y_0=\alpha_0\one_{n_0}+X_0\beta+\epsilon_0\), the intercept vanishes because \((X_0^c)^\top\one_{n_0}=0\).
    Using \Cref{asm:predictive-model} and conditioning on
    \(\mathcal G_n\), the signal-residual cross term vanishes and
    \[
        \EE\left\{|\partial_\lambda\widehat\mu_{0,\lambda}|^2
        \,\middle|\,\mathcal G_n\right\}
        \leq C_K\left\{
        \frac{\|\Delta\|_2^2}{p}
        +\frac{\|\Delta\|_2^2}{n_0}
        \right\}.
    \]
    Since
    \(\partial_\lambda^2\widehat\mu_{0,\lambda}
    =2n_0^{-1}\Delta^\top (M_\lambda^c)^3
    (X_0^c)^\top y_0\), the same calculation gives this bound for its conditional second moment.
    Both calculations hold on
    \(\mathcal E_K=\{\|S_0^c\|_{\oper}\leq K\}\).
    The sample-covariance operator-norm bound and the one-dimensional Sobolev inequality then give the asserted uniform derivative rate.

    By \Cref{thm:risk-estimation}~\eqref{thm:risk-estimation-iii}, \(\widehat\lambda\pto\lambda_{\eta,\star}^*\). 
    Uniform convergence of \(n_0^\eta R_n\), continuity of \(\sR_{\eta,\star}\), and positivity at \(\lambda_{\eta,\star}^*\) give the first ratio.
    Together with \Cref{thm:risk-estimation}~\eqref{thm:risk-estimation-i}, this gives the second ratio.
\end{proof}

\subsubsection[Conditional perturbation and plug-in studentisation (Lemmas \ref{lem:conditional-kolmogorov-perturbation,lem:plugin-studentization})]{Conditional perturbation and plug-in studentisation (\Cref{lem:conditional-kolmogorov-perturbation,lem:plugin-studentization})}

\begin{lemma}[Conditional perturbation in Kolmogorov distance]
    \label{lem:conditional-kolmogorov-perturbation}
    Let \(Y_n,Z_n\) be real-valued random variables and let \(\mathcal G_n\) be
    a sigma-field. Suppose, for a continuous distribution function \(F\),
    \[
        \sup_{t\in\RR}
        \left|\PP(Y_n\leq t\mid\mathcal G_n)-F(t)\right|\pto0,
    \]
    and, for every \(\varepsilon>0\),
    \[
        \PP(|Z_n|>\varepsilon\mid\mathcal G_n)\pto0.
    \]
    Then
    \[
        \sup_{t\in\RR}
        \left|\PP(Y_n+Z_n\leq t\mid\mathcal G_n)-F(t)\right|\pto0.
    \]
\end{lemma}

\begin{proof}[Proof of \Cref{lem:conditional-kolmogorov-perturbation}]
    For every \(\varepsilon>0\),
    \[
        \begin{split}
        \sup_t\left|\PP(Y_n+Z_n\leq t\mid\mathcal G_n)-F(t)\right|
        \leq{}&
        \sup_t\left|\PP(Y_n\leq t\mid\mathcal G_n)-F(t)\right|\\
        &+\PP(|Z_n|>\varepsilon\mid\mathcal G_n)
        +\sup_t\{F(t+\varepsilon)-F(t-\varepsilon)\}.
        \end{split}
    \]
    The first two terms converge to zero in probability.
    The final deterministic term tends to zero as \(\varepsilon\downarrow0\) by uniform continuity of \(F\).
\end{proof}

\begin{lemma}[Plug-in studentisation]
    \label{lem:plugin-studentization}
    Let \(T_n\) be real-valued and let \(s_n,\widehat s_n>0\) with probability
    tending to one. Suppose
    \[
        \sup_{t\in\RR}
        \left|
        \PP(T_n/s_n\leq t\mid\mathcal G_n)-\Phi(t)
        \right|\pto0
    \]
    and \(\widehat s_n/s_n\pto1\). Then the same conclusion holds with
    \(\widehat s_n\) in place of \(s_n\).
\end{lemma}

\begin{proof}[Proof of \Cref{lem:plugin-studentization}]
    Let
    \[
        Y_n=\frac{T_n}{s_n},
        \qquad
        A_n=\frac{s_n}{\widehat s_n},
        \qquad
        Z_n=(A_n-1)Y_n.
    \]
    The assumed conditional Kolmogorov convergence implies \(Y_n=\Op(1)\).
    Since \(A_n-1=\op(1)\), \(Z_n=\op(1)\).
    Thus, for every \(\varepsilon>0\),
    \[
        \EE\{\PP(|Z_n|>\varepsilon\mid\mathcal G_n)\}
        =\PP(|Z_n|>\varepsilon)\longrightarrow0,
    \]
    so \(\PP(|Z_n|>\varepsilon\mid\mathcal G_n)\pto0\).
    Since \(A_nY_n=Y_n+Z_n\), \Cref{lem:conditional-kolmogorov-perturbation} with \(F=\Phi\) gives the result.
\end{proof}

\subsubsection[Normal approximation (Proposition \ref{prop:post-tuning-transfer})]{Normal approximation (\Cref{prop:post-tuning-transfer})}

\begin{proposition}[Normal approximation]
    \label{prop:post-tuning-transfer}
    Let \(\Lambda\) be a compact interval and let
    \(\widetilde\lambda_n\in\Lambda\) be \(\mathcal G_n\)-measurable.
    Let \(s_n(\widetilde\lambda_n)\) be a positive scale.
    No
    \(\mathcal G_n\)-measurability condition is imposed on
    \(\widehat\lambda\); in the application below it may depend on \(y_0\).
    Suppose
    \[
        \sup_{t\in\RR}
        \left|
        \PP\left(
        \frac{\widehat\mu_{0,\widetilde\lambda_n}-\mu_{0,\beta}}
        {s_n(\widetilde\lambda_n)}
        \leq t
        \,\middle|\,\mathcal G_n
        \right)-\Phi(t)
        \right|\pto0.
    \]
    If
    \[
        |\widehat\lambda-\widetilde\lambda_n|=\op(1),
        \qquad
        \sup_{\lambda\in\Lambda}
        \frac{|\partial_\lambda\widehat\mu_{0,\lambda}|}
        {s_n(\widetilde\lambda_n)}
        =\Op(1),
    \]
    then the same conditional conclusion holds for
    \[
        \frac{\widehat\mu_{0,\widehat\lambda}-\mu_{0,\beta}}
        {s_n(\widetilde\lambda_n)}.
    \]
\end{proposition}

\begin{proof}[Proof of \Cref{prop:post-tuning-transfer}]
    Put
    \[
        Y_n=\frac{\widehat\mu_{0,\widetilde\lambda_n}-\mu_{0,\beta}}
        {s_n(\widetilde\lambda_n)},
        \qquad
        Z_n=\frac{\widehat\mu_{0,\widehat\lambda}
        -\widehat\mu_{0,\widetilde\lambda_n}}
        {s_n(\widetilde\lambda_n)}.
    \]
    The mean-value theorem gives
    \[
        |Z_n|
        \leq
        |\widehat\lambda-\widetilde\lambda_n|
        \sup_{\lambda\in\Lambda}
        \frac{|\partial_\lambda\widehat\mu_{0,\lambda}|}
        {s_n(\widetilde\lambda_n)}
        =\op(1).
    \]
    For every \(\varepsilon>0\),
    \[
        \EE\{\PP(|Z_n|>\varepsilon\mid\mathcal G_n)\}
        =\PP(|Z_n|>\varepsilon)\longrightarrow0,
    \]
    so \(\PP(|Z_n|>\varepsilon\mid\mathcal G_n)\pto0\).
    Applying \Cref{lem:conditional-kolmogorov-perturbation} with \(F=\Phi\) gives the asserted conditional convergence for \(Y_n+Z_n\).
\end{proof}

\subsubsection[Proof of Theorem \ref{thm:predictive-calibration}]{Proof of \Cref{thm:predictive-calibration}}
\begin{proof}[Proof of \Cref{thm:predictive-calibration}]
    For \Cref{thm:predictive-calibration}~\eqref{thm:predictive-calibration-i}, conditional on \(\mathcal G_n\),
    \[
        \widehat\mu_{0,\lambda}-\mu_{0,\beta}
        =d_n(\lambda)^\top\beta+\gamma_\lambda^\top\epsilon_0
    \]
    is centred Gaussian with variance \(R_n(\lambda)\), since the two terms are conditionally independent centred Gaussian variables.
    Thus the exact conditional Gaussian statement follows.
    Since \(\one_{n_0}^\top\gamma_\lambda=1\), Cauchy--Schwarz gives \(\|\gamma_\lambda\|_2^2\geq n_0^{-1}\), and hence \(R_n(\lambda)\geq\sigma_0^2/n_0>0\).
    Therefore,
    \[
        \PP\{\widehat R_n(\lambda)\leq0\}
        \leq
        \PP\left\{
        \left|
        \frac{\widehat R_n(\lambda)}{R_n(\lambda)}-1
        \right|\geq1
        \right\}
        \to0.
    \]
    Define
    \[
        \widehat s_n
        =
        \begin{cases}
        \widehat R_n(\lambda)^{1/2},
            & \widehat R_n(\lambda)>0,\\
        1,  & \widehat R_n(\lambda)\leq0.
        \end{cases}
    \]
    Then
    \[
        \frac{\widehat s_n}{R_n(\lambda)^{1/2}}\pto1,
    \]
    so \Cref{lem:plugin-studentization} applies with \(s_n=R_n(\lambda)^{1/2}\).
    Any other definition on \(\{\widehat R_n(\lambda)\leq0\}\) gives the same limit because the conditional probability of this event converges to zero in probability.

    For \Cref{thm:predictive-calibration}~\eqref{thm:predictive-calibration-ii}, set
    \(\widetilde\lambda_n=\lambda_{\eta,\star}^*\) and
    \(s_n=R_n(\lambda_{\eta,\star}^*)^{1/2}\).
    By \Cref{thm:risk-estimation}~\eqref{thm:risk-estimation-iii},
    \(\widehat\lambda\pto\lambda_{\eta,\star}^*\).
    By \Cref{lem:ood-post-tuning-stability},
    \(n_0^{\eta/2}\sup_{\lambda\in\Lambda}|\partial_\lambda\widehat\mu_{0,\lambda}|=\Op(1)\) and
    \(\widehat R_n(\widehat\lambda)/R_n(\lambda_{\eta,\star}^*)\pto1\).
    The uniform risk limit and \(\sR_{\eta,\star}(\lambda_{\eta,\star}^*)>0\) give
    \[
        n_0^{\eta/2}s_n
        \pto
        \sR_{\eta,\star}(\lambda_{\eta,\star}^*)^{1/2}>0,
    \]
    so the derivative condition in \Cref{prop:post-tuning-transfer} holds.
    Part~\eqref{thm:predictive-calibration-i} gives the conditional normal pivot at the deterministic penalty \(\widetilde\lambda_n\), and \Cref{prop:post-tuning-transfer} therefore gives the selected-penalty conditional Gaussian limit with denominator \(s_n\).

    The risk ratio above implies \(\PP\{\widehat R_n(\widehat\lambda)>0\}\to1\).
    Define
    \[
        \widehat s_n
        =
        \begin{cases}
        \widehat R_n(\widehat\lambda)^{1/2},
            & \widehat R_n(\widehat\lambda)>0,\\
        1,  & \widehat R_n(\widehat\lambda)\leq0.
        \end{cases}
    \]
    Then \(\widehat s_n/s_n\pto1\), so \Cref{lem:plugin-studentization} gives the stated selected-penalty result.
    Any other definition on \(\{\widehat R_n(\widehat\lambda)\leq0\}\) gives the same limit because the conditional probability of this event converges to zero in probability.
\end{proof}

\subsubsection[Non-Gaussian coefficient extension (Theorem \ref{thm:predictive-normality-nongaussian})]{Non-Gaussian coefficient extension (\Cref{thm:predictive-normality-nongaussian})}

For deterministic or \(\mathcal G_n\)-measurable \(\lambda\), \Cref{asm:variance-component-model} gives
\begin{equation}
    \widehat\mu_{0,\lambda}-\mu_{0,\beta}
    =\sum_{j=1}^p\frac{r d_{n,j}(\lambda)}{\sqrt p}\xi_{n,j}
    +\gamma_\lambda^\top\epsilon_0.
    \label{eq:predictive-representation}
\end{equation}

\begin{lemma}[Conditional characteristic-function convergence]
    \label{lem:conditional-characteristic-functions}
    Let \(\mu_n\) be random probability measures with characteristic functions
    \(\varphi_n\). Suppose that, for every fixed \(t\in\RR\),
    \(\varphi_n(t)\pto\varphi(t)\), where \(\varphi\) is the characteristic
    function of a deterministic probability law \(F\). Then
    \(\mu_n\Rightarrow F\) weakly in probability. If \(F\) has a continuous
    distribution function, then
    \[
        \sup_{t\in\RR}|\mu_n(({-}\infty,t])-F(t)|\pto0.
    \]
\end{lemma}

\begin{proof}[Proof of \Cref{lem:conditional-characteristic-functions}]
    Since \(|\varphi_n(t)-\varphi(t)|\leq2\), convergence in probability implies
    \(\EE|\varphi_n(t)-\varphi(t)|\to0\) for every \(t\in\RR\).
    Hence, for every \(T<\infty\), dominated convergence and Fubini's theorem give
    \[
        \int_{-T}^T|\varphi_n(t)-\varphi(t)|\,\rd t\pto0.
    \]
    From any subsequence, extract a further subsequence along which these integrals converge almost surely to zero for every positive integer \(T\).
    Since \(\varphi\) is continuous at zero, the characteristic-function tightness criterion gives tightness of the corresponding probability measures almost surely.
    Every weakly convergent further subsequence has characteristic function equal to \(\varphi\) almost everywhere by the local \(L_1\) convergence, hence everywhere by continuity.
    L\'evy's continuity theorem therefore identifies its limit as \(F\).
    The subsequence principle yields \(\mu_n\Rightarrow F\) weakly in probability.
    If \(F\) has a continuous distribution function, P\'olya's theorem and the same subsequence argument give uniform convergence of the distribution functions in probability.
\end{proof}

\begin{theorem}[Non-Gaussian coefficient extension]
    \label{thm:predictive-normality-nongaussian}
    Suppose \Cref{asm:predictive-model} and the coefficient conditions in
    \Cref{asm:variance-component-model} hold, and suppose
    \(\epsilon_0\mid\mathcal G_n\sim\cN(0,\sigma_0^2I_{n_0})\), independently of
    \(\beta\). Let \(\lambda\) be deterministic or
    \(\mathcal G_n\)-measurable, set \(s_n^2=R_n(\lambda)\), and assume
    \[
        \max_{j\leq p}
        \frac{|d_{n,j}(\lambda)|}{\sqrt p\,s_n}\pto0.
    \]
    Then
    \[
        \sup_{t\in\RR}
        \left|
        \PP\left(
            \frac{\widehat\mu_{0,\lambda}-\mu_{0,\beta}}{s_n}\leq t
            \,\middle|\,\mathcal G_n
        \right)-\Phi(t)
        \right|\pto0.
    \]
    If \(n_0^\eta s_n^2\pto\sR_{\eta,\star}\in(0,\infty)\), then
    \[
        n_0^{\eta/2}(\widehat\mu_{0,\lambda}-\mu_{0,\beta})
        \dto\cN(0,\sR_{\eta,\star}).
    \]
    If, in addition, \(\widehat s_n>0\) with probability tending to one and
    \(\widehat s_n/s_n\pto1\), the same conditional conclusion holds with
    \(\widehat s_n\) in the denominator.
    The studentised statistic may be defined arbitrarily on \(\{\widehat s_n\leq0\}\), whose probability tends to zero.
\end{theorem}

\begin{proof}[Proof of \Cref{thm:predictive-normality-nongaussian}]
    Write \(\zeta_{n,i}=\epsilon_{0i}/\sigma_0\), and define
    \[
        A_{n,j} = \frac{rd_{n,j}}{\sqrt{p}s_n}, \quad B_{n,i} = \frac{\sigma_0\gamma_{\lambda,i}}{s_n}.
    \]
    By \eqref{eq:predictive-representation}, 
    \[
        \frac{\widehat\mu_{0,\lambda} - \mu_{0,\beta}}{s_n}
        =\sum_{j=1}^pA_{n,j}\xi_{n,j}
        +\sum_{i=1}^{n_0}B_{n,i}\zeta_{n,i}.
    \]
    Set \(\alpha_n^2=\|A_n\|_2^2\) and
    \(\beta_n^2=\|B_n\|_2^2\), so
    \(\alpha_n^2+\beta_n^2=1\). Conditional on \(\mathcal G_n\),
    \(\sum_{i=1}^{n_0}B_{n,i}\zeta_{n,i}\sim \cN(0,\beta_n^2)\).

    Since \(\EE\xi_{n,j}^2=1\), \(\EE\xi_{n,j}=0\), and the third moments
    are uniformly bounded, there exists a constant \(C<\infty\) such
    that for all \(u\),
    \[
        \left| \EE e^{iu\xi_{n,j}} - \left(1 - \frac{u^2}{2}\right) \right|
        \leq C|u|^3.
    \]
    Therefore, for each fixed \(t\),
    \[
        \EE(e^{itA_{n,j}\xi_{n,j}}|\mathcal{G}_n)
        =1-\frac{t^2A_{n,j}^2}{2}+\cO(|t|^3|A_{n,j}|^3).
    \]
    Since \(\max_j|A_{n,j}|\pto0\) and \(\alpha_n^2\leq1\),
    \[
        \sum_{j=1}^p|A_{n,j}|^3
        \leq \max_{j\leq p}|A_{n,j}|\,\alpha_n^2\pto0.
    \]
    We then obtain 
    \[
        \EE(e^{it\sum_{j=1}^pA_{n,j}\xi_{n,j}}|\mathcal{G}_n)  = e^{-t^2\alpha_n^2/2 + \op(1)}.
    \]
    Since \((\sum_{j=1}^pA_{n,j}\xi_{n,j})\) is independent of
    \(\sum_{i=1}^{n_0}B_{n,i}\zeta_{n,i}\) conditional on
    \(\mathcal G_n\), and
    \(\EE[\exp\{it\sum_iB_{n,i}\zeta_{n,i}\}\mid\mathcal G_n]
    =\exp(-t^2\beta_n^2/2)\),
    \[
        \begin{aligned}
            \EE\left[
            \exp\left( 
            it\frac{\widehat\mu_{0,\lambda} - \mu_{0,\beta}}{s_n}
            \right)\,\middle|\, \mathcal{G}_n
            \right]
            &= \EE[e^{it\sum_{j=1}^pA_{n,j}\xi_{n,j}}|\mathcal{G}_n]\EE[e^{it\sum_{i=1}^{n_0}B_{n,i}\zeta_{n,i}}|\mathcal{G}_n]\\
            &= \exp\left\{-t^2(\alpha^2_n + \beta_n^2)/2 + \op(1)\right\}\\
            &= \exp\left\{-\frac{t^2}{2} + \op(1)\right\}\\
            & \pto e^{-t^2/2}.
        \end{aligned}
    \]
    Thus the conditional characteristic function converges in probability to \(e^{-t^2/2}\) for every fixed \(t\in\RR\).
    By \Cref{lem:conditional-characteristic-functions}, the conditional laws converge weakly in probability to \(\cN(0,1)\), and continuity of \(\Phi\) gives
    \[
        \sup_{t\in\RR}
        \left|
            \PP\left(
                \frac{\widehat\mu_{0,\lambda}-\mu_{0,\beta}}{s_n}
                \leq t
                \,\middle|\,
                \mathcal G_n
            \right)
            -\Phi(t)
        \right|
        \pto0.
    \]

    If \(n_0^\eta s_n^2\pto\sR_{\eta,\star}\in(0,\infty)\), then
    \[
        n_0^{\eta/2}\{\widehat\mu_{0,\lambda}-\mu_{0,\beta}\}
        =\{n_0^{\eta/2}s_n\}
        \frac{\widehat\mu_{0,\lambda}-\mu_{0,\beta}}{s_n}
        \dto\cN(0,\sR_{\eta,\star})
    \]
    by Slutsky's theorem. If, in addition, \(\widehat s_n>0\) with
    probability tending to one and \(\widehat s_n/s_n\pto1\), then the
    conditional plug-in conclusion follows from
    \Cref{lem:plugin-studentization}.
\end{proof}

\begin{remark}[Fixed-surface versus integrated squared bias]
    \label{rem:fixed-surface-bias}
    For fixed \(\lambda\), suppose \(\beta_j=r\xi_{n,j}/\sqrt p\), where the
    \(\xi_{n,j}\) are independent and standardised. Because
    \(d_n(\lambda)\) is design-dependent, impose, for every
    \(\varepsilon>0\), the conditional Lindeberg condition
    \[
        \frac{1}{\|d_n(\lambda)\|_2^2}
        \sum_{j=1}^p d_{n,j}(\lambda)^2
        \EE\!\left[
        \xi_{n,j}^2
        \ind{\{|d_{n,j}(\lambda)\xi_{n,j}|>
        \varepsilon\|d_n(\lambda)\|_2\}}
        \,\middle|\,\mathcal G_n
        \right]\pto0
    \]
    on \(\{B_n(\lambda)>0\}\). Then, conditionally on \(\mathcal G_n\),
    \[
        \frac{\beta^\top d_n(\lambda)}
        {(r/\sqrt p)\|d_n(\lambda)\|_2}
        \Rightarrow\cN(0,1),
        \qquad
        \frac{B(\beta;\lambda)}{B_n(\lambda)}
        \Rightarrow\chi_1^2
    \]
    in probability.
    Thus the fixed-surface squared conditional bias need not concentrate around its random-effects integral.
    The random-matrix norm limits do not imply this coordinatewise condition.
\end{remark}

\clearpage
\section{Random-Matrix and Concentration Tools}
\label{app:rmt-concentration}

This appendix collects the random-matrix and concentration results used in Appendices~\ref{app:main-proofs} and~\ref{app:profile-transfer}.
The source-design arguments are organised around the technical source-resolvent condition in \Cref{asm:technical-source-resolvent}.
The sub-Gaussian source condition in \Cref{asm:global-source-resolvent} verifies this condition, while the subsequent source-design results use \Cref{asm:technical-source-resolvent} directly.

Fix a compact interval \(\Lambda=[\lambda_-,\lambda_+]\subset(0,\infty)\), choose \(0<r_\Lambda<\min\{\lambda_-/8,1/2\}\), and define
\[
    \mathcal D_\Lambda(r_\Lambda)
    :=
    \{z\in\mathbb C:\operatorname{dist}(z,\Lambda)\leq r_\Lambda\}.
\]
Our resolvents use the convention \((G+zI)^{-1}\), so the conventional random-matrix spectral parameter is \(-z\), and \(-\mathcal D_\Lambda(2r_\Lambda)\) stays a fixed positive distance from \([0,\infty)\).
The larger neighbourhood \(\mathcal D_\Lambda(2r_\Lambda)\) provides the analytic extension needed for Cauchy's formula on \(\mathcal D_\Lambda(r_\Lambda)\).

For every \(H\succeq0\) and \(z\in\mathcal D_\Lambda(2r_\Lambda)\),
\[
    \|(H+zI)^{-1}\|_{\oper}
    \leq
    (\Re z)^{-1}
    \leq
    (\lambda_--2r_\Lambda)^{-1}.
\]
For the fixed-point solution \(v_n(\lambda)\) defined in Section~\ref{sec:risk:asymptotic-risk}, write
\[
    F_n(v,z)
    :=
    v^{-1}
    -
    z
    -
    \phi_{0,n}
    \otr\{\Sigma_0(I_p+v\Sigma_0)^{-1}\}.
\]
Under \Cref{asm:rmt}, on \(\Lambda\),
\[
    (\lambda_++C_\phi C_\Sigma)^{-1}
    \leq
    v_n(\lambda)
    \leq
    \lambda_-^{-1},
\]
and its stability denominator satisfies
\[
    v_n(\lambda)^{-2}
    -
    \phi_{0,n}
    \otr\{\Sigma_0^2(I_p+v_n(\lambda)\Sigma_0)^{-2}\}
    =
    \frac{\lambda}{v_n(\lambda)}
    +
    \phi_{0,n}
    \otr\left\{
        \frac{\Sigma_0}
        {v_n(\lambda)(I_p+v_n(\lambda)\Sigma_0)^2}
    \right\}
    \geq
    \lambda_-^2.
\]
After decreasing \(r_\Lambda\) if necessary, a uniform implicit-function argument therefore gives a unique analytic extension \(v_n(z)\) to \(\mathcal D_\Lambda(2r_\Lambda)\) and a constant \(c>0\) not depending on \(n\) such that
\[
    \inf_{z\in\mathcal D_\Lambda(2r_\Lambda)}
    |\partial_vF_n\{v_n(z),z\}|
    \geq
    c.
\]
These deterministic bounds use only the bounded-spectrum and proportional-growth conditions in \Cref{asm:rmt}.

\begin{assumption}[Technical source-resolvent concentration]
    \label{asm:technical-source-resolvent}
    Let
    \[
        B_0
        =
        n_0^{-1/2}Z_0\Sigma_0^{1/2},
        \qquad
        G_0
        =
        B_0B_0^\top,
    \]
    and define
    \[
        R_0(z)
        =
        (G_0+zI_{n_0})^{-1},
        \qquad
        M_0(z)
        =
        (B_0^\top B_0+zI_p)^{-1}.
    \]
    For some \(\zeta\in(0,1/4)\),
    \begin{align*}
        \sup_{z\in\mathcal D_\Lambda(2r_\Lambda)}
        \left|
            a_n^\top\{R_0(z)-v_n(z)I_{n_0}\}b_n
        \right|
        &=
        \op(1)\|a_n\|_2\|b_n\|_2,\\
        \sup_{z\in\mathcal D_\Lambda(2r_\Lambda)}
        \left|
            n_0^{-1}\tr R_0(z)-v_n(z)
        \right|
        &=
        \op(1),\\
        \sup_{z\in\mathcal D_\Lambda(2r_\Lambda)}
        \left|
            \xi_n^\top M_0(z)B_0^\top a_n
        \right|
        &=
        \Op(n_0^{-1/2+\zeta})
        \|\xi_n\|_2\|a_n\|_2
    \end{align*}
    for deterministic real vectors \(a_n\), \(b_n\), and \(\xi_n\) of conformable dimensions.
    For every sequence of sigma-fields \(\mathcal A_n\) independent of \(Z_0\), the same bounds hold, for \(\PP\)-almost every realisation of \(\mathcal A_n\), under the conditional law of \(Z_0\) given \(\mathcal A_n\) whenever \(a_n\), \(b_n\), and \(\xi_n\) are \(\mathcal A_n\)-measurable.
\end{assumption}

Appendix~\ref{app:uncentered-resolvent} verifies \Cref{asm:technical-source-resolvent} under \Cref{asm:rmt,asm:global-source-resolvent} and records the higher-resolvent and cross-block bounds used below.
Appendix~\ref{app:centered-resolvent} develops the corresponding centred-source and source-mean results, while Appendix~\ref{app:treated-mean-concentration} establishes treated-mean concentration.

\subsection{Verification of the Source-Resolvent Condition}
\label{app:uncentered-resolvent}

Under the sub-Gaussian source condition, the anisotropic resolvent result below verifies \Cref{asm:technical-source-resolvent}; we then record the higher-resolvent and cross-block bounds used with the centred source design.

\begin{proposition}[Verification of source-resolvent concentration]
    \label{prop:global-source-resolvent-sufficient}
    Under \Cref{asm:rmt,asm:global-source-resolvent}, \Cref{asm:technical-source-resolvent} holds.
\end{proposition}

\begin{proof}[Proof of \Cref{prop:global-source-resolvent-sufficient}]
    For \(i=1,\ldots,n_0\), let \(z_{0i}\in\RR^p\) denote the transpose of the \(i\)th row of \(Z_0\), and define
    \[
        g_i
        :=
        \Sigma_0^{1/2}z_{0i},
        \qquad
        \mathsf G
        :=
        [g_1,\ldots,g_{n_0}]
        =
        \Sigma_0^{1/2}Z_0^\top.
    \]
    To match the notation of \citet{fan2026anisotropic}, set \(n_{\mathrm F}=p\), \(N_{\mathrm F}=n_0\), \(G_{\mathrm F}=\mathsf G\), and \(\Sigma_{\mathrm F}=\Sigma_0\).
    Their sample covariance and Gram matrices are then
    \[
        K_{\mathrm F}
        =
        \frac{\mathsf G\mathsf G^\top}{n_0}
        =
        B_0^\top B_0,
        \qquad
        \widetilde K_{\mathrm F}
        =
        \frac{\mathsf G^\top\mathsf G}{n_0}
        =
        B_0B_0^\top
        =
        G_0.
    \]
    We apply \citet[Theorem~2.10]{fan2026anisotropic}, which requires their Assumptions~1 and~2.
    By \Cref{asm:rmt}, \(n_{\mathrm F}/N_{\mathrm F}=p/n_0\) is bounded above and away from zero and \(\|\Sigma_{\mathrm F}\|_{\oper}\leq C_\Sigma\).
    Choose \(0<c_{\mathrm F}<\min\{c_\phi,c_\Sigma\}\); since \(\Sigma_0\succeq c_\Sigma I_p\), the spectral nondegeneracy condition in their Assumption~1 holds.
    For their separable model, take \(d_{\mathrm F}=p\), \(X_{\mathrm F}=\Sigma_0^{1/2}\), and \(w_i=z_{0i}\), so \(g_i=X_{\mathrm F}w_i\) and \(X_{\mathrm F}X_{\mathrm F}^\top=\Sigma_0\).
    The coordinates of \(w_i\) are independent, centred, and standardised by \Cref{asm:rmt}, and \Cref{asm:global-source-resolvent} gives uniform sub-Gaussianity.
    Hence \citet[Proposition~2.12]{fan2026anisotropic} verifies their Assumption~2, with its dimension-comparability condition immediate from \(d_{\mathrm F}=n_{\mathrm F}=p\).

    Let \(\widetilde m_{0,\mathrm F}(w)\) denote the deterministic companion Stieltjes transform in \citet{fan2026anisotropic}.
    Its fixed-point equation is
    \[
        w
        =
        -\frac{1}{\widetilde m_{0,\mathrm F}(w)}
        +
        \frac{1}{n_0}
        \tr\left\{
            \Sigma_0
            \left(I_p+\widetilde m_{0,\mathrm F}(w)\Sigma_0\right)^{-1}
        \right\}.
    \]
    Substituting \(w=-z\) and using \(n_0^{-1}\tr(\cdot)=\phi_{0,n}\otr(\cdot)\) gives
    \[
        \frac{1}{\widetilde m_{0,\mathrm F}(-z)}
        =
        z
        +
        \phi_{0,n}
        \otr\left\{
            \Sigma_0
            \left(I_p+\widetilde m_{0,\mathrm F}(-z)\Sigma_0\right)^{-1}
        \right\},
    \]
    which is the defining equation for \(v_n(z)\).
    By the uniqueness and uniform stability of the analytic branch established above,
    \[
        v_n(z)
        =
        \widetilde m_{0,\mathrm F}(-z),
        \qquad
        z\in\mathcal D_\Lambda(2r_\Lambda).
    \]

    Put \(d_\Lambda=\lambda_--2r_\Lambda>0\), \(M_\Lambda=\lambda_++2r_\Lambda\), and
    \[
        \delta_*
        :=
        \frac12
        \min\{d_\Lambda,M_\Lambda^{-1},1\}.
    \]
    For \(z\in\mathcal D_\Lambda(2r_\Lambda)\) and \(w=-z\), we have \(\operatorname{dist}\{w,[0,\infty)\}\geq d_\Lambda\), \(|w|\geq d_\Lambda\), \(|\Re w|\leq M_\Lambda\), and \(|\Im w|\leq2r_\Lambda<1\).
    Since the deterministic companion spectral measure is supported on \([0,\infty)\), these bounds give
    \[
        -\mathcal D_\Lambda(2r_\Lambda)
        \subset
        \overline{\mathcal D}_{\mathrm F}^{\,o}(\delta_*),
    \]
    the separated complex domain of \citet[Theorem~2.10]{fan2026anisotropic}.

    Write
    \[
        R_{\mathrm F}(w)
        =
        (K_{\mathrm F}-wI_p)^{-1},
        \qquad
        \widetilde R_{\mathrm F}(w)
        =
        (\widetilde K_{\mathrm F}-wI_{n_0})^{-1}.
    \]
    At \(w=-z\), \(R_{\mathrm F}(-z)=M_0(z)\) and \(\widetilde R_{\mathrm F}(-z)=R_0(z)\), while the corresponding linearised resolvent satisfies
    \[
        \begin{pmatrix}
            zI_p & B_0^\top\\
            B_0 & -I_{n_0}
        \end{pmatrix}^{-1}
        =
        \begin{pmatrix}
            M_0(z) & M_0(z)B_0^\top\\
            B_0M_0(z) & -zR_0(z)
        \end{pmatrix}.
    \]
    The lower-right, off-diagonal, and averaged conclusions of \citet[Theorem~2.10]{fan2026anisotropic} therefore give, for every fixed \(\varepsilon>0\), uniformly over \(z\in\mathcal D_\Lambda(2r_\Lambda)\),
    \begin{align*}
        \left|
            a_n^\top
            \{R_0(z)-v_n(z)I_{n_0}\}
            b_n
        \right|
        &=
        \Op(n_0^{-1/2+\varepsilon})
        \|a_n\|_2\|b_n\|_2,\\
        \left|
            \xi_n^\top M_0(z)B_0^\top a_n
        \right|
        &=
        \Op(n_0^{-1/2+\varepsilon})
        \|\xi_n\|_2\|a_n\|_2,\\
        \left|
            n_0^{-1}\tr R_0(z)-v_n(z)
        \right|
        &=
        \Op(n_0^{-1+\varepsilon}),
    \end{align*}
    where the first bound uses \(\inf_{z\in\mathcal D_\Lambda(2r_\Lambda)}|z|\geq d_\Lambda\), and the deterministic approximation to the off-diagonal block is zero.
    Choose \(0<\varepsilon<\zeta\).
    Then the first two rates are respectively \(\op(1)\|a_n\|_2\|b_n\|_2\) and \(\Op(n_0^{-1/2+\zeta})\|\xi_n\|_2\|a_n\|_2\), while the averaged rate is \(\op(1)\).
    This proves the three unconditional bounds in \Cref{asm:technical-source-resolvent}.

    Finally, let \(\mathcal A_n\) be independent of \(Z_0\).
    Conditional on any admissible realisation of \(\mathcal A_n\), the law of \(Z_0\) is unchanged and the \(\mathcal A_n\)-measurable directions are deterministic.
    The constants in \citet[Theorem~2.10]{fan2026anisotropic} depend on the model and spectral-domain bounds, not on the particular deterministic unit directions, so the same estimates hold under the conditional source-design law after normalising each nonzero direction; zero directions are immediate.
    Homogeneity restores the direction norms, proving the conditional clause of \Cref{asm:technical-source-resolvent}, including \(\mathcal A_n=\sigma(\mathcal F_{\mathrm{aux}},X_1)\).
\end{proof}

The next lemma records the higher-resolvent and cross-block consequences of \Cref{asm:technical-source-resolvent}.

\begin{lemma}[Higher-resolvent and cross-block bounds]
    \label{lem:global-linearized-resolvent}
    Under \Cref{asm:rmt,asm:technical-source-resolvent}, write
    \[
        X
        =
        Z_0\Sigma_0^{1/2}
        =
        \sqrt{n_0}B_0,
        \qquad
        m_{k,n}(z)
        =
        n_0^{-1}\tr\{R_0(z)^k\}.
    \]
    For \(k=1,2\) and deterministic \(a_n,b_n\in\RR^{n_0}\),
    \[
        \sup_{z\in\mathcal D_\Lambda(r_\Lambda)}
        \left|
            a_n^\top R_0(z)^k b_n
            -
            m_{k,n}(z)a_n^\top b_n
        \right|
        =
        \op(1)\|a_n\|_2\|b_n\|_2,
    \]
    and
    \[
        \sup_{z\in\mathcal D_\Lambda(r_\Lambda)}
        |m_{1,n}(z)-v_n(z)|
        =
        \op(1),
        \qquad
        \sup_{z\in\mathcal D_\Lambda(r_\Lambda)}
        |m_{2,n}(z)+v_n'(z)|
        =
        \op(1).
    \]
    For every fixed integer \(q\geq1\) and deterministic \(u_n\in\RR^p\),
    \[
        \sup_{\lambda\in\Lambda}
        \left|
            u_n^\top M_0(\lambda)^qX^\top a_n
        \right|
        =
        \Op(n_0^\zeta)\|u_n\|_2\|a_n\|_2.
    \]
    For every sequence of sigma-fields \(\mathcal A_n\) independent of \(Z_0\), all conclusions remain valid, for \(\PP\)-almost every realisation of \(\mathcal A_n\), under the conditional law of \(Z_0\) given \(\mathcal A_n\) whenever the directions are \(\mathcal A_n\)-measurable.
\end{lemma}

\begin{proof}[Proof of \Cref{lem:global-linearized-resolvent}]
    It suffices to prove the conditional statement after fixing an admissible realisation of \(\mathcal A_n\); the unconditional statement follows by taking \(\mathcal A_n\) trivial.
    The row-resolvent and normalised-trace bounds in \Cref{asm:technical-source-resolvent} give
    \[
        \sup_{z\in\mathcal D_\Lambda(2r_\Lambda)}
        |m_{1,n}(z)-v_n(z)|
        =
        \op(1),
    \]
    and, by subtraction, the asserted trace-centred bilinear bound for \(k=1\).

    Since \(X=\sqrt{n_0}B_0\), the cross-block bound in \Cref{asm:technical-source-resolvent} gives directly
    \[
        \sup_{\lambda\in\Lambda}
        \left|
            u_n^\top M_0(\lambda)X^\top a_n
        \right|
        =
        \Op(n_0^\zeta)\|u_n\|_2\|a_n\|_2.
    \]
    Cauchy's formula on contours contained in \(\mathcal D_\Lambda(2r_\Lambda)\) differentiates the preceding analytic bounds uniformly on \(\mathcal D_\Lambda(r_\Lambda)\).
    Using \(R_0'(z)=-R_0(z)^2\) and \(\{M_0(z)^q\}'=-qM_0(z)^{q+1}\) gives the \(k=2\) bilinear bound, \(m_{2,n}(z)+v_n'(z)=\op(1)\), and the stated cross-block bound for every fixed \(q\geq2\).
    Homogeneity gives the displayed direction norms and completes the conditional and unconditional conclusions.
\end{proof}

\subsection{Centred Source Resolvents and Source-Mean Identities}
\label{app:centered-resolvent}

Set \(u=n_0^{-1/2}\one_{n_0}\) and \(G_0^c=C_0G_0C_0\), with \(B_0\), \(G_0\), \(R_0\), and \(M_0\) as in \Cref{asm:technical-source-resolvent}.
Since \(C_0=I_{n_0}-uu^\top\) is a rank-one modification of the identity, the identities below reduce the centred source quantities to the uncentred resolvents controlled in Appendix~\ref{app:uncentered-resolvent}.

\begin{lemma}[Exact centred row-resolvent compression]
    \label{lem:centered-row-compression}
    Let \(Q\in\RR^{n_0\times(n_0-1)}\) satisfy \(Q^\top Q=I_{n_0-1}\), \(QQ^\top=C_0\), and \(Q^\top u=0\), and set
    \[
        R_\perp(z)
        =
        Q\{Q^\top G_0Q+zI_{n_0-1}\}^{-1}Q^\top,
        \qquad
        z\in\mathcal D_\Lambda(r_\Lambda).
    \]
    Then, exactly,
    \[
        R_\perp(z)
        =
        R_0(z)
        -
        \frac{R_0(z)uu^\top R_0(z)}
        {u^\top R_0(z)u},
        \qquad
        (G_0^c+zI_{n_0})^{-1}
        =
        z^{-1}uu^\top+R_\perp(z).
    \]
    Consequently, for \(g=C_0\gamma\),
    \[
        g^\top R_\perp(z)g
        =
        g^\top R_0(z)g
        -
        \frac{\{g^\top R_0(z)u\}^2}
        {u^\top R_0(z)u},
        \qquad
        \tr R_\perp(z)
        =
        \tr R_0(z)
        -
        \frac{u^\top R_0(z)^2u}
        {u^\top R_0(z)u}.
    \]
\end{lemma}

\begin{proof}[Proof of \Cref{lem:centered-row-compression}]
    For \(z\in\mathcal D_\Lambda(r_\Lambda)\), \(\Re\{u^\top R_0(z)u\}>0\), so the denominator below is nonzero.
    Let
    \[
        P
        =
        R_0
        -
        R_0u(u^\top R_0u)^{-1}u^\top R_0,
    \]
    with all resolvents evaluated at \(z\).
    Then \(Pu=0\), \(P=C_0PC_0\), and direct multiplication gives \(Q^\top(G_0+zI_{n_0})PQ=I_{n_0-1}\).
    Hence
    \[
        P
        =
        Q\{Q^\top G_0Q+zI_{n_0-1}\}^{-1}Q^\top.
    \]
    Since \(G_0^c=C_0G_0C_0\) vanishes on the span of \(u\) and has compressed matrix \(Q^\top G_0Q\) on its orthogonal complement, the full centred identity follows.
    The quadratic-form and trace identities are immediate.
\end{proof}

The next result replaces a design-independent direction \(g=C_0\gamma\) by the corresponding centred trace average.

\begin{lemma}[Design-independence trace replacement]
    \label{lem:external-trace-replacement}
    Under \Cref{asm:rmt,asm:technical-source-resolvent}, together with \Cref{asm:design-independent-regime}~\eqref{asm:design-independent-regime-i}, for \(k=1,2\), conditionally on \(\mathcal F_{\mathrm{aux}}\) and for \(\PP\)-almost every auxiliary realisation,
    \[
        \sup_{\lambda\in\Lambda}
        \left|
            g^\top(G_0^c+\lambda I_{n_0})^{-k}g
            -
            \|g\|_2^2
            \frac{1}{n_0-1}
            \tr\{C_0(G_0^c+\lambda I_{n_0})^{-k}\}
        \right|
        =
        \op(\|g\|_2^2),
    \]
    where \(G_0^c=X_0^c(X_0^c)^\top/n_0\).
\end{lemma}

\begin{proof}[Proof of \Cref{lem:external-trace-replacement}]
    Fix an admissible auxiliary realisation.
    If \(g=0\), the result is immediate.
    Otherwise apply \Cref{lem:global-linearized-resolvent} to \(g/\|g\|_2\) and \(u\).
    Since \(g^\top u=0\), uniformly on \(\Lambda\),
    \[
        g^\top R_0g
        =
        m_{1,n}\|g\|_2^2+\op(\|g\|_2^2),
        \qquad
        g^\top R_0u
        =
        \op(\|g\|_2),
        \qquad
        u^\top R_0u
        =
        m_{1,n}+\op(1).
    \]
    The denominator is uniformly bounded away from zero with high probability because
    \[
        u^\top R_0u
        \geq
        \{\|G_0\|_{\oper}+\lambda_+\}^{-1},
        \qquad
        \|G_0\|_{\oper}
        =
        \Op(1).
    \]
    Substitution in \Cref{lem:centered-row-compression} proves the \(k=1\) result and its trace counterpart.

    For \(k=2\), use the same identities on \(\mathcal D_\Lambda(r_\Lambda)\).
    If \(G_0=\sum_j\kappa_jq_jq_j^\top\), then
    \[
        \Re\{u^\top R_0(z)u\}
        =
        \sum_j
        |q_j^\top u|^2
        \frac{\kappa_j+\Re z}
        {|\kappa_j+z|^2}
        >
        0.
    \]
    On \(\{\|G_0\|_{\oper}\leq C\}\), this real part is uniformly bounded below on \(\mathcal D_\Lambda(r_\Lambda)\).
    Hence the compression denominator has no zeros there, and Cauchy's formula gives the \(k=2\) statement.
    The initial conditioning gives the stated conditional formulation.
\end{proof}

\begin{remark}[Centred finite-dimensional normalisation]
    Although \(X_0^c\) has at most \(n_0-1\) nonzero row directions, the finite-dimensional formulas use \(p/n_0\) and the penalty in \(n_0^{-1}(X_0^c)^\top X_0^c\).
    Under proportional growth, \(p/(n_0-1)-p/n_0=\cO(n_0^{-1})\), so using the exact rank-\((n_0-1)\) normalisation changes these deterministic equivalents by \(o(1)\).
\end{remark}

The remaining source-side fluctuation is the sample-mean term \(e_0\), whose centred-resolvent contributions can also be reduced to the uncentred row resolvent.

\begin{lemma}[Exact source-mean resolvent reduction]
    \label{lem:source-mean-resolvent}
    Under \Cref{asm:rmt,asm:technical-source-resolvent}, write \(e_0=B_0^\top u\) and define \(a_n(\lambda)=u^\top R_0(\lambda)u\) and \(b_n(\lambda)=u^\top R_0(\lambda)^2u\).
    Uniformly over \(\Lambda\),
    \begin{align*}
        M_\lambda^c e_0
        &=
        \frac{B_0^\top R_0u}{\lambda a_n},
        &
        e_0^\top M_\lambda^c e_0
        &=
        \frac1{\lambda a_n}-1,\\
        e_0^\top(M_\lambda^c)^2e_0
        &=
        \frac1{\lambda^2a_n}
        -
        \frac{b_n}{\lambda a_n^2},
        &
        a_n-v_n
        &=
        \op(1),
        \qquad
        b_n-\widetilde v_{v,n}
        =
        \op(1).
    \end{align*}
    The two source-mean contributions satisfy
    \begin{align*}
        \sup_{\lambda\in\Lambda}
        \left|
            \frac{n_0}{p}
            \lambda^2e_0^\top(M_\lambda^c)^2e_0
            -
            K_{2,n}(\lambda;\Sigma_0)
        \right|
        &\pto
        0,\\
        \sup_{\lambda\in\Lambda}
        \left|
            \frac{b_n}{a_n^2}
            -
            1
            -
            \frac{\phi_{0,n}}{\lambda}
            \{K_{1,n}(\lambda;\Sigma_0)-K_{2,n}(\lambda;\Sigma_0)\}
        \right|
        &\pto
        0.
    \end{align*}
    Finally,
    \[
        C_0B_0M_\lambda^c e_0
        =
        -\frac{C_0R_0u}{a_n},
        \qquad
        n_0
        \left\|
            \frac{C_0R_0u}{\sqrt{n_0}a_n}
        \right\|_2^2
        =
        \frac{b_n}{a_n^2}-1.
    \]
    Moreover, for every fixed integer \(q\geq1\) and every sequence of sigma-fields \(\mathcal A_n\) independent of \(Z_0\), for \(\PP\)-almost every realisation of \(\mathcal A_n\),
    \[
        \sup_{\lambda\in\Lambda}
        \left|
            w_n^\top(M_\lambda^c)^q e_0
        \right|
        =
        \Op(n_0^{-1/2+\zeta})\|w_n\|_2
    \]
    under the conditional law of \(Z_0\) given \(\mathcal A_n\) whenever \(w_n\in\RR^p\) is \(\mathcal A_n\)-measurable.
\end{lemma}

\begin{proof}[Proof of \Cref{lem:source-mean-resolvent}]
    Put \(S=B_0^\top B_0\) and \(M_0=(S+\lambda I_p)^{-1}\).
    Then \(S_0^c=S-e_0e_0^\top\), \(M_0e_0=B_0^\top R_0u\), and \(e_0^\top M_0e_0=1-\lambda a_n\).
    Sherman--Morrison gives the first two identities.
    Differentiating the second, using \(a_n'=-b_n\) and \(\partial_\lambda M_\lambda^c=-(M_\lambda^c)^2\), gives the third.

    Applying \Cref{lem:global-linearized-resolvent} with both directions equal to \(u\) gives \(a_n-m_{1,n}=\op(1)\) and \(b_n-m_{2,n}=\op(1)\), and the same lemma gives \(m_{1,n}-v_n=\op(1)\) and \(m_{2,n}+v_n'=\op(1)\), all uniformly on \(\Lambda\).
    Differentiating the fixed-point equation yields \(-v_n'=\widetilde v_{v,n}\), and hence \(b_n-\widetilde v_{v,n}=\op(1)\).
    Substitution in the definitions of \(K_{1,n}\) and \(K_{2,n}\) gives
    \[
        \frac1{\phi_{0,n}}
        \left\{
            \frac1{v_n}
            -
            \frac{\lambda\widetilde v_{v,n}}{v_n^2}
        \right\}
        =
        K_{2,n}(\lambda;\Sigma_0),
        \qquad
        \frac{\widetilde v_{v,n}}{v_n^2}-1
        =
        \frac{\phi_{0,n}}{\lambda}
        \{K_{1,n}-K_{2,n}\},
    \]
    which proves the two limits.
    Finally, \(B_0M_\lambda^c e_0=G_0R_0u/(\lambda a_n)\), \(G_0R_0=I_{n_0}-\lambda R_0\), and \(\|C_0R_0u\|_2^2=b_n-a_n^2\), proving the remaining identities.

    For the cross form,
    \[
        w_n^\top M_\lambda^c e_0
        =
        \frac{w_n^\top M_0B_0^\top u}
        {\lambda a_n}.
    \]
    The denominator is uniformly bounded away from zero with high probability, and the cross-block conclusion of \Cref{lem:global-linearized-resolvent} proves the claim for \(q=1\).
    The same identity holds on \(\mathcal D_\Lambda(r_\Lambda)\), where the cross-block bound in \Cref{asm:technical-source-resolvent} and the lower bound on \(\Re a_n(z)\) from the proof of \Cref{lem:external-trace-replacement} extend the \(q=1\) bound; Cauchy's formula applied to the analytic vector \((S_0^c+zI_p)^{-1}e_0\) then gives every fixed \(q\geq2\).
    Conditioning on an admissible realisation of \(\mathcal A_n\) and applying the deterministic-direction bound gives the conditional version.
    Integrating the conditional error probabilities by the tower property gives the corresponding unconditional statement.
\end{proof}

\subsection{Treated-Mean Concentration}
\label{app:treated-mean-concentration}

Under \Cref{asm:rmt}, the treated-mean fluctuation \(\epsilon_1=\bar x_1-\nu_1\) is independent of the source design.
Conditioning on source-measurable quantities therefore reduces quadratic and bilinear terms involving \(\epsilon_1\) to moment bounds, with uniform control over the ridge path.

\begin{lemma}[Treated-mean quadratic-form concentration]
    \label{lem:treated-mean-qf}
    Under \Cref{asm:rmt}, let \(\epsilon_1=\bar x_1-\nu_1\).
    Suppose \(\{A_\lambda:\lambda\in\Lambda\}\) is a symmetric, \(X_0\)-measurable matrix family with uniformly bounded operator norm and a uniform deterministic operator-norm Lipschitz constant.
    Then
    \[
        \sup_{\lambda\in\Lambda}
        \left|
            \epsilon_1^\top A_\lambda\epsilon_1
            -
            \frac1{n_1}\tr(\Sigma_1A_\lambda)
        \right|
        \pto
        0.
    \]
    This applies to \(I_p\), \(M_\lambda^c\), \((M_\lambda^c)^2\), and \(M_\lambda^c(I_p-\lambda M_\lambda^c)\).
\end{lemma}

\begin{proof}[Proof of \Cref{lem:treated-mean-qf}]
    Conditional on \(X_0\), write \(\epsilon_1=\Sigma_1^{1/2}\bar z_1\).
    The coordinates of \(\bar z_1\) are independent and centred, with variance \(n_1^{-1}\) and fourth moments bounded by \(Cn_1^{-2}\).
    Hence, for fixed \(\lambda\),
    \[
        \Var(\epsilon_1^\top A_\lambda\epsilon_1\mid X_0)
        \leq
        \frac{C}{n_1^2}
        \tr\{(\Sigma_1^{1/2}A_\lambda\Sigma_1^{1/2})^2\}
        \leq
        \frac{Cp}{n_1^2}
        =
        o(1).
    \]
    The conditional mean is \(n_1^{-1}\tr(\Sigma_1A_\lambda)\), so Chebyshev's inequality gives pointwise convergence.
    The Lipschitz condition, \(\|\epsilon_1\|_2^2=\Op(1)\), \(p/n_1=\cO(1)\), and \(\|\Sigma_1\|_{\oper}=\cO(1)\) give a tight modulus of continuity for the centred quadratic form, and a finite-net argument gives uniform convergence over \(\Lambda\).
    The stated resolvent families satisfy the required bounds because \(\lambda\geq\lambda_->0\).
\end{proof}

\begin{lemma}[Treated-mean bilinear concentration]
    \label{lem:treated-mean-bilinear}
    Under \Cref{asm:rmt}, let \(\mathcal H_0\) be a sigma-field that contains \(X_0\) and is independent of the treated sample.
    If \(A_n\in\RR^{p\times p}\) and \(w_n\in\RR^p\) are \(\mathcal H_0\)-measurable, then, conditionally on \(\mathcal H_0\),
    \[
        w_n^\top A_n\epsilon_1
        =
        \Op\left(
            n_1^{-1/2}
            \|\Sigma_1^{1/2}A_n^\top w_n\|_2
        \right).
    \]
    For the uniform version, suppose \(A_n(\lambda)\) and \(w_n(\lambda)\) are continuously differentiable \(\mathcal H_0\)-measurable families and \(r_n>0\) is deterministic.
    If
    \[
        \sup_{\lambda\in\Lambda}
        \{\|w_n(\lambda)\|_2+\|w_n'(\lambda)\|_2\}
        =
        \Op(r_n),
        \qquad
        \sup_{\lambda\in\Lambda}
        \{\|A_n(\lambda)\|_{\oper}+\|A_n'(\lambda)\|_{\oper}\}
        =
        \Op(1),
    \]
    then
    \[
        \sup_{\lambda\in\Lambda}
        |w_n(\lambda)^\top A_n(\lambda)\epsilon_1|
        =
        \Op(r_n/\sqrt{n_1}).
    \]
\end{lemma}

\begin{proof}[Proof of \Cref{lem:treated-mean-bilinear}]
    Conditional on \(\mathcal H_0\),
    \[
        \Var\{w_n^\top A_n\epsilon_1\mid\mathcal H_0\}
        =
        \frac1{n_1}
        w_n^\top A_n\Sigma_1A_n^\top w_n.
    \]
    Chebyshev's inequality gives the pointwise bound.
    For the uniform result, put \(f_n(\lambda)=w_n(\lambda)^\top A_n(\lambda)\epsilon_1\).
    On the event where the two displayed suprema are bounded by \(Kr_n\) and \(K\), respectively, the conditional variance identity, the product rule, and \(\|\Sigma_1\|_{\oper}\leq C_\Sigma\) give
    \[
        \EE\left[
            \int_\Lambda
            \{|f_n(\lambda)|^2+|f_n'(\lambda)|^2\}
            \,\rd\lambda
            \,\middle|\,
            \mathcal H_0
        \right]
        \leq
        \frac{C_Kr_n^2}{n_1}.
    \]
    The one-dimensional Sobolev inequality and conditional Markov's inequality therefore give the asserted supremum rate on this event.
    Its complement has probability tending to zero as \(K\to\infty\), which completes the proof.
\end{proof}

\newpage
\section{Deterministic Risk Limits for Adaptive Ridge Base Weights}
\label{app:profile-transfer}

For covariate-adaptive base weights, \Cref{thm:risk-estimation} establishes uniform risk estimation and vanishing scaled oracle excess risk under the stated base-weight rate conditions, while penalty consistency additionally requires a deterministic limit for the scaled risk path.
This appendix gives sufficient component-level conditions for this limit and verifies them for the ridge base-weight constructions considered here.
Appendix~\ref{app:general-profile-transfer} derives risk limits from component limits and develops common ridge verification tools, and Appendix~\ref{app:target-split-ridge-profile-verification} uses these tools for the target-split ridge construction used for honest target-aware tuning.
Appendix~\ref{app:ridge-profile-verification} treats same-sample ridge base weights as a separate theoretical case, while the target-aware tuning guarantee in Section~\ref{sec:risk-estimation} remains based on evaluation-sample separation.

\subsection{Risk-Path Components and Ridge Verification Tools}
\label{app:general-profile-transfer}

By \Cref{prop:conditional-risk}, the exact conditional risk identity remains valid for any normalised \(\mathcal G_n\)-measurable base-weight vector \(\gamma\), including covariate-adaptive weights.
The deterministic-equivalent analysis in \Cref{thm:risk-asym} additionally uses design independence, so here we work directly with the realised residual imbalance \(\Delta=\bar x_1-X_0^\top\gamma\) and centred base weights.
Write \(g=C_0\gamma\), \(G_0^c=X_0^c(X_0^c)^\top/n_0\), and
\[
    u_{\lambda}
    =
    \lambda(G_0^c+\lambda I_{n_0})^{-1}g,
    \qquad
    h_{\lambda}
    =
    n_0^{-1}X_0^c M_\lambda^c(\bar x_1-\bar x_0).
\]
Then \(X_0^\top\gamma_\lambda-\nu_1=\epsilon_1-\lambda M_\lambda^c\Delta\) and \(\gamma_\lambda=\one_{n_0}/n_0+u_{\lambda}+h_{\lambda}\), so
\[
    \begin{aligned}
        B_n(\lambda)
        &=
        \frac{r^2}{p}\|\epsilon_1-\lambda M_\lambda^c\Delta\|_2^2,\\
        V_n(\lambda)
        &=
        \sigma_0^2\{n_0^{-1}+\|u_{\lambda}+h_{\lambda}\|_2^2\},\\
        R_n(\lambda)
        &=
        B_n(\lambda)+V_n(\lambda).
    \end{aligned}
\]

For \(\eta\in[0,1]\), define the six scaled risk-path components
\begin{equation}\label{eq:profiles}
    \begin{gathered}
        Q_{\epsilon,n}
        =
        \frac{n_0^\eta}{p}\epsilon_1^\top\epsilon_1,
        \qquad
        Q_{\epsilon\Delta,n}(\lambda)
        =
        \frac{n_0^\eta}{p}\lambda\,\epsilon_1^\top M_\lambda^c\Delta,
        \qquad
        Q_{\Delta,n}(\lambda)
        =
        \frac{n_0^\eta}{p}\lambda^2\,\Delta^\top(M_\lambda^c)^2\Delta,\\
        W_{\gamma,n}(\lambda)
        =
        n_0^\eta\|u_{\lambda}\|_2^2,
        \qquad
        W_{x,n}(\lambda)
        =
        n_0^\eta\|h_{\lambda}\|_2^2,
        \qquad
        W_{\gamma x,n}(\lambda)
        =
        n_0^\eta u_{\lambda}^\top h_{\lambda}.
    \end{gathered}
\end{equation}
The three \(Q\)-components determine \(n_0^\eta B_n\), while the three \(W\)-components, together with the affine-normalisation term, determine \(n_0^\eta V_n\).
Thus uniform convergence of these six components is sufficient for a deterministic scaled risk path; no closed-form limit is imposed at this stage.

The following stochastic Lipschitz bound promotes pointwise component limits on a fixed dense subset to uniform limits over the tuning interval.

\begin{lemma}[Stochastic Lipschitz bounds for the risk-path components]
    \label{lem:profile-lipschitz}
    Let \(\Lambda=[\lambda_-,\lambda_+]\subset(0,\infty)\), suppose \(p/n_0\) is bounded, and assume
    \[
        \frac{\|X_0^c\|_{\oper}^2}{n_0}
        =
        \Op(1),\qquad
        n_0^\eta\|C_0\gamma\|_2^2
        =
        \Op(1),\qquad
        \frac{n_0^\eta}{p}
        \left(
            \|\epsilon_1\|_2^2
            +
            \|\Delta\|_2^2
            +
            \|\bar x_1-\bar x_0\|_2^2
        \right)
        =
        \Op(1).
    \]
    Then \(\partial_{\lambda}Q_{\epsilon,n}=0\) and
    \[
        L_n
        :=
        \sup_{\lambda\in\Lambda}
        \left\{
            |\partial_{\lambda}Q_{\epsilon\Delta,n}(\lambda)|
                +
                |\partial_{\lambda}Q_{\Delta,n}(\lambda)|
                +
                |\partial_{\lambda}W_{\gamma,n}(\lambda)|
                +
                |\partial_{\lambda}W_{x,n}(\lambda)|
                +
                |\partial_{\lambda}W_{\gamma x,n}(\lambda)|
        \right\}
        =
        \Op(1).
    \]
    Consequently, if the six components converge pointwise in probability on a fixed dense subset \(\Lambda_0\subset\Lambda\) to finite deterministic limits, then those limits admit unique Lipschitz extensions to \(\Lambda\), and the component convergence to those extensions is uniform over \(\Lambda\).
\end{lemma}

\begin{proof}[Proof of \Cref{lem:profile-lipschitz}]
    Let \(a_n=n_0^\eta/p\), \(M_{\lambda}=M_\lambda^c\), and \(R_{\lambda}=(G_0^c+\lambda I_{n_0})^{-1}\).
    Uniformly over \(\lambda\in\Lambda\),
    \[
        \|M_{\lambda}\|_{\oper}
        \le
        \lambda_-^{-1},
        \qquad
        \|M_{\lambda}'\|_{\oper}
        =
        \|M_{\lambda}^2\|_{\oper}
        \le
        \lambda_-^{-2},
    \]
    and the same bounds hold for \(R_{\lambda}\).
    Define
    \[
        \begin{gathered}
            E_n
            =
            a_n\|\epsilon_1\|_2^2,
            \qquad
            D_n
            =
            a_n\|\Delta\|_2^2,
            \qquad
            G_n
            =
            n_0^\eta\|g\|_2^2,\\
            H_n
            =
            n_0^\eta
            \frac{\|X_0^c\|_{\oper}^2}{n_0^2}
            \|\bar x_1-\bar x_0\|_2^2
            =
            \frac{\|X_0^c\|_{\oper}^2}{n_0}
            \frac{p}{n_0}
            \frac{n_0^\eta}{p}
            \|\bar x_1-\bar x_0\|_2^2.
        \end{gathered}
    \]
    The assumptions give \(E_n,D_n,G_n,H_n=\Op(1)\).

    Since \(M_{\lambda}'=-M_{\lambda}^2\), differentiation and Cauchy--Schwarz give
    \begin{align*}
        \sup_{\lambda\in\Lambda}
        |\partial_{\lambda}Q_{\epsilon\Delta,n}(\lambda)|
        &=
        \sup_{\lambda\in\Lambda}
        a_n
        \left|
            \epsilon_1^\top
            (M_{\lambda}-\lambda M_{\lambda}^2)
            \Delta
        \right|
        \le
        C_\Lambda(E_nD_n)^{1/2},\\
        \sup_{\lambda\in\Lambda}
        |\partial_{\lambda}Q_{\Delta,n}(\lambda)|
        &=
        \sup_{\lambda\in\Lambda}
        a_n
        \left|
            2\lambda\Delta^\top M_{\lambda}^2\Delta
            -
            2\lambda^2\Delta^\top M_{\lambda}^3\Delta
        \right|
        \le
        C_\Lambda D_n.
    \end{align*}

    For the \(W\)-components, recall \(u_{\lambda}=\lambda R_{\lambda}g\) and \(h_{\lambda}=n_0^{-1}X_0^cM_{\lambda}(\bar x_1-\bar x_0)\).
    The resolvent bounds imply
    \[
        \begin{aligned}
            \sup_{\lambda\in\Lambda}
            \{\|u_{\lambda}\|_2+\|u_{\lambda}'\|_2\}
            &\le
            C_\Lambda\|g\|_2,\\
            \sup_{\lambda\in\Lambda}
            \{\|h_{\lambda}\|_2+\|h_{\lambda}'\|_2\}
            &\le
            C_\Lambda
            \frac{\|X_0^c\|_{\oper}}{n_0}
            \|\bar x_1-\bar x_0\|_2.
        \end{aligned}
    \]
    Therefore,
    \begin{align*}
        \sup_{\lambda\in\Lambda}
        |\partial_{\lambda}W_{\gamma,n}(\lambda)|
        &\le
        C_\Lambda G_n,\\
        \sup_{\lambda\in\Lambda}
        |\partial_{\lambda}W_{x,n}(\lambda)|
        &\le
        C_\Lambda H_n,\\
        \sup_{\lambda\in\Lambda}
        |\partial_{\lambda}W_{\gamma x,n}(\lambda)|
        &\le
        C_\Lambda(G_nH_n)^{1/2}.
    \end{align*}
    Together with \(\partial_{\lambda}Q_{\epsilon,n}=0\), these inequalities prove \(L_n=\Op(1)\).

    Let \(P_n\) denote any one of the six components and suppose \(P_n(\lambda)\pto p(\lambda)\) for each \(\lambda\) in a fixed dense subset \(\Lambda_0\subset\Lambda\), where \(p\) is deterministic on \(\Lambda_0\).
    Since \(L_n=\Op(1)\), choose \(K<\infty\) such that
    \[
        \limsup_{n\to\infty}\PP(L_n>K)<\frac12.
    \]
    On \(\{L_n\le K\}\),
    \[
        |P_n(\lambda)-P_n(\lambda')|
        \le
        K|\lambda-\lambda'|.
    \]
    If \(p\) were not \(K\)-Lipschitz on \(\Lambda_0\), there would exist \(\lambda,\lambda'\in\Lambda_0\) such that
    \[
        |p(\lambda)-p(\lambda')|
        >
        K|\lambda-\lambda'|.
    \]
    Pointwise convergence at these two points would then imply \(\PP(L_n>K)\to1\), a contradiction.
    Hence \(p\) is \(K\)-Lipschitz on \(\Lambda_0\) and has a unique Lipschitz extension \(\bar p\) to \(\Lambda\).

    For uniform convergence, fix \(\varepsilon,\xi>0\).
    Tightness provides \(K_\varepsilon\ge K\) such that \(\limsup_n\PP(L_n>K_\varepsilon)<\varepsilon\).
    Choose a finite \(\delta\)-net from \(\Lambda_0\) with \((K_\varepsilon+K)\delta<\xi/2\).
    Pointwise convergence on this net and the two Lipschitz bounds give
    \[
        \limsup_{n\to\infty}
        \PP\left\{
            \sup_{\lambda\in\Lambda}
            |P_n(\lambda)-\bar p(\lambda)|
            >
            \xi
        \right\}
        \le
        \varepsilon.
    \]
    Since \(\varepsilon\) is arbitrary,
    \[
        \sup_{\lambda\in\Lambda}
        |P_n(\lambda)-\bar p(\lambda)|
        \pto
        0.
    \]
\end{proof}

We collect the required uniform component convergence in the following assumption.

\begin{assumption}[Base-weight component limits]
    \label{asm:admissible-profiles}
    Let \(\Lambda=[\lambda_-,\lambda_+]\subset(0,\infty)\) be compact.
    For some \(\eta\in[0,1]\), suppose the normalised \(\mathcal G_n\)-measurable base weights \(\gamma\) satisfy
    \[
        \begin{aligned}
            |Q_{\epsilon,n}-q_\epsilon|
            &\pto
            0,\\
            \sup_{\lambda\in\Lambda}
            |Q_{\epsilon\Delta,n}(\lambda)-q_{\epsilon\Delta}(\lambda)|
            &\pto
            0,\\
            \sup_{\lambda\in\Lambda}
            |Q_{\Delta,n}(\lambda)-q_{\Delta}(\lambda)|
            &\pto
            0,\\
            \sup_{\lambda\in\Lambda}
            |W_{\gamma,n}(\lambda)-w_\gamma(\lambda)|
            &\pto
            0,\\
            \sup_{\lambda\in\Lambda}
            |W_{x,n}(\lambda)-w_x(\lambda)|
            &\pto
            0,\\
            \sup_{\lambda\in\Lambda}
            |W_{\gamma x,n}(\lambda)-w_{\gamma x}(\lambda)|
            &\pto
            0,
        \end{aligned}
    \]
    where \(q_\epsilon\) is finite and \(q_{\epsilon\Delta}\), \(q_\Delta\), \(w_\gamma\), \(w_x\), and \(w_{\gamma x}\) are finite deterministic functions on \(\Lambda\).
\end{assumption}

\begin{proposition}[Conditional risk limits from component limits]
    \label{prop:risk-profile-transfer}
    Under \Cref{asm:predictive-model,asm:admissible-profiles}, let \(\pi=(q_\epsilon,q_{\epsilon\Delta},q_\Delta,w_\gamma,w_x,w_{\gamma x})\) denote the limiting component collection and define
    \[
        \begin{aligned}
            \sB_{\eta,\pi}(\lambda)
            &:=
            r^2
            \{q_\epsilon-2q_{\epsilon\Delta}(\lambda)+q_\Delta(\lambda)\},\\
            \sV_{\eta,\pi}(\lambda)
            &:=
            \sigma_0^2
            \{\ind{(\eta=1)}+w_\gamma(\lambda)
            +2w_{\gamma x}(\lambda)+w_x(\lambda)\},
        \end{aligned}
    \]
    and
    \[
        \sR_{\eta,\pi}(\lambda)
        :=
        \sB_{\eta,\pi}(\lambda)+\sV_{\eta,\pi}(\lambda).
    \]
    Then
    \begin{align*}
        \sup_{\lambda\in\Lambda}
        |n_0^\eta B_n(\lambda)-\sB_{\eta,\pi}(\lambda)|
        &\pto
        0,\\
        \sup_{\lambda\in\Lambda}
        |n_0^\eta V_n(\lambda)-\sV_{\eta,\pi}(\lambda)|
        &\pto
        0,\\
        \sup_{\lambda\in\Lambda}
        |n_0^\eta R_n(\lambda)-\sR_{\eta,\pi}(\lambda)|
        &\pto
        0.
    \end{align*}
\end{proposition}

\begin{proof}[Proof of \Cref{prop:risk-profile-transfer}]
    The proof of \Cref{prop:conditional-risk} gives
    \[
        X_0^\top\gamma_\lambda-\bar x_1
        =
        -\lambda M_\lambda^c\Delta,
    \]
    and hence
    \[
        X_0^\top\gamma_\lambda-\nu_1
        =
        \epsilon_1-\lambda M_\lambda^c\Delta.
    \]
    Expanding the square in \eqref{eq:conditional-risk} gives
    \[
        n_0^\eta B_n(\lambda)
        =
        r^2
        \{Q_{\epsilon,n}
        -2Q_{\epsilon\Delta,n}(\lambda)
        +Q_{\Delta,n}(\lambda)\}.
    \]
    For the variance, substituting \(\Delta=\bar x_1-X_0^\top\gamma\) into \eqref{eq:gamma-aug} and using
    \[
        n_0^{-1}X_0^cM_\lambda^c(X_0^c)^\top
        =
        G_0^c(G_0^c+\lambda I_{n_0})^{-1}
    \]
    gives
    \[
        \gamma_\lambda
        =
        \frac1{n_0}\one_{n_0}
        +
        \lambda(G_0^c+\lambda I_{n_0})^{-1}g
        +
        \frac1{n_0}X_0^cM_\lambda^c(\bar x_1-\bar x_0)
        =
        \frac1{n_0}\one_{n_0}
        +
        u_{\lambda}
        +
        h_{\lambda}.
    \]
    Therefore,
    \[
        n_0^\eta V_n(\lambda)
        =
        \sigma_0^2
        \{n_0^{\eta-1}
        +W_{\gamma,n}(\lambda)
        +2W_{\gamma x,n}(\lambda)
        +W_{x,n}(\lambda)\}.
    \]
    Both identities are exact.
    Taking suprema over \(\Lambda\) and applying \Cref{asm:admissible-profiles} gives the asserted uniform conclusions for \(\sB_{\eta,\pi}\) and \(\sV_{\eta,\pi}\).
    The conclusion for \(\sR_{\eta,\pi}\) follows from \(R_n(\lambda)=B_n(\lambda)+V_n(\lambda)\) and the triangle inequality.
\end{proof}

Thus deterministic risk-path limits follow from the six component limits.
For ridge constructions, these components can be reduced to a fixed collection of basic quadratic forms.

\begin{lemma}[Basic quadratic forms for ridge verification]
    \label{lem:ridge-generator-forms}
    Fix \(\eta\in[0,1]\), \(\alpha>0\), \(\lambda>0\), and a nonempty finite index set \(\mathcal K\).
    For each \(k\in\mathcal K\), let \(\epsilon_{1,k}\) be an empirical target-mean fluctuation based on a sample of size \(n_{1,k}\).
    Suppose \Cref{asm:rmt} holds, the source and target standardised coordinates are uniformly sub-Gaussian, and the target samples indexed by \(\mathcal K\) are mutually independent and independent of the source design.
    Suppose also that
    \[
        \phi_{0,n}
        \to
        \phi_0
        \in
        (0,\infty),
        \qquad
        \frac{p}{n_{1,k}}
        \to
        \phi_{1,k}
        \in
        (0,\infty),
        \quad
        k\in\mathcal K,
        \qquad
        \rho_{\eta,n}^2
        \to
        \rho_\eta^2
        \in
        [0,\infty),
    \]
    and that \(H_{0,p}\), \(G_{\nu,p}\), and \(H_{1\mid0,p}\) converge weakly as in \Cref{asm:spectral-limits}.
    Define
    \[
        \mathcal T_{\alpha,\lambda}
        =
        \left\{
            M_\alpha^c,\,
            (M_\alpha^c)^2,\,
            M_\lambda^c,\,
            (M_\lambda^c)^2
        \right\},
        \qquad
        \mathcal V_n
        =
        \{\nu_\Delta,e_0\}
        \cup
        \{\epsilon_{1,k}:k\in\mathcal K\}.
    \]
    For every \(T\in\mathcal T_{\alpha,\lambda}\), each scaled pure form
    \[
        \frac{n_0^\eta}{p}v^\top Tv,
        \qquad
        v\in\mathcal V_n,
    \]
    admits a finite deterministic limit, and each scaled mixed form
    \[
        \frac{n_0^\eta}{p}v^\top Tw,
        \qquad
        v,w\in\mathcal V_n,
        \quad
        v\ne w,
    \]
    is \(\op(1)\).
    All of these conclusions hold jointly.
\end{lemma}

\begin{proof}[Proof of \Cref{lem:ridge-generator-forms}]
    Write \(n=n_0\).
    Under proportional growth, \(p\asymp n\) and \(n_{1,k}\asymp n\) for every \(k\in\mathcal K\).
    Put
    \[
        c_n
        =
        \frac{n^\eta}{p}
        \asymp
        n^{\eta-1}.
    \]
    The scaling and moment conditions give
    \[
        \|\nu_\Delta\|_2
        =
        O\left(\sqrt p\,n^{-\eta/2}\right),
        \qquad
        \|e_0\|_2
        =
        \Op(1),
        \qquad
        \|\epsilon_{1,k}\|_2
        =
        \Op(1),
        \quad
        k\in\mathcal K.
    \]
    For every \(T\in\mathcal T_{\alpha,\lambda}\), \(\|T\|_{\oper}\le C_{\alpha,\lambda}<\infty\).

    \Cref{prop:global-source-resolvent-sufficient} verifies the source-resolvent condition under the stated source assumptions.
    Let \(\zeta\in(0,1/4)\) denote the exponent in \Cref{asm:technical-source-resolvent}.
    The pure-form orders are
    \begin{center}
        \small
        \setlength{\tabcolsep}{4pt}
        \begin{tabular}{@{}llll@{}}
            \toprule
            Pure-form term & Unscaled order & \(c_n\)-scaled order & Controlling result \\
            \midrule
            \(\nu_\Delta,\nu_\Delta\)
            & \(\cO(pn^{-\eta})\) & \(\cO(1)\)
            & \Cref{lem:uniform-ridge-de,lem:quadratic-form-extraction} \\
            \(e_0,e_0\)
            & \(\Op(1)\) & \(\Op(n^{\eta-1})\)
            & \Cref{lem:source-mean-resolvent} \\
            \(\epsilon_{1,k},\epsilon_{1,k}\)
            & \(\Op(1)\) & \(\Op(n^{\eta-1})\)
            & \Cref{lem:treated-mean-qf,lem:uniform-ridge-de} \\
            \bottomrule
        \end{tabular}
    \end{center}
    The cited deterministic-equivalent and quadratic-form results give finite deterministic limits for all pure forms.
    In particular, the source- and target-mean terms vanish for \(\eta<1\) and may contribute at \(\eta=1\).

    For mixed forms involving one target empirical mean and either \(\nu_\Delta\) or \(e_0\), condition on the source design and apply \Cref{lem:treated-mean-bilinear}.
    For distinct \(k,\ell\in\mathcal K\), condition additionally on the \(k\)th target sample.
    Then \(T\epsilon_{1,k}\) is fixed under the conditional law of the \(\ell\)th sample, while \(\epsilon_{1,\ell}\) is independent and centred, so the same lemma applies.
    Together with the source-mean cross bound in \Cref{lem:source-mean-resolvent}, the unscaled orders follow from
    \begin{align*}
        \nu_\Delta^\top T\epsilon_{1,k}
        &=\Op\left(n_{1,k}^{-1/2}\|T\nu_\Delta\|_2\right)
        =\Op(n^{-\eta/2}),\\
        \nu_\Delta^\top Te_0
        &=\Op(n^{-1/2+\zeta})\|\nu_\Delta\|_2
        =\Op(n^{-\eta/2+\zeta}),\\
        e_0^\top T\epsilon_{1,k}
        &=\Op\left(n_{1,k}^{-1/2}\|Te_0\|_2\right)
        =\Op(n^{-1/2}),\\
        \epsilon_{1,k}^\top T\epsilon_{1,\ell}
        &=\Op\left(n_{1,\ell}^{-1/2}
        \|T\epsilon_{1,k}\|_2\right)
        =\Op(n^{-1/2}).
    \end{align*}
    The resulting mixed-form orders are summarised below.
    \begin{center}
        \small
        \setlength{\tabcolsep}{4pt}
        \begin{tabular}{@{}llll@{}}
            \toprule
            Mixed pair & Unscaled order & \(c_n\)-scaled order & Controlling input \\
            \midrule
            \(\nu_\Delta,\epsilon_{1,k}\)
            & \(\Op(n^{-\eta/2})\)
            & \(\Op(n^{-1+\eta/2})\)
            & \Cref{lem:treated-mean-bilinear} \\
            \(\nu_\Delta,e_0\)
            & \(\Op(n^{-\eta/2+\zeta})\)
            & \(\Op(n^{-1+\eta/2+\zeta})\)
            & \Cref{lem:source-mean-resolvent} \\
            \(e_0,\epsilon_{1,k}\)
            & \(\Op(n^{-1/2})\)
            & \(\Op(n^{-3/2+\eta})\)
            & \Cref{lem:treated-mean-bilinear} \\
            \(\epsilon_{1,k},\epsilon_{1,\ell}\), \(k\ne\ell\)
            & \(\Op(n^{-1/2})\)
            & \(\Op(n^{-3/2+\eta})\)
            & \Cref{lem:treated-mean-bilinear} \\
            \bottomrule
        \end{tabular}
    \end{center}
    Each scaled order vanishes for \(\eta\in[0,1]\); in the second row, \(-1+\eta/2+\zeta\le-1/2+\zeta<0\) because \(\zeta<1/4\).
    Thus every distinct mixed pair in \(\mathcal V_n\) is negligible after scaling.
    Because \(\mathcal K\), \(\mathcal T_{\alpha,\lambda}\), and \(\mathcal V_n\) have fixed cardinality, all pure- and mixed-form conclusions hold jointly.
\end{proof}

The following lemma extends these basic quadratic-form limits to the resolvent combinations appearing in the risk-path components.

\begin{lemma}[Ridge-resolvent reduction for quadratic forms]
    \label{lem:ridge-rational-filter-closure}
    Fix \(\alpha>0\) and \(\lambda>0\) with \(\lambda\ne\alpha\), and let
    \[
        \mathcal R_{\alpha,\lambda}
        =
        \operatorname{span}
        \{
            M_\alpha^c,
            (M_\alpha^c)^2,
            M_\lambda^c,
            (M_\lambda^c)^2
        \}.
    \]
    Let \(c_n>0\) be a deterministic scaling and let
    \[
        \mathcal V_n
        =
        \{v_{1,n},\ldots,v_{J,n}\}
    \]
    be a collection of vector sequences of fixed cardinality \(J\).
    Suppose that, for each of the four resolvent terms \(T_n\) spanning \(\mathcal R_{\alpha,\lambda}\), every pure form \(c_nv^\top T_nv\), with \(v\in\mathcal V_n\), admits a deterministic equivalent, and every mixed form \(c_nv^\top T_nw\), with distinct \(v,w\in\mathcal V_n\), is \(\op(1)\).
    Then the same conclusions hold with \(T_n\) replaced by any finite linear combination of these terms.
    The coefficients may depend on the fixed pair \((\alpha,\lambda)\), but they are deterministic, independent of \(n\), and finite for \(\lambda\ne\alpha\).

    In particular, the conclusion applies to any proper rational function of \(S_0^c\) whose only poles are \(-\alpha\) and \(-\lambda\), each of order at most two.
    If the conclusions for the four terms are uniform and the partial-fraction coefficients are uniformly bounded on a parameter set, the conclusion is uniform there as well.
\end{lemma}

\begin{proof}[Proof of \Cref{lem:ridge-rational-filter-closure}]
    Partial fractions express such a rational function as a finite linear combination of the four resolvent terms.
    The corresponding pure deterministic equivalent is the same linear combination of the equivalents for the four terms, and each mixed form is a finite linear combination of \(\op(1)\) terms.
    Uniformity follows after taking the supremum when the coefficients are uniformly bounded.
\end{proof}

\begin{remark}[Coincident penalties]
    \label{rem:ridge-filter-coincident-penalties}
    The partial-fraction representation applies for fixed \(\lambda\ne\alpha\), and its coefficients need not remain bounded as \(\lambda\to\alpha\).
    When the conditions of \Cref{lem:profile-lipschitz} hold, pointwise convergence on a dense subset of \(\Lambda\setminus\{\alpha\}\) extends uniformly over \(\Lambda\).
    The limit at \(\lambda=\alpha\) is then the unique continuous extension of the off-diagonal limit, so no additional higher-order resolvent equivalents are required.
\end{remark}

The remaining subsections use these componentwise limits and ridge verification tools for the ridge base-weight constructions considered in this appendix.

\subsection{Target-Split Ridge Base Weights}
\label{app:target-split-ridge-profile-verification}

We now study ridge base weights constructed from the pilot target sample in the honest target-splitting procedure of Appendix~\ref{app:honest-target-split}.
The fixed penalty \(\alpha>0\) controls the base ridge weights, while \(\lambda\in\Lambda\) indexes the subsequent augmentation path evaluated using the independent evaluation sample.
We verify the component limits needed for a deterministic scaled risk path and the resulting tuning conclusions.

Put
\[
    d_{\mathrm P}
    =
    \bar x_{1,\mathrm P}-\bar x_0
    =
    \nu_\Delta+\epsilon_{1,\mathrm P}-e_0,
\]
and define the pilot ridge base weights
\begin{equation}
    \label{eq:target-split-ridge-base-weight}
    \gamma_{\alpha,\mathrm P}
    =
    \frac1{n_0}\one_{n_0}
    +
    \frac1{n_0}X_0^c M_\alpha^cd_{\mathrm P}.
\end{equation}
Since \(S_0^c M_\alpha^c=I_p-\alpha M_\alpha^c\), the evaluation imbalance has the exact decomposition
\begin{equation}
    \label{eq:honest-evaluation-imbalance}
    \begin{split}
        \Delta_{\mathrm E}
        &=
        \bar x_{1,\mathrm E}
        -
        X_0^\top\gamma_{\alpha,\mathrm P}\\
        &=
        \epsilon_{1,\mathrm E}
        -
        \epsilon_{1,\mathrm P}
        +
        \alpha M_\alpha^c
        (\nu_\Delta+\epsilon_{1,\mathrm P}-e_0).
    \end{split}
\end{equation}
Thus the pilot and evaluation fluctuations enter separately, with \(\epsilon_{1,\mathrm P}\) and \(\epsilon_{1,\mathrm E}\) independent of one another and of the source design.
Every component expansion is therefore generated by
\[
    \nu_\Delta,\qquad
    e_0,\qquad
    \epsilon_{1,\mathrm P},\qquad
    \epsilon_{1,\mathrm E}.
\]

For this construction, write
\[
    b_{\mathrm P}
    =
    C_0\gamma_{\alpha,\mathrm P}
    =
    \frac1{n_0}X_0^c M_\alpha^cd_{\mathrm P},
    \qquad
    a_{\alpha,\lambda}
    =
    \frac1{n_0}X_0^c M_\lambda^c\Delta_{\mathrm E}.
\]
Then
\[
    \gamma_{\mathrm E,\lambda}
    =
    \frac1{n_0}\one_{n_0}
    +
    b_{\mathrm P}
    +
    a_{\alpha,\lambda}.
\]
Define the six target-split components by
\begin{equation}
    \label{eq:target-split-profiles}
    \begin{gathered}
        Q_{\mathrm E,n}
        =
        \frac{n_0^\eta}{p}
        \|\epsilon_{1,\mathrm E}\|_2^2,
        \qquad
        Q_{\mathrm E\Delta,n}(\lambda)
        =
        \frac{n_0^\eta}{p}\lambda
        \epsilon_{1,\mathrm E}^\top M_\lambda^c\Delta_{\mathrm E},\\
        Q_{\Delta,\mathrm E,n}(\lambda)
        =
        \frac{n_0^\eta}{p}\lambda^2
        \Delta_{\mathrm E}^\top(M_\lambda^c)^2\Delta_{\mathrm E},
        \qquad
        W_{\mathrm P,n}
        =
        n_0^\eta\|b_{\mathrm P}\|_2^2,\\
        W_{\mathrm A,n}(\lambda)
        =
        n_0^\eta\|a_{\alpha,\lambda}\|_2^2,
        \qquad
        W_{\mathrm{PA},n}(\lambda)
        =
        n_0^\eta b_{\mathrm P}^\top a_{\alpha,\lambda}.
    \end{gathered}
\end{equation}
Under \Cref{asm:predictive-model}, these components give the exact target-split identities
\begin{align}
    n_0^\eta B_n(\lambda)
    &=
    r^2
    \{
        Q_{\mathrm E,n}
        -
        2Q_{\mathrm E\Delta,n}(\lambda)
        +
        Q_{\Delta,\mathrm E,n}(\lambda)
    \},
    \label{eq:target-split-bias-transfer}\\
    n_0^\eta V_n(\lambda)
    &=
    \sigma_0^2
    \{
        n_0^{\eta-1}
        +
        W_{\mathrm P,n}
        +
        2W_{\mathrm{PA},n}(\lambda)
        +
        W_{\mathrm A,n}(\lambda)
    \}.
    \label{eq:target-split-variance-transfer}
\end{align}
By \eqref{eq:target-split-bias-transfer}--\eqref{eq:target-split-variance-transfer}, uniform deterministic limits for these six components are sufficient to obtain a deterministic scaled risk path.

\begin{proposition}[Sub-Gaussian verification for target-split ridge weights]
    \label{prop:target-split-ridge-profile-verification}
    Fix \(\eta\in[0,1]\), \(\alpha>0\), and a compact interval \(\Lambda\subset(0,\infty)\).
    Suppose the pilot/evaluation split is constructed as in Appendix~\ref{app:honest-target-split}.
    Suppose \Cref{asm:rmt} holds and the source and treated standardised coordinates are uniformly sub-Gaussian.
    Suppose also that
    \[
        \phi_{0,n}
        \to
        \phi_0
        \in
        (0,\infty),
        \qquad
        \frac{p}{n_{1,\mathrm P}}
        \to
        \phi_{1,\mathrm P}
        \in
        (0,\infty),
        \qquad
        \frac{p}{n_{1,\mathrm E}}
        \to
        \phi_{1,\mathrm E}
        \in
        (0,\infty),
    \]
    that \(\rho_{\eta,n}^2\to\rho_\eta^2\in[0,\infty)\), and that \(H_{0,p}\), \(G_{\nu,p}\), and \(H_{1\mid0,p}\) converge weakly as in \Cref{asm:spectral-limits}.
    Then all six components in \eqref{eq:target-split-profiles} converge uniformly on \(\Lambda\) to finite deterministic limits, with the constant components interpreted uniformly.
    Denote these limits by
    \[
        q_{\mathrm E},\quad
        q_{\mathrm E\Delta,\alpha}(\lambda),\quad
        q_{\Delta,\mathrm E,\alpha}(\lambda),\quad
        w_{\mathrm P,\alpha},\quad
        w_{\mathrm A,\alpha}(\lambda),\quad
        w_{\mathrm{PA},\alpha}(\lambda),
    \]
    respectively.
\end{proposition}

The proposition is design-side: \(\gamma_{\alpha,\mathrm P}\) may depend on the source design and pilot target sample, while the response model is used only for the risk conclusions below.

\begin{proof}[Proof of \Cref{prop:target-split-ridge-profile-verification}]
    Substitute \eqref{eq:honest-evaluation-imbalance} into \eqref{eq:target-split-profiles}.
    The corresponding quadratic forms are summarised by
    \begin{center}
        \small
        \setlength{\tabcolsep}{4pt}
        \begin{tabular}{@{}lllll@{}}
            \toprule
            Component & Left vector & Right vector & Function \(f_{\alpha,\lambda}(s)\) & Prefactor \\
            \midrule
            \(Q_{\mathrm E,n}\)
            & \(\epsilon_{1,\mathrm E}\) & \(\epsilon_{1,\mathrm E}\)
            & \(1\) & \(n_0^\eta/p\) \\
            \(Q_{\mathrm E\Delta,n}\)
            & \(\epsilon_{1,\mathrm E}\) & \(\Delta_{\mathrm E}\)
            & \(\lambda/(s+\lambda)\) & \(n_0^\eta/p\) \\
            \(Q_{\Delta,\mathrm E,n}\)
            & \(\Delta_{\mathrm E}\) & \(\Delta_{\mathrm E}\)
            & \(\lambda^2/(s+\lambda)^2\) & \(n_0^\eta/p\) \\
            \(W_{\mathrm P,n}\)
            & \(d_{\mathrm P}\) & \(d_{\mathrm P}\)
            & \(s/(s+\alpha)^2\) & \(n_0^\eta/n_0\) \\
            \(W_{\mathrm A,n}\)
            & \(\Delta_{\mathrm E}\) & \(\Delta_{\mathrm E}\)
            & \(s/(s+\lambda)^2\) & \(n_0^\eta/n_0\) \\
            \(W_{\mathrm{PA},n}\)
            & \(d_{\mathrm P}\) & \(\Delta_{\mathrm E}\)
            & \(s/\{(s+\alpha)(s+\lambda)\}\) & \(n_0^\eta/n_0\) \\
            \bottomrule
        \end{tabular}
    \end{center}
    Thus, after substituting \eqref{eq:honest-evaluation-imbalance}, each component is a finite linear combination of quadratic and bilinear forms
    \[
        \frac{n_0^\eta}{p}
        v^\top f_{\alpha,\lambda}(S_0^c)w,
        \qquad
        v,w
        \in
        \{
            \nu_\Delta,
            e_0,
            \epsilon_{1,\mathrm P},
            \epsilon_{1,\mathrm E}
        \},
    \]
    possibly multiplied by the bounded factor \(p/n_0\), where \(f_{\alpha,\lambda}\) is a proper rational function with poles only at \(-\alpha\) and \(-\lambda\), each of order at most two.
    For fixed \(\lambda\ne\alpha\), \Cref{lem:ridge-rational-filter-closure} reduces each such form to the basic quadratic forms in \Cref{lem:ridge-generator-forms}.
    The case \(\lambda=\alpha\) is handled below by the Lipschitz extension.

    The constant evaluation component is handled separately.
    Since \(n_0^\eta/p=\cO(1)\), applying \Cref{lem:treated-mean-qf} to the evaluation sample with \(A_\lambda=I_p\) gives
    \[
        Q_{\mathrm E,n}
        -
        \frac{n_0^\eta}{p\,n_{1,\mathrm E}}
        \tr(\Sigma_1)
        \pto
        0.
    \]
    The aspect-ratio and spectral assumptions therefore imply
    \[
        Q_{\mathrm E,n}
        \pto
        q_{\mathrm E},
        \qquad
        q_{\mathrm E}
        :=
        \ind{(\eta=1)}
        \frac{\phi_{1,\mathrm E}}{\phi_0}
        \int 1\,\rd H_{1\mid0}.
    \]

    For the remaining five components, take \(\mathcal K=\{\mathrm P,\mathrm E\}\).
    Then \Cref{lem:ridge-generator-forms}, together with \Cref{lem:ridge-rational-filter-closure}, gives deterministic limits for all pure terms and makes every mixed term negligible.
    The scaling for the \(W\)-components satisfies
    \[
        \frac{n_0^\eta}{n_0}
        =
        \phi_{0,n}
        \frac{n_0^\eta}{p},
    \]
    so the same pure- and mixed-form conclusions apply.
    The pilot and evaluation aspect ratios determine their trace coefficients separately.
    This establishes pointwise convergence for each fixed \(\lambda\ne\alpha\).
    The rate table in \Cref{lem:ridge-generator-forms} is applied pointwise, so no uniform bound on the partial-fraction coefficients near \(\lambda=\alpha\) is needed.

    The primitive vector rates and \eqref{eq:honest-evaluation-imbalance} give
    \[
        \frac{n_0^\eta}{p}
        \|d_{\mathrm P}\|_2^2
        =
        \Op(1),
        \qquad
        \frac{n_0^\eta}{p}
        \|\Delta_{\mathrm E}\|_2^2
        =
        \Op(1).
    \]
    Moreover, the source operator-norm bound and the resolvent bounds imply
    \begin{align*}
        n_0^\eta\|b_{\mathrm P}\|_2^2
        &\leq
        \frac{\|X_0^c\|_{\oper}^2}{n_0}
        \alpha^{-2}
        \frac{n_0^\eta}{n_0}
        \|d_{\mathrm P}\|_2^2
        =
        \Op(1),\\
        n_0^\eta
        \sup_{\lambda\in\Lambda}
        \|a_{\alpha,\lambda}\|_2^2
        &\leq
        \frac{\|X_0^c\|_{\oper}^2}{n_0}
        \lambda_-^{-2}
        \frac{n_0^\eta}{n_0}
        \|\Delta_{\mathrm E}\|_2^2
        =
        \Op(1).
    \end{align*}
    Using the original matrix products,
    \[
        \partial_\lambda M_\lambda^c
        =
        -(M_\lambda^c)^2,
        \qquad
        \sup_{\lambda\in\Lambda}
        \|M_\lambda^c\|_{\oper}
        \leq
        \lambda_-^{-1},
        \qquad
        \|M_\alpha^c\|_{\oper}
        \leq
        \alpha^{-1}.
    \]
    Consequently, each matrix function and its first derivative are uniformly bounded by a constant depending only on \(\lambda_-\), \(\lambda_+\), and \(\alpha\).
    Cauchy--Schwarz and the preceding norm bounds therefore give a common \(\Op(1)\) Lipschitz envelope for the four nonconstant components, as in \Cref{lem:profile-lipschitz}.
    Pointwise convergence on \((\Lambda\cap\mathbb Q)\setminus\{\alpha\}\) and the finite-net argument in \Cref{lem:profile-lipschitz} therefore give uniform convergence over \(\Lambda\).
    The component limits at \(\lambda=\alpha\) are their unique Lipschitz extensions.
\end{proof}

The verification of the component limits first yields the base-weight rate condition, then the deterministic scaled risk path and its consequences for penalty selection.

\begin{corollary}[Penalty selection with target-split ridge base weights]
    \label{cor:target-split-ridge-post-tuning}
    Under the conditions of \Cref{prop:target-split-ridge-profile-verification}:
    \begin{enumerate}[(a)]
        \item\label{cor:target-split-ridge-post-tuning-i}
        \Cref{asm:risk-rate-admissibility} holds for \(\gamma=\gamma_{\alpha,\mathrm P}\);

        \item\label{cor:target-split-ridge-post-tuning-ii}
        if, in addition, \Cref{asm:predictive-model} holds, \(\widehat r^2\pto r^2\), and \(\widehat\sigma_0^2\pto\sigma_0^2\), define
        \begin{align*}
            \sB_{\eta,\alpha}(\lambda)
            &=
            r^2
            \{
                q_{\mathrm E}
                -
                2q_{\mathrm E\Delta,\alpha}(\lambda)
                +
                q_{\Delta,\mathrm E,\alpha}(\lambda)
            \},\\
            \sV_{\eta,\alpha}(\lambda)
            &=
            \sigma_0^2
            \{
                \ind{(\eta=1)}
                +
                w_{\mathrm P,\alpha}
                +
                2w_{\mathrm{PA},\alpha}(\lambda)
                +
                w_{\mathrm A,\alpha}(\lambda)
            \},\\
            \sR_{\eta,\alpha}(\lambda)
            &=
            \sB_{\eta,\alpha}(\lambda)
            +
            \sV_{\eta,\alpha}(\lambda).
        \end{align*}
        Then
        \begin{align*}
            \sup_{\lambda\in\Lambda}
            \left|
                n_0^\eta B_n(\lambda)
                -
                \sB_{\eta,\alpha}(\lambda)
            \right|
            &\pto
            0,\\
            \sup_{\lambda\in\Lambda}
            \left|
                n_0^\eta V_n(\lambda)
                -
                \sV_{\eta,\alpha}(\lambda)
            \right|
            &\pto
            0,\\
            \sup_{\lambda\in\Lambda}
            \left|
                n_0^\eta R_n(\lambda)
                -
                \sR_{\eta,\alpha}(\lambda)
            \right|
            &\pto
            0.
        \end{align*}
        Moreover, \Cref{thm:risk-estimation} applies to the honest construction in \Cref{alg:target-split-ridge-augmentation} and its target-aware tuning specialisation in Appendix~\ref{app:honest-target-split} when the balancing map is instantiated by
        \[
            \mathsf{Balance}(X_0,\bar x_{1,\mathrm P})
            =
            \gamma_{\alpha,\mathrm P}
        \]
        as defined in \eqref{eq:target-split-ridge-base-weight};

        \item\label{cor:target-split-ridge-post-tuning-iii}
        if, further, \(\sR_{\eta,\alpha}\) has a unique minimiser \(\lambda_{\eta,\alpha}^*\) with \(\sR_{\eta,\alpha}(\lambda_{\eta,\alpha}^*)>0\), then \Cref{lem:ood-post-tuning-stability} applies with \(\sR_{\eta,\star}=\sR_{\eta,\alpha}\).
        Under \Cref{asm:gaussian-predictive}, \Cref{thm:predictive-calibration} gives the selected conditional predictive limit with the selected plug-in scale.
        Under the non-Gaussian conditions of \Cref{thm:predictive-normality-nongaussian} at \(\lambda_{\eta,\alpha}^*\), that theorem gives the normal approximation at the reference penalty, and \Cref{prop:post-tuning-transfer} together with \Cref{lem:plugin-studentization} gives the same selected-penalty conclusion.
    \end{enumerate}
\end{corollary}

\begin{proof}[Proof of \Cref{cor:target-split-ridge-post-tuning}]
    The bounds in the proof of \Cref{prop:target-split-ridge-profile-verification} give
    \[
        n_0^\eta
        \|C_0\gamma_{\alpha,\mathrm P}\|_2^2
        =
        n_0^\eta
        \|b_{\mathrm P}\|_2^2
        =
        \Op(1).
    \]
    Moreover,
    \[
        \nu_1-X_0^\top\gamma_{\alpha,\mathrm P}
        =
        -\epsilon_{1,\mathrm P}
        +
        \alpha M_\alpha^cd_{\mathrm P}
        =
        \Delta_{\mathrm E}
        -
        \epsilon_{1,\mathrm E}.
    \]
    The bounds for \(\Delta_{\mathrm E}\) and the constant evaluation component therefore imply
    \[
        n_0^\eta
        \frac{
            \|\nu_1-X_0^\top\gamma_{\alpha,\mathrm P}\|_2^2
        }{p}
        =
        \Op(1).
    \]
    Thus \Cref{asm:risk-rate-admissibility} holds.
    Together with \Cref{lem:target-split-separation}, \Cref{lem:base-to-augmented-rate} supplies the corresponding uniform evaluation-augmented path rate.
    This proves \Cref{cor:target-split-ridge-post-tuning}~\eqref{cor:target-split-ridge-post-tuning-i}.

    Under the response-model assumptions in \Cref{cor:target-split-ridge-post-tuning}~\eqref{cor:target-split-ridge-post-tuning-ii}, the exact identities \eqref{eq:target-split-bias-transfer}--\eqref{eq:target-split-variance-transfer} and the component convergence in \Cref{prop:target-split-ridge-profile-verification} give the three uniform risk limits.
    Together with \Cref{cor:target-split-ridge-post-tuning}~\eqref{cor:target-split-ridge-post-tuning-i}, variance-component consistency verifies every condition of \Cref{thm:risk-estimation}; the uniform deterministic risk limit supplies the additional condition in \Cref{thm:risk-estimation}~\eqref{thm:risk-estimation-iii}.
    The explicit instantiation of \(\mathsf{Balance}\) records the ridge construction covered by the verification of the component limits.

    Finally, identify \(\lambda_{\eta,\star}^*=\lambda_{\eta,\alpha}^*\).
    Under \Cref{asm:gaussian-predictive}, \Cref{thm:predictive-calibration} gives the stated conclusion.
    In the non-Gaussian case, \Cref{thm:predictive-normality-nongaussian} gives the normal approximation at the reference penalty, \Cref{lem:ood-post-tuning-stability} and \Cref{prop:post-tuning-transfer} transfer it to the selected penalty, and \Cref{lem:plugin-studentization} gives the selected plug-in scale.
\end{proof}

The fixed-\(\alpha\) verification extends to an \(\mathcal H_n\)-measurable selection over a fixed finite grid because the fixed-\(\alpha\) conclusions hold jointly over that grid.

\begin{corollary}[Pilot-selected ridge base weights on a finite grid]
    \label{cor:target-split-ridge-selected-alpha}
    Let \(\mathcal A\subset(0,\infty)\) be a fixed finite set, and suppose the conditions of \Cref{prop:target-split-ridge-profile-verification} hold for every \(\alpha\in\mathcal A\).
    Let \(\widehat\alpha\) be \(\mathcal H_n\)-measurable with \(\widehat\alpha\in\mathcal A\) almost surely, and set \(\gamma=\gamma_{\widehat\alpha,\mathrm P}\).
    Then \Cref{asm:risk-rate-admissibility} holds.
    Under the additional conditions of \Cref{thm:risk-estimation}, conclusions \eqref{thm:risk-estimation-i}--\eqref{thm:risk-estimation-ii} of that theorem hold for the resulting target-aware tuning procedure.

    If, in addition, \(\widehat\alpha\pto\alpha_\star\in\mathcal A\), then the deterministic risk-limit conclusion of \Cref{cor:target-split-ridge-post-tuning} holds with \(\alpha=\alpha_\star\), and \Cref{thm:risk-estimation}~\eqref{thm:risk-estimation-iii} applies whenever its limiting-risk conditions hold.
\end{corollary}

\begin{proof}[Proof of \Cref{cor:target-split-ridge-selected-alpha}]
    Because \(\mathcal A\) is finite, the fixed-\(\alpha\) conclusions of \Cref{prop:target-split-ridge-profile-verification} and \Cref{cor:target-split-ridge-post-tuning} hold jointly over \(\alpha\in\mathcal A\).
    Substitution of the \(\mathcal H_n\)-measurable selector \(\widehat\alpha\) therefore preserves \Cref{asm:evaluation-sample} and \Cref{asm:risk-rate-admissibility}, which gives the first conclusion.
    If \(\widehat\alpha\pto\alpha_\star\), finiteness of \(\mathcal A\) implies \(\PP(\widehat\alpha=\alpha_\star)\to1\), so the fixed-\(\alpha\) deterministic risk limit carries over directly.
\end{proof}

\subsection{Same-Sample Ridge Base Weights}
\label{app:ridge-profile-verification}

We finally consider normalised centred ridge base weights that use the same treated mean in both base-weight construction and augmentation.
Because the same target sample enters both stages, this construction lies outside the evaluation-sample separation setting of \Cref{asm:evaluation-sample}.
Here we use the tools of Appendix~\ref{app:general-profile-transfer} to establish its deterministic scaled conditional-risk path.

Fix \(\alpha>0\) and define the ridge base weights
\begin{equation}
    \label{eq:ridge-base-weight}
    \gamma_\alpha
    =
    \frac1{n_0}\one_{n_0}
    +
    \frac1{n_0}X_0^c M_\alpha^c(\bar x_1-\bar x_0),
    \qquad
    M_\alpha^c
    =
    \left(S_0^c+\alpha I_p\right)^{-1},
    \qquad
    S_0^c
    =
    \frac1{n_0}(X_0^c)^\top X_0^c.
\end{equation}
Their imbalance is
\begin{equation}
    \label{eq:ridge-base-imbalance}
    \Delta_\alpha
    =
    \bar x_1-X_0^\top\gamma_\alpha
    =
    \alpha M_\alpha^c(\bar x_1-\bar x_0).
\end{equation}
The two centred components in the decomposition of \(\gamma_{\alpha,\lambda}\) are
\begin{align}
    u_{\alpha,\lambda}
    &=
    \lambda(G_0^c+\lambda I_{n_0})^{-1}C_0\gamma_\alpha
    =
    \frac{\lambda}{n_0}X_0^c M_\lambda^cM_\alpha^c(\bar x_1-\bar x_0),
    \label{eq:ridge-u-profile}\\
    h_\lambda
    &=
    \frac1{n_0}X_0^c M_\lambda^c(\bar x_1-\bar x_0).
    \label{eq:ridge-h-profile}
\end{align}
Consequently,
\begin{equation}
    \label{eq:ridge-aug-filter}
    \gamma_{\alpha,\lambda}
    =
    \frac1{n_0}\one_{n_0}
    +
    \frac1{n_0}X_0^c
    \{M_\alpha^c+\alpha M_\lambda^cM_\alpha^c\}
    (\bar x_1-\bar x_0).
\end{equation}
Writing \(d=\bar x_1-\bar x_0=\nu_\Delta+\epsilon_1-e_0\), the nonconstant components therefore reduce to pure and mixed resolvent forms generated by \(\{\nu_\Delta,e_0,\epsilon_1\}\).

\begin{proposition}[Sub-Gaussian verification for same-sample ridge weights]
    \label{prop:ridge-profile-verification}
    Fix \(\eta\in[0,1]\), \(\alpha>0\), and a compact interval \(\Lambda\subset(0,\infty)\).
    Suppose \Cref{asm:rmt} holds and the source and treated standardised coordinates are uniformly sub-Gaussian.
    Suppose also that
    \[
        \phi_{t,n}
        \to
        \phi_t
        \in
        (0,\infty),
        \qquad
        t=0,1,
        \qquad
        \rho_{\eta,n}^2
        \to
        \rho_\eta^2
        \in
        [0,\infty),
    \]
    and that \(H_{0,p}\), \(G_{\nu,p}\), and \(H_{1\mid0,p}\) converge weakly as in \Cref{asm:spectral-limits}.
    Then the ridge base weights in \eqref{eq:ridge-base-weight} satisfy \Cref{asm:admissible-profiles} at scaling \(\eta\).
\end{proposition}

Thus the same-sample ridge construction admits a deterministic scaled conditional-risk path through \Cref{prop:risk-profile-transfer}.
The target-aware risk estimator and tuning guarantee in Section~\ref{sec:risk-estimation} use the evaluation-split construction of Appendix~\ref{app:honest-target-split}.

\begin{proof}[Proof of \Cref{prop:ridge-profile-verification}]
    Put
    \[
        d
        =
        \bar x_1-\bar x_0
        =
        \nu_\Delta+\epsilon_1-e_0.
    \]
    By \eqref{eq:ridge-base-imbalance}--\eqref{eq:ridge-h-profile}, the left/right pairs are \((\epsilon_1,\epsilon_1)\) for \(Q_{\epsilon,n}\), \((\epsilon_1,d)\) for \(Q_{\epsilon\Delta,n}\), and \((d,d)\) for the remaining four components.
    In the order
    \[
        Q_{\epsilon,n},\quad
        Q_{\epsilon\Delta,n},\quad
        Q_{\Delta,n},\quad
        W_{\gamma,n},\quad
        W_{x,n},\quad
        W_{\gamma x,n},
    \]
    their exact functions are
    \[
        1,\quad
        \frac{\alpha\lambda}{(s+\alpha)(s+\lambda)},\quad
        \frac{\alpha^2\lambda^2}{(s+\alpha)^2(s+\lambda)^2},\quad
        \frac{\lambda^2s}{(s+\alpha)^2(s+\lambda)^2},\quad
        \frac{s}{(s+\lambda)^2},\quad
        \frac{\lambda s}{(s+\alpha)(s+\lambda)^2},
    \]
    with prefactor \(n_0^\eta/p\) for the three \(Q\)-components and \(n_0^\eta/n_0\) for the three \(W\)-components.
    Expanding \(d\) therefore expresses each nonconstant component as a finite linear combination of pure and mixed forms generated by \(\{\nu_\Delta,e_0,\epsilon_1\}\).
    The constant component \(Q_{\epsilon,n}\) is handled separately.

    Since \(n_0^\eta/p=\cO(1)\), \Cref{lem:treated-mean-qf} with \(A_\lambda=I_p\) gives
    \[
        Q_{\epsilon,n}
        -
        \frac{n_0^\eta}{p\,n_1}\tr(\Sigma_1)
        \pto
        0.
    \]
    Consequently,
    \[
        Q_{\epsilon,n}
        \pto
        q_\epsilon,
        \qquad
        q_\epsilon
        :=
        \ind{(\eta=1)}
        \frac{\phi_1}{\phi_0}
        \int 1\,\rd H_{1\mid0}.
    \]

    For each of the five nonconstant components and fixed \(\lambda\ne\alpha\), each function is a proper rational function with poles only at \(-\alpha\) and \(-\lambda\), each of order at most two.
    Hence \Cref{lem:ridge-rational-filter-closure} applies, and no deterministic equivalent beyond second resolvent powers is needed.
    The closure is applied at each fixed \(\lambda\ne\alpha\), so no uniform bound on the partial-fraction coefficients near \(\lambda=\alpha\) is required.

    Apply \Cref{lem:ridge-generator-forms} with \(\mathcal K=\{1\}\) and \(\epsilon_{1,1}=\epsilon_1\).
    Together with \Cref{lem:ridge-rational-filter-closure}, this gives finite deterministic limits for all pure terms in the five nonconstant components and makes every corresponding mixed term negligible.
    Since
    \[
        \frac{n_0^\eta}{n_0}
        =
        \phi_{0,n}\frac{n_0^\eta}{p},
    \]
    the conclusions of \Cref{lem:ridge-generator-forms} apply unchanged to the three \(W\)-components.
    Together with the preceding limit for \(Q_{\epsilon,n}\), the spectral convergences in the proposition therefore give finite deterministic limits for all six components at every point of
    \[
        \Lambda_0
        =
        (\Lambda\cap\mathbb Q)\setminus\{\alpha\}.
    \]

    It remains to establish uniform convergence.
    Uniform sub-Gaussianity and proportional growth give
    \[
        \frac{\|X_0^c\|_{\oper}^2}{n_0}
        =
        \Op(1),
        \qquad
        \frac{n_0^\eta}{p}\|\epsilon_1\|_2^2
        =
        \Op(1),
        \qquad
        \frac{n_0^\eta}{p}
        \|\bar x_1-\bar x_0\|_2^2
        =
        \Op(1).
    \]
    The resolvent bound \(\|M_\alpha^c\|_{\oper}\leq\alpha^{-1}\) then gives
    \[
        \frac{n_0^\eta}{p}\|\Delta_\alpha\|_2^2
        \leq
        \frac{n_0^\eta}{p}\|\bar x_1-\bar x_0\|_2^2
        =
        \Op(1).
    \]
    Also,
    \[
        n_0^\eta\|C_0\gamma_\alpha\|_2^2
        =
        n_0^{\eta-1}
        (\bar x_1-\bar x_0)^\top
        M_\alpha^cS_0^cM_\alpha^c
        (\bar x_1-\bar x_0)
        =
        \Op(1),
    \]
    because \(\sup_{s\geq0}s/(s+\alpha)^2<\infty\), \(n_0^\eta\|\nu_\Delta\|_2^2/p=\cO(1)\), and \(\|\epsilon_1\|_2^2+\|e_0\|_2^2=\Op(1)\).
    These bounds verify the conditions of \Cref{lem:profile-lipschitz}.
    The six pointwise limits on \(\Lambda_0\) therefore admit unique Lipschitz extensions through \(\lambda=\alpha\), and the component convergence is uniform over \(\Lambda\).
    This gives \Cref{asm:admissible-profiles}.
\end{proof}

\newpage
\section{Supplementary Numerical Results}\label{app:sup-exp}
\setcounter{table}{0}
\setcounter{figure}{0}

\subsection{Simulation implementation and Monte Carlo conventions}\label{subsec:mc-conventions}

For a Monte Carlo statistic \(Z_m\), \(m=1,\ldots,M\), we report the mean
\[
    \bar Z=M^{-1}\sum_{m=1}^M Z_m
\]
and Monte Carlo standard error \(\operatorname{MCSE}(\bar Z)=s_Z/\sqrt{M}\), where \(s_Z\) is the replication standard deviation.
Unless stated otherwise, displayed Monte Carlo error bars are \(\pm2\operatorname{MCSE}\).

In the risk-path experiments of Section~\ref{subsec:num-benchmark}, \(B_n(\lambda)\), \(V_n(\lambda)\), and \(R_n(\lambda)\) are deterministic conditional on the realised design and design-independent base weights.
The shaded bands show the analytic \(\pm1\) standard error \(R_n(\lambda)\sqrt{2/200}\) of an empirical mean of 200 squared prediction errors under the Gaussian predictive law, whereas error bars for \(\widehat R_n(\lambda)\) reflect its outcome-dependent sampling variation conditional on the realised design, rather than variation across independently generated designs.

For Section~\ref{subsec:num-adaptive}, let \((s_j,w_j)_{j=1}^p\) be the eigensystem of \(\Sigma_0\) and set
\[
    u_{\mathrm{diffuse}}=p^{-1/2}\sum_{j=1}^p w_j,
    \qquad
    \nu_\delta=
    \frac{\delta\,u_{\mathrm{diffuse}}}
    {\{p^{-1}\sum_{j=1}^p s_j^{-1}\}^{1/2}},
\]
so that \(\nu_\delta^\top\Sigma_0^{-1}\nu_\delta=\delta^2\), and the population effective-sample-size fraction \(\exp(-\delta^2)\) is unchanged from the isotropic design.
The adaptive experiments use 120 replications, 41 logarithmically spaced augmentation penalties over \(\Lambda=[0.01,100]\), and common random seeds across the AR(1) and isotropic sensitivity experiments.

The \(\ell_2\) base penalty is selected by ten-fold Riesz-loss cross-validation over
\[
    \mathcal A=\{j/999:j=1,\ldots,999\}\subset(0,1],
\]
using only the source design and pilot-target information, so the finite-grid selection is covered by \Cref{cor:target-split-ridge-selected-alpha}.
For the source-only comparator, define
\[
    H_\lambda
    =
    X_0^c\{(X_0^c)^\top X_0^c+n_0\lambda I_p\}^{-1}(X_0^c)^\top
\]
and
\[
    \operatorname{GCV}_0(\lambda)
    =
    \frac{n_0^{-1}\|(I_{n_0}-H_\lambda)y_0^c\|_2^2}
    {[1-\{1+\operatorname{tr}(H_\lambda)\}/n_0]^2},
    \qquad
    \widehat\lambda_{\mathrm{GCV}}
    \in
    \argmin_{\lambda\in\Lambda}\operatorname{GCV}_0(\lambda),
\]
where the additional degree of freedom accounts for the unpenalised intercept.

\subsection{Isotropic risk-path and aspect-ratio controls}\label{subsec:isotropic-controls}

The two experiments in Section~\ref{subsec:num-benchmark} are repeated under the isotropic design \(\Sigma_0=\Sigma_1=I_p\).
For \Cref{fig:risk-isotropic}, both the covariance geometry and mean-shift orientation change relative to \Cref{fig:risk}, whereas \Cref{fig:aspect-ratio-isotropic} uses a diffuse shift in both designs and therefore isolates covariance sensitivity.

\begin{figure}[!t]
    \centering
    \includegraphics[width=0.9\linewidth]{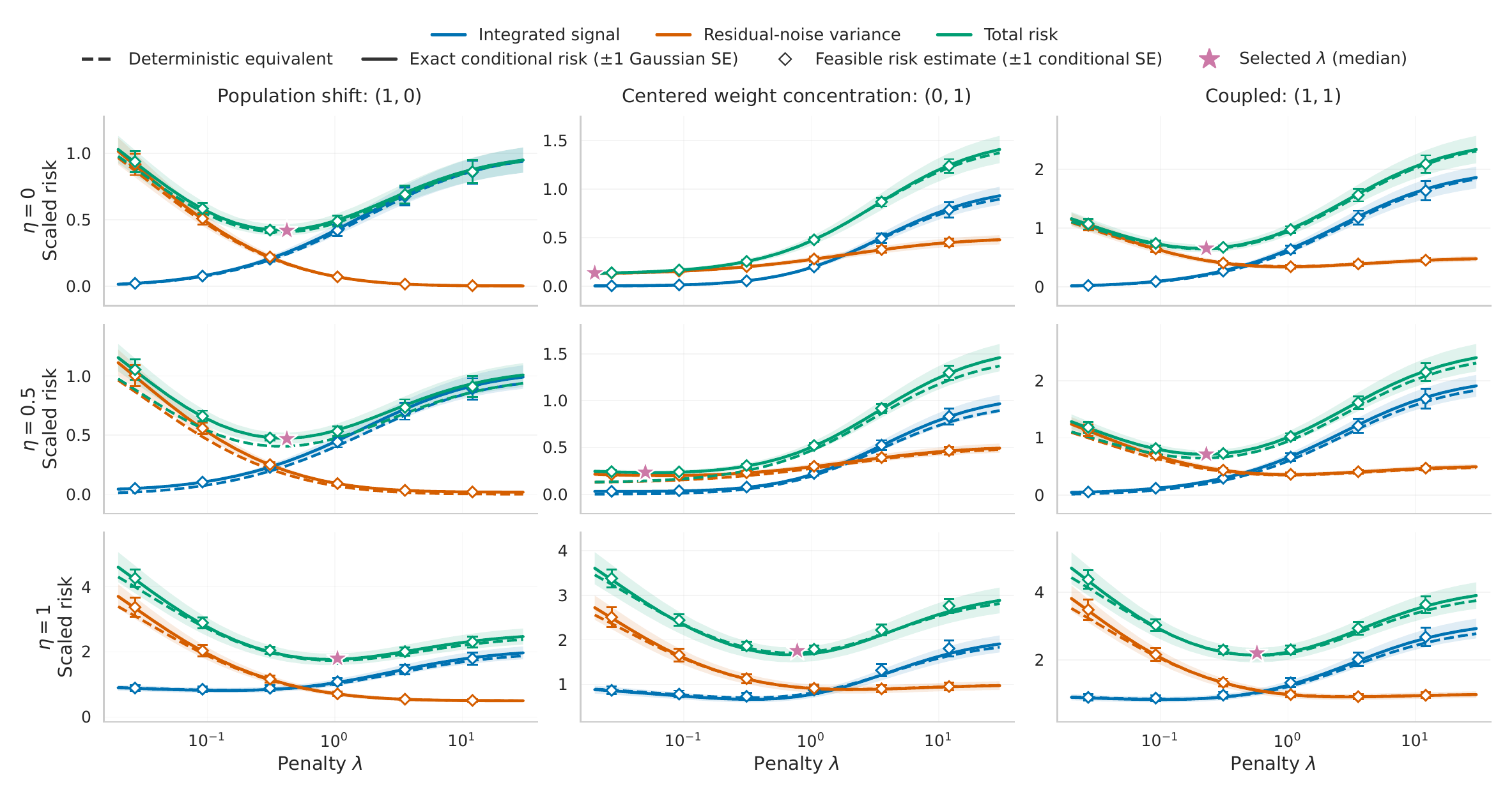}
    \caption{Isotropic counterpart of \Cref{fig:risk}.
    We set \(p=640\), \(n_0=n_1=853\), and \(\phi_0=\phi_1\approx0.75\).
    Rows index \(\eta\in\{0,0.5,1\}\), and columns set \((\rho_\eta^2,\varrho_\eta^2)\) to \((1,0)\), \((0,1)\), and \((1,1)\).
    All risks are scaled by \(n_0^\eta\); curve, marker, and uncertainty conventions match \Cref{fig:risk}.}
    \label{fig:risk-isotropic}
\end{figure}

In \Cref{fig:risk-isotropic}, the deterministic equivalents remain close to the exact conditional paths and the feasible criterion identifies the same low-risk regions across the three extrapolation regimes.
Thus isotropy changes risk levels but preserves the qualitative risk tradeoff in \Cref{fig:risk}.

For \Cref{fig:aspect-ratio-isotropic}, we fix \(p=640\), \(\lambda=1\), and \(\varrho_\eta^2=1\), vary either \(\phi_0\) or \(\phi_1\) over \(\{0.5,1,2,5,10\}\), and hold the other at \(0.75\).
At \(\eta=0\), total risk is nearly invariant to \(\phi_1\), whereas at \(\eta=1\) target empirical-mean variation becomes visible; \(\phi_0\) affects risk in both regimes.
The source-target aspect-ratio asymmetry in \Cref{fig:aspect ratio risk} therefore persists under isotropy.

\begin{figure}[!t]
    \centering
    \includegraphics[width=0.9\linewidth]{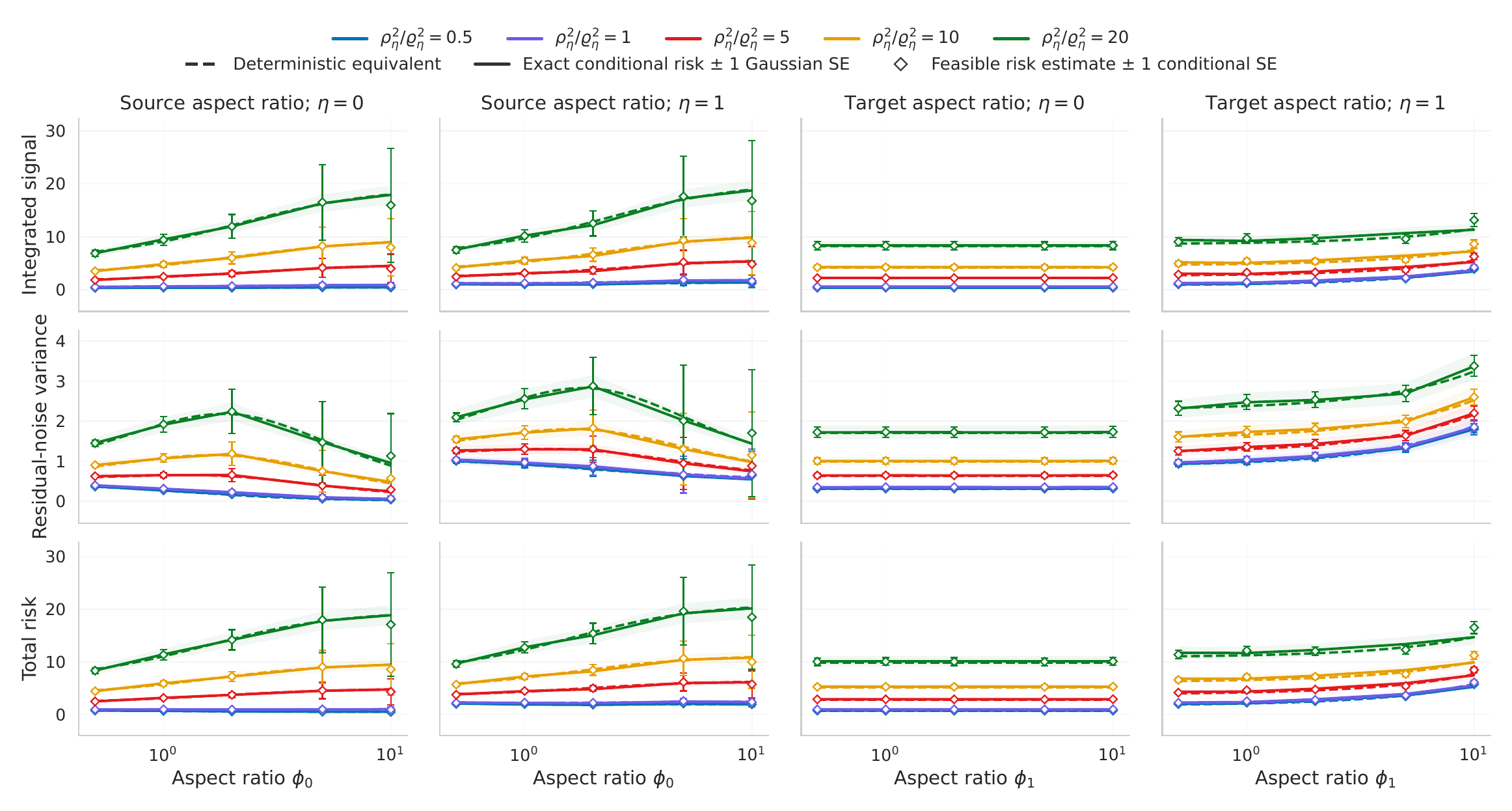}
    \caption{Isotropic counterpart of \Cref{fig:aspect ratio risk}.
    We fix \(p=640\), \(\lambda=1\), and \(\varrho_\eta^2=1\), vary either \(\phi_0\) or \(\phi_1\) while holding the other at \(0.75\), and index curves by \(\rho_\eta^2/\varrho_\eta^2\in\{0.5,1,5,10,20\}\).
    Curve and uncertainty conventions match \Cref{fig:aspect ratio risk}.}
    \label{fig:aspect-ratio-isotropic}
\end{figure}

\clearpage
\subsection{Prediction-interval coverage}\label{subsec:exp-predictive}

We examine the coverage of the selected-penalty prediction intervals in \Cref{thm:predictive-calibration} under the Gaussian outcome model of Section~\ref{subsec:num-benchmark}, with isotropic covariates, \(\Sigma_0=\Sigma_1=I_p\), in place of the AR(1) covariance.
Set \(p=320\), \(\eta\in\{0,0.5,0.75,1\}\), and
\[
    n_0=n_{1,\mathrm{total}}=\operatorname{round}(p/0.75).
\]
Let \(k=\operatorname{round}(0.1p)\) and use the sparse target shift
\[
    \nu_{1,n,j}
    =
    \left(\frac{p}{2kn_0^\eta}\right)^{1/2}
    \mathbf 1\{j\leq k\},
    \qquad
    \frac{n_0^\eta}{p}\|\nu_{1,n}\|_2^2=\frac12.
\]
The target sample is split approximately equally into pilot and evaluation subsets as in \Cref{app:honest-target-split}, giving \(p/n_{1,\mathrm E}\approx1.5\).
We construct the base weights from the pilot subset using the fixed ridge penalty \(1\),
\[
    \widehat\gamma
    =
    \frac{1}{n_0}\one_{n_0}
    +
    \frac{1}{n_0}X_0^cM_1^c
    (\bar x_{1,\mathrm P}-\bar x_0),
\]
and use the evaluation subset for estimating and tuning the conditional prediction risk as in \Cref{alg:ood-tuning}, writing \(\widehat\gamma_{\mathrm E,\lambda}\) for the resulting ridge-augmented weights.

For each \(\eta\), we generate 36 independent covariate designs and, conditional on each design, 100 independent draws of \((\beta,\varepsilon_0)\), re-estimating the variance components and selecting \(\widehat\lambda\) for every outcome draw.
We report the selected studentised prediction error
\[
    T_n
    =
    \frac{
    \widehat\gamma_{\mathrm E,\widehat\lambda}^{\top}y_0-\mu_{0,\beta}
    }{
    \widehat R_n(\widehat\lambda;\widehat\gamma)^{1/2}
    }.
\]
As shown in \Cref{fig:hist}, the distributions are broadly close to the standard normal reference across the four extrapolation regimes.

\begin{figure}[!ht]
    \centering
    \includegraphics[width=0.6\linewidth]{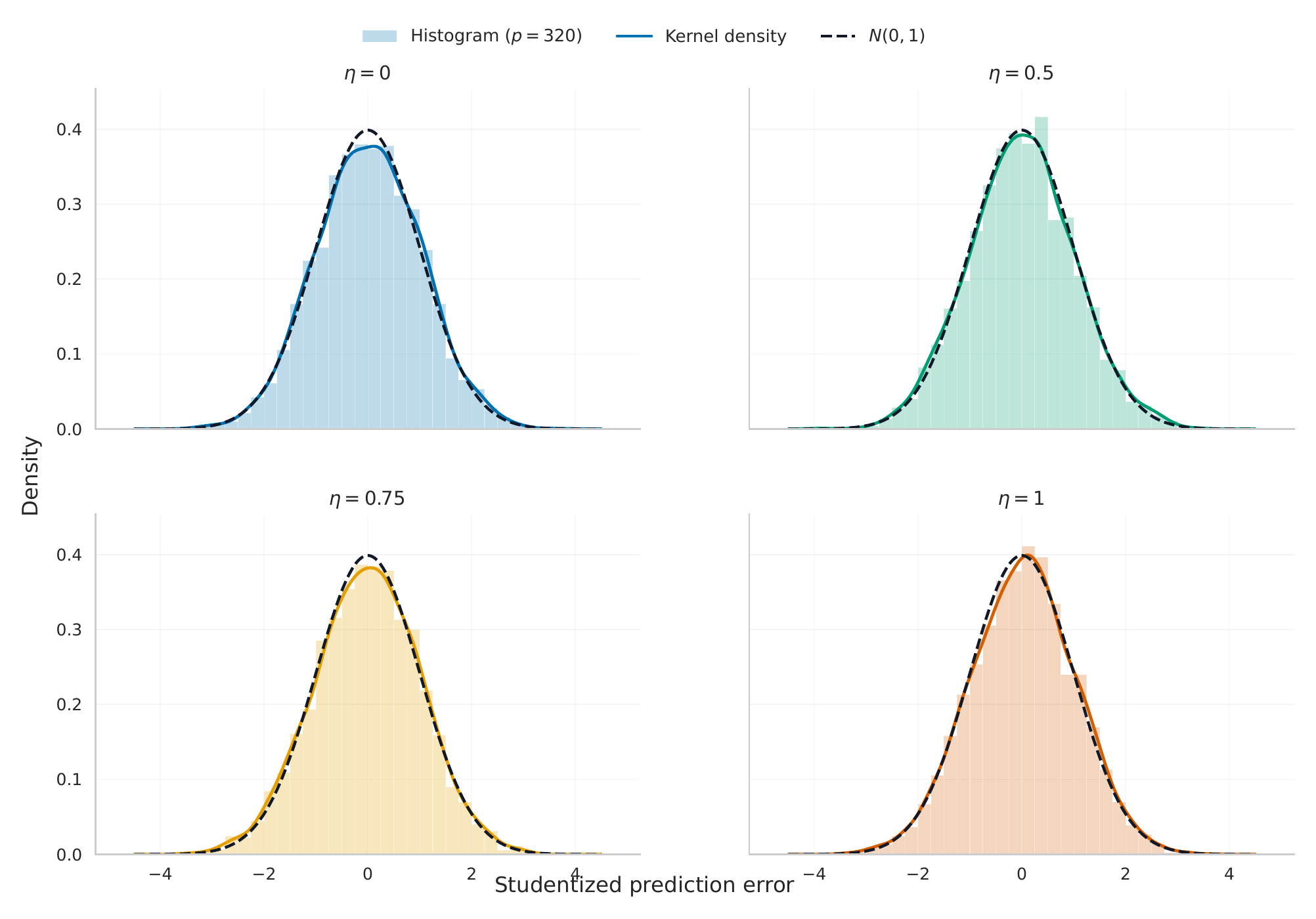}
    \caption{Coverage of selected-penalty prediction intervals at \(p=320\).
    Each panel shows the normalised histogram and kernel-density estimate of \(T_n\) for the indicated extrapolation exponent \(\eta\); the dashed curve is the standard normal density.
    The target sample is split approximately equally into pilot and evaluation subsets, giving \(p/n_{1,\mathrm E}\approx1.5\).
    Each display pools 100 outcome draws within each of 36 independently generated covariate designs; the pooled distributions are descriptive and do not treat the 3600 outcome draws as independent Monte Carlo units.}
    \label{fig:hist}
\end{figure}

\clearpage
\subsection{Adaptive-balancing diagnostics and sensitivity}\label{subsec:sup-num}

\Cref{fig:comparison-phi-1-25} complements \Cref{fig:complete-estimator-rmse} by showing the selected-risk decomposition in the overparameterised source design.
Across weighting families, augmentation lowers the integrated squared bias and generally lowers total risk, while weaker overlap is accompanied by greater weight concentration and evaluation-sample imbalance.

\begin{figure}[!ht]
    \centering
    \includegraphics[width=\linewidth]{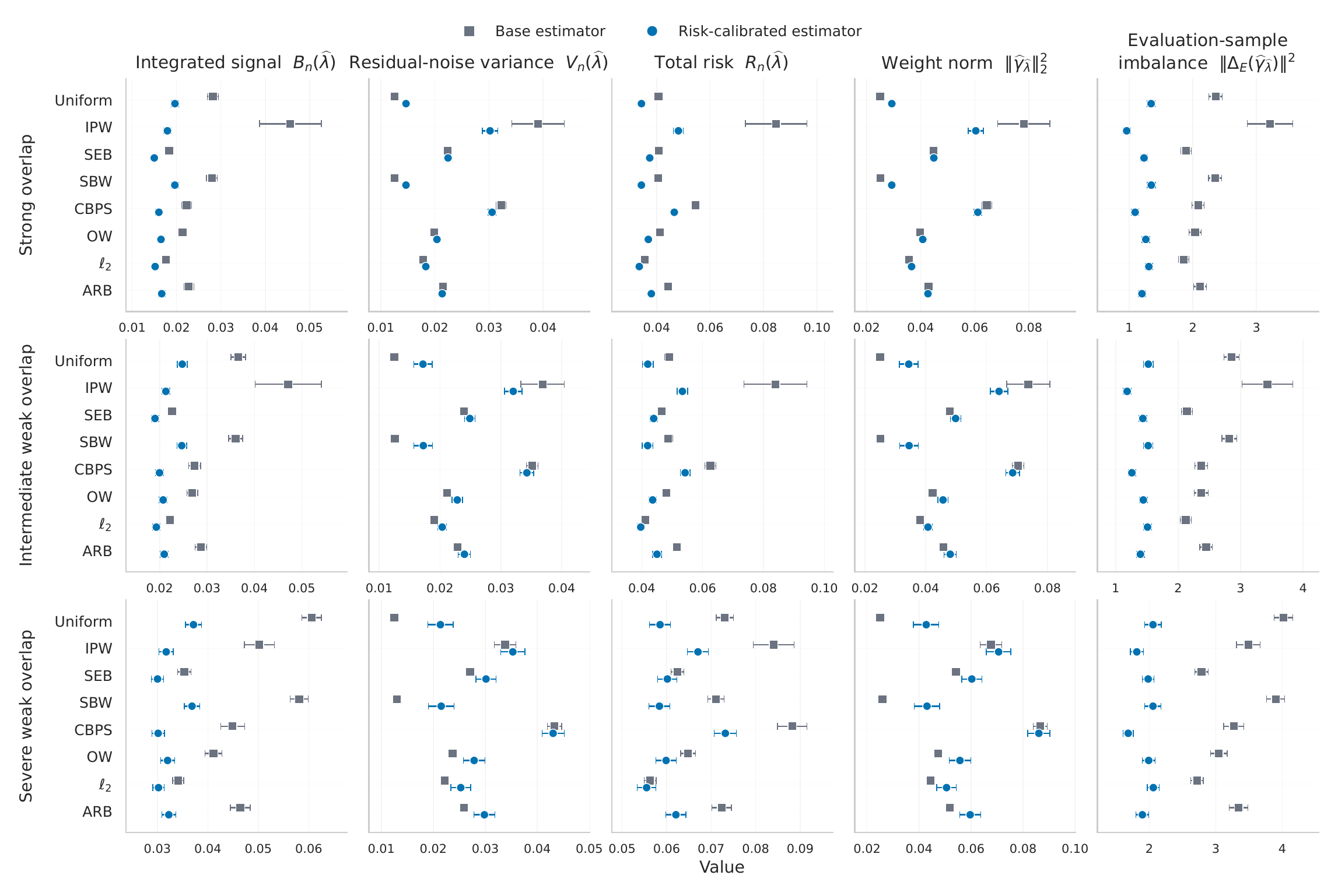}
    \caption{Selected risk and weight diagnostics for \(p=50\), \(n_0=40\), and \(n_1=100\).
    Rows correspond to \(\delta\in\{0.5,1,\sqrt{3}\}\), and columns report \(B_n(\widehat\lambda)\), \(V_n(\widehat\lambda)\), \(R_n(\widehat\lambda)\), \(\|\widehat\gamma_{\widehat\lambda}\|_2^2\), and \(\|\Delta_E(\widehat\gamma_{\widehat\lambda})\|_2^2\).
    Gray squares denote base estimators and blue circles their target-aware ridge augmentations.
    Points are means over 120 replications; bars are \(\pm2\) Monte Carlo standard errors.}
    \label{fig:comparison-phi-1-25}
\end{figure}

\clearpage
\Cref{fig:iso-exp-tuning} repeats the tuning comparison in \Cref{fig:exp_tuning} under \(\Sigma_0=\Sigma_1=I_p\) and \(\nu_\delta=\delta p^{-1/2}\one_p\).
The spectral target-aware rule has the smallest mean excess exact risk in 43 of 48 configurations.
In the same isotropic sweep, augmentation lowers the integrated squared bias in all 48 configurations and total risk in 46 of 48, so isotropy changes risk levels but not the principal conclusions about augmentation or target-aware tuning.

\begin{figure}[htbp]
    \centering
    \begin{minipage}[t]{\linewidth}
        \centering
        \includegraphics[width=\linewidth]{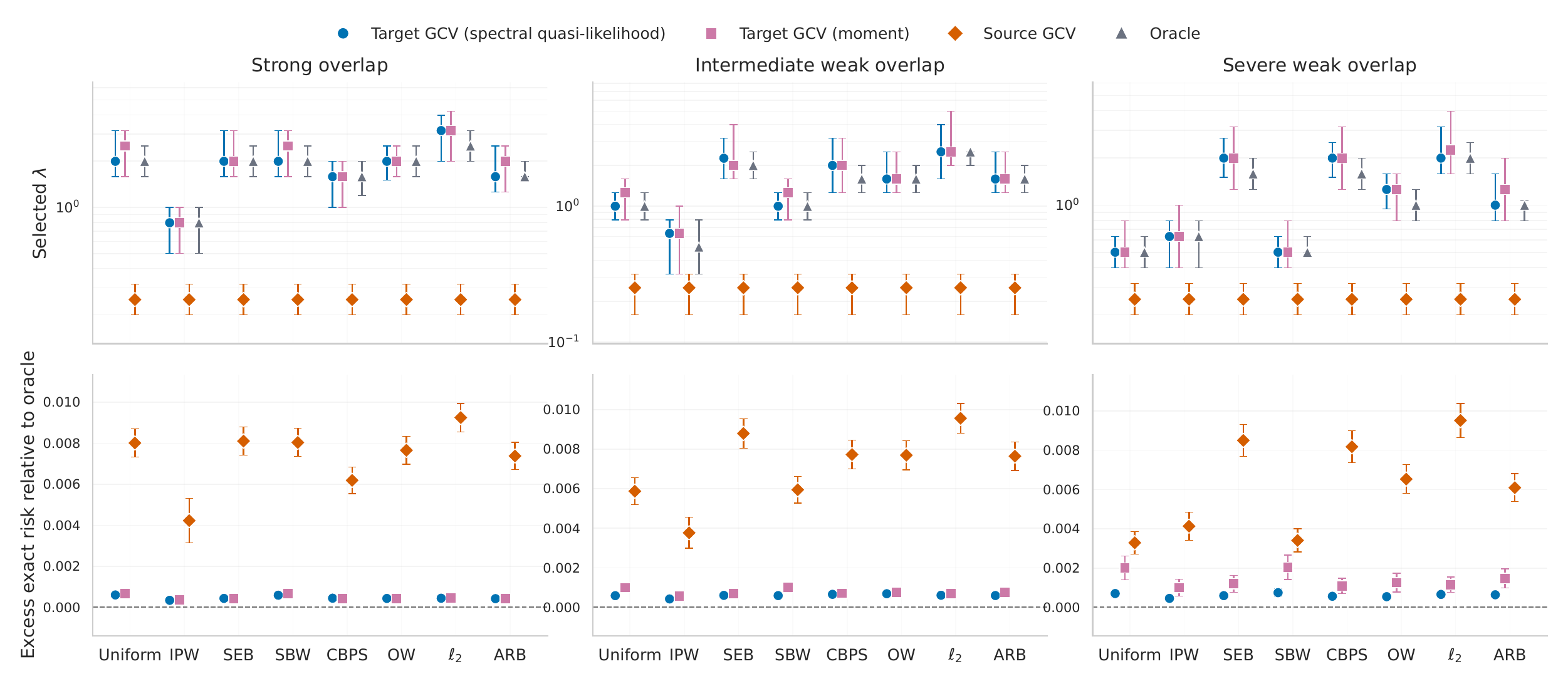}
        \par\smallskip\textbf{(a)} Underparameterised source design: \(p/n_0=0.5\).
    \end{minipage}
    \par\medskip
    \begin{minipage}[t]{\linewidth}
        \centering
        \includegraphics[width=\linewidth]{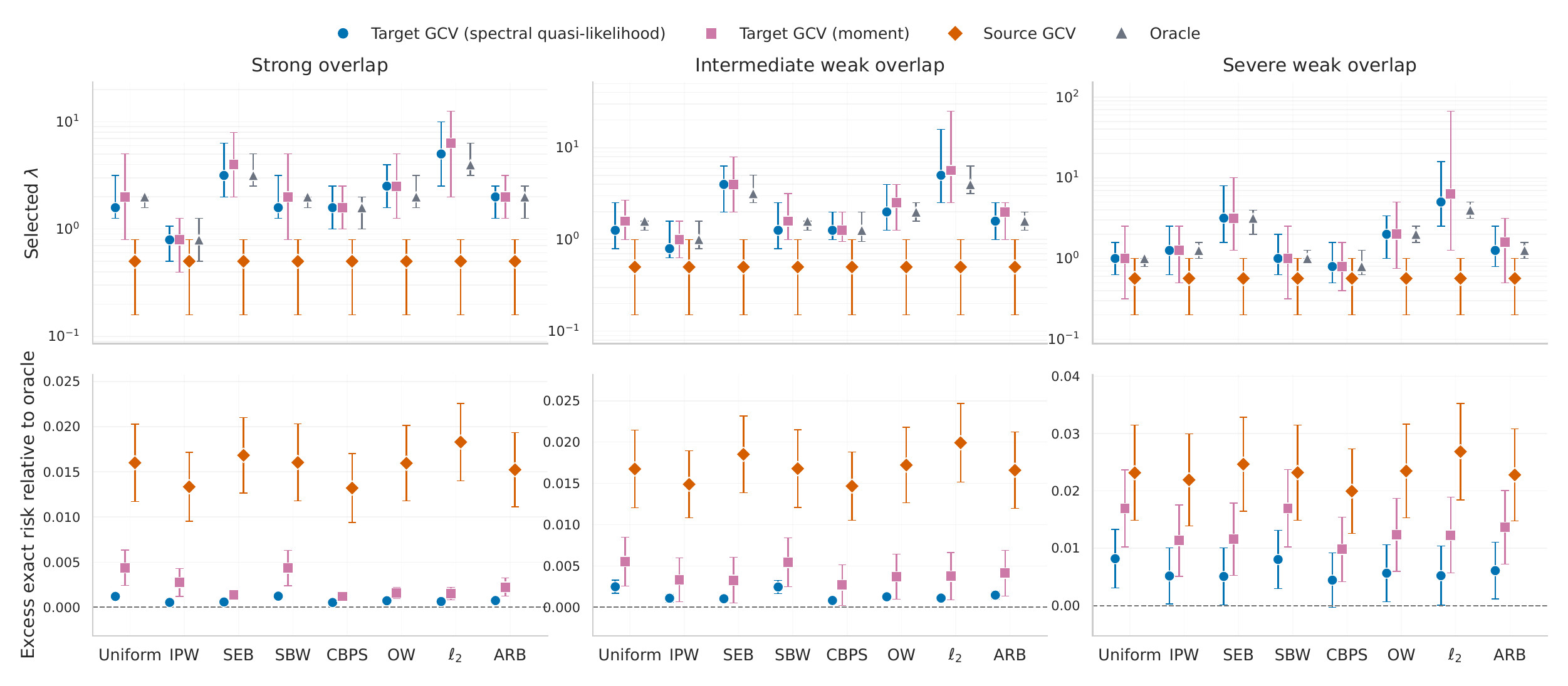}
        \par\smallskip\textbf{(b)} Overparameterised source design: \(p/n_0=1.25\).
    \end{minipage}
    \caption{Isotropic sensitivity of the tuning-rule comparison in \Cref{fig:exp_tuning}, with \(\Sigma_0=\Sigma_1=I_p\).}
    \label{fig:iso-exp-tuning}
\end{figure}

\clearpage
\subsection{Response-model robustness}\label{subsec:response-robustness}

The adaptive experiments above draw the outcome from the affine random-effects law of \Cref{asm:predictive-model}, so the tuning rule is evaluated under its own working model.
We therefore repeat the overparameterised experiment with \(p=50\), \(n_0=40\), \(n_1=100\), the same target split, overlap grid, \(\ell_2\) base weights, penalty grid, seeds, and variance-component estimator, replacing the random-effects outcome model by fixed response surfaces.
Let \(s_1\geq\cdots\geq s_p\) and \(w_1,\ldots,w_p\) denote the eigensystem of \(\Sigma_0\), let \(u_S=5^{-1/2}\sum_{j=1}^5e_j\), and define
\[
    \beta_H=\frac{w_1}{\sqrt{s_1}},
    \qquad
    \beta_L=\frac{w_p}{\sqrt{s_p}},
    \qquad
    \beta_S=\frac{u_S}{(u_S^\top\Sigma_0u_S)^{1/2}},
\]
\[
    m_H(x)=x^\top\beta_H,
    \quad
    m_L(x)=x^\top\beta_L,
    \quad
    m_S(x)=x^\top\beta_S,
    \quad
    m_N(x)=\frac{1}{\sqrt2}
    \left\{
    x^\top\beta_S+
    \frac{(w_1^\top x)^2-s_1}{\sqrt2\,s_1}
    \right\},
\]
each normalised so that \(\operatorname{var}_{\PP_0}\{m(x)\}=r^2=1\), and generate \(y_{0i}=m(x_{0i})+\varepsilon_{0i}\) with \(\sigma_0^2=0.5\).
For each response surface \(m\), let \(\mu_m=\EE_{\PP_1}\{m(x)\}\), which is available in closed form because \(\Sigma_1=\Sigma_0\), and define the fixed-surface conditional risk
\[
    R_n^{\mathrm{fix}}(\lambda;m)
    =
    \{\gamma_\lambda^\top m(X_0)-\mu_m\}^2
    +
    \sigma_0^2\|\gamma_\lambda\|_2^2,
    \qquad
    \lambda_{\mathrm{fix}}^\star
    \in
    \argmin_{\lambda\in\Lambda}R_n^{\mathrm{fix}}(\lambda;m).
\]
\Cref{fig:response-robustness} reports actual counterfactual-mean prediction error \((\gamma_{\widehat\lambda}^\top y_0-\mu_m)^2\) for the feasible spectral rule and source GCV relative to this infeasible fixed-surface oracle, rather than the criterion \(\widehat R_n\), which no longer measures the estimand's error once the response surface is fixed.

\begin{figure}[!ht]
    \centering
    \includegraphics[width=0.75\linewidth]{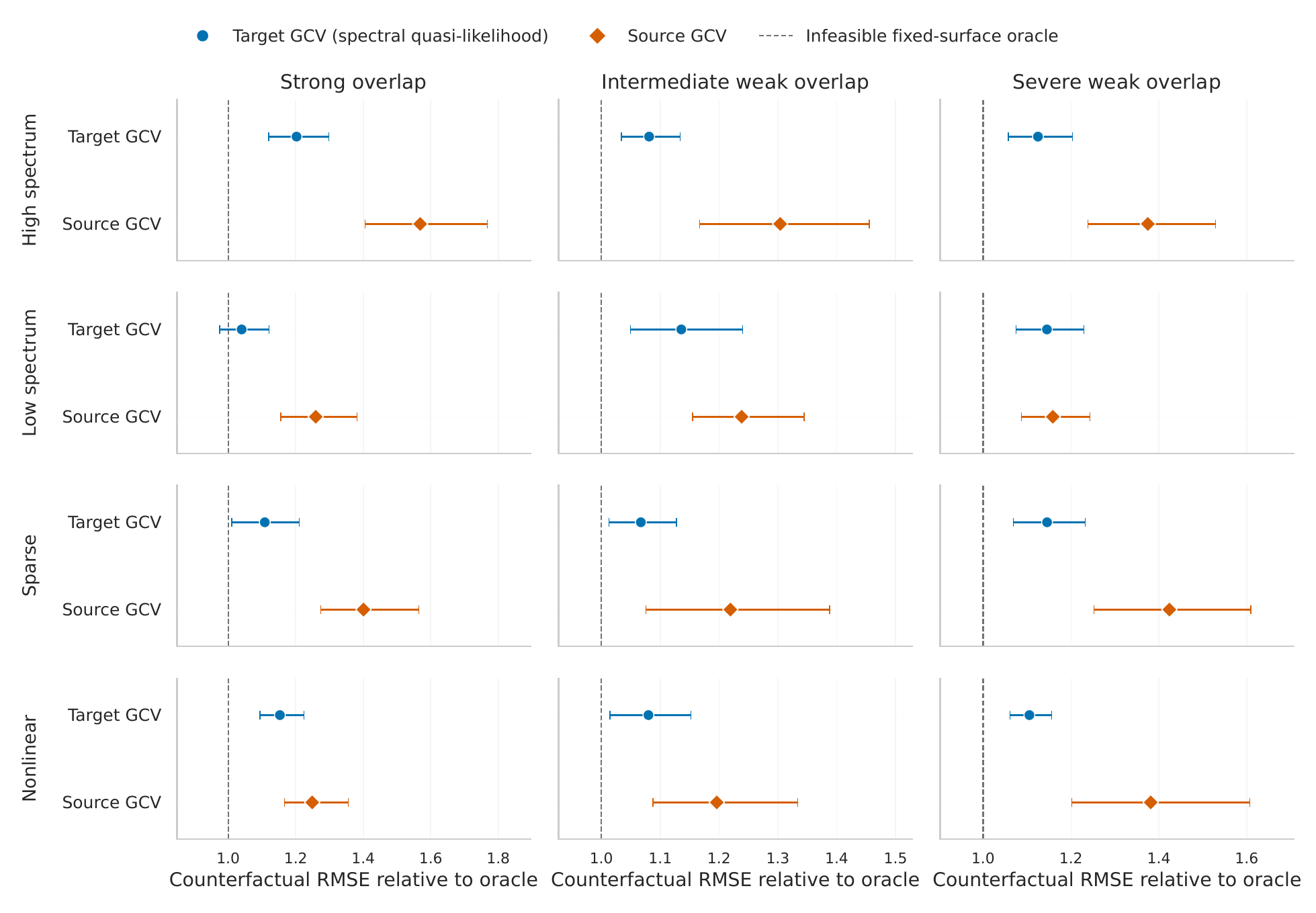}
    \caption{Response-model robustness in the overparameterised design \(p=50\), \(n_0=40\), and \(n_1=100\).
    Rows correspond to high-spectrum, low-spectrum, sparse, and nonlinear fixed response surfaces, and columns correspond to strong, intermediate, and severe weak overlap.
    Points report counterfactual-mean RMSE relative to the infeasible fixed-surface oracle over 120 common Monte Carlo replications; intervals are paired Monte Carlo bootstrap intervals.}
    \label{fig:response-robustness}
\end{figure}

Across all twelve configurations the spectral rule attains a counterfactual-mean RMSE between \(1.04\) and \(1.20\) times the infeasible fixed-surface oracle, compared with \(1.16\) to \(1.57\) for source GCV, and remains closer to the fixed-surface oracle throughout.

\clearpage
\section{Supplementary Real-Data Analyses}
\label{app:real-data}

\subsection{LaLonde analyses and diagnostics}
\label{app:lalonde-observational}

\noindent\textbf{Observational ATT.}
\Cref{tab:lalonde171-full} reports the observational estimates and weight diagnostics for the three base-weight constructions and their corresponding RCB estimators.

\begin{table}[!ht]
    \centering
    \caption{Observational estimates and diagnostics in the LaLonde-171 analysis.
    All monetary quantities are in 1978 US dollars.
    Intervals are percentile intervals from the paired full-procedure bootstrap.
    Uniform- and \(\ell_2\)-based procedures use 500 bootstrap replications, and entropy-based procedures use 200.
    For RCB, the uniform base uses the full target sample for tuning, while the \(\ell_2\) and regularised-entropy bases use the pilot/evaluation split.
    For possibly negative weights, ESS is interpreted as a weight-concentration stability index.
    The randomised NSW estimate \(1{,}794\) (SE \(=632\)) is an external reference only.
    }
    \label{tab:lalonde171-full}
    \scriptsize
    \setlength{\tabcolsep}{2.4pt}
    \resizebox{\linewidth}{!}{%
    \begin{tabular}{@{}l l c c c c c c@{}}
    \toprule
    Category
    & Estimator
    & \(\widehat{\mu}_0\)
    & \(\widehat{\tau}_{\mathrm{ATT}}\)
    & Bootstrap 95\% interval
    & \(\widehat{\lambda}\)
    & ESS
    & Max SMD \\
    \midrule
    \multirow{3}{*}{Base weights}
    & Uniform
    & 14,121 & \(-7{,}772\) & [\(-9{,}228\), \(-6{,}273\)]
    & -- & 727.0 & 1.47 \\
    & \(\ell_2\)
    & 5,437 & 912 & [\(-916\), \(2{,}574\)]
    & -- & 79.1 & 0.57 \\
    & Regularised entropy
    & 4,574 & 1,775 & [79, 4,076]
    & -- & 14.9 & 0.30 \\ \cmidrule(lr){2-8}
    \multirow{3}{*}{RCB}
    & Uniform base
    & 4,269 & 2,080 & [\(-1{,}422\), \(3{,}966\)]
    & 0.0631 & 31.3 & 0.14 \\
    & \(\ell_2\) base
    & 3,902 & 2,447 & [\(-999\), \(4{,}017\)]
    & 0.1995 & 35.8 & 0.29 \\
    & Entropy base
    & 4,539 & 1,810 & [\(-448\), \(4{,}164\)]
    & 17.7828 & 16.7 & 0.29 \\
    \bottomrule
    \end{tabular}%
    }
\end{table}

\noindent\textbf{Held-out outcome validation.}
Because untreated outcomes for the NSW treated participants are unavailable, we use the 260 randomised NSW controls as a pseudo-target with observed untreated earnings.
This provides an observed-outcome evaluation under the empirical NSW--PSID shift.
Each replication resamples 727 PSID source observations and 185 pseudo-target observations, with target outcomes withheld from fitting and tuning and used only for evaluation.
We use 500 paired replications for the uniform- and \(\ell_2\)-based RCB procedures and 200 for the entropy-based procedure.
The weighting and RCB RMSEs are \(9{,}604\) and \(1{,}570\) for the uniform base, \(1{,}116\) and \(1{,}524\) for the \(\ell_2\) base, and \(1{,}604\) and \(1{,}608\) for the regularised-entropy base, respectively.

\begin{table}[!t]
\centering
\caption{Performance under classifier-induced empirical discrepancy scaling over 120 paired replications.
RMSE ratios are relative to imbalance-CV double ridge using common Monte Carlo draws within each regime.
Intervals are paired Monte Carlo bootstrap intervals and do not quantify sampling uncertainty in the original LaLonde application.
The reported discrepancy levels are finite-sample classifier-based quantities, not population \(\chi^2(\PP_1\Vert\PP_0)\) divergences.}\label{tab:exp-chi}
\small
\setlength{\tabcolsep}{7pt}
\resizebox{\linewidth}{!}{%
\begin{tabular}{llrrrr}
\toprule
Setting &
Estimator
& RMSE
& MCSE
& RMSE ratio & (95\% interval) \\
\midrule
\multirow{9}{*}{Moderate shift} &
Double ridge (Outcome CV)
& 2,009 & 128 & 1.131 & [1.085, 1.176] \\
& Double ridge (Imbalance CV)
& 1,777 & 105 & 1.000 & [1.000, 1.000] \\
& Double ridge (Riesz CV)
& 2,164 & 157 & 1.218 & [1.125, 1.326] \\
& RCB (uniform base, full target)
& 1,772 & 106 & 0.997 & [0.950, 1.044] \\
& RCB (\(\ell_2\) base, pilot/evaluation)
& 1,805 & 102 & 1.016 & [0.956, 1.080] \\
& RCB (entropy base, pilot/evaluation)
& 2,935 & 168 & 1.652 & [1.459, 1.875] \\
& Uniform weighting
& 10,942 & 41 & 6.159 & [5.532, 6.971] \\
& \(\ell_2\) weighting
& 2,495 & 79 & 1.405 & [1.261, 1.579] \\
& Regularised entropy weighting
& 2,778 & 151 & 1.564 & [1.383, 1.776] \\
\cmidrule(lr){1-6}
\multirow{9}{*}{Strong shift} &
Double ridge (Outcome CV)
& 2,765 & 198 & 1.076 & [1.018, 1.131] \\
& Double ridge (Imbalance CV)
& 2,569 & 159 & 1.000 & [1.000, 1.000] \\
& Double ridge (Riesz CV)
& 2,857 & 202 & 1.112 & [1.037, 1.188] \\
& RCB (uniform base, full target)
& 2,423 & 161 & 0.943 & [0.875, 0.997] \\
& RCB (\(\ell_2\) base, pilot/evaluation)
& 2,350 & 154 & 0.915 & [0.853, 0.967] \\
& RCB (entropy base, pilot/evaluation)
& 3,643 & 218 & 1.418 & [1.266, 1.601] \\
& Uniform weighting
& 11,781 & 42 & 4.586 & [4.089, 5.216] \\
& \(\ell_2\) weighting
& 3,752 & 97 & 1.461 & [1.330, 1.612] \\
& Regularised entropy weighting
& 3,478 & 221 & 1.354 & [1.183, 1.558] \\
\bottomrule
\end{tabular}%
}
\end{table}

\par\smallskip
\noindent\textbf{Controlled overlap-sensitivity analysis.}
We progressively tilt the empirical LaLonde source design away from the observed target covariates using a five-fold cross-fitted ridge-logistic classifier fitted only to covariates and source-target labels.
After prior-odds correction and clipping, the source-sample odds define a normalised density-ratio proxy \(\widehat r_i\).
Let
\[
    p_i=\frac{\widehat r_i}{\sum_j\widehat r_j},
    \qquad
    q_i(\alpha)=\frac{\widehat r_i^\alpha}{\sum_j\widehat r_j^\alpha},
    \qquad
    \widehat\chi^2(\alpha)=\sum_i\frac{p_i^2}{q_i(\alpha)}-1.
\]
The proxy \(\widehat r_i\), and hence \(p_i\) and \(q_i(\alpha)\), are computed once from the observed covariates and held fixed throughout.
The value \(\alpha=0\) gives the observed empirical source distribution \(q_i(0)=1/n_0\), with baseline discrepancy \(\widehat\chi^2(0)\approx44.8\), and negative \(\alpha\) moves sampling mass away from target-like source units.
We choose \(\alpha_4\) and \(\alpha_{16}\) by a grid search over \(\alpha\in[-3,0]\) in steps of \(0.001\), taking the value whose \(\widehat\chi^2(\alpha)\) is closest to \(4\widehat\chi^2(0)\) or \(16\widehat\chi^2(0)\) subject to a sampling effective sample size \(1/\sum_iq_i(\alpha)^2\ge0.2n_0\).
This gives \(\alpha_4=-0.169\) and \(\alpha_{16}=-0.328\), with achieved multiples \(4.00\) and \(15.95\) and sampling effective sample sizes \(663\) and \(568\).
We refer to the resulting \(4\times\) and \(16\times\) regimes as the moderate and strong controlled shifts, respectively.
These are finite-sample classifier-based quantities, not estimates of the population divergence \(\chi^2(\PP_1\Vert\PP_0)\).

For each regime, the nine procedures reported in \Cref{tab:exp-chi} are refitted over 120 paired replications.
In each replication, \(n_0=727\) source observations are drawn with replacement from the PSID covariates with probabilities \(q_i(\alpha)\), and their outcomes are the fitted values of a LaLonde-calibrated ridge outcome model plus independent Gaussian noise with the residual standard deviation of that fit; the estimand is the mean of the same fitted surface over the 185 NSW treated covariates.
The target covariates, prediction target, outcome model, and resampling and target-split seeds are held fixed across regimes, so replications are paired both across procedures and across regimes.
\Cref{tab:exp-chi} reports the complete results.

\subsection{K562 Perturb-seq implementation}
\label{app:k562-implementation}

The raw single-cell screen is preprocessed with the \texttt{crispyx} Python package \citep{du2026crispyx}: expression is aggregated within \((\text{perturbation},\text{batch})\) groups to obtain pseudobulk observations, and a batch-corrected differential-expression analysis against non-targeting controls gives, for each perturbation and each of the 8{,}563 measured genes, a marginal \(\log_2\) fold change, a Wilcoxon \(z\)-score, and a Benjamini--Hochberg adjusted Wilcoxon \(p\)-value computed within that perturbation across all 8{,}563 genes.
The five perturbations are among the eleven with the maximal 48 pseudobulk observations, and the non-targeting controls also have 48.
All outcome and feature selection uses only these marginal differential-expression summaries and is completed once, before any resampling, base-weight construction, or RCB fit; it uses no RCB estimates, risk estimates, or selected penalties.
The 25-gene display panel is selected from the genes other than the five perturbation genes by ranking each gene by its largest absolute \(z\)-score across the five perturbations, breaking ties by the smallest adjusted \(p\)-value, and retaining the top 25; the same panel is used for all five perturbations.
The overlap-sensitivity analysis uses perturbation-specific outcomes instead: for each perturbation, the 25 genes other than the perturbation gene ranked first by adjusted \(p<0.05\) and then by absolute \(z\)-score, giving 125 perturbation-outcome pairs on 111 distinct genes that include all 25 display-panel genes.
Both outcome sets are high-signal by construction, and the results in both analyses are conditional on them.
The \(p=500\) features are chosen from the remaining genes after excluding the five perturbation genes, every outcome gene used in either analysis, and, for each perturbation, the 20 genes with the smallest adjusted \(p\)-values, with ties broken by absolute \(z\)-score and then absolute \(\log_2\) fold change.
The remaining genes are ranked by decreasing average adjusted \(p\)-value across the five perturbations, with ties broken by increasing average absolute \(z\)-score and then increasing average absolute \(\log_2\) fold change, and the first 500 are retained.
Each feature is standardised by its control-sample mean and standard deviation.

\begin{table}[!t]
    \centering
    \caption{Balance diagnostics in the K562 overlap-sensitivity analysis.
    Entries are medians over the 12{,}500 fits in each regime of the effective sample size \(\operatorname{ESS}(\gamma)=\|\gamma\|_2^{-2}\) and of the imbalance \(\|\bar x-X_0^\top\gamma\|_2\) of the 500 standardised covariates, where \(\bar x\) is the pilot-fold mean (Pilot) or the mean of all 48 perturbed profiles (Population).
    For the three double-ridge ABW rules, the base rows are the complete double-ridge weights and the RCB rows augment their balancing components.
    The population reference includes the evaluation fold used for RCB tuning.}
    \label{tab:k562-balance}
    \begin{tabular}{@{}llrrrrrr@{}}
        \toprule
        & & \multicolumn{3}{c}{Good overlap} & \multicolumn{3}{c}{Weak overlap} \\
        \cmidrule(lr){3-5}\cmidrule(l){6-8}
        Base & Weights & ESS & Pilot & Population & ESS & Pilot & Population \\
        \midrule
        Uniform & Base & 48.0 & 28.00 & 19.44 & 48.0 & 29.07 & 20.34 \\
        & RCB & 48.0 & 26.77 & 18.72 & 28.0 & 27.45 & 19.28 \\
        \(\ell_2\) & Base & 5.1 & 25.62 & 19.28 & 9.7 & 26.90 & 19.65 \\
        & RCB & 6.2 & 26.44 & 18.72 & 11.6 & 27.30 & 19.44 \\
        Entropy & Base & 13.1 & 26.69 & 19.02 & 12.2 & 27.36 & 19.51 \\
        & RCB & 13.8 & 26.91 & 18.82 & 13.1 & 27.57 & 19.42 \\
        ABW, outcome & Base & 4.4 & 25.61 & 19.44 & 8.2 & 26.89 & 19.74 \\
        & RCB & 6.7 & 26.53 & 18.76 & 10.6 & 27.33 & 19.46 \\
        ABW, imbalance & Base & 4.5 & 25.61 & 19.37 & 8.3 & 26.89 & 19.74 \\
        & RCB & 21.8 & 26.70 & 18.77 & 18.3 & 27.44 & 19.41 \\
        ABW, Riesz & Base & 4.1 & 25.59 & 19.52 & 8.2 & 26.89 & 19.74 \\
        & RCB & 6.0 & 26.44 & 18.72 & 10.9 & 27.32 & 19.46 \\
        \bottomrule
    \end{tabular}
\end{table}

The overlap-sensitivity analysis uses 100 Monte Carlo replications, with 24 target observations independently resampled and split equally into pilot and evaluation folds in each replication, so the Monte Carlo variation includes sensitivity to the realised target sample and partition.
Each pair-level RMSE in panels (a)--(b) is computed over the 100 replications against the all-data reference contrast, the difference between the perturbed and control means of the outcome over all 48 observations of each.
Because feature screening uses the observed Perturb-seq dataset, this experiment is an empirical finite-sample evaluation rather than a direct application of \Cref{thm:risk-estimation}.
RCB has lower pair-level RMSE than its base in \(89\%\) and \(72\%\) of the 125 pairs for \(\ell_2\) weighting under good and weak overlap, in \(95\%\)--\(96\%\) and \(82\%\)--\(93\%\) for the three double-ridge rules, in \(56\%\) and \(57\%\) for regularised-entropy weighting, and in \(48\%\) and \(83\%\) for uniform weighting.
For uniform-base RCB, the selected penalty is noticeably sensitive to the realised target sample and split: among finite selections, the pooled median \(\log_{10}\widehat\lambda\) is \(0.90\) (interquartile range \(0.75\)--\(1.20\)) under good overlap and \(1.05\) (\(0.75\)--\(1.50\)) under weak overlap, while the median within-pair interquartile width among finite selections across replications increases from \(0.25\) to \(0.45\).
The no-augmentation endpoint is selected in \(55.8\%\) and \(18.7\%\) of all selections, respectively.

\Cref{tab:k562-balance} reports the effective sample size and covariate imbalance of each weight vector.
Relative to uniform weighting, the base procedures reduce the imbalance against their own pilot fold by only \(5\%\)--\(9\%\), and this reduction carries over little to the perturbed population: at most \(2.2\%\) under good overlap and \(4.1\%\) under weak overlap, at the cost of reducing the effective sample size from 48 to between 4 and 13.
RCB reduces the population imbalance for every base, and with uniform base weights it does so while retaining all 48 effective controls under good overlap and 28 under weak overlap.

\begin{figure}[!t]
    \centering
    \includegraphics[width=\textwidth]{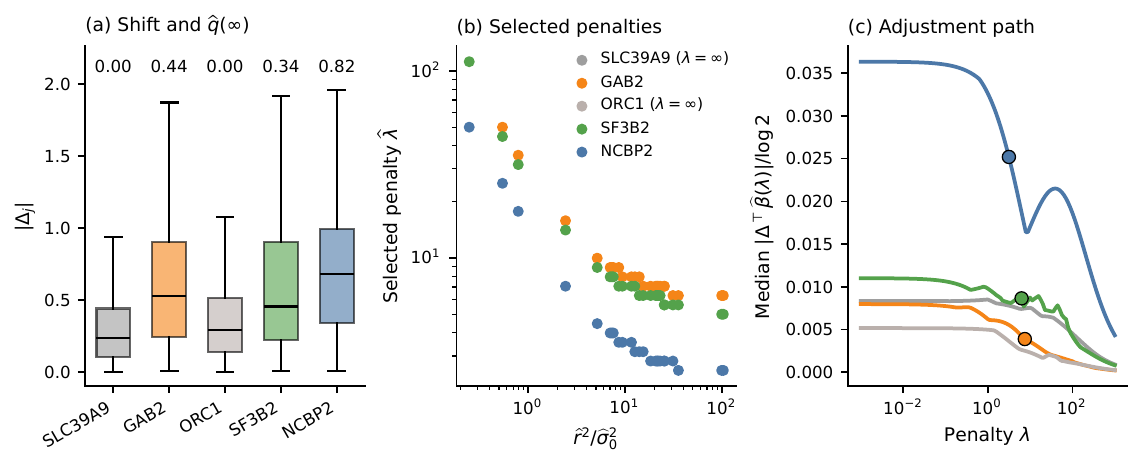}
    \caption{Full-data K562 adjustment diagnostics with uniform base weights.
    (a) Absolute components \(|\Delta_j|\) of the uniform-weight imbalance \(\Delta=\bar x_1-\bar x_0\) in \eqref{eq:Delta} over the 500 covariates, each standardised by its control mean and standard deviation, with \(\widehat q(\infty)\) above each box.
    (b) Selected penalty against the estimated outcome signal-to-noise ratio \(\widehat r^2/\widehat\sigma_0^2\) for each outcome gene; \textit{SLC39A9} and \textit{ORC1} select \(\lambda=\infty\) for every gene.
    (c) Median absolute adjustment \(|\Delta^\top\widehat\beta(\lambda)|/\log 2\) across the 25 outcome genes along the penalty path, in mean \(\log_2(1+\mathrm{count})\) units, with the median selected penalty marked.}
    \label{fig:k562-adjustment-diagnostics}
\end{figure}

The full-data profiles use uniform base weights, so the full perturbation sample evaluates the estimated conditional prediction risk, and the candidate penalties are the logarithmic grid \(\lambda\in[10^{-3},10^{3}]\) together with the exact no-augmentation endpoint \(\lambda=\infty\).
For display, panel (c) reports \(\widehat\tau_{PG}/\log 2\), and panel (d) reports its difference from the unadjusted difference between the perturbed and control means, both in mean \(\log_2(1+\mathrm{count})\) units.
On the display panel, the RCB contrasts are close to the unadjusted differences in mean \(\log_2(1+\mathrm{count})\): they have the same sign for all 125 perturbation-outcome pairs, coincide with them for \textit{SLC39A9} and \textit{ORC1}, for which the no-augmentation endpoint is selected for every outcome gene, and otherwise differ from them by less than \(0.2\).
Across the 25 outcome genes, the median effective sample size \(\operatorname{ESS}(\gamma_{\widehat\lambda})=\|\gamma_{\widehat\lambda}\|_2^{-2}\) of the selected weights is \(48\) for \textit{SLC39A9} and \textit{ORC1}, and \(22.2\), \(16.0\), and \(8.1\) for \textit{GAB2}, \textit{SF3B2}, and \textit{NCBP2}, respectively.
\Cref{fig:k562-adjustment-diagnostics} shows why the amount of adjustment differs across perturbations.
For uniform base weights, the endpoint risk estimate in \Cref{rem:no-augmentation-endpoint} contains the bias-corrected shift term \(\widehat q(\infty):=\left[\|\Delta\|_2^2/p-\otr(\widehat\Sigma_1)/n_1\right]_+\), which is zero for \textit{SLC39A9} and \textit{ORC1}, whose observed covariate shift is no larger than target sampling variability, and equals \(0.44\), \(0.34\), and \(0.82\) for \textit{GAB2}, \textit{SF3B2}, and \textit{NCBP2}, respectively.
A larger estimated shift increases the estimated bias of the unadjusted estimator relative to the variance cost of concentrating the weights, and \textit{NCBP2} selects the smallest penalty for every outcome gene.
At a common penalty, its adjustment is also three to six times larger than for the other perturbations, indicating that its covariate shift lies in directions that predict the outcome genes, most of which encode ribosomal proteins.

\end{document}